\documentclass{lmcs}
\pdfoutput=1

\usepackage{lastpage}
\lmcsdoi{22}{4}{23}
\lmcsheading{}{\pageref{LastPage}}{}{}%
{Mar.~18,~2025}{Nov.~26,~2026}{}


\usepackage[utf8]{inputenc}

\usepackage[T1]{fontenc}
\usepackage{relsize}

\usepackage{silence}
\usepackage{microtype}
\usepackage{parskip}
\usepackage[all]{nowidow}

\usepackage{tabularx}
\usepackage{booktabs}
\usepackage[shortlabels]{enumitem}
\usepackage{float}
\usepackage{xspace}

\usepackage{xparse}
\usepackage{listofitems}

\usepackage{amsmath}
\usepackage{mathtools}
\usepackage{amssymb}
\usepackage{tensor}
\usepackage{centernot}
\usepackage{cancel}

\usepackage[useregional,showisoZ=false]{datetime2}
\usepackage{bussproofs}

\usepackage{tikz}
\usetikzlibrary{shapes, arrows, positioning, decorations.pathmorphing, backgrounds, fit, calc}

\usepackage{xcolor}
\usepackage{color}

\usepackage{ifplatform}
\usepackage{ifthen}

\usepackage{hyperref}
\hypersetup{
  colorlinks=true,
  linkcolor=black,
  citecolor=black,
  urlcolor=gray,
  linktoc=all,
  breaklinks=true,
}

\newenvironment{emphasis}{\begin{quote}\em}{\end{quote}}

\newcommand{\defiff}{\text{if}}

\NewDocumentCommand{\reflist}{m m m}{%
  \readlist*\refitems{#3}%
  \ifnumequal{\refitemslen}{1}%
    {#1~\eqref{\refitems[1]}}%
    {#2~\foreachitem\p\in\refitems{%
      \ifnumequal{\pcnt}{\refitemslen}%
        {,~and~\eqref{\p}}%
        {\ifnumequal{\pcnt}{1}%
          {\eqref{\p}}%
          {,~\eqref{\p}}%
        }%
    }}%
}

\newcommand{\factRef}[1]{\reflist{Statement}{Statements}{#1}}
\newcommand{\claimRef}[1]{\reflist{Claim}{Claims}{#1}}
\newcommand{\equivalenceRef}[1]{\reflist{Equivalence}{Equivalences}{#1}}
\newcommand{\observationRef}[1]{\reflist{Statement}{Statements}{#1}}

\newcommand{\muML}{$\mu$ML}

\newcommand{\Gm}{\mathcal{G}}
\newcommand{\Gms}{\mathcal{G}_s}
\newcommand{\Gmt}{\mathcal{G}_t}
\newcommand{\Stree}{\mathbb{T}}

\newcommand{\isdef}{\mathrel{:=}}
\newcommand{\ol}[1]{\vec{#1}}

\newcommand{\eq}{\leftrightarrow}
\newcommand{\PDL}{\mathbf{PDL}} 
\newcommand{\Ax}[1]{(\mathrm{Ax}(#1))} 
\newcommand{\AxInd}{\mathrm{(Ind)}} 

\newcommand{\Test}{\mathit{Tests}}
\newcommand{\Prog}{\mathit{Progs}}
\newcommand{\TT}{\mathit{TP}}
\newcommand{\FL}{\mathit{FL}}

\newcommand{\Q}{\mathcal{Q}}
\newcommand{\qedge}{\edge_{\Q}}
\newcommand{\subst}[2]{#2\langle#1\rangle}
\newcommand{\substAt}[2]{#1 \mapsto #2}
\newcommand{\iitp}{\iota}
\newcommand{\itpfmas}{\mathcal{L}_{\Q}}
\newcommand{\Simple}{\mathsf{Spl}}
\newcommand{\nf}[1]{#1^{\textit{nf}}}

\newcommand{\witdis}{\mathsf{wd}}

\newcommand{\Bset}[3]{\ensuremath{B^{#1}_{#2,#3}}}
\newcommand{\PsetName}{\ensuremath{P}}
\newcommand{\Pset}[2]{\ensuremath{P^{#1}{(#2)}}}
\newcommand{\FsetName}{\ensuremath{F}}
\newcommand{\Fset}[2]{\ensuremath{F^{#1}{(#2)}}}

\newcommand{\DsetName}{\ensuremath{D}}
\newcommand{\Dset}[1]{\ensuremath{D_{#1}}}

\newcommand{\pdlforms}{\mathcal{F}}
\newcommand{\loaded}[1]{\underline{#1}}
\newcommand{\kmodel}{\mathcal{M}}
\renewcommand{\phi}{\varphi}
\newcommand{\emptylist}{\varepsilon}

\newcommand{\seqc}{\mathop{;}} 

\newcommand{\edge}{\lessdot} 
\newcommand{\edgeT}{<} 
\newcommand{\edgeRT}{\leq} 

\newcommand{\cEdge}{\lhd} 
\newcommand{\cEdgeT}{\lhd^{+}} 
\newcommand{\cEdgeRT}{\lhd^{*}} 

\newcommand{\tedge}{\lessdot_{\mathbb{T}}} 
\newcommand{\tedgeT}{<_{\mathbb{T}}} 
\newcommand{\tedgeRT}{\leq_{\mathbb{T}}} 

\newcommand{\simpler}[2]{#2 <_{c} #1}
\newcommand{\cEquiv}{\equiv_c}

\newcommand{\comp}{\hspace{0.2em}\rotatebox[origin=c]{90}{$\heartsuit$}\hspace{0.1em}}

\newcommand{\compan}[1]{c(#1)}

\newcommand{\lpr}{loaded-path repeat\xspace}
\newcommand{\lprs}{loaded-path repeats\xspace}

\newcommand{\lprCond}{T1\xspace}
\newcommand{\freeCond}{T2\xspace}
\newcommand{\lprCondSplit}{T1S\xspace}
\newcommand{\freeCondSplit}{T2S\xspace}

\makeatletter
\newcommand{\BNF}[2]{\begin{equation*}#1 \, \, \Coloneqq \, \, #2\BNFchecknextarg}
\newcommand{\BNFchecknextarg}{\@ifnextchar\bgroup{\BNFgobblenextarg}{\end{equation*}}}
\newcommand{\BNFgobblenextarg}[1]{\, \, | \, \, #1\@ifnextchar\bgroup{\BNFgobblenextarg}{\end{equation*}}}
\makeatother

\newcommand{\arr}[3]{
  \draw [->, thick, shorten >=2pt, shorten <=2pt] (#1) -- node[right, xshift=1pt]{\footnotesize { #2 }} (#3);
}

\newcommand{\+}{ {,} \ }

\newcommand{\rel}[1]{\xrightarrow{#1}}
\newcommand{\reach}[1]{R_{#1}}

\newcommand{\progs}{\mathcal{P}}

\newcommand{\tab}{\mathcal{T}}
\newcommand{\powerset}{\mathcal{P}}

\newcommand{\diam}{\mathop{\scalebox{1}{\ensuremath{\Diamond}}}}
\newcommand{\unfold}{\mathsf{unfold}}

\newcommand{\boxvec}{\overrightarrow{\Box\strut_{\phantom{x}}\kern-1ex}}

\newcommand{\voc}{\mathsf{voc}}

\newcommand{\cycs}[1]{L_{<#1}}

\newcommand{\splitCase}{\hspace{0.5em}\mid\hspace{0.5em}}

\newcommand{\twopartdef}[4]{
  \left\{
    \begin{array}{ll}
      #1 & #2 \\
      #3 & #4
    \end{array}
  \right.
}

\newread\fid
\newcommand{\readfile}[1]
{\bgroup
  \endlinechar=-1
  \openin\fid=#1
  \read\fid to\filetext
  \loop\ifx\empty\filetext\relax
    \read\fid to\filetext
  \repeat
  \closein\fid
  \global\let\filetext=\filetext
\egroup}

\newcommand{\whenlean}[2]{%
  \ifshellescape%
    \IfFileExists{./tools/finder}{#1}{#2}%
  \else%
    #2\fi%
}

\whenlean{}{\let\filetext\relax}

\whenlean{\immediate\write18{rm -f lean-cache.txt}}{}

\whenlean{}{\IfFileExists{./lean-cache.txt}{

}{}}

\newcommand{\lean}[1]{%
  \href{https://m4lvin.github.io/lean4-pdl/docs/find/?pattern=#1\#doc}{\texttt{\detokenize{#1}}}%
}

\usepackage{expl3}
\ExplSyntaxOn
\NewDocumentCommand{\getFromCache}{m}{%
  \tl_set:Nn \l_tmpa_tl {#1}
  \regex_replace_all:nnN { [._] } { } \l_tmpa_tl
  \tl_set:Nx \l_tmpb_tl { LeanFor \tl_use:N \l_tmpa_tl }
  \cs_if_exist:cT { \tl_use:N \l_tmpb_tl } {
    \use:c { \tl_use:N \l_tmpb_tl }
  }
}
\ExplSyntaxOff

\newcommand{\checklean}[1]{%
  \texttt{\href{https://m4lvin.github.io/lean4-pdl/docs/find/?pattern=#1\#doc}{\detokenize{#1}}}%
  \whenlean{%
    \immediate\write18{bash ./tools/checker #1 > line.out}%
    \readfile{line.out}%
    \ \filetext%
  }{%
    \getFromCache{#1}
  }%
}

\begin{document}

\title[Propositional Dynamic Logic has Craig Interpolation]{Propositional Dynamic Logic has Craig Interpolation: \newline a tableau-based proof}

\author[MB]{Manfred Borzechowski}[a]
\author[MG]{Malvin Gattinger\lmcsorcid{0000-0002-2498-5073}}[b]
\author[HHH]{Helle Hvid Hansen\lmcsorcid{0000-0001-7061-1219}}[c]
\author[RR]{Revantha Ramanayake\lmcsorcid{0000-0002-7940-9065}}[c]
\author[FTD]{Francisco Trucco Dalmas\lmcsorcid{0000-0003-2493-023X}}[c]
\author[YV]{Yde Venema}[b]

\address{EDV-Beratung Manfred Borzechowski, Berlin, Germany}
\email{mail@borzechowski.de}
\address{University of Amsterdam, The Netherlands}
\email{malvin@w4eg.eu, y.venema@uva.nl}
\address{University of Groningen, The Netherlands}
\email{h.h.hansen@rug.nl, d.r.s.ramanayake@rug.nl, f.c.trucco.dalmas@rug.nl}


\begin{abstract}
  We show that Propositional Dynamic Logic (PDL) has the Craig Interpolation Property.
  This question has been open for many years.
  Three proof attempts were published, but later criticized in the literature or retracted.
  Our proof is based on the main ideas from Borzechowski (1988, master thesis).
  We define a cyclic tableau system for PDL with a loading mechanism to recognize successful repeats.
  For this system, we show soundness and completeness via a game.

  To show interpolation, we modify Maehara's method to work for tableaux with repeats: we first define pre-interpolants at each node, and then use a quasi-tableau to define interpolants for clusters (strongly connected components).
  In different terms, our method solves the fixpoint equations that characterize the desired interpolants, and the method ensures that the solutions to these equations can be expressed within PDL\@.

  The proof is constructive and we show how to compute interpolants.
  We make available a Haskell implementation of the proof system that provides interpolants,
  and a formal verification of our proof in the interactive theorem prover Lean.
  Lastly, we discuss several questions for future work.
\end{abstract}

\maketitle

\section{Introduction}

\paragraph{Propositional Dynamic Logic}

Propositional Dynamic Logic (abbreviated: PDL) is the propositional version of a first-order modal logic that was designed by Pratt~\cite{Pratt1976:Semantical} as
a formalism to reason about the behaviour of programs.
Since its introduction by Fischer and Ladner~\cite{FL:PDL1979} PDL has been studied as a logic in its own right, and it has been recognized as an interesting and important modal logic for quite some time, see for instance Harel, Kozen and Tiuryn~\cite{hare:dyna00} or Troquard and Balbiani~\cite{troq:prop23} for surveys.

The basic idea underlying PDL is that its syntax features a \emph{structured family of modalities}, each of which is indexed by some regular expression that represents a \emph{program} --- it is in fact customary to refer to these
expressions themselves as ``programs''.
Then, where $\alpha$ is such a program and $\phi$ is a formula, the intended meaning of the boxed formula $[\alpha]\phi$ is that ``after every execution of the program $\alpha$, the formula $\phi$ holds''.

The collection of all programs is generated from a set $\progs_{0}$ of so-called \emph{atomic programs}.
What defines PDL is that we close the set $\progs_{0}$ under three \emph{program constructors} which correspond to the standard regular operations
of sequential composition, non-deterministic choice and unbounded iteration
(to be denoted by the symbols ${;}$, $\cup$ and $^*$, respectively).
In \emph{test-free} PDL these are the only program constructors, so that the
programs correspond to the regular expressions over the set $\progs_{0}$.
The language of \emph{full} PDL (which we will usually refer to as ``PDL'' without qualification) also admits so-called \emph{test programs};
that is, for every formula $\tau$, we have an expression ${\tau?}$ which amounts to the program ``proceed if $\tau$ is true, else fail''.

The standard semantics for PDL is given by (multi-)modal Kripke models, which feature an accessibility relation $R_{a}$ for each atomic program $a$.
This interpretation is extended to arbitrary programs by a straightforward induction on the complexity of programs
by interpreting ${;}$, $\cup$ and $^*$ on relations as composition, union and reflexive, transitive closure, respectively.

The most important properties of PDL were established within a few years after its introduction.
The decidability of its satisfiability problem was already proved by Fischer and Ladner~\cite[Theorem~3.1]{FL:PDL1979}, who also established an exponential-time lower bound~\cite[Theorem~4.4]{FL:PDL1979}.
The \textsc{ExpTime} upper bound was established by Pratt~\cite{prat:near80}.
A natural axiomatisation was proposed by Segerberg~\cite{sege:comp77} and proved
to be complete by Parikh~\cite[Theorem~1]{pari:comp78} and others --- see
Pratt~\cite{prat:dyna17} for an overview of the early history of PDL\@.
In contrast to these results, the status of the metalogical property of
\emph{interpolation} has been much less clear.

\paragraph{Craig interpolation and Beth definability}

A logic has the Craig Interpolation Property (CIP) if for any validity $\phi\to\psi$ there exists a formula $\theta$ in the common vocabulary of $\phi$ and $\psi$ such that $\phi \to \theta$ and $\theta \to \psi$ are valid.
Any such formula $\theta$ is then called an \emph{interpolant} of $\phi$ and $\psi$.
The property was named after Craig who established it for first-order logic~\cite[Lemma~1]{Craig1957:UsesHG}.
Informally we will often say that a logic ``has interpolation'' if it enjoys the Craig Interpolation Property.

Interpolation is considered to be an attractive property of a logic.
Applications in computer science are generally related to modularity: if one needs to reason about two different parts of a system (for instance, in the setting of knowledge representation or program specification), interpolation guarantees that this reasoning can be done in the common vocabulary~\cite{mcmi:appl05,rena:inte08,jung:sepa21}.

A related property is that of Beth Definability, which informally states that any notion that can be implicitly defined in the logic, in fact also has an explicit definition.
This property owes its name to Beth, who established it for classical first-order logic~\cite[Theorem~4.2]{beth:pao53}; it is studied in the area of description logics, where it is used to optimize reasoning~\cite{cate:beth13}.

Craig Interpolation and Beth Definability have been studied in logic, see for instance the work of Hoogland and that of van Benthem~\cite{hoog:defi01,vB2008:ManyFaces}. It is well known that, in the presence of a deduction theorem, the Craig Interpolation Property implies the Beth Definability Property.
For an overview of interpolation and definability properties in the setting of modal logics we refer to the monograph by Gabbay and Maksimova~\cite{gabb:inte05}.

To prove that a given logic has interpolation, proof methods from various areas have been developed, including algebra, model theory, proof theory and automata theory.
A widely applied proof-theoretic method originates with Maehara~\cite{Maehara1961}: here one obtains an interpolant for a derivable implication, by taking some derivation for the implication and defining interpolants for every node of the derivation tree by a leaf-to-root induction.

Interpolation is a rather \emph{fragile} property.
For instance, it is not preserved under axiomatic extensions of a logic or inherited by its fragments.
To elaborate on the latter: if a logic has interpolation, this does not necessarily mean that the interpolant of two formulas that both belong to some distinguished fragment of the language also belongs to this fragment.
To see a concrete example of this, observe that although full first-order logic has the Craig interpolation property~\cite[Lemma~1]{Craig1957:UsesHG}, the two-variable fragment $\mathrm{FO}^2$ does not: there exist sentences $\varphi, \psi \in \mathrm{FO}^2$ where $\varphi \to \psi$ is a valid sentence for which no $\mathrm{FO}^2$-interpolant exists, so that any first-order interpolant must use at least three variables~\cite[Corollary~4]{Comer1969}; for a more modern reference see for example the work of ten Cate and Comer~\cite{tenCateComer2024}.

Returning to the question whether PDL has interpolation, we must observe that the literature on this matter is rather puzzling.
As far as we know, affirmative answers have been claimed on three occasions.
The first of these claims was made in 1981 by Leivant~\cite[Theorem~5.3.2]{Leivant1981} and refuted in 1999 by Kracht~\cite{Kracht1999tools}, but for the wrong
reasons~\cite{Gat2014:notesAboutLeivant}.
There is, however, a different gap in Leivant's proof for which no bridge has been found --- see \autoref{apdx:history}.
The second claim was made seven years later by Borzechowski, one of this paper's coauthors, in his Diplomarbeit~\cite[Theorem~5]{MB1988} (comparable to an MSc thesis).
Written in German and never properly published, this work received little attention; the general opinion, if any, seems to have been that this work could be dismissed since it was not ``possible to verify the argument''~\cite[p.~493]{Kracht1999tools}.
Finally, the third claim, made by Kowalski~\cite[Theorem~2]{Kowalski2002} in 2002, was refuted and officially withdrawn in 2004~\cite{KowalskiRetraction2004}.
In \autoref{apdx:history} we review this intriguing history in more detail.

\paragraph{Propositional Dynamic Logic as a fixpoint logic}
One way to understand why the interpolation problem for PDL constitutes such a challenging problem is by taking the perspective of modal fixpoint logic.
Since ${}*$-free PDL is just a notational variation of basic multi-modal logic, many if not all of the challenges in the study of PDL involve the iteration operator ${}*$, and it is this operator that takes Propositional Dynamic Logic inside the field of modal fixpoint logics (and outside the area of compact logics such as first-order logic).
The key step to see this is to think of the formula $[\alpha^{*}]\phi$ as the greatest fixpoint of the formula $\phi \land [\alpha] x$.
That is, writing $\equiv$ for semantic equivalence, we think of $[\alpha^{*}]\phi$ as a solution to the ``equation'' $x \equiv \phi \land [\alpha] x$, in the sense that
\[
[\alpha^{*}]\phi \equiv \phi \land [\alpha] [\alpha^{*}]\phi,
\]
and it is the \emph{largest} such, in the sense that for any formula $\psi$ we have the following:
\[
\text{if } \psi \equiv \phi \land [\alpha]\psi
\text{ then } \psi \text{ implies } [\alpha^{*}]\phi.
\]

Based on this observation, one may define a truth-preserving embedding of PDL into the modal $\mu$-calculus (\muML).
This logic, which extends basic poly-modal logic with arbitrary least- and greatest fixpoint operators, is a formalism of considerable importance, both because of its applications in computer science and for its metalogical properties, see Bradfield and Stirling~\cite{brad:moda07} for an overview.
In fact, PDL corresponds to a natural fragment of the modal $\mu$-calculus, see Carreiro and Venema~\cite{Carreiro2014:PDLInsideMu} for the details.

\paragraph{Interpolation for fixpoint logics}

Over the past three decades various results concerning interpolation properties
for modal fixpoint logics have been obtained, both positive and negative.
A quite striking result was obtained by D'Agostino and Hollenberg, who proved, using automata-theoretic tools, that the modal $\mu$-calculus enjoys \emph{uniform} interpolation~\cite[Corollary~3.8]{AH2000:muCalcIP}, a very strong version of interpolation where the interpolant only depends on one of the two formulas and the vocabulary of the other.
These authors also showed that any extension of PDL with the uniform interpolation property must have the same expressive power as the full $\mu$-calculus.
From this it is immediate that the uniform interpolation property fails for PDL\@.
Another negative result worth mentioning concerns a fragment ML${}^{*}$ of PDL which features only a minimal trace of the iteration operator.
ML${}^{*}$ is obtained by adding the master modality $[*]$ to the language of basic poly-modal logic --- this modality can be expressed in PDL as $[(a_{1} \cup \cdots \cup a_{n})^{*}]$, if we restrict the language to a finite number of atomic programs $a_{1}, \ldots,a_{n}$.
It was proved that this fragment of PDL does not have Craig interpolation by Studer~\cite[Theorem~4]{stud:comm09}, using an argument going back to Maksimova~\cite{maks:temp91}.

Moving back to the positive side, in recent years Maehara's method has become applicable to fixpoint logics through the introduction of so-called \emph{cyclic} proof systems~\cite{brot:cycl05,BGP2012:GenCyclicProver}.
Cyclic proof and tableau systems have been introduced for various modal fixpoint logics; in particular, the focus games of Lange and Stirling can be seen as cyclic tableau systems, respectively for linear time logic and computation tree logic~\cite{lang:focu01}, and for PDL~\cite{lang:sati03}.
In a cyclic proof, open leaves are allowed, as long as from each of these there is a back edge to an identically labeled ancestor (which is called the leaf's companion), and each path from a companion to a repeat leaf satisfies some global soundness (or ``progress'') condition.

Afshari and Leigh established the Lyndon interpolation property (a slight strengthening of the CIP) for the modal $\mu$-calculus using a cyclic sequent calculus~\cite[Theorem~2]{afsh:lynd19}.
Similarly, Marti and Venema introduced a focus-style derivation system for the alternation-free $\mu$-calculus (a fragment of \muML\ extending PDL), and used this to give a Maehara-style interpolation proof for this
logic~\cite[Theorem~7.2]{mart:focu21}.
(Note that the property itself had been established earlier by D'Agostino~\cite[Theorem 3]{dago:mule18}, who based her proof on automata-theoretic work by Santocanale and Arnold~\cite{sant:ambi05}.)
The idea underlying these interpolation proofs is to think of a finite, cyclic derivation of an implication as a \emph{system of equations}, where each node of the derivation tree induces an equation, the shape of which is determined by the rule applied at the node --- in the style of Maehara's method.
In order to find the interpolant for the implication at the root of the derivation tree, one needs to \emph{solve these equations}, that is, find a certain set of formulas, indexed by the nodes of the derivation tree, that
satisfy these equations.

\paragraph{Interpolation for PDL}
We may now give an explanation why establishing the Craig Interpolation Property is significantly harder for PDL than for either the full $\mu$-calculus or its alternation-free fragment.
Given the developments in cyclic proof theory, it is relatively easy to come up with a well-behaved cyclic proof system (or tableau calculus) for PDL, such that for any cyclic derivation of an implication $\phi\to\psi$ one may define a system of equations as sketched above.
The problem is that, while the systems of equations one obtains can easily be solved in the alternation-free $\mu$-calculus, it seems to be impossible to set up a proof system such that the system of equations admits solutions
\emph{inside PDL}.
It is precisely here that the insights from Borzechowski's \emph{Diplomarbeit} play an important role.

\paragraph{Contribution}
Based on ideas from Borzechowski's \emph{Diplomarbeit}~\cite{MB1988}, we show that Propositional Dynamic Logic has the Craig Interpolation Property, and, as a corollary, the Beth Definability Property.

We first introduce a new, cyclic, tableau system for PDL, with finitary rules such that proofs are finite trees.
We use a certain \emph{loading} mechanism, which allows one to place a focus on certain formulas in the tableau, in order to formulate an appropriate soundness condition for companion-repeat paths.
While the rules for the boolean connectives and for the atomic modal programs are standard, a novelty in our approach is the treatment of complex programs.
With every formula of the form $[\alpha]\psi$, where $\alpha$ is non-atomic, we associate a finite collection $\unfold_{\square}(\alpha,\psi)$ of special sets of formulas such that $[\alpha]\psi$ is equivalent to the disjunction $\displaystyle{\bigvee} \big\{ \bigwedge \Gamma \mid \Gamma \in \unfold_{\square}(\alpha,\psi) \big\}$, and the \emph{unfolding rule} then simply replaces $[\alpha]\psi$ with the respective members of
$\unfold_{\square}(\alpha,\psi)$.
A similar rule takes care of formulas of the form $\neg[\alpha]\psi$.

Turning to the interpolation proof, we follow Maehara's proof-theoretical method for establishing interpolation, and develop a version of the tableau system that operates on \emph{split sequents}, that is, pairs of finite sets of formulas.
Given the cyclic nature of the tableau system, closed tableaux have a natural notion of a cluster, that is, a maximal strongly connected component of the graph that is obtained by adding the back edges from repeats to their companion to the parent-child relation of the tableau.
The idea of the proof is to define interpolants, by a leaf-to-root induction on the tableau structure, for the roots of these clusters.
The key insight of Borzechowski's approach is that, in order to define the interpolant for the root of a (non-trivial) cluster, we use an \emph{auxiliary structure} rather than the cluster itself to set up an appropriate system of equations.
This shift ensures that solutions for these equations can be found inside the language of PDL\@.
We will provide more intuitions on our approach in \autoref{subsec:ProofIdea}.

\paragraph{Outline}
This article is structured as follows. We introduce preliminaries in \autoref{sec:Preliminaries}, including the standard syntax and semantics of PDL, and notation and conventions for tableau systems.
In \autoref{sec:LocalTableaux}, we define the local part of our tableau system, and in \autoref{sec:PDLTableaux}, a full tableau system for PDL\@.
We then show soundness of the system in \autoref{sec:Soundness}, completeness in \autoref{sec:Completeness} and finally interpolation in \autoref{sec:Interpolation}.
We conclude in \autoref{sec:Conclusion} and provide an overview of notation and symbols in~\autoref{apdx:notation}.
As we also discuss in \autoref{sec:Conclusion}, together with other collaborators we formalized the proof given here in the interactive theorem prover Lean 4.
In the present document, labels such as ``\checklean{truthLemma}'' provide a link to the corresponding definition or statement in the formalization.
A check mark next to the label means that this lemma or theorem has been fully formalized and verified in Lean. We note that the Lean proof may sometimes proceed differently than the paper proofs, and we refer to the code documentation for details about these differences.

\section{Preliminaries}\label{sec:Preliminaries}

We will write $\mathbb{N}$ for the set of natural numbers (including $0$).
For finite sets, instead of writing $\{x_1, \dots, x_k\}$ we may write ${\{x_i\}}_{i \leq k}$. For countably infinite sets, we sometimes use $\omega$ in place of $k$.
To denote sequences (or lists) we will simply use juxtaposition; instead of $x_1 \cdots x_k$ we may sometimes write ${(x_i)}_{i\leq k}$.
We will use the notation~$\emptylist$ to represent the empty sequence.

Given a set $A$, we denote the identity relation on $A$ by $\mathsf{Id}_A \isdef \{(a, a) \mid a \in A \}$.
Let $A, B, C$ be sets. Let $R \subseteq A \times B$ and $S \subseteq B \times C$ be two binary relations.
The composition of $R$ and $S$, denoted $R \seqc S$, is defined as:
\begin{equation*}
R \seqc S \isdef \{ (a, c) \mid \exists b \in B \text{ such that } (a, b) \in R \text{ and } (b, c) \in S \}
\end{equation*}
Let $n \in \mathbb{N}$ and let $R \subseteq A \times A$.
The $n$-fold composition of $R$, denoted $R^n$, is defined as $R^0 \isdef \mathsf{Id}_A$ and $R^n \isdef R^{n-1} \seqc R$ for $n \geq 1$.
The reflexive-transitive closure of $R$, denoted $R^*$, is defined as $R^* \isdef \bigcup \{R^n \mid n \in \mathbb{N}\}$.
Given a binary relation $R$ over a set $A$, we define an \emph{$R$-path} of length $n$ as a finite sequence of elements $x_0 x_1 \ldots x_n$ such that $(x_i, x_{i+1}) \in R $ for all $i \in \{0, 1, \ldots, n-1\}$. Given $a, b \in A$ we will say that the path is \emph{from $a$ to $b$} whenever $x_0 = a$ and $x_n = b$.

\subsection{Syntax of PDL}

Fix $\pdlforms_0$ to be a countably infinite set of proposition letters and $\progs_0$ to be a countably infinite set of atomic programs.
\begin{defi}\label{d:pdl-syntax}
  The set $\pdlforms$ of PDL formulas and the set $\progs$ of PDL programs are defined by mutual recursion as follows:
  \BNF{\pdlforms \ni \phi, \psi, \tau}{\bot \mid p \mid \lnot \phi \mid \phi \land \psi \mid [\alpha] \phi}
  \BNF{\progs \ni \alpha, \beta}{a \mid \tau? \mid \alpha \cup \beta \mid \alpha ; \beta \mid \alpha^\ast}
  where $p \in \pdlforms_0$ and $a \in \progs_0$.
  When clear from the context, we will talk of formulas and programs rather than PDL formulas and PDL programs.
  Other connectives $\lor$, $\to$ and $\eq$ are defined as usual in terms of $\neg$ and $\land$.
\end{defi}

PDL is a logic of programs where formulas of the form $[\alpha]\phi$ are to be read as: $\phi$ is true after all possible executions of the program $\alpha$.
The program $\alpha$ in this formula can either be an \emph{atomic} program $a \in \progs_0$ or a complex program.
Next, we describe all possible complex programs.
A program of the form~$\tau?$, called a \emph{test}, represents a program that only succeeds when~$\tau$ is true.
A program of the form~$\alpha \cup \beta$, called a \emph{non-deterministic choice}, represents a program that executes~$\alpha$ or~$\beta$ non-deterministically.
A program of the form~$\alpha; \beta$, called a \emph{sequential composition}, represents a program that sequentially executes~$\alpha$ followed by~$\beta$.
A program of the form~$\alpha^*$, called an \emph{iteration}, represents a program that executes~$\alpha$ a non-deterministically determined finite number of times.

For tests~$\tau?$, we could have used~$\phi$ or~$\psi$ for the formula~$\tau$, but we prefer to use the symbol~$\tau$ to stress that it is the formula of a test.

The logic $\PDL$ is usually defined via the following Hilbert-style axiomatisation.
We identify the logic with the set of its validities and will provide semantics later in \autoref{subsec:semantics}.
\begin{defi}\label{d:pdl-logic}
The logic $\PDL$ is the least set of formulas that contains all propositional tautologies,
all instances of the following axiom schemas,
\[\begin{array}{ll}
  (K)  &  [\alpha](\phi \to \psi) \to ([\alpha]\phi \to [\alpha]\psi) \\
  \Ax{?}  &  [\tau?]\phi \eq \neg \tau \lor \phi\\
  \Ax{\cup} & [\alpha \cup \beta]\phi \eq [\alpha]\phi \land [\beta]\phi\\
  \Ax{;} & [\alpha; \beta]\phi \eq [\alpha][\beta]\phi\\
  \Ax{^*} & [\alpha^*]\phi \eq \phi \land [\alpha][\alpha^*]\phi\\
  \AxInd & \phi \land [\alpha^*](\phi \to [\alpha]\phi) \to [\alpha^*]\phi
\end{array}
\]
and is closed under modus ponens and necessitation (that is, $\phi \in \PDL \text{ implies } [\alpha]\phi \in \PDL$).
\end{defi}

Next, we give a series of auxiliary definitions that will be used throughout the document.

As usual, the \emph{subformulas} of a formula~$\phi$ (or a program~$\alpha$) refer to the set of formulas occurring in~$\phi$ (respectively in~$\alpha$).
Similarly, the \emph{subprograms} of a formula $\phi$ (or a program~$\alpha$) refer to the set of programs occurring in~$\phi$ (respectively in~$\alpha$).

Due to the mutual recursion between formulas and programs in \autoref{d:pdl-syntax},
tests and programs can be nested inside test formulas.
For instance, $p?$ is a subprogram of $[p?]q?$ and $p$ is a subformula of $[p?]q?$.
However, when doing an inductive proof based on \autoref{d:pdl-syntax}, the inductive step often only considers the \emph{shallow} program structure for which the base cases are atomic programs and shallow tests.
For a program $\alpha$, the shallow tests and shallow subprograms of $\alpha$ are the tests and programs that occur in $\alpha$ that are not nested inside other tests.
Formally, we define shallow tests as follows:

\begin{defi}[\lean{testsOfProgram}]\label{d:TestSet}
We define the set $\Test(\alpha)$ of \emph{shallow tests of $\alpha$} as:
\[\begin{array}{lll}
\Test(a) & \isdef & \emptyset
\\ \Test(\tau?) &\isdef& \{ \tau \}
\\ \Test(\alpha\cup\beta) &\isdef& \Test(\alpha) \cup \Test(\beta)
\\ \Test(\alpha;\beta) &\isdef& \Test(\alpha) \cup \Test(\beta)
\\ \Test(\alpha^{\ast}) &\isdef& \Test(\alpha)
\end{array}\]
where $a \in \progs_0$ and $\alpha, \beta \in \progs$ and $\tau \in \pdlforms$.
\end{defi}

We analogously define shallow subprograms as follows:

\begin{defi}[\lean{subprograms}]\label{d:ProgSet}
We define the set $\Prog(\alpha)$ of shallow subprograms of $\alpha$ as:
\[\begin{array}{lll}
\Prog(a) & \isdef & \{ a \}
\\ \Prog(\tau?) &\isdef& \{ \tau? \}
\\ \Prog(\alpha\cup\beta) &\isdef& \{ \alpha\cup\beta \} \cup \Prog(\alpha) \cup \Prog(\beta)
\\ \Prog(\alpha;\beta) &\isdef& \{ \alpha;\beta \} \cup \Prog(\alpha) \cup \Prog(\beta)
\\ \Prog(\alpha^{\ast}) &\isdef& \{ \alpha^{\ast} \} \cup \Prog(\alpha)
\end{array}\]
where $a \in \progs_0$ and $\alpha, \beta \in \progs$ and $\tau \in \pdlforms$.
\end{defi}

\begin{exa}
We illustrate how $\Test(\alpha)$ does not include the tests that appear inside tests,
and how $\Prog(\alpha)$ does not include programs that occur inside tests.
We have:
\begin{enumerate}[(a)]
\item $\Test([q?]p?;a) = \{[q?]p\}$. Note that $q \notin \Test([q?]p?;a)$.
\item $\Prog([q?]p?;a) = \{[q?]p?;a, [q?]p?, a\}$. Note that $q? \notin \Prog([q?]p?;a)$.
\end{enumerate}
\end{exa}

Given a list of programs $\alpha_1 \cdots \alpha_n$ and a formula $\phi$ we write $\Box(\alpha_1 \cdots \alpha_n, \phi)$ to mean the formula $[\alpha_1]\dots [\alpha_n]\phi$.
Formally:

\begin{defi}[\lean{Formula.boxes}]\label{d:boxes}
Given a list $\ol{\alpha}$ of programs and a formula $\phi$, we define the
formula $\Box(\ol{\alpha},\phi)$ as follows:
\[\begin{array}{lll}
   \Box(\emptylist,\phi) &\isdef& \phi
\\ \Box(\alpha\ol{\beta},\phi) & \isdef & [\alpha]\Box(\ol{\beta},\phi)
\end{array}\]
\end{defi}

We have the usual definition of uniform substitution of proposition letters for formulas.

\begin{defi}[\lean{Substitution}]\label{d:subst}
A \emph{substitution} is a function mapping proposition letters to formulas.
Any such map $\sigma \colon \pdlforms_0 \to \pdlforms$ can be extended to a function $\langle\sigma\rangle \colon \pdlforms \cup \progs \to \pdlforms \cup \progs$
mapping formulas to formulas, and programs to programs, as follows.
\[
\begin{array}{lcl}
   \subst{\sigma}{\bot} & \isdef  & \bot
\\ \subst{\sigma}{p}    & \isdef  & \sigma(p)
\\ \subst{\sigma}{(\lnot\phi)} & \isdef  & \lnot \subst{\sigma}{\phi}
\\ \subst{\sigma}{(\phi \land \psi)} & \isdef &
    \subst{\sigma}{\phi} \land \subst{\sigma}{\psi}
\\ \subst{\sigma}{([\alpha]\phi)} & \isdef &
   [\subst{\sigma}{\alpha}]\subst{\sigma}{\phi}
\end{array}
\hspace{2em}
\begin{array}{lcl}
\subst{\sigma}{a} & \isdef  & a
\\ \subst{\sigma}{(\alpha ; \beta) }& \isdef&
   \subst{\sigma}{\alpha} ; \subst{\sigma}{\beta}
\\ \subst{\sigma}{(\alpha \cup \beta) } & \isdef &
\subst{\sigma}{\alpha} \cup \subst{\sigma}{\beta}
\\ \subst{\sigma}{\alpha^\ast} & \isdef  &
   {(\subst{\sigma}{\alpha})}^{\ast}
\\ \subst{\sigma}{(\tau?)} & \isdef& (\subst{\sigma}{\tau})?
\end{array}
\]
The substitution $\substAt{p}{\theta}$ denotes the function that maps $p$ to $\theta$ and every other proposition letter to itself.
In the case that $p$ occurs in $\phi$ we generally write $\phi(q)$ to denote the formula $\subst{\substAt{p}{q}}{\phi}$.
\end{defi}

Note that in \autoref{d:subst} we use postfix notation and write the substitution after the formula to which it is applied.
For example, let $\sigma$ be a substitution such that we have $\sigma(p)=r$ and $\sigma(r)=p$.
Then $\subst{\sigma}{([[r?]p?]r)} = [[p?]r?]p$.

\begin{defi}
  Given a formula $\phi$, the \emph{single negation $\sim\!\phi$ of $\phi$} is the formula
  \[
  \sim\!\phi \isdef
  \twopartdef
  {\psi}
  {\text{if } \phi \text{ is of the form } \lnot \psi,}
  {\lnot \phi}
  {\text{otherwise.}}
  \]
  A set $X$ of formulas is \emph{closed under single negations} if for all formulas $\phi$:
  \[
  \phi \in X \text{ implies } \sim\!\phi \in X.
  \]
\end{defi}

The following definition of the Fischer-Ladner closure is adapted from~\cite[Definition~4.79]{BRV}.
We will use it to ensure termination of our proof system, e.g.\ for~\autoref{l:pdlTabFinite}.

\begin{defi}[\lean{FL}, \lean{FLb}, \lean{FLL}]\label{d:fischerLadner}
  A set $X$ of formulas is called \emph{Fischer-Ladner closed} if it is closed under single negations and subformulas and we have that:
  \[
    \begin{array}{lcl}
      [\alpha\cup\beta] \phi \in X  & \text{ implies } & [\alpha]\phi \in X \text{ and } [\beta]\phi \in X \\[0.3em]
      [\alpha;\beta] \phi \in X  & \text{ implies } & [\alpha][\beta]\phi \in X \\[0.3em]
      [\alpha^\ast] \phi \in X  & \text{ implies } &  [\alpha][\alpha^\ast]\phi \in X
    \end{array}
  \]
For any set $X$ of formulas, we define its \emph{Fischer-Ladner closure} denoted $\FL(X)$ as the smallest set containing $X$ that is Fischer-Ladner closed.
\end{defi}

The Fischer-Ladner closure of a finite set is finite.
For proofs, see~\cite[Theorem~3.2]{FL:PDL1979} and~\cite[Section 6.1]{hare:dyna00},
and note that \lean{FLL} in Lean returns a \texttt{List} that must be finite.

We end this subsection with the definition of basic formulas, a notion we will use later to restrict when specific tableau rules may be applied.

\begin{defi}[\lean{Formula.basic}]\label{d:isBasic}
A set of formulas is \emph{basic} if it only contains formulas of the
forms $\bot$, $\lnot\bot$, $p$, $\lnot p$, $[a]\phi$ and $\lnot [a]\phi$ where $p \in \pdlforms_0$ and $a \in \progs_0$.
\end{defi}

\subsection{Semantics of PDL}\label{subsec:semantics}

PDL formulas of the form $[\alpha]\phi$ are to be read as stating that \emph{all possible executions of the program $\alpha$ from the current state end in a state where $\phi$ holds}.
Intuitively, these states represent states in some computing environment, the programs represent computations that transform one state into another, and atomic programs represent atomic computations like primitive instructions in a programming language.
It is important to mention that these programs can behave non-deterministically, and for this reason we talk of \emph{all} possible executions when interpreting the meaning of $[\alpha]\phi$.

More abstractly, we view these states as states in a Kripke model, where the edges between these states are labeled by atomic programs and where in each state a set of proposition letters holds.

\begin{defi}[\lean{KripkeModel}]\label{d:KripkeModel}
  A \emph{Kripke model} $\kmodel = (W, \{\reach{a} \mid a \in \progs_0\}, V)$ consists of a set $W$ of states, an accessibility relation $\reach{a} \subseteq W^2$ for every atomic program $a \in \progs_0$, and a valuation $V \colon \pdlforms_0 \to \powerset(W)$ that assigns to each proposition letter a subset of $W$.
  A~pointed model is a pair $(\kmodel, v)$ where $\kmodel$ is a Kripke model and $v$ is a state in $\kmodel$.
\end{defi}

Next, we give the usual Kripke semantics that matches the previously discussed intuitions about the meaning of PDL formulas and programs.

\begin{defi}[\lean{evaluate}, \lean{relate}]\label{d:evaluate}
  Let $\kmodel= (W, \{\reach{a} \mid a \in \progs_0\}, V)$ be a Kripke model.
  The \emph{truth} of a formula $\phi$ in $\kmodel$ at a state $v$ (denoted as $\kmodel, v \Vdash \phi$)
  and the interpretation of complex programs are defined by a mutual induction as follows:
\[
\begin{array}{lll}
\kmodel, v \Vdash p            & \text{ iff }& v \in V(p) \\
\kmodel, v \Vdash \lnot \phi      & \text{ iff }& \kmodel, v \not\Vdash \phi \\
\kmodel, v \Vdash \phi \land \psi    & \text{ iff }&  \kmodel, v \Vdash \phi \text{ and } \kmodel, v \Vdash \psi \\
\kmodel, v \Vdash [\alpha] \phi        & \text{ iff }& \kmodel, w \Vdash \phi \text{ for all $w$ with } (v,w) \in \reach{\alpha} \\[1em]
\end{array}
\]
  \vspace{-1em}
\[
\begin{array}{lll}
\reach{\phi?}       & \isdef & \{ (v,v) \mid v \Vdash \phi \}  \hfill (\text{test}) \\
\reach{\alpha \cup \beta} & \isdef & \reach{\alpha} \cup \reach{\beta}  \hfill (\text{choice}) \\
\reach{\alpha ; \beta}    & \isdef & R_\alpha \seqc R_\beta = \{(w,v) \mid \exists u: (w,u) \in R_\alpha \land (u,v) \in R_{\beta}\}  \;\; \hfill (\text{composition})\\
\reach{\alpha^\ast}    & \isdef & \reach{\alpha}^\ast = \bigcup_{n=0}^\omega R_\alpha^n\hfill (\text{reflexive transitive closure})
\end{array}
\]
  Note that $R_a$ for atomic programs is given by $\kmodel$.
  For $(v, w) \in \reach{\alpha}$ we also write $v \rel{\alpha} w$.
  Given a set $X$ of formulas, we write $\kmodel, v \Vdash X$ if $\kmodel, v \Vdash \phi$ for all $\phi \in X$.
\end{defi}

We extend the interpretation of single programs to lists of programs.

\begin{defi}[\lean{relateSeq}]\label{d:relateSeq}
  For any list of programs $\ol{\delta}$ and for any model $\kmodel$ with set of states $W$, we define $R_{\ol{\delta}}$ inductively as follows:
\[
\begin{array}{lll}
R_\emptylist & \isdef & \mathsf{Id}_W = \{ (v,v) \mid v \in W \} \\
R_{\alpha\ol{\delta}} & \isdef &R_\alpha \seqc R_{\ol{\delta}}
    \end{array}
  \]
  where $\alpha$ is a program.
\end{defi}

Notice that for any Kripke model $(W, \{R_a \mid  a \in \progs_0\}, V)$ and for any list $\alpha_1\cdots\alpha_n$ of programs it holds that $R_{\alpha_1\cdots\alpha_n} = R_{\alpha_1; \dots; \alpha_n}$.

Our consequence relation is defined as follows.

\begin{defi}\label{d:validEquiv}
Let $\phi$ and $\psi$ be formulas and let $X$ be a set of formulas.
We say that $\phi$ is a \emph{(local) semantic consequence} of $X$ (written $X \models \phi$), if for all pointed models $(\kmodel, v)$, if $\kmodel, v \Vdash X$ then $\kmodel, v \Vdash \phi$.
  We say $\phi$ is \emph{semantically valid} (written $\models \phi$), if for all pointed models $(\kmodel, v)$, we have that $\kmodel, v \Vdash \phi$.
  Finally, we say that $\phi$ and $\psi$ are \emph{semantically equivalent} (written $\phi \equiv \psi$), if for all pointed models $(\kmodel, v)$ we have that $\kmodel, v \Vdash \phi$ if and only if $\kmodel, v \Vdash \psi$.
 \end{defi}

The Hilbert style axiomatisation of $\PDL$ from \autoref{d:pdl-logic} is known to be sound and complete with respect to the above Kripke semantics:  $\phi \in \PDL$ iff $\models \phi$, for all formulas $\phi$.
For a proof, see~\cite{pari:comp78} for test-free PDL, and e.g.~\cite[Sec.~10.2]{Kracht1999tools} for full PDL\@.

The next theorem, called the Local Deduction Theorem, is used to show that the Beth definability property follows from the interpolation property in the proof of \autoref{cor:beth}.

\begin{thm}[\checklean{deduction}]\label{t:deduction}
  Let $\phi$ and $\psi$ be formulas and let $X$ be a set of formulas.
  Then $X \cup \{ \psi \} \models \phi$ if and only if $X \models \psi \to \phi$.
\end{thm}
\begin{proof}
  This follows directly from~\autoref{d:evaluate} and \autoref{d:validEquiv}.
\end{proof}

The next lemma states that given a substitution $\sigma$ (see \autoref{d:subst}), a Kripke model $\kmodel$, and a formula $\phi$, the formula $\subst{\sigma}{\phi}$ is true in state $w$ if and only if the formula $\phi$ is true in the same state $w$ in a model equal to $\kmodel$ except that the valuation is updated with the substitution as given by $V^\sigma$ below.

\begin{lem}[\checklean{substitutionLemma}]\label{l:subst}
Let $\kmodel = (W,R,V)$ be a Kripke model and $\sigma$ a~substitution.
For every state $w$ in $\kmodel$ and every formula $\phi$, we have
\[
\kmodel, w \Vdash \subst{\sigma}{\phi} \text{ iff }
\kmodel^{\sigma}, w \Vdash \phi,
\]
where $\kmodel^{\sigma}$ is the model $(W,R,V^{\sigma})$ with $V^{\sigma}$
the valuation given by
\[
V^{\sigma}(p) \isdef \{ v \in W \mid \kmodel, v \Vdash \sigma(p) \}.
\]
\end{lem}
\begin{proof}
  We also claim that for any two states $v$ and $w$ in $\kmodel$ and for any program $\alpha$, we have
  \[
    (v,w) \in R_{\subst{\sigma}{\alpha}}
    \text{ iff }
    (v,w) \in R^{\kmodel^{\sigma}}_\alpha.
  \]
  where $R^{\kmodel^{\sigma}}_\alpha$ is the accessibility relation for $\alpha$ in the model $\kmodel^{\sigma}$.
  This claim and the lemma can then be shown by a parallel mutual induction on the formula structure of $\phi$ and the program structure of $\alpha$.
\end{proof}

We finish this section with one more useful technical lemma about the semantics of the iteration operator.

\begin{lem}[\checklean{stepToStar}]\label{l:stepToStar}
For all formulas $\phi$ and $\psi$, and all programs $\alpha$,
if $\phi \models \psi \land [\alpha]\phi $ then $\phi \models [\alpha^{\ast}]\psi$.
\end{lem}
\begin{proof}
  This follows directly from the semantics in \autoref{d:evaluate}.
\end{proof}

\subsection{Tableau Calculi}

Trees are the fundamental structures underlying tableaux. We fix notation and recall concepts.

A \emph{rooted tree} is a triple $(V, \edge, r)$ where $V$ is a set of nodes, $\edge$ is a binary relation on $V$ and $r \in V$, called the \emph{root} of the tree, is such that for every node $s \in V$ there is a unique $\edge$-path from $r$ to $s$.
If $s \edge t$, then $s$ is called the \emph{parent} of $t$, and $t$ a \emph{child} of $s$.
If a node has no children then it is called a \emph{leaf}, otherwise it is called an \emph{interior node}.
We denote the transitive closure of $\edge$ by $\edgeT$, and the reflexive and transitive closure by $\edgeRT$.
If $s \edgeT t$ we say that $s$ is an \emph{ancestor} of $t$ and that $t$ is a \emph{descendant} of $s$.

In this document, we visualize trees as growing downwards from their root, so e.g., an ancestor of a node $s$ is depicted above $s$.

A \emph{sequent} is a finite set of formulas.
We will use $\Gamma, \Sigma, \Delta$ to denote sequents.
Instead of $\Gamma \cup \Delta$ we will write $\Gamma, \Delta$, and instead of $\Gamma \cup \{\phi\}$ we will write $\Gamma, \phi$.
Also, instead of $\Gamma \setminus \{\phi\}$ we will write $\Gamma \setminus \phi$.

A tableau calculus is a finite set of \emph{schematic inference rules}.
Schematic inference rules specify how to transform a sequent into a list of new sequents based on the form of the formulas contained within the original sequent.

In \autoref{f:ruleExamples}, we present some examples of inference rules for tableau calculi.
The sequent above the horizontal line in a rule is called the \emph{parent} of the rule, while the sequents below the line that are separated by the symbol $\splitCase$ are called \emph{children} of the rule.
Sometimes, instead of separating the rule children with $\splitCase$, we might write a set below the horizontal line, as exemplified by Rule $(P)$.
In this case, each element in this set is a rule child.
A rule is said to be \emph{branching} if it has more than one rule child.
For instance, the rule $(\lnot \land)$ is branching whereas $(\lnot)$ is not.

\begin{figure}[H]
\begin{center}
  \AxiomC{$\Delta , \lnot \lnot \phi$}
  \LeftLabel{$(\lnot)$}
  \UnaryInfC{$\Delta,\phi$}
  \DisplayProof
  \hspace{1.5em}
  \AxiomC{$\Delta , \lnot (\phi \land \psi)$}
  \LeftLabel{$(\lnot \land)$}
  \UnaryInfC{$\Delta , \lnot \phi \splitCase \Delta , \lnot \psi$}
  \DisplayProof
  \hspace{1.5em}
  \AxiomC{$\Delta , [\alpha]\phi$}
  \LeftLabel{$(P)$}
  \UnaryInfC{$\big\{\Delta, \Gamma
      \mid  \Gamma \in P(\alpha, \phi) \big\}$}
  \DisplayProof
\end{center}
\caption{Examples of tableau rules.}\label{f:ruleExamples}
\end{figure}

The formula that appears in the rule parent is called the \emph{principal formula} of the rule.
For instance, $\lnot \lnot \phi$ is the principal formula of rule $(\lnot)$.
The sequent $\Delta$ that appears in the rules of \autoref{f:ruleExamples} is called the \emph{context} of the rule.
More generally, the context is the sequent that appears on the rule parent next to the principal formula.

A \emph{rule instance} is an instantiation of the parent and children of the rule with specific sequents that match the pattern of the rule.
We will say that a rule can be applied to a sequent whenever the sequent matches the parent of the rule.
For instance, consider the rule $(\lnot \land)$.
The sequent $[a]p, \lnot(p \land q)$ matches the rule parent, $[a]p, \lnot p$ the first rule child and $[a]p, \lnot q$ the second rule child, and thus the three sequents together are a rule instance of $(\lnot \land)$.
For brevity, we will write the rule instances as $\Gamma_0 ~/~ \Gamma_1 | \dots | \Gamma_n$ where ${\{\Gamma_i\}}_{i \leq n}$ are sequents.
Here the symbol $~/~$ is used to separate the sequent that instantiates the rule parent and the ones that instantiate the rule children whereas $|$ is used to separate the rule children.
Using this notation we can say that $[a]p, \lnot(p \land q) ~/~ [a]p, \lnot p \mid  [a]p, \lnot q$ is a rule instance of $(\lnot \land)$.

We use the same terminology that we used for rules for rule instances.
For example, the principal formula of a rule instance is simply the formula that substitutes the principal formula in the rule.

A key convention we follow when interpreting inference rules is that:

\begin{emphasis}
In all rule instances the principal formula is not an element of the context.
\end{emphasis}

For example, $\lnot \lnot p ~/~ \lnot \lnot p, p$ is not a rule instance of $(\lnot)$ whereas $\lnot \lnot p ~/~ p$ is.

Rules might have side conditions under which they are instantiated.
For example, we could add a side condition to the rule $(\lnot)$ stating that the context $\Delta$ does not contain a negated formula.
Then, with this side condition, the sequent $\lnot q, \lnot \lnot p ~/~ \lnot q, p$ is not a rule instance of $(\lnot)$.

Next, we will give a formal definition of a \emph{tableau} over a set $\mathsf{L}$ of inference rules.
Informally, tableaux are rooted trees in which the nodes are labeled with sequents, and their construction adheres to the rules in $\mathsf{L}$.

\begin{defi}\label{d:tableau-general}
  Given a set $\mathsf{L}$ of inference rules,
  a \emph{tableau $\tab$ for a sequent $\Gamma$} is a tuple $(V, \edge, r, \Lambda, L)$ where
  $(V, \edge, r)$ is a rooted tree,
  $L$ is a function that assigns to every interior node $t \in V$ a rule $L(t) \in \mathsf{L}$,
  $\Lambda$ is a function that assigns to every node $t \in V$ a sequent $\Lambda(t)$ such that
  $\Lambda(r) = \Gamma$, and
  if $s \in V$ is an interior node with children $t_1, \dots, t_n$ then $\Lambda(s) ~/~ \Lambda(t_1) | \dots | \Lambda(t_n)$ is a rule instance of the rule $L(s)$ (respecting any side conditions).
\end{defi}

The tableau in \autoref{f:tableauExample} is an example of how we represent tableaux.
The root is always the uppermost node in the drawing.
Note that we do not write names for the nodes and instead write the corresponding sequent $\Lambda(s)$ associated with a node $s$.
Whenever $s \edge t$ we draw an arrow from $s$ to $t$.
The rule $L(s)$ associated to an interior node $s$ is written next to the arrows to the children of $s$.

\begin{figure}[H]
\begin{center}
  \begin{tikzpicture}[node distance=4em, >=latex]
    \node (1) {$\lnot \lnot p \+ \lnot (p \land q)$};
    \node [below of=1] (2) {$p \+ \lnot (p \land q)$};
    \arr{1}{$(\lnot)$}{2}
    \node [below of=2] (3) {};
    \node [left of=3, node distance=8em] (3a) {$p, \lnot p$}; \arr{2}{$(\lnot \land)$}{3a}
    \node [right of=3, node distance=7em] (3b) {$p, \lnot q$};  \arr{2}{$(\lnot \land)$}{3b}
  \end{tikzpicture}
\end{center}
\caption{An example of a tableau over the rules $(\lnot)$ and $(\lnot \land)$.}\label{f:tableauExample}
\end{figure}

\begin{defi}
  A tableau is \emph{maximal} if no rules can be applied to any of the sequents of the leaves.
\end{defi}

For example, the tableau in \autoref{f:tableauExample} is maximal, as neither $(\lnot)$ nor $(\lnot \land)$ can be applied to $p, \lnot p$ nor to $p, \lnot q$.

Intuitively, a tableau represents an attempt to construct a Kripke model satisfying the sequent at the root.
The following two definitions make this precise and will be used later.

\begin{defi}
A tableau rule is \emph{locally sound} if for all its instances $\Delta ~/~ \Gamma_1 | \dots | \Gamma_n$ and for all pointed models $(\kmodel, v)$: if $\kmodel, v \Vdash \Delta$ then $\kmodel, v \Vdash \Gamma_i$ for some $1 \leq i \leq n$.
It is \emph{locally invertible} if conversely, for all its instances $\Delta ~/~ \Gamma_1 | \dots | \Gamma_n$ and for all pointed models $(\kmodel, v)$: If $\kmodel, v \Vdash \Gamma_i$ for some $1 \leq i \leq n$ then $\kmodel, v \Vdash \Delta$.
\end{defi}

As an example, both the rules $(\lnot)$ and $(\lnot \land)$ are locally sound and locally invertible.
We~include ``locally'' here because both definitions above are about the local truth at a specific state in a Kripke model --- not about satisfiability.

\section{Local Tableaux}\label{sec:LocalTableaux}

We now start to define a tableau proof system.
We make a key distinction between local and modal reasoning.
Local means that we reason about the same state in a potential Kripke model.
For example, we can go from $[a \cup b^\ast]p$ to $[a]p \+ p \+ [b][b^\ast]p$ by local reasoning, but from $\lnot[a]p$ to $\lnot p$ we need a modal step.
Here in \autoref{sec:LocalTableaux}, we only consider local reasoning.

As we are defining a tableau system for PDL the reader may expect reduction rules corresponding to the PDL axioms in \autoref{d:pdl-logic}.
Each of these rules would be applicable to formulas of the form $[\alpha]\psi$ where $\alpha$ has the corresponding top-level constructor.
For example, a $(\ast)$ rule corresponding to the axiom $\Ax{^*}$ would allow us to replace $\Gamma, [\alpha^\ast]\psi$ with $\Gamma, \psi, [\alpha][\alpha^\ast]\psi$.
With such rules we obtain the (not necessarily maximal) tableau shown in \autoref{f:exLocRules}.

\begin{figure}[H]
\begin{center}
  \begin{tikzpicture}[node distance=4em, >=latex]
    \node (0) {$[(a \cup p?); b^\ast]\psi$};
    \node [below of=0] (1) {$[a \cup p?][b^\ast]\psi$}; \arr{0}{$(;)$}{1}
    \node [below of=1] (2) {$[a][b^\ast]\psi \+ [p?][b^\ast]\psi$}; \arr{1}{$(\cup)$}{2}
    \node [below of=2, node distance=2em] (2-3) {$(?)$};
    \node [below of=2] (3) {};
    \node [left of=3, node distance=7em] (3a) {$[a][b^\ast]\psi  \+  \lnot p$}; \arr{2}{}{3a}
    \node [right of=3, node distance=7em] (3b) {$[a][b^\ast]\psi  \+  [b^\ast]\psi$};  \arr{2}{}{3b}
    \node [below of=3b] (4b) {$[a][b^\ast]\psi  \+  \psi  \+  [b][b^\ast]\psi$}; \arr{3b}{$(\ast)$}{4b}
  \end{tikzpicture}
\end{center}
\caption{A tableau with reduction rules for each program constructor --- not our system.}\label{f:exLocRules}
\end{figure}

We do \emph{not} define our system with a rule for each program constructor as in \autoref{f:exLocRules}.
Instead, we have only \emph{one general rule $(\Box)$} which is applicable to all formulas of the form $[\alpha]\psi$ where $\alpha$ is a non-atomic program.
The rule $(\Box)$ is called \emph{local}, since it only incorporates propositional reasoning, and no modal reasoning.
Applying the rule $(\Box)$ to $[\alpha]\psi$ can be understood as reducing $\alpha$ maximally by applying $({;})$, $(\cup)$, $(?)$, $(\ast)$ exhaustively, while treating $\psi$ and all $\tau \in \Test(\alpha)$ as atomic.
To illustrate, instead of \autoref{f:exLocRules}, our system will generate the tableau in \autoref{f:exUnfoldRule}.

\begin{figure}[H]
\begin{center}
  \begin{tikzpicture}[node distance=4em, >=latex]
    \node (0) {$[(a \cup p?); b^\ast]\psi$};
    \node [below of=0, node distance=2em] {($\Box$)};
    \node [below of=0] (3) {};
    \node [left of=3, node distance=7em] (3a) {$[a][b^\ast]\psi  \+  \lnot p$}; \arr{0}{}{3a}
    \node [right of=3, node distance=7em] (3b) {$[a][b^\ast]\psi  \+  \psi  \+  [b][b^\ast]\psi$}; \arr{0}{}{3b}
  \end{tikzpicture}
\end{center}
\caption{A tableau in our system consisting of a single application of the $(\Box)$ rule.}\label{f:exUnfoldRule}
\end{figure}

A key challenge when designing any PDL proof system is that naive rewriting of formulas of the form $[\alpha^\ast]\psi$ may lead to this formula being regenerated locally, i.e., the formula $[\alpha^\ast]\psi$ may occur again without a modal step having been applied.
For example, if we apply the rule $(\ast)$ (corresponding to axiom $\Ax{^*}$ from \autoref{d:pdl-logic}) twice to $[{(a^\ast)}^\ast]p$ then besides $p$ and $[a][a^*][{(a^\ast)}^\ast]p$ we also obtain $[{(a^\ast)}^\ast]p$ again.
We show this in the top left of \autoref{fig:double-star} and note that, unless we build other constraints into our system, such a repeat would lead to an infinite tableau.
To avoid such issues, our rule $(\Box)$ will eliminate the repeated unfolded formula as part of its definition, resulting in the bottom left of \autoref{fig:double-star}.

\begin{figure}[H]
\centering
\begin{minipage}{0.4\textwidth}
  \begin{tikzpicture}[node distance=4em, >=latex]
    \node (0) {$[(a^\ast)^\ast]p$};
    \node [below of=0] (1) {$p \+ [a^\ast][(a^\ast)^{\ast}]p$}; \arr{0}{$(\ast)$}{1}
    \node [below of=1] (2) {$p \+ [(a^\ast)^{\ast}]p \+  [a][a^\ast][(a^\ast)^\ast]p$}; \arr{1}{$(\ast)$}{2}
    \node [below of=2](3) {$[(a^\ast)^\ast]p$};
    \node [below of=3] (4) {$p \+ [a][a^\ast][(a^\ast)^\ast]p$}; \arr{3}{$(\Box)$}{4}
  \end{tikzpicture}
\end{minipage}
\hspace{1cm}
\begin{minipage}{0.45\textwidth}
  \begin{tikzpicture}[node distance=4em, >=latex]
    \node (0) {$\lnot[(a^\ast)^\ast]p$};
    \node [below of=0, node distance=2em] (1L) {$(\lnot \ast)$};
    \node [below of=0] (1) {};
    \node [left of=1, node distance=4em] (1a) {$\lnot p$}; \arr{0}{}{1a}
    \node [right of=1, node distance=4em] (1b) {$\lnot[a^\ast][(a^\ast)^{\ast}]p$};  \arr{0}{}{1b}
    \node [below of=1b, node distance=2em] (1bL) {$(\lnot \ast)$};
    \node [below of=1b] (2) {};
    \node [left of=2, node distance=4em] (2a) {$\lnot[(a^\ast)^{\ast}]p$}; \arr{1b}{}{2a}
    \node [right of=2, node distance=4em] (2b) {$\lnot[a][a^\ast][(a^\ast)^\ast]p$}; \arr{1b}{}{2b}
    \node [below of=2a, node distance=1em] (3) {$\times$};
    \node [below of=2a] (4) {{$\lnot[(a^\ast)^\ast]p$}};
    \node [below of=4, node distance=2em] (4L) {$(\diam)$};
    \node [below of=4] (5) {};
    \node [left of=5, node distance=5em] (5a) {$\lnot p$}; \arr{4}{}{5a}
    \node [right of=5, node distance=5em] (5b) {$\lnot[a][a^\ast][(a^\ast)^\ast]p$}; \arr{4}{}{5b}
  \end{tikzpicture}
\end{minipage}
\caption{Two tableaux with reduction rules for $[{(a^*)}^*]p$ and $\lnot[{(a^\ast)}^\ast]p$ and below them their corresponding local tableaux in our system.}\label{fig:double-star}
\end{figure}

Similarly to the local rule $(\Box)$, we will define a local rule $(\diam)$ to reduce formulas of the shape $\lnot[\alpha]\psi$ which we call \emph{diamond formulas}.
To avoid getting into similar infinite repeats when unfolding diamond formulas, the rule $(\diam)$ will eliminate sequents that contain the formula we are unfolding.
In \autoref{fig:double-star} on the right the $\times$ indicates such a sequent.

The soundness of both these eliminations can informally be explained by the fact that formulas of the form $[\beta^*]\phi$ are greatest fixpoints, and are therefore satisfied when unfolded infinitely often, whereas formulas of the form $\lnot[\beta^*]\phi$ are least fixpoints, and are therefore not satisfied when unfolded infinitely often.

We prepare the rules for boxes in \autoref{subsec:Boxes}, and in \autoref{subsec:Diamonds} and \autoref{subsec:LoadedDiamonds} we prepare the rules for diamonds.
In \autoref{subsec:LocalRules} we then add the usual rules for the propositional connectives $\lnot$ and $\land$ to define local tableaux.

\subsection{Boxes}\label{subsec:Boxes}

In this section, our goal is to unfold programs in boxes.
We start by sketching our approach.
For a formula $[\alpha]\psi$, where $\alpha$ is a complex program, we will define a finite set $\unfold_{\square}(\alpha,\psi)$ of sequents such that the following holds (cf.~\autoref{l:localBoxTruth}):
\begin{equation}\label{equiv:unfold-box}
  [\alpha]\psi
  \equiv
  \bigvee \Big\{ \textstyle{\bigwedge} \Gamma
  \mid \Gamma \in \unfold_{\square}(\alpha,\psi) \Big\}
\end{equation}
and every formula in any sequent in $\unfold_{\square}(\alpha,\psi)$ is either equal to~$\psi$, or of the form $\lnot\tau$ where $\tau \in \Test(\alpha)$, or of the form $[a][\delta_{1}]\cdots[\delta_{n}]\psi$ where $a, \delta_1, \ldots, \delta_n \in \Prog(\alpha)$ and $a$ is atomic. This will be made precise in \autoref{l:unfoldBoxContent}.

As a corollary of \autoref{l:localBoxTruth}, the following rule will be locally invertible and locally sound:
\begin{center}
  \AxiomC{$\Delta , [\alpha]\psi$}
  \LeftLabel{$(\Box)$}
  \UnaryInfC{$
    \big\{\Delta, \Gamma \mid \Gamma \in \unfold_{\square}(\alpha,\psi)
    \big\}
    $}
  \DisplayProof
\end{center}

Each $\Gamma \in \unfold_{\Box}(\alpha,\psi)$ will correspond to a subset $\ell \subseteq \Test(\alpha)$, which we will refer to as a test profile. Tests define control flow in the execution of $\alpha$, and $\ell$ represents a situation where the tests in $\ell$ succeed and the tests not in $\ell$ fail.
Given a test profile $\ell$, the corresponding $\Gamma$ will contain the formulas of the form $[a][\delta_{1}]\cdots[\delta_{n}]\psi$,
where the program list $a \delta_1 \ldots \delta_n \in \Prog(\alpha)$ is a possible execution of $\alpha$ in the situation represented by $\ell$.

We give some further intuitions for the definition of $\unfold_{\Box}(\alpha,\psi)$ in terms of the reduction axioms from \autoref{d:pdl-logic}. The right-hand side of \equivalenceRef{equiv:unfold-box} is \emph{almost} the disjunctive normal form obtained by rewriting $[\alpha]\psi$ exhaustively using the reduction axioms $\Ax{?}$, $\Ax{\cup}$, $\Ax{;}$, $\Ax{^*}$ while ensuring that formulas of the form $[\beta^*]\phi$
are only rewritten once, and considering test formulas and $\psi$ atomic, since these formulas will be reduced recursively by other local rule applications. Multiple disjuncts result from applying the axiom $\Ax{?}$.
The ``almost'' refers to two things:
First, each disjunct in the resulting disjunctive normal form corresponds to one or more $\Gamma \in \unfold_{\Box}(\alpha,\psi)$.\footnote{This can be seen by considering $[p?;q?]r$ for which we obtain a disjunctive normal form $\neg p \lor \neg q \lor \neg r$. The disjunct $\neg p$ corresponds to both $\ell=\emptyset$ and $\ell = \{q\}$.
Informally, this can be explained as ``if the test $p?$ fails then the outcome of $q?$ does not matter''.}
Secondly, as mentioned above, we must also take care to avoid infinite unfolding of formulas of the form $[\beta^*]\phi$ (which would cause non-termination). This is done by excluding from $\Gamma$ all formula occurrences of the form $[\beta^*]\phi$ that were
already reduced. We will explain this mechanism below before \autoref{d:unfoldBox}.

First, we define the notion of a test profile $\ell$ for $\alpha$ and prove a number of useful results.

\begin{defi}[\lean{TP}]\label{d:TP}
Let $\alpha$ be a program.
A \emph{test profile} for $\alpha$ is a subset $\ell \subseteq \Test(\alpha)$.
The set of all test profiles of $\alpha$ is denoted as $\TT(\alpha)$.
Given a profile $\ell \in \TT(\alpha)$, we define its \emph{signature formula}
$\sigma^{\ell}$
as
\[
\sigma^{\ell} \isdef \bigwedge \{ \pm_{\ell}\tau \mid \tau \in \Test(\alpha) \},
\]
where $\pm_{\ell}\tau$ denotes the formula $\tau$ if $\tau \in \ell$ and
the formula $\lnot\tau$ if $\tau \notin \ell$.
\end{defi}

In the following two remarks, we gather some basic observations on test profiles and their signature formulas.

\begin{rem}
We will frequently use the following facts:
\begin{equation}\label{fact:test-profile-complete}
  \text{For all } \alpha \in \progs: \quad \top \equiv \bigvee \big\{ \sigma^{\ell}
\mid \ell \in \TT(\alpha) \big\},
\end{equation}
\begin{equation}\label{fact:test-profile-orthogonal}
\text{For all }  \alpha \in \progs \text{ and all } \ell,\ell' \in \TT(\alpha): \quad \ell \neq \ell' \text{ implies } \sigma^{\ell} \land \sigma^{\ell'} \equiv \bot.
\end{equation}
Note that in \factRef{fact:test-profile-orthogonal},
the right to left direction does not hold, because $\sigma^{\ell}$
does not have to be consistent.
Consider for example $\alpha \isdef (p \land q)?; (\lnot p)?; {(\lnot q)?}$ and the test profile $\ell \isdef \{ p \land q \}$.
Then $\sigma^\ell = (p \land q) \land \lnot p \land \lnot q$, and therefore $\sigma^\ell \land \sigma^\ell \equiv \bot$,
but still $\ell = \ell$.

Moreover, we will use the fact that for all formulas $\phi$ and $\psi$, and all programs $\alpha$:
\begin{equation}
\label{fact:equiv_iff_TPequiv}
\phi\models \psi \text{ \ iff \ }
\text{ for all } \ell \in \TT(\alpha):
\phi \land \sigma^{\ell} \models \psi \land \sigma^{\ell}.
\end{equation}
We leave the verification of \factRef{fact:test-profile-complete,fact:test-profile-orthogonal,fact:equiv_iff_TPequiv} as an exercise in propositional logic to the reader.
\end{rem}

\begin{rem}
Below we will need to relate the test profiles of a complex program to those of its subprograms.
Here the following facts will be useful.

Assume that $\beta \in \Prog(\alpha)$.
Every test profile $\ell \in \TT(\alpha)$ induces a (possibly smaller) test profile $\ell|_{\beta} \isdef \ell \cap \Test(\beta) \in \TT(\beta)$, and every test profile in $\TT(\beta)$ is of this kind.
For the signature formulas, we find for $\ell \in \TT(\alpha)$ that
\[
\sigma^{\ell} \equiv
\sigma^{\ell|_{\beta}} \land \bigwedge \{ \pm_{\ell}\tau \mid \tau \in \Test(\alpha) \setminus \Test(\beta)\}.
\]
From this it follows by propositional logic that:
\begin{equation}
\label{fact:tp1}
\begin{array}{l}
\text{For all $\ell \in \TT(\alpha)$, all $\beta \in \Prog(\alpha)$, and all formulas $\phi,\psi$,}\\
\text{if $\phi \land \sigma^{\ell|_{\beta}} \equiv \psi \land \sigma^{\ell|_{\beta}}$
then $\phi \land \sigma^{\ell} \equiv \psi \land \sigma^{\ell}$.}
\end{array}
\end{equation}
\end{rem}

Given a formula $[\alpha]\psi$ and a test profile $\ell \in \TT(\alpha)$,
we will define a set $\Bset{\ell}{\alpha}{\psi}$ of formulas that suffice to make $[\alpha]\psi$ true
when the tests that are not in $\ell$ fail and those in $\ell$ succeed.
Every formula in $\Bset{\ell}{\alpha}{\psi}$ will be either of the form $\neg \tau$ where $\tau \in \Test(\alpha)\setminus\ell$ or
of the form $[a][\alpha_1] \ldots[\alpha_n]\psi$ for a possibly empty list $a\alpha_1\ldots \alpha_n$ of subprograms of $\alpha$.
Formulas of the form $[a][\alpha_1] \ldots[\alpha_n]\psi$ will be represented by the list $a\alpha_1\ldots \alpha_n$, and we compute them recursively.
The set $\Pset{\ell}{\alpha}$ will contain all the program lists of possible executions of $\alpha$ when precisely the tests in $\ell$ succeed.
Moreover, we must avoid infinite unfolding of formulas of the form $[\beta^*]\phi$.
This is taken care of in the definition of $\Pset{\ell}{\alpha}$ below in the clause for $\Pset{\ell}{\beta^\ast}$
where removing the empty list $\emptylist$ from $\Pset{\ell}{\beta}$ ensures that $[\beta^\ast]\phi$ is only unfolded once.
Put differently, this clause excludes regenerated occurrences of $[\beta^\ast]\phi$ from $\Bset{\ell}{\alpha}{\psi}$.

We now define $\Pset{\ell}{\alpha}$ and we will provide examples to illustrate \autoref{d:unfoldBox} below.

\begin{defi}[\lean{unfoldBox}]\label{d:unfoldBox}
Let $\alpha_0 \in \progs$ be a program and let $\ell \in \TT(\alpha_0)$ be a test profile.
We define for all $\alpha \in \Prog(\alpha_0)$, the set $\Pset{\ell}{\alpha}$ of lists of programs,
by induction on the shallow program structure of $\alpha$:
\[\begin{array}{lll}
\Pset{\ell}{a} &\isdef& \{ a \}
\\ \Pset{\ell}{\tau?} &\isdef&
     \left\{\begin{array}{cl}
	    \emptyset & \text{ if } \tau \notin \ell
	 \\ \{ \emptylist \} & \text{ if } \tau \in \ell
	 \end{array}\right.
\\ \Pset{\ell}{\beta\cup\gamma} &\isdef& \Pset{\ell}{\beta} \cup \Pset{\ell}{\gamma}
\\ \Pset{\ell}{\beta;\gamma} &\isdef&
    \big\{ \ol{\beta} \gamma \mid
	     \ol{\beta} \in \Pset{\ell}{\beta} \setminus \{\emptylist \} \big\}
	\cup \big\{ \ol{\gamma} \mid
	     \ol{\gamma} \in \Pset{\ell}{\gamma}, \emptylist \in \Pset{\ell}{\beta} \big\}
\\ \Pset{\ell}{\beta^{\ast}} &\isdef&
    \{ \emptylist \} \cup
    \big\{ \ol{\beta}\beta^{\ast} \mid
	    \ol{\beta} \in \Pset{\ell}{\beta} \setminus \{\emptylist \} \big\}
\end{array}\]

Furthermore, we define the set $\Fset{\ell}{\alpha}$ of negations of ``failed test formulas''
\[
\Fset{\ell}{\alpha} \isdef \{ \lnot\tau \mid \tau \in \Test(\alpha) \text{ and } \tau \notin \ell \}.
\]
Finally, for a formula $\psi \in \pdlforms$,
we define
\[
\Bset{\ell}{\alpha}{\psi} \isdef \Fset{\ell}{\alpha} \cup
   \big\{ \Box(\ol{\alpha},\psi) \mid \ol{\alpha} \in \Pset{\ell}{\alpha} \big\},
\]
and
\[
\unfold_{\square}(\alpha,\psi) \isdef
   \big\{ \Bset{\ell}{\alpha}{\psi} \mid \ell \in \TT(\alpha) \big\}.
\]
\end{defi}

\begin{exa}\label{ex:unfoldBox}
  Consider the program $\alpha \isdef (a \cup p?); b^\ast$.
  Then $\Test(\alpha) = \{ p \}$, and hence there are two test profiles for $\alpha$, namely $\ell_0 \isdef \emptyset$ and $\ell_1 \isdef \{p\}$,
corresponding to ``${p?}$ fails'' and ``${p?}$ succeeds''.
For the test profile $\ell_0=\emptyset$, we have $\Fset{\ell_0}{(a \cup p?); b^\ast} = \{ \lnot p \}$ and:
  \[
    \begin{array}{rl}
      & \Pset{\ell_0}{ (a \cup p?); b^\ast } \\
      = &
          \big\{ \ol{\alpha} b^\ast \mid
          \ol{\alpha} \in \Pset{\ell_0}{a \cup p?} \setminus \{\emptylist \} \big\}
          \cup \big\{ \ol{\beta} \mid
          \ol{\beta} \in \Pset{\ell_0}{b^\ast}, \emptylist \in \Pset{\ell_0}{a \cup p?} \big\} \\
      = &
          \big\{ \ol{\alpha} b^\ast \mid
          \ol{\alpha} \in ( \Pset{\ell_0}{a} \cup \Pset{\ell_0}{p?} ) \setminus \{ \emptylist \} \big\}
          \ \cup \ \big\{ \ol{\beta} \mid
          \ol{\beta} \in (\ldots)
          , \emptylist \in \Pset{\ell_0}{a} \cup \Pset{\ell_0}{p?} \big\} \\
      = &
          \big\{ \ol{\alpha} b^\ast \mid
          \ol{\alpha} \in ( \{ a \} \cup \emptyset ) \setminus \{ \emptylist \} \big\}
          \ \cup \ \big\{ \ol{\beta} \mid
          \ol{\beta} \in (\ldots)
          , \emptylist \in \{ a \} \cup \emptyset \big\} \\
      = &
          \big\{ a b^\ast \big\}
          \ \cup \ \emptyset \\
      = & \{ a b^\ast \}.
    \end{array}
  \]
  The set $\Pset{\ell_0}{ (a \cup p?); b^\ast} =\{ a b^\ast \}$ can be seen as the possible executions of $(a \cup p?); b^\ast$ when the test~${p?}$ fails.

  For the test profile $\ell_1=\{p\}$, we get $\Fset{\ell_1}{(a \cup p?); b^\ast} = \emptyset$ and:
  \[
    \begin{array}{rl}
      & \Pset{\ell_1}{ (a \cup p?); b^\ast } \\
      = &
          \big\{ \ol{\alpha} b^\ast \mid
          \ol{\alpha} \in \Pset{\ell_1}{a \cup p?} \setminus \{\emptylist \} \big\}
          \cup \big\{ \ol{\beta} \mid
          \ol{\beta} \in \Pset{\ell_1}{b^\ast}, \emptylist \in \Pset{\ell_1}{a \cup p?} \big\} \\
      = &
          \big\{ \ol{\alpha} b^\ast \mid
          \ol{\alpha} \in ( \Pset{\ell_1}{a} \cup \Pset{\ell_1}{p?} ) \setminus \{ \emptylist \} \big\} \\
        & \ \cup \ \big\{ \ol{\beta} \mid
          \ol{\beta} \in ( \{ \emptylist \} \cup \{ \ol{\gamma} b^\ast \mid \ol{\gamma} \in \Pset{\ell_1}{b} \setminus \{ \emptylist \} \} )
          , \emptylist \in \Pset{\ell_1}{a} \cup \Pset{\ell_1}{p?} \big\} \\
      = &
          \big\{ \ol{\alpha} b^\ast \mid
          \ol{\alpha} \in ( \{ a \} \cup \{ \emptylist \} ) \setminus \{ \emptylist \} \big\} \\
        & \ \cup \ \big\{ \ol{\beta} \mid
          \ol{\beta} \in ( \{ \emptylist \} \cup \{ \ol{\gamma} b^\ast \mid \ol{\gamma} \in \{ b \} \setminus \{ \emptylist \} \} )
          , \emptylist \in \{ a \} \cup \{ \emptylist \} \big\} \\
      = &
          \big\{ a b^\ast \big\} \ \cup \ \big\{ \ol{\beta} \mid \ol{\beta} \in ( \{ \emptylist \} \cup \{ b b^\ast \} ) \big\} \\
      = &
          \{ a b^\ast, \emptylist, b b^\ast \}.
    \end{array}
  \]
  The set $\Pset{\ell_1}{ (a \cup p?); b^\ast} =\{ a b^\ast, \emptylist, b b^\ast \}$ can be seen as the possible executions of $(a \cup p?); b^\ast$ when the test~${p?}$ succeeds.

  This can then be used to unfold a formula $[(a \cup p?);b^\ast]\psi$, for arbitrary $\psi$, as follows:
  \[
    \Bset{\ell_0}{(a \cup p?);b^\ast}{\psi} = \Fset{\ell_0}{(a \cup p?);b^\ast} \cup \big\{ \Box(\ol{\alpha},\psi) \mid \ol{\alpha} \in \Pset{\ell_0}{(a \cup p?);b^\ast} \big\} = \{ \lnot p, [a][b^\ast]\psi \},
  \]
  \[
    \Bset{\ell_1}{(a \cup p?);b^\ast}{\psi} = \Fset{\ell_1}{(a \cup p?);b^\ast} \cup \big\{ \Box(\ol{\alpha},\psi) \mid \ol{\alpha} \in \Pset{\ell_1}{(a \cup p?);b^\ast} \big\} = \{ \psi, [a][b^\ast]\psi, [b][b^\ast]\psi \}.
  \]
  Now by \autoref{d:unfoldBox} we get:
  \[
    \unfold_{\square}((a \cup p?);b^\ast, \psi) = \{ \{ \lnot p, [a][b^\ast]\psi \} ,  \{ \psi, [a][b^\ast]\psi,  [b][b^\ast]\psi \} \}.
  \]
  This means that our $(\Box)$ rule will split the formula $[(a \cup p?);b^\ast]\psi$ into two branches (cf.~\autoref{f:exUnfoldRule}), one containing
  $\{ \lnot p, [a][b^\ast]\psi \}$
  and another containing
  $\{ \psi, [a][b^\ast]\psi,  [b][b^\ast]\psi \}$.

  Lastly, the forthcoming \autoref{l:localBoxTruth} applied to our example gives us
  \[ [(a \cup p?);b^\ast]\psi \ \equiv \  ( \lnot p \land [a][b^\ast]\psi ) \lor ( \psi \land [a][b^\ast]\psi \land  [b][b^\ast]\psi ), \]
  which can also be verified directly.
\end{exa}

\begin{exa}\label{ex:P-box-double-star}
  Consider $\alpha \isdef {(a^\ast)}^\ast$.
  As $\alpha$ does not contain tests there is only the empty test profile $\ell \isdef \emptyset$ for $\alpha$.
  Now first, we have $\Fset{\emptyset}{\alpha} = \emptyset$.
  Second, one easily computes
  $\Pset{\emptyset}{a} = \{ a \}$,
  $\Pset{\emptyset}{a^\ast} = \{\emptylist\} \cup \{aa^\ast\}$ and
  $\Pset{\emptyset}{{(a^\ast)}^\ast} = \{\emptylist\} \cup \{aa^\ast{(a^\ast)}^\ast\}$.
  Together we thus have $\unfold_{\square}({(a^\ast)}^\ast, \psi) = \{\psi, [a][a^\ast][{(a^\ast)}^\ast]\psi\}$.
  Compare this to the left side of~\autoref{fig:double-star}.
\end{exa}

In the following \autoref{l:boxHelperTermination}, we gather some syntactic information about $\PsetName^\ell$ and the $\unfold_{\Box}$ operation.
Intuitively, we will use this to show that programs without $\ast$ become simpler when unfolded, and that will be useful to show termination of our proof system.

\begin{lem}[\checklean{boxHelperTermination}]\label{l:boxHelperTermination}
Let $\alpha$ be a program, let $\ell$ be a test profile for $\alpha$, and let $\ol{\delta} \in \Pset{\ell}{\alpha}$.
Then we have:
\begin{enumerate}[(a)]
\item If $\alpha = a$ for some atomic $a$, then $\ol{\delta} = a$.
\item If $\alpha = \beta^\ast$ for some $\beta$, then either $\ol{\delta} = \emptylist$ or $\ol{\delta}$ is of the form $\ol{\delta} = a\delta_1\cdots\delta_n\beta^\ast$, where $n \geq 0$, $a \in \Prog(\alpha)$ is atomic and $\delta_i \in \Prog(\alpha) \setminus \{\alpha\}$, for each $i$.
\item Otherwise, either $\ol{\delta} = \emptylist$ or $\ol{\delta}$ is of the form $\ol{\delta} = a\delta_1\cdots\delta_n$, where $n \geq 0$, $a$ is atomic and $\delta_i \in \Prog(\alpha) \setminus \{\alpha\}$, for each $i$.
\end{enumerate}
\end{lem}
We note that case (c) applies if $\alpha$ is a test, or a choice or composition of simpler programs.
\begin{proof}
  Follows directly from \autoref{d:unfoldBox} of $\Pset{\ell}{\alpha}$.
  We omit the details.
\end{proof}

From \autoref{l:boxHelperTermination} we also get the following two results which we will use to prove termination of local tableaux.

\begin{lem}[\checklean{unfoldBoxContent}]\label{l:unfoldBoxContent}
  Let $\alpha$ be a program and let $\psi$ and $\phi$ be formulas such that $\phi \in \bigcup \unfold_{\square}(\alpha,\psi)$.
  Then we have $\phi \in \FL([\alpha]\psi)$.
  Furthermore, we find that $\phi = \psi$,
  or $\phi$ is of the form $\phi = \lnot\tau$ with $\tau \in \Test(\alpha)$,
  or $\phi$ is of the form $\phi = [a][\delta_{1}]\cdots[\delta_{n}]\psi$
  with $a,\delta_{1},\ldots,\delta_{n} \in \Prog(\alpha)$ and $n \geq 0$.
\end{lem}
\begin{proof}
Take any $\phi \in \bigcup \unfold_{\square}(\alpha,\psi)$.
To show $\phi \in \FL([\alpha]\psi)$ we distinguish two cases.
Either $\phi$ comes from the function $\PsetName^\ell$ or from $\Fset{\ell}{\cdot}$ in \autoref{d:unfoldBox}.
If it comes from $\PsetName^\ell$ then it can be obtained using the three cases given in \autoref{d:fischerLadner} and closure under subformulas.
If it comes from $\Fset{\ell}{\alpha}$, then it is a subformula of $[\alpha]\psi$.
In both cases it must be in $\FL([\alpha]\psi)$.

For the remaining claim, distinguish the same cases.
If $\phi$ comes from $\Pset{\ell}{\alpha}$, then we apply \autoref{l:boxHelperTermination} and note that in each of the three cases (a) to (c) the claim follows.
If $\phi$ comes from $\Fset{\ell}{\alpha}$, then by definition of $\FsetName$ we must have $\phi = \lnot\tau$ for some $\tau \in \Test(\alpha)$.
\end{proof}

\begin{lem}[\checklean{PgoesDown}]\label{l:PgoesDown}
  For a program $\alpha$, let $|\alpha|$ denote its length (\lean{lengthOfProgram}).
  Let $\alpha$ be a program,
  let $\ell$ be a test profile for $\alpha$,
  let $\ol{\delta} \in \Pset{\ell}{\alpha}$,
  and let $\gamma \in \ol{\delta}$.
  If $\alpha = a$ for some atomic program $a$, then  $\gamma = a$.
  If $\alpha = \beta^\ast$ for some program $\beta$, then $|\gamma| \leq |\alpha|$.
In all other cases we have $|\gamma| < |\alpha|$.
\end{lem}
\begin{proof}
  Take any $\alpha$ and $\ell$,
  let $\ol{\delta} \in \Pset{\ell}{\alpha}$,
  and let $\gamma \in \ol{\delta}$.
  Now distinguish the three cases.
  \begin{itemize}
  \item If $\alpha = a$ for some $a$, then by \autoref{l:boxHelperTermination}(a) we have $\ol{\delta} = a$ and thus must have $\gamma = a$.
  \item If $\alpha = \beta^\ast$ for some $\beta$, then by \autoref{l:boxHelperTermination}(b) we have $\gamma \in \Prog(\alpha)$ and thus $|\gamma| \leq |\alpha|$.
  \item In any other case by \autoref{l:boxHelperTermination}(c) we have $\gamma \in \Prog(\alpha) \setminus \{\ \alpha \}$ and thus $|\gamma| < |\alpha|$. \qedhere
  \end{itemize}
\end{proof}

The following technical lemma will be needed in the proof of \autoref{l:localBoxTruth} below.

\begin{lem}[\checklean{boxHelperTP}]\label{l:boxHelperTP}
Let $\alpha$ be a program, and let $\ell$ be a test profile for $\alpha$.
For all formulas $\psi$, we have:
\[
\bigwedge \Bset{\ell}{\alpha}{\psi} \land \sigma^{\ell} \equiv
\bigwedge \big\{ \Box(\ol{\alpha},\psi) \mid
    \ol{\alpha} \in \Pset{\ell}{\alpha} \big\}
	\land \sigma^{\ell}.
\]
\end{lem}
\begin{proof}
  By \autoref{d:unfoldBox}, we have
  $\Bset{\ell}{\alpha}{\psi} \isdef \Fset{\ell}{\alpha} \cup
  \big\{ \Box(\ol{\alpha},\psi) \mid \ol{\alpha} \in \Pset{\ell}{\alpha} \big\}$, so left-to-right is immediate.
  For right-to-left, we also use that $\sigma^\ell \to \Fset{\ell}{\alpha}$ is valid.
\end{proof}

A crucial part of how we defined $\PsetName^\ell$ in \autoref{d:unfoldBox} is the use of the empty sequence $\emptylist$ while unfolding tests, sequential compositions and iteration programs.

Readers familiar with $\mu$-calculus may want to know that these conditions allow our approach to deal with unguardedness in PDL formulas and programs.
We discuss this in $\mu$-calculus terms in \autoref{remark:mu-unguardedness} and then return to PDL in \autoref{l:guard} below.

\begin{rem}\label{remark:mu-unguardedness}
Using notation from $\mu$-calculus we note that
\[
[\beta^{\ast}]\psi \equiv \nu x. \psi \land [\beta] x
\]
where $x$ is a fresh proposition letter not occurring in either $\psi$ or $\beta$.
From this it follows that
\[
[\beta^{\ast}]\psi \equiv \nu x. \psi \land \phi(x)
\]
for any formula $\phi$ such that $[\beta]x \equiv \phi(x)$.
Now suppose that $\phi'$ is a formula we obtain from $\phi$ by replacing any number of \emph{unguarded} occurrences of $x$ with $\top$.
Then it follows from the general theory of fixpoint logics that
\[
\nu x. \psi \land \phi'(x) \equiv \nu x. \psi \land \phi(x).
\]
In particular, if $\phi(x)$ is of the form $\phi = (x \land \chi_{0}) \lor \chi_{1}$ (where $x$ may or may not have occurrences in $\chi_0$ and $\chi_1$), we find that
\[
[\beta^{\ast}]\psi \equiv \nu x. \psi \land (\chi_0 \lor \chi_1).
\]
From this it follows that, in order to show that $\rho\models [\beta^{\ast}]\psi$, for some formula $\rho$, it suffices to show that $\rho\models\psi$ and that $\rho \models\ \subst{\substAt{x}{\rho}}{(\chi_{0}\lor\chi_{1})}$.
We spell this out in the lemma below, and we also supply a direct proof.
\end{rem}

\begin{lem}[\checklean{guardToStar}]\label{l:guard}
Let $\chi_0$ and $\chi_1$ be formulas, let $\beta$ be a program, and let $x$ be a proposition letter that does not occur in $\beta$ (but $x$ may occur in $\chi_0$ and in $\chi_1$) such that $[\beta] x \equiv (x \land \chi_0) \lor \chi_1$.
Then, for all formulas $\rho$ and $\psi$,
if $\rho \models \subst{\substAt{x}{\rho}}{(\chi_{0}\lor\chi_{1})}$  and $\rho\models\psi$
then $\rho \models [\beta^{\ast}]\psi$.
\end{lem}

\begin{proof}
The key observation in this proof is the following:
\begin{equation}\label{obs:ind}
\text{if $\kmodel, w \Vdash \rho$ and $(w,v) \in R_{\beta}$ then
$\kmodel, v \Vdash \rho$.}
\end{equation}
To see why this is the case, assume that $\kmodel, w \Vdash \rho$.
Then by assumption we have $\kmodel, w \Vdash \subst{\substAt{x}{\rho}}{(\chi_{0}\lor\chi_{1})}$,
and since obviously the formula $\subst{\substAt{x}{\rho}}{x} = \rho$ holds at $w$, it follows that
$\kmodel, w \Vdash \subst{\substAt{x}{\rho}}{((x \land \chi_{0}) \lor\chi_{1})}$.
Now by the two assumptions we have $\kmodel, w \Vdash \subst{\substAt{x}{\rho}}{([\beta]x)}$, which just
means that $\kmodel, w \Vdash [\beta]\rho$.
It is then immediate by $(w,v) \in R_{\beta}$ that $\kmodel, v \Vdash \rho$.
This finishes the proof of \observationRef{obs:ind}.

To see how the lemma follows from this, let $\kmodel, w$ be a pointed model such
that $\kmodel, w \Vdash \rho$.
Take an arbitrary state $v$ such that $(w,v) \in R_{\beta^{\ast}}$, then there is
a finite sequence of states that forms a path of $\beta$-steps from $w$ to $v$.
Now by induction on the length of this path we can apply \observationRef{obs:ind} to each of the $R_\beta$ steps.
We find that $\kmodel, v \Vdash \rho$, so that $\kmodel, v \Vdash \psi$ by assumption.
Hence $\kmodel, w \Vdash [\beta^\ast]\psi$.
\end{proof}

Before the main lemma of this section, we prove a technical helper lemma that relates test profiles to the set $X$.

\begin{lem}\label{l:localBoxTruthAux}
Let $\gamma$ be a program, and let $\psi$ be a formula.
Then the following hold:
\begin{enumerate}[(a)]
\item[(a)] If $[\gamma]\psi \land \sigma^{\ell} \equiv \textstyle{\bigwedge} \Bset{\ell}{\gamma}{\psi} \land \sigma^{\ell}$ for all $\ell \in \TT(\gamma)$, then $[\gamma]\psi \equiv \bigvee  \Big\{ \textstyle{\bigwedge}\Bset{\ell}{\gamma}{\psi} \mid \ell \in \TT(\gamma) \Big\}$;
\item[(b)] If $[\gamma]\psi \land \sigma^{\ell} \models \textstyle{\bigwedge}\Bset{\ell}{\gamma}{\psi} \land \sigma^{\ell}$ for all $\ell \in \TT(\gamma)$, then $[\gamma]\psi \models \bigvee  \Big\{ \textstyle{\bigwedge}\Bset{\ell}{\gamma}{\psi} \mid \ell \in \TT(\gamma) \Big\}$.
\end{enumerate}
\end{lem}
Note that part (b) is almost the same as part (a), only that $\equiv$ is replaced by $\models$ in two places.
\begin{proof}
First, we show part (a):
\begin{align*}
[\gamma]\psi
& \equiv [\gamma]\psi \land \top & \text{by propositional logic} \\
& \equiv [\gamma]\psi \land \bigvee \big\{ \sigma^{\ell}  \mid \ell \in \TT(\gamma) \big\}
& \text{by \factRef{fact:test-profile-complete}} \\
& \equiv \bigvee \big\{ [\gamma]\psi \land \sigma^{\ell}  \mid \ell \in \TT(\gamma) \big\}
& \text{by propositional logic} \\
& \equiv \bigvee \big\{ \textstyle{\bigwedge}\Bset{\ell}{\gamma}{\psi} \land \sigma^{\ell}  \mid \ell \in \TT(\gamma) \big\}
& \text{by assumption} \\
& \equiv \bigvee
  \big\{ \textstyle{\bigwedge}\Bset{\ell}{\gamma}{\psi} \mid \ell \in \TT(\gamma)
\big\} & \text{(see below)}.
\end{align*}

For the last step, the downward implication is trivial.
For the other direction, let $\ell \in \TT(\gamma)$ be such that $\kmodel, w \Vdash \textstyle{\bigwedge}\Bset{\ell}{\gamma}{\psi}$.
Let $\ell' \in \TT(\gamma)$ be the test profile of $w$; that is, $\ell' \isdef \{ \tau \in \Test(\gamma) \mid \kmodel, w \Vdash \tau \}$.
Then by definition we have $\kmodel,w \Vdash \sigma^{\ell'}$.

We claim that $\kmodel, w \Vdash \textstyle{\bigwedge}\Bset{\ell'}{\gamma}{\psi}$.
By definition of $\ell'$ it is clear that $\kmodel, w \Vdash \textstyle{\bigwedge}\Fset{\ell'}{\gamma}$, so it is left to show that
$\kmodel, w \Vdash \textstyle{\bigwedge}\{ \Box(\ol{\alpha},\psi) \mid \ol{\alpha} \in \Pset{\ell'}{\gamma} \}$.
Note that from $\kmodel, w \Vdash \textstyle{\bigwedge}\Bset{\ell}{\gamma}{\psi}$ we may infer that $\kmodel, w \Vdash \textstyle{\bigwedge}\{ \Box(\ol{\alpha},\psi) \mid \ol{\alpha} \in \Pset{\ell}{\gamma} \}$, so that it suffices to show $\Pset{\ell'}{\gamma} \subseteq \Pset{\ell}{\gamma}$.
By \autoref{d:unfoldBox} it now suffices to show $\ell' \subseteq \ell$.
To show this, take any $\tau \notin \ell$.
Note that by assumption we also have $\kmodel, w \Vdash \textstyle{\bigwedge}\Fset{\ell}{\gamma}$,
hence $\kmodel, w \Vdash \lnot \tau$  and thus by definition of $\ell'$ also $\tau \notin \ell'$, so we are done.

To show part (b), note that the same proof works if we replace the fourth $\equiv$ with a $\models$.
\end{proof}

\begin{lem}[\checklean{localBoxTruth}]\label{l:localBoxTruth}
For every program $\gamma$ and every formula $\psi$  we have
\[
[\gamma]\psi \equiv \bigvee
\Big\{ \textstyle{\bigwedge}\Bset{\ell}{\gamma}{\psi} \mid \ell \in \TT(\gamma)\Big\}.
\]
\end{lem}

\begin{proof}
By \autoref{l:localBoxTruthAux}(a) it suffices to show that for all programs $\gamma$,
for all formulas $\psi$ and all test profiles $\ell \in \TT(\gamma)$ we have
\begin{equation}\label{claim:bxeq}
[\gamma]\psi \land \sigma^{\ell} \equiv
\textstyle{\bigwedge}\Bset{\ell}{\gamma}{\psi} \land \sigma^{\ell}.
\end{equation}

We proceed to show~\claimRef{claim:bxeq} by induction on the shallow program structure of $\gamma$.

\textbf{Base Case.}

\begin{description}
\item[Case $\gamma=a$]
In this case, the only test profile is $\ell=\emptyset$, and for this $\ell$ we find
that $\Bset{\ell}{\gamma}{\psi} = \emptyset \cup \{ \Box(a,\psi) \}$, so that
$\bigwedge \Bset{\ell}{\gamma}{\psi} \equiv [a]\psi$ as required.
\item[Case $\gamma={\tau?}$]
Here we make a case distinction.
If $\tau \notin \ell$ then $\Bset{\ell}{\gamma}{\psi} = \{ \lnot\tau \}$, so that
$\bigwedge \Bset{\ell}{\gamma}{\psi} \land \sigma^{\ell}
\equiv \lnot\tau \land \pm_{\ell}\tau
\equiv \lnot\tau
\equiv [\tau?]\psi \land \lnot\tau$.
On the other hand, if $\tau \in \ell$, then
$\Bset{\ell}{\gamma}{\psi} = \emptyset \cup \{ \Box(\emptylist,\psi) \} = \{\psi\}$,
so that
$\bigwedge \Bset{\ell}{\gamma}{\psi} \land \sigma^{\ell}
\equiv \psi \land \pm_{\ell}\tau
\equiv \psi \land \tau
\equiv [\tau?]\psi \land \tau$.
\end{description}

\textbf{Induction step.}
Here we make the usual case distinction.

\begin{description}
\item[Case $\gamma=\alpha\cup\beta$]
Let $\ell \in \TT(\gamma)$. By the induction hypothesis and \autoref{l:boxHelperTP}, we have that
\[
\begin{array}{llll}
[\alpha]\psi \land \sigma^{\ell|_{\alpha}} &\equiv&
\displaystyle{\bigwedge} \big\{ \Box(\ol{\delta},\psi) \mid \ol{\delta} \in \Pset{\ell|_{\alpha}}{\alpha} \big\} \land \sigma^{\ell|_{\alpha}}, \text{ and} \\
{[\beta]}\psi \land \sigma^{\ell|_{\beta}} &\equiv&
\displaystyle{\bigwedge} \big\{ \Box(\ol{\delta},\psi) \mid \ol{\delta} \in \Pset{\ell|_{\beta}}{\beta} \big\} \land \sigma^{\ell|_{\beta}}.
\end{array}
\]
It then follows from \factRef{fact:tp1} that
\[
\begin{array}{llll}
[\alpha]\psi \land \sigma^{\ell} &\equiv&
\displaystyle{\bigwedge} \big\{ \Box(\ol{\delta},\psi) \mid \ol{\delta} \in \Pset{\ell|_{\alpha}}{\alpha} \big\} \land \sigma^{\ell}, \text{ and} \\
{[\beta]}\psi \land \sigma^{\ell} &\equiv&
\displaystyle{\bigwedge} \big\{ \Box(\ol{\delta},\psi) \mid \ol{\delta} \in \Pset{\ell|_{\beta}}{\beta} \big\} \land \sigma^{\ell}.
\end{array}
\]
Furthermore, since $\Pset{\ell|_{\alpha}}{\alpha} = \Pset{\ell}{\alpha}$ and  $\Pset{\ell|_{\beta}}{\beta} = \Pset{\ell}{\beta}$, we have
\[
\begin{array}{llll}
[\alpha]\psi \land \sigma^{\ell} &\equiv&
\displaystyle{\bigwedge} \big\{ \Box(\ol{\delta},\psi) \mid \ol{\delta} \in \Pset{\ell}{\alpha} \big\} \land \sigma^{\ell}, \text{ and} \\
{[\beta]}\psi \land \sigma^{\ell} &\equiv&
\displaystyle{\bigwedge} \big\{ \Box(\ol{\delta},\psi) \mid \ol{\delta} \in \Pset{\ell}{\beta} \big\} \land \sigma^{\ell}.
\end{array}
\]
On the basis of this, we calculate as follows:
\begin{align*}
 & [\gamma]\psi \land \sigma^{\ell}
\\ & = [\alpha \cup \beta]\psi \land \sigma^\ell
\\ & \equiv ([\alpha]\psi \land [\beta]\psi) \land \sigma^\ell
\\ & \equiv \big([\alpha]\psi \land \sigma^{\ell}\big) \land
     \big([\beta]\psi \land \sigma^{\ell}\big)
\\ & \equiv \Big(\bigwedge \big\{ \Box(\ol{\delta},\psi) \mid
    \ol{\delta} \in \Pset{\ell}{\alpha} \big\}
	\land \sigma^{\ell} \Big)
    \land
	\Big(\bigwedge \big\{ \Box(\ol{\delta},\psi) \mid
    \ol{\delta} \in \Pset{\ell}{\beta} \big\}
	\land \sigma^{\ell} \Big)
\\ & \equiv \bigwedge \big\{ \Box(\ol{\delta},\psi) \mid
    \ol{\delta} \in \Pset{\ell}{\alpha} \cup \Pset{\ell}{\beta} \big\}
	\land \sigma^{\ell}
\\ & \equiv \bigwedge \big\{ \Box(\ol{\delta},\psi) \mid
    \ol{\delta} \in \Pset{\ell}{\gamma} \big\}
        \land \sigma^{\ell}
\\ & = \bigwedge \Bset{\ell}{\gamma}{\psi}
     \land \sigma^{\ell} \hspace{6.5cm} \text{by \autoref{l:boxHelperTP}.}
\end{align*}

\item[Case $\gamma=\alpha;\beta$]
Let $\ell \in \TT(\gamma)$. By the induction hypothesis and \autoref{l:boxHelperTP}, we have that
\[
\begin{array}{lll}
[\alpha][\beta]\psi \land \sigma^{\ell|_{\alpha}} &\equiv&
\bigwedge \big\{ \Box(\ol{\delta},[\beta]\psi) \mid
    \ol{\delta} \in \Pset{\ell|_{\alpha}}{\alpha} \big\}
	\land \sigma^{\ell|_{\alpha}}, \text{ and}
\\ {[\beta]}\psi \land \sigma^{\ell|_{\beta}} &\equiv&
\bigwedge \big\{ \Box(\ol{\delta},\psi) \mid
    \ol{\delta} \in \Pset{\ell|_{\beta}}{\beta} \big\}
	\land \sigma^{\ell|_{\beta}}.
\end{array}
\]
Using again \factRef{fact:tp1} and $\Pset{\ell|_{\alpha}}{\alpha} = \Pset{\ell}{\alpha}$ and  $\Pset{\ell|_{\beta}}{\beta} = \Pset{\ell}{\beta}$, we have
\[
\begin{array}{lll}
[\alpha][\beta]\psi \land \sigma^{\ell} &\equiv&
\bigwedge \big\{ \Box(\ol{\delta},[\beta]\psi) \mid
    \ol{\delta} \in \Pset{\ell}{\alpha} \big\}
	\land \sigma^{\ell}, \text{ and}
\\ {[\beta]}\psi \land \sigma^{\ell} &\equiv&
\bigwedge \big\{ \Box(\ol{\delta},\psi) \mid
    \ol{\delta} \in \Pset{\ell}{\beta} \big\}
	\land \sigma^{\ell}.
\end{array}
\]
We now argue as follows:
\begin{align*} 
\hspace{.6em}[\gamma]\psi \land \sigma^{\ell}
& \equiv [\alpha][\beta]\psi \land \sigma^{\ell}
\\ & \equiv \bigwedge \big\{ \Box(\ol{\delta},[\beta]\psi) \mid
    \ol{\delta} \in \Pset{\ell}{\alpha} \big\}
	\land \sigma^{\ell} & \text{as argued above}
\\ & \equiv
    \Big(\bigwedge \big\{ \Box(\ol{\delta},[\beta]\psi) \mid
      \ol{\delta} \in \Pset{\ell}{\alpha} \setminus \{ \emptylist \} \big\}
	  \land \sigma^{\ell} \Big)
\\ & \hspace*{4mm} \land
	\Big(\bigwedge \big\{ \Box(\emptylist,[\beta]\psi) \mid
      \emptylist \in \Pset{\ell}{\alpha} \big\}
	  \land \sigma^{\ell} \Big)
\\ & \equiv \Big(\bigwedge \big\{ \Box(\ol{\delta}\beta,\psi) \mid
    \ol{\delta} \in \Pset{\ell}{\alpha} \setminus \{ \emptylist \} \big\}
	\land \sigma^{\ell} \Big)
\\ & \hspace*{4mm}	\land
	\Big(\bigwedge \big\{ [\beta]\psi \mid
    \emptylist \in \Pset{\ell}{\alpha} \big\}
	\land \sigma^{\ell} \Big)
\\ & \equiv \Big(\bigwedge \big\{ \Box(\ol{\delta}\beta,\psi) \mid
    \ol{\delta} \in \Pset{\ell}{\alpha} \setminus \{ \emptylist \} \big\}
	\land \sigma^{\ell} \Big)
\\ & \hspace*{4mm}	\land
	\Big(\bigwedge \big\{ \Box(\ol{\delta},\psi) \mid
    \ol{\delta} \in \Pset{\ell}{\beta}, \emptylist \in \Pset{\ell}{\alpha} \big\}
	\land \sigma^{\ell} \Big) & \text{as argued above}
\\ & \equiv \bigwedge \big\{ \Box(\ol{\eta},\psi) \mid
    \ol{\eta} \in \Pset{\ell}{\gamma} \big\}
	\land \sigma^{\ell}  & \text{by definition of $\PsetName^\ell$}
\\ & \equiv \bigwedge \Bset{\ell}{\gamma}{\psi}
     \land \sigma^{\ell} & \text{by \autoref{l:boxHelperTP}.}
\end{align*}

\item[Case $\gamma=\beta^{\ast}$]
By the induction hypothesis for $\beta$ we have:
\begin{equation}\label{obs:lbtStarIH}
[\beta]\phi \land \sigma^{\ell} \equiv
\textstyle{\bigwedge}\Bset{\ell}{\beta}{\phi} \land \sigma^{\ell},
\text{for all $\phi$ and all $\ell \in \TT(\beta)$},
\end{equation}
and, as we saw before, by \autoref{l:localBoxTruthAux} we obtain the following equivalence, for all $\phi$:
\begin{equation}\label{equiv:IHbetaThm}
{[\beta]}\phi  \equiv
\bigvee \big\{
\textstyle{\bigwedge}\Bset{\ell}{\beta}{\phi} \mid \ell \in \TT(\beta) \big\}.
\end{equation}

We now have to show that \claimRef{claim:bxeq} holds for $\gamma = \beta^*$.
Note that $\TT(\beta^\ast) = \TT(\beta)$, so we fix any test profile $\ell \in \TT(\beta)$ and any formula $\psi$.
We claim that
\begin{equation}\label{claim:bxToShow}
[\beta^\ast]\psi \land \sigma^{\ell} \equiv
\textstyle{\bigwedge}\Bset{\ell}{\beta^\ast}{\psi} \land \sigma^{\ell}.
\end{equation}

For the implication from left to right of \claimRef{claim:bxToShow}, by unfolding the definition of $\Bset{\ell}{\beta^\ast}{\psi}$ we only need to show that
$[\beta^\ast]\psi \land \sigma^\ell \models
\bigwedge \big\{ \Box(\ol{\alpha},\psi) \mid \ol{\alpha}
\in \Pset{\ell}{\beta^\ast} \big\}$.
Take an arbitrary program list $\ol{\alpha} \in \Pset{\ell}{\beta^{\ast}}$.
If $\ol{\alpha} = \emptylist$ then $\Box(\ol{\alpha},\psi) = \psi$,
and so we have $[\beta^{\ast}]\psi \models \Box(\ol{\alpha},\psi)$.
Otherwise $\ol{\alpha}$ by \autoref{l:boxHelperTermination}(b) must be of the form $\ol{\delta}\beta^{\ast}$, for some $\ol{\delta} \neq \emptylist$ in $\Pset{\ell}{\beta}$.
In this case, we apply \observationRef{obs:lbtStarIH} where we take $[\beta^\ast]\psi$ as the formula $\phi$, and obtain, for all $\ell' \in \TT(\beta)$:
\begin{align*}
[\beta][\beta^\ast]\psi \land \sigma^{\ell'}
&\equiv
\textstyle{\bigwedge}\Bset{\ell'}{\beta}{[\beta^\ast]\psi} \land \sigma^{\ell'}
\\ &=
\textstyle{\bigwedge}
\Fset{\ell'}{\beta}
\land
\textstyle{\bigwedge}
\big\{ \Box(\ol{\beta},[\beta^\ast]\psi) \mid \ol{\beta} \in \Pset{\ell'}{\beta} \big\}
\land \sigma^{\ell'}
\\ &\models
\textstyle{\bigwedge}
\big\{ \Box(\ol{\beta},[\beta^\ast]\psi) \mid \ol{\beta} \in \Pset{\ell'}{\beta} \big\}
\land \sigma^{\ell'}.
\end{align*}
Here the last step follows because $\sigma^{\ell'}$ implies $\bigwedge \Fset{\ell'}{\beta}$.

Now:
\begin{align*}
[\beta^\ast]\psi \land \sigma^\ell
& \models
  ([\beta][\beta^\ast]\psi \land \top) \land \sigma^{\ell} \\
& \equiv
  ( [\beta][\beta^\ast]\psi \land \bigvee_{\ell'} \sigma^{\ell'} ) \land \sigma^\ell \\
& \models
\bigvee_{\ell'} ( [\beta][\beta^\ast]\psi \land \sigma^{\ell'} ) \land \sigma^\ell \\
& \models
\bigvee_{\ell'} (
\textstyle{\bigwedge}
\big\{ \Box(\ol{\beta},[\beta^\ast]\psi) \mid \ol{\beta} \in \Pset{\ell'}{\beta} \big\} \land \sigma^{\ell'}
) \land \sigma^{\ell} \\
& \models
\bigvee_{\ell'} \Big(
\textstyle{\bigwedge}
\big\{ \Box(\ol{\beta},[\beta^\ast]\psi) \mid \ol{\beta} \in \Pset{\ell'}{\beta} \big\} \land \sigma^{\ell'}
 \land \sigma^{\ell} \Big) \\
& \models
\textstyle{\bigwedge}
\big\{ \Box(\ol{\beta},[\beta^\ast]\psi) \mid \ol{\beta} \in \Pset{\ell}{\beta} \big\}
\\
& \models
\Box(\ol{\delta}\beta^{\ast},\psi)
\end{align*}
where the second step uses \factRef{fact:test-profile-complete}, the third step uses distributivity, and the penultimate step uses \factRef{fact:test-profile-orthogonal}.

This finishes the proof that $[\beta^\ast]\psi \land \sigma^\ell \models
\bigwedge \big\{ \Box(\ol{\alpha},\psi) \mid \ol{\alpha} \in \Pset{\ell}{\beta^\ast} \big\}$, and hence the direction from left to right of \claimRef{claim:bxToShow}.

For the implication from right to left of \claimRef{claim:bxToShow} more work is required.
Let $\rho$ be the formula
\[
\rho \isdef \bigvee
  \Big\{ \textstyle{\bigwedge}\Bset{\ell'}{\beta^{\ast}}{\psi}
  \mid \ell' \in \TT(\beta^{\ast}) \Big\}.
\]
Our key claim is that
\begin{equation}
\label{claim:rho}
\rho \models [\beta^\ast]\psi.
\end{equation}

To see why this suffices to prove the right-to-left implication of \claimRef{claim:bxToShow}, observe that we have $\bigwedge \Bset{\ell}{\beta^\ast}{\psi} \models \rho$ by definition of $\rho$.
But then by \claimRef{claim:rho} we obtain $\bigwedge \Bset{\ell}{\beta^\ast}{\psi} \models [\beta^\ast]\psi$, and from this and \factRef{fact:equiv_iff_TPequiv} we arrive at the required $\bigwedge \Bset{\ell}{\beta^\ast}{\psi} \land \sigma^\ell \models [\beta^\ast]\psi \land \sigma^\ell$.

Our aim then is to show $\rho \models [\beta^\ast]\psi$ by applying \autoref{l:guard}.
For this, let $x$ be a fresh proposition letter that does not occur in $\beta$.
We now need to find suitable formulas $\chi_{0}$ and $\chi_{1}$ such that
  $[\beta]x \equiv (x \land \chi_{0}) \lor \chi_{1}$.
From \equivalenceRef{equiv:IHbetaThm} we have
\begin{align*}
{[\beta]}x & \equiv
\bigvee \big\{
\textstyle{\bigwedge}\Bset{\ell}{\beta}{x} \mid \ell \in \TT(\beta) \big\}.
\end{align*}
For any $\ell \in \TT(\beta)$ we will write
\[
\phi_{\ell} \isdef \bigwedge \big\{ \Box(\ol{\alpha},x) \mid \ol{\alpha}
\in \Pset{\ell}{\beta} \big\},
\]
so that by definition of $\Bset{\ell}{\beta}{x}$ we find
\begin{align*}
{[\beta]}x & \equiv
\bigvee \big\{
\textstyle{\bigwedge}\Fset{\ell}{\beta} \land \phi_{\ell}
   \mid \ell \in \TT(\beta) \big\}.
\end{align*}

Now partition the set $\TT(\beta)$ into $T_{0}$ and $T_{1}$ as follows:
\[\begin{array}{lll}
T_{0} &\isdef& \{ \ell \in \TT(\beta) \mid \emptylist \in \Pset{\ell}{\beta}\}, \text{ and}
\\
T_{1} &\isdef& \{ \ell \in \TT(\beta) \mid \emptylist \not\in \Pset{\ell}{\beta}\},
\end{array}\]
so that we have $\TT(\beta) = T_{0} \cup T_{1}$ and $T_{0} \cap T_{1} = \emptyset$.

Clearly, this gives
\[
{[\beta]}x \equiv
\bigvee_{\ell\in T_{0}} ({\textstyle{\bigwedge}}\Fset{\ell}{\beta} \land \phi_{\ell})
\lor
\bigvee_{\ell\in T_{1}} (\textstyle{\bigwedge}\Fset{\ell}{\beta} \land \phi_{\ell}).
\]

Now abbreviate
\[
\phi'_{\ell} \isdef \bigwedge \big\{ \Box(\ol{\alpha},x) \mid \ol{\alpha}
\in \Pset{\ell}{\beta} \setminus \{ \emptylist \} \big\},
\]
then we have
\[
\phi_{\ell} \equiv
   \left\{\begin{array}{ll}
   x \land \phi'_{\ell} & \text{if } \ell \in T_{0}
   \\ \phi'_{\ell}   & \text{otherwise}
\end{array}\right.
\]
so that using distributivity we obtain
\begin{align*}
{[\beta]}x & \equiv
\bigvee_{\ell\in T_{0}}
   ({\textstyle{\bigwedge}}\Fset{\ell}{\beta} \land x \land \phi'_{\ell})
\lor
\bigvee_{\ell\in T_{1}}
  (\textstyle{\bigwedge}\Fset{\ell}{\beta} \land \phi'_{\ell})
\\ & \equiv
\big(x \land \bigvee_{\ell\in T_{0}}
   ({\textstyle{\bigwedge}}\Fset{\ell}{\beta} \land \phi'_{\ell}) \big)
\lor
\bigvee_{\ell\in T_{1}}
  (\textstyle{\bigwedge}\Fset{\ell}{\beta} \land \phi'_{\ell}).
\end{align*}
Hence, if we define, for $i \in \{0, 1\}$:
\[
\chi_{i} \isdef
\bigvee_{\ell\in T_{i}}
  (\textstyle{\bigwedge}\Fset{\ell}{\beta} \land \phi'_{\ell}),
\]
then we obtain
\begin{equation}\label{obs:xx0}
[\beta]x \equiv (x \land \chi_{0}) \lor \chi_{1},
\end{equation}
as required.

It is left to verify that with this definition of $\chi_{0}$ and $\chi_{1}$, the other conditions for \autoref{l:guard} are also satisfied.
For starters, it is easy to check that
\begin{equation}\label{obs:xx1}
\rho \models \psi.
\end{equation}
To see this, recall that $\varepsilon$ belongs to the set $\Pset{\ell'}{\beta^*}$, for every $\ell'$.
It follows that the formula $\Box(\varepsilon,\psi) =\psi$ belongs to $\Bset{\ell'}{\beta^{\ast}}{\psi}$, for every $\ell' \in \TT(\beta^*)$, and from this \observationRef{obs:xx1} is immediate by definition of $\rho$.

The key claim then is the following:
\begin{equation}\label{claim:xx2}
\rho \models \subst{\substAt{x}{\rho}}{(\chi_{0}\lor\chi_{1})}.
\end{equation}
To prove this, by the definition of $\rho$ it suffices to show
\[
\bigwedge \Bset{\ell}{\beta^{\ast}}{\psi} \models  \subst{\substAt{x}{\rho}}{(\chi_{0}\lor\chi_{1})},
\]
for every test profile $\ell \in \TT(\beta^{\ast})$.
Fix such an $\ell$, then we find
\begin{align*}
{\textstyle{\bigwedge}} \Bset{\ell}{\beta^{\ast}}{\psi}
& = {\textstyle{\bigwedge}} \Fset{\ell}{\beta^{\ast}} \land
    {\textstyle{\bigwedge}} \big\{ \Box(\ol{\delta},\psi)
	  \mid \ol{\delta} \in \Pset{\ell}{\beta^{\ast}} \big\}
\\ & \models {\textstyle{\bigwedge}} \Fset{\ell}{\beta} \land
    {\textstyle{\bigwedge}} \big\{ \Box(\ol{\alpha}\beta^{\ast},\psi)
	  \mid \ol{\alpha} \in \Pset{\ell}{\beta} \setminus \{ \emptylist \}
	  \big\} \hspace{1em} \text{ by def.~of $\FsetName$ and $\PsetName^\ell$}
\\ & = {\textstyle{\bigwedge}} \Fset{\ell}{\beta} \land
    {\textstyle{\bigwedge}} \big\{ \Box(\ol{\alpha},[\beta^{\ast}]\psi)
	  \mid \ol{\alpha} \in \Pset{\ell}{\beta} \setminus \{ \emptylist \}
	  \big\}
\\ & = {\textstyle{\bigwedge}} \Fset{\ell}{\beta}
        \land (\subst{\substAt{x}{[\beta^\ast]\psi}}{\phi'_{\ell}})
\\ & \models \subst{\substAt{x}{[\beta^\ast]\psi}}{(\chi_{0} \lor \chi_{1})}
\\ & \models \subst{\substAt{x}{\rho}}{(\chi_{0} \lor \chi_{1})}
\end{align*}
Notice that the last step follows from the left-to-right direction of \claimRef{claim:bxToShow} (shown before) and \autoref{l:localBoxTruthAux}(b), which together imply $[\beta^*]\psi\models \rho$.

With \observationRef{obs:xx0,obs:xx1}, and \claimRef{claim:xx2} it follows immediately from \autoref{l:guard} that $\rho \models  [\beta^{\ast}]\psi$, as required.
\qedhere
\end{description}
\end{proof}

With \autoref{l:localBoxTruth} we reached the main goal of this section: local soundness and local invertibility of the $(\Box)$ rule which we will define in the next section.
We now provide one more lemma about boxes that will be needed in the proof of the Truth Lemma (\autoref{l:truthLemma}).
Informally, \autoref{l:FP} says the following: no matter how a Kripke model executes a program, the definitions of $\FsetName$ and $P$ will provide a list of failed tests and a list of programs that match the program execution in the model.

\begin{lem}[\checklean{existsBoxFP}]\label{l:FP}
Let $\gamma$ be a program, let $\kmodel = (W,R,V)$ be a Kripke model, and let $v$ and $w$ be states in $\kmodel$ such that $(v,w) \in R_{\gamma}$.
Furthermore, let $\ell \in \TT(\gamma)$ be such that $\kmodel, v \Vdash \bigwedge \Fset{\ell}{\gamma}$.
Then we find $(v,w) \in R_{\ol{\delta}}$ for some $\ol{\delta} \in \Pset{\ell}{\gamma}$.
\end{lem}

\begin{proof}
We prove the lemma by induction on the shallow program structure of $\gamma$.

\textbf{Base case.}

\begin{description}
\item[\it Case $\gamma = a$]
  Suppose $(v,w) \in R_a$.
  Note that $\Pset{\ell}{\gamma} = \{ a \}$ for any $\ell$.
  Hence we can choose $\ol{\delta} \isdef a$ to get $(v,w) \in R_{\ol{\delta}}$ and are done.
\item[\it Case $\gamma = {\tau?}$]
Here  $(v,w) \in R_{\gamma}$ implies that $v=w$ and $\kmodel, v \Vdash \tau$.
Now assume that $\ell \in \TT(\tau?)$ is such that $\kmodel, v \Vdash \bigwedge \Fset{\ell}{\tau?}$.
It cannot be the case that $\tau \notin \ell$ since then $\Fset{\ell}{\tau?} = \{ \neg\tau \}$ would imply $\kmodel, v \Vdash \lnot\tau$.
This means that $\tau \in \ell$ and we are fine since now $\emptylist \in
\Pset{\ell}{\tau?}$ and $(v,w) \in R_{\emptylist}$.
\end{description}

\textbf{Induction step.}
\begin{description}
\item[\it Case $\gamma = \alpha\cup\beta$]
If $\kmodel, v \Vdash \bigwedge \Fset{\ell}{\gamma}$ then since $\Fset{\ell}{\gamma} = \Fset{\ell}{\alpha} \cup \Fset{\ell}{\beta}$, obviously $\kmodel, v \Vdash \bigwedge \Fset{\ell}{\alpha}$ and $\kmodel, v \Vdash \bigwedge \Fset{\ell}{\beta}$.
Furthermore, $(v,w) \in R_{\gamma}$ implies that $(v,w) \in R_{\alpha}$ or $(v,w) \in R_{\beta}$; without loss of generality, we may assume the first.
Then by the induction hypothesis on $\alpha$ we may infer that $(v,w) \in
R_{\ol{\delta}}$ for some $\ol{\delta} \in \Pset{\ell}{\alpha}$.
Clearly, then by the definition of $\Pset{\ell}{\gamma}$ this $\ol{\delta}$ also belongs
to $\Pset{\ell}{\gamma}$.
\item[\it Case $\gamma = \alpha; \beta$]
Let $u$ be such that $(v,u) \in R_{\alpha}$ and $(u,w) \in R_{\beta}$.
By definition of $\Fset{\ell}{\gamma}$ we have $\kmodel, v \Vdash \bigwedge
\Fset{\ell}{\alpha}$, so that by the induction hypothesis on $\alpha$ we find
$(v,u) \in R_{\ol{\delta}}$ for some $\ol{\delta} \in \Pset{\ell}{\alpha}$.
We make a further case distinction.

If $\ol{\delta} = \emptylist$ then we have $v=u$ and so $(v,w) \in R_{\beta}$.
Furthermore, by definition of $\Fset{\ell}{\gamma}$ we have $\kmodel, v \Vdash
\bigwedge \Fset{\ell}{\beta}$.
Hence, by induction hypothesis on $\beta$ we find $(v,w) \in R_{\ol{\eta}}$,
for some $\ol{\eta}\in \Pset{\ell}{\beta}$.
It follows that $\ol{\eta}\in \Pset{\ell}{\gamma}$, and so we are done.

If $\ol{\delta} \neq \emptylist$, then we have $(v,w) \in
R_{\ol{\delta}}\seqc R_{\beta} = R_{\ol{\delta}\beta}$, and this suffices since
$\ol{\delta}\beta \in \Pset{\ell}{\gamma}$.

\item[\it Case $\gamma=\beta^{\ast}$]
Since $(v,w) \in R_{\beta^{\ast}}$ we can make the following case
distinction.

If $v = w$ then we have $(v,w) \in R_{\emptylist}$ and we are done since
$\emptylist \in \Pset{\ell}{\beta^{\ast}}$.

If $v \neq w$ then there is some $u \neq v$ such that $(v,u) \in R_{\beta}$ and
$(u,w) \in R_{\beta^{\ast}}$.
Now let $\ell \in \TT(\beta^{\ast})$ be such that $\kmodel, v \Vdash
\bigwedge \Fset{\ell}{\beta^{\ast}}$.
Since $\Fset{\ell}{\beta^{\ast}} = \Fset{\ell}{\beta}$ this implies, by the induction hypothesis on $\beta$, that $(v,u) \in R_{\ol{\delta}}$, for some
$\ol{\delta} \in \Pset{\ell}{\beta}$.
Note that $\ol{\delta} \neq \emptylist$ since $v \neq u$.
It then follows that $\ol{\delta}\beta^{\ast} \in \Pset{\ell}{\beta^{\ast}}$ and so
we are done since $(v,w) \in R_{\ol{\delta}} \seqc R_{\beta^{\ast}} =
R_{\ol{\delta}\beta^{\ast}}$.
\qedhere
\end{description}
\end{proof}

\subsection{Diamonds}\label{subsec:Diamonds}

Analogously to the previous section about boxes, in this section our goal is to locally unfold programs in ``diamonds''.
That is, given a formula $\lnot[\alpha]\psi$, we will define a finite set  $\unfold_{\diam}(\alpha,\psi)$ of sequents such that
\[
\neg[\alpha]\psi \equiv
\bigvee \Big\{ \textstyle{\bigwedge} \Gamma \mid \Gamma \in
\unfold_{\diam}(\alpha,\psi) \Big\},
\]
and for every formula $\phi$ in $\bigcup \unfold_{\diam}(\alpha,\psi)$, we have either $\phi \in \Test(\alpha)$ or $\phi$ is of the form $\lnot[a][\alpha_{1}]\cdots[\alpha_{n}]\psi$, where the atomic program $a$ and all $\alpha_{i}$ are in $\Prog(\alpha)$.

As a corollary, the following rule will be locally invertible and locally sound:
\begin{center}
  \AxiomC{$\Delta , \lnot[\alpha]\phi$}
  \LeftLabel{$(\diam)$}
  \UnaryInfC{$
    \big\{\Delta, \Gamma \mid \Gamma \in \unfold_{\diam}(\alpha,\phi)
    \big\}
    $}
  \DisplayProof
\end{center}

We again give some intuitions with respect to the well-known reduction axioms from \autoref{d:pdl-logic}.
Each $\Gamma \in \unfold_{\diam}(\alpha,\psi)$ corresponds to a disjunct in the disjunctive normal form obtained by rewriting
$\neg [\alpha]\psi$ exhaustively using the reduction axioms $\Ax{?}$, $\Ax{\cup}$, $\Ax{;}$, $\Ax{^*}$ while ensuring
that formulas of the form $[\beta^*]\phi$ are only rewritten once and treating test formulas and $\psi$ as atoms.
To prevent infinite local unfolding, disjuncts containing a formula of the form $\neg[\beta^\ast]\phi$ that has been reduced before will not be represented by any $\Gamma \in \unfold_{\diam}(\alpha,\psi)$.
Inspecting the axioms, we see that multiple disjuncts in the normal form arise from applications of $\Ax{\cup}$ and $\Ax{^*}$,
and that every $\Gamma \in \unfold_{\diam}(\alpha,\psi)$ will consist of exactly one formula of the form
$\lnot[a][\alpha_{1}]\cdots[\alpha_{n}]\psi$ (for a possibly empty list $a \alpha_1 \cdots \alpha_n$) together with the test formulas (zero or more) that must be true for $a \alpha_1 \cdots \alpha_n$ to be a successful execution of $\alpha$.
We can therefore represent such a sequent as a pair $(X,\ol{\delta})$
where $X \subseteq \Test(\alpha)$ and $\ol{\delta}$ is a list of programs from $\Prog(\alpha)$.
We compute the set $\Dset{\alpha}$ of these pairs recursively using the function $\DsetName$ defined below.
Note that the clause for $\Dset{\alpha;\beta}$ ensures that $\beta$ is reduced iff $\alpha$ does not entail executing an atomic program.
Note also that in the clause for $\Dset{\alpha^{\ast}}$, the requirement $\ol{\delta} \neq \emptylist$ corresponds to the exclusion
of sequents that contain a regenerated formula of the form $\neg[\beta^\ast]\phi$.

\begin{defi}[\lean{Dset}]\label{d:Dset}
For all programs $\alpha$, we recursively define $\Dset{\alpha}$ as a set of pairs $(X, \ol{\delta})$
where $X \subseteq \Test(\alpha)$ and $\ol{\delta}$ is a list of programs from $\Prog(\alpha)$:
\[\begin{array}{lll}
\Dset{a} &\isdef& \big\{ (\emptyset, a) \big\}
\\ \Dset{\tau?} &\isdef& \big\{ ( \{ \tau \}, \emptylist )\big\}
\\ \Dset{\alpha\cup\beta} &\isdef& \Dset{\alpha} \cup \Dset{\beta}
\\ \Dset{\alpha;\beta} &\isdef&
   \big\{(X,\ol{\delta}\beta ) \mid (X,\ol{\delta}) \in \Dset{\alpha}, \ol{\delta} \neq
      \emptylist \big\}
\\ && \cup
   \big\{(X \cup Y, \ol{\delta} ) \mid (X,\emptylist) \in \Dset{\alpha},
      (Y, \ol{\delta}) \in \Dset{\beta} \big\}
\\ \Dset{\alpha^{\ast}} &\isdef& \big\{ (\emptyset,\emptylist) \big\} \cup
     \big\{ (X, \ol{\delta}\alpha^{\ast})
	   \mid (X,\ol{\delta}) \in \Dset{\alpha}, \ol{\delta} \neq
	   \emptylist \big\}
\end{array}\]
\end{defi}

\begin{defi}[\lean{unfoldDiamond}]\label{d:unfoldDiamond}
For all programs $\alpha$ and all formulas $\psi$, we define:
\[
\unfold_{\diam}(\alpha,\psi) \isdef
\{ X \cup \{ \lnot\Box(\ol{\delta},\psi) \} \mid (X, \ol{\delta}) \in \Dset{\alpha} \}
\]
\end{defi}

We illustrate the definition of $\DsetName$ and $\unfold_{\diam}$ with a few examples.

\begin{exa}\label{ex:H-dia-double-star}
To illustrate the exclusion of sequents to prevent infinite unfolding via the requirement $\ol\delta \neq \emptylist$,
consider the program ${(a^\ast)}^\ast$.
Note that we have $\Dset{a^\ast} = \{(\emptyset,\emptylist), (\emptyset,aa^*)\}$ and therefore we get
\[
  \begin{array}{rcl}
    \Dset{{(a^*)}^*}
        &=& \{ (\emptyset,\emptylist) \} \cup  \{ (X, \ol{\delta}{(a^{\ast})}^\ast)
	   \mid (X,\ol{\delta}) \in \Dset{a^*}, \ol{\delta} \neq \emptylist \}
       \\
        &=& \{ (\emptyset,\emptylist) \} \cup  \{ (X, \ol{\delta}{(a^{\ast})}^\ast)
	   \mid (X,\ol{\delta}) \in \{(\emptyset,\emptylist), (\emptyset,aa^*)\}, \ol{\delta} \neq \emptylist \}
       \\
        &=& \{ (\emptyset,\emptylist) \} \cup  \{ (X, \ol{\delta}{(a^{\ast})}^\ast)
	   \mid (X,\ol{\delta}) \in \{ (\emptyset,aa^*)\} \}
       \\
       &=& \{(\emptyset,\emptylist), (\emptyset,aa^*{(a^*)}^*) \}
  \end{array}
\]
and hence for all $\psi$, $\unfold_{\diam}({(a^\ast)}^\ast, \psi) =  \{ \{\neg\psi\}, \{\neg[a][a^\ast][{(a^\ast)}^\ast]\psi \}\}$, cf.~\autoref{fig:double-star}.
\end{exa}

\begin{exa}\label{ex:H-comp-star}
To illustrate how the clause for $\Dset{\alpha;\beta}$ reduces the program $\beta$ if and only if $\alpha$ does not entail the execution of an atomic program, consider the program $(p? \cup a); b^\ast$.
\[
  \begin{array}{rcl}
      \Dset{ (p? \cup a);b^\ast }
        & = & \{ (X, \ol{\delta} b^*) \mid (X, \ol{\delta}) \in \Dset{p? \cup a}, \ol{\delta} \neq \emptylist\}\\
            &&  \cup \;\; \{ (X \cup Y, \ol{\delta}) \mid (X, \emptylist) \in \Dset{p? \cup a}, (Y, \ol\delta) \in \Dset{b^*}\}\\
        & = & \{ (X, \ol{\delta} b^*) \mid (X, \ol{\delta}) \in \{(\{p\}, \emptylist), (\emptyset,a)\}, \ol{\delta} \neq \emptylist\}\\
            &&  \cup \;\; \{ (X \cup Y, \ol{\delta}) \mid (X, \emptylist) \in \{(\{p\}, \emptylist), (\emptyset,a)\}, (Y, \ol\delta) \in \{ (\emptyset,\emptylist), (\emptyset,bb^*)\}\}\\
        & = & \{ (\emptyset, ab^*)\} \cup \{(\{p\}, \emptylist), (\{p\}, bb^*)\}
  \end{array}
 \]
hence  for all $\psi$, $\unfold_{\diam}((p? \cup a);b^\ast, \psi) = \{ \{ \neg[a][b^*]\psi\} , \{p, \neg\psi\}, \{p, \neg[b][b^*]\psi\}\}$.
\end{exa}

\begin{exa}\label{ex:H}
We illustrate with the program ${(p? \cup a)}^\ast$ how nesting of ${?}$ and $^*$ leads to exclusion of a sequent via the requirement $\ol\delta \neq \emptylist$:
  \[
    \begin{array}{rcl}
      \Dset{ {(p? \cup a)}^\ast } & = & \{ (\emptyset,\emptylist) \} \cup  \{ (X, \ol{\delta}{(p? \cup a)}^{\ast})
                                     \mid (X,\ol{\delta}) \in \Dset{(p? \cup a)}, \ol{\delta} \neq \emptylist \} \\
                               & = & \{ (\emptyset,\emptylist) \} \cup  \{ (X, \ol{\delta}{(p? \cup a)}^{\ast})
                                     \mid (X,\ol{\delta}) \in \Dset{p?} \cup \Dset{a} , \ol{\delta} \neq \emptylist \}\\
                               & = & \{ (\emptyset,\emptylist) \} \cup  \{ (X, \ol{\delta}{(p? \cup a)}^{\ast})
                                     \mid (X,\ol{\delta}) \in \{(\{p\},\emptylist), (\emptyset,a)\} , \ol{\delta} \neq \emptylist \}\\
                               & = & \{ (\emptyset,\emptylist) \} \cup  \{ (X, \ol{\delta}{(p? \cup a)}^{\ast})
                                     \mid (X,\ol{\delta}) = (\emptyset,a) \}\\
                               & = & \{ (\emptyset,\emptylist), (\emptyset, a {(p? \cup a)}^{\ast}) \}
    \end{array}
  \]
\autoref{l:localDiaTruth} below, applied to $\neg[{(p? \cup a)}^\ast]\psi$, will then tell us that for all $\psi$:
\[ \lnot[{(p? \cup a)}^\ast]\psi \ \equiv \ \lnot \psi \lor \lnot[a][{(p? \cup a)}^\ast]\psi. \]
\end{exa}

We also note the following fact about \autoref{d:Dset} that will later be useful to show the termination of our system and interpolation.

\begin{fact}\label{f:terminationHelperDiamond}
For all $(X, \ol{\delta}) \in \Dset{\alpha}$,
if $\ol{\delta} \neq \emptylist$, then the first element of $\ol{\delta}$ is an atomic program.
\end{fact}

The following lemma will also be used to show termination of our system.

\begin{lem}[\checklean{unfoldDiamond_in_FL}]\label{l:unfoldDiamondInFL}
  Let $\alpha$ be a program and let $\psi$ and $\phi$ be formulas such that $\phi \in \bigcup \unfold_{\diam}(\alpha, \psi)$.
  Then we have $\phi \in \FL(\lnot[\alpha]\psi)$ (as in \autoref{d:fischerLadner}).
\end{lem}
\begin{proof}
By inspection of \autoref{d:unfoldDiamond} and \autoref{d:Dset}.
\end{proof}

Later we will use \autoref{d:unfoldDiamond} only for non-atomic programs to ensure termination of our tableau system, but the following \autoref{l:localDiaTruth} also holds for atomic programs.

\begin{lem}[\checklean{localDiamondTruth}]\label{l:localDiaTruth}
For all programs $\alpha$ and all formulas $\psi$, we have:
\[
\neg[\alpha]\psi \equiv \bigvee
  \Big\{ \textstyle{\bigwedge} X \land \lnot\Box(\ol{\delta},\psi) \mid (X,\ol{\delta}) \in \Dset{\alpha}
\Big\}.
\]
\end{lem}

We will prove \autoref{l:localDiaTruth} by induction on the shallow program structure of $\alpha$.
In the inductive case where $\alpha = \beta^{\ast}$ we shall need the following lemma.

\begin{lem}[\checklean{guardToStarDiamond}]\label{l:guardToStarDiamond}
Suppose we have two formulas $\chi_0$ and $\chi_1$, a~program $\beta$ and a proposition letter $x$ that does not occur in $\beta$ (but may occur in $\chi_0$ and in $\chi_1$), such that $\lnot[\beta]\lnot x \equiv (x \land \chi_0) \lor \chi_1$.
Then for any two formulas $\rho$ and $\psi$ such that $\subst{\substAt{x}{\rho}}{\chi_1} \models \rho$ and $\lnot\psi \models \rho$ we have $\lnot[\beta^{\ast}]\psi \models \rho$.
\end{lem}

\begin{proof}
Assume that $\kmodel, w \Vdash \lnot[\beta^{\ast}]\psi$, then we need to show that $\kmodel, w \Vdash \rho$.
Note that since $\kmodel, w \Vdash \lnot[\beta^{\ast}]\psi$ there is a $\beta$-path from $w$ to some state where $\psi$ does not hold, and hence, by assumption, $\rho$ does hold.
Therefore, it suffices to prove the following claim.
\begin{equation}\label{claim:dd1}
\text{For all states $u, v$ in $\kmodel$, if } \kmodel, v \Vdash \rho \text{ and } (u,v) \in R_{\beta} \text{ then } \kmodel, u \Vdash \rho.
\end{equation}
To prove this, assume $\kmodel, v \Vdash \rho$ and $(u,v) \in R_{\beta}$.
It follows that $\kmodel, u \Vdash \neg[\beta]\neg\rho$, and so by the assumption on the $\chi$-formulas we find that $\kmodel, u \Vdash (\rho \land \subst{\substAt{x}{\rho}}{\chi_{0}}) \lor \subst{\substAt{x}{\rho}}{\chi_{1}}$.

This means that we either have $\kmodel, u \Vdash \rho \land \subst{\substAt{x}{\rho}}{\chi_{0}}$, in which case we find $\kmodel, u \Vdash \rho$ immediately, or $\kmodel, u \Vdash \subst{\substAt{x}{\rho}}{\chi_{1}}$, in which case we find $\kmodel, u \Vdash \rho$ by the assumption that $\subst{\substAt{x}{\rho}}{\chi_{1}} \models \rho$.
This shows that in both cases we have $\kmodel, u \Vdash \rho$, as required.
\end{proof}

We will also need the following fact, which can easily be verified given \autoref{d:Dset} of $\Dset{\alpha}$ above.
\begin{fact}[\checklean{keepFreshDset}]\label{f:keepFreshDset}
For any proposition letter $x$ and program $\alpha$, if $x$ does not occur in $\alpha$ then $x$ also does not occur in \Dset{\alpha} in the sense that for all $(X, \ol{\delta}) \in \Dset{\alpha}$, for all $\phi \in X$ and all $\beta \in \ol{\delta}$, $x$ occurs in neither $\phi$ nor $\beta$.
\end{fact}

\begin{proof}[Proof of \autoref{l:localDiaTruth}]
By induction on the shallow program structure of $\alpha$.

\textbf{Base Case.}
\begin{description}
\item[Case $\alpha=a$]
In this case we find $\lnot[a]\psi \equiv \top \land \lnot[a]\psi \equiv \bigwedge\emptyset \land \lnot\Box(a,\psi)$, as required.
\item[Case $\alpha={\tau?}$]
Here we obtain $\lnot[\tau?]\psi \equiv \tau \land \lnot\psi \equiv \tau \land \neg\Box(\emptylist,\psi)$.
\end{description}

Induction hypothesis:
Suppose that for two programs $\beta$ and $\gamma$ we have, for all formulas $\psi$:
\begin{gather}
  \tag{IH $\beta$}
  \neg[\beta]\psi \equiv \bigvee
  \Big\{ \textstyle{\bigwedge} X \land \lnot\Box(\ol{\delta},\psi) \mid (X,\ol{\delta}) \in \Dset{\beta}
  \Big\},
  \\
  \tag{IH $\gamma$}
  \neg[\gamma]\psi \equiv \bigvee
  \Big\{ \textstyle{\bigwedge} X \land \lnot\Box(\ol{\delta},\psi) \mid (X,\ol{\delta}) \in \Dset{\gamma}
  \Big\}.
\end{gather}

\textbf{Induction step.}
\begin{description}
\item[Case $\alpha=\beta\cup\gamma$]
  First we use that $\Dset{\beta  \cup \gamma} = \Dset{\beta} \cup  \Dset{\gamma}$, then we can rewrite the disjuncts using the IH for $\beta$ and $\gamma$, so we are done.
\item[Case $\alpha=\beta;\gamma$] This case follows from the following calculation:
\begin{align*}
\hspace{1em}\lnot[\beta;\gamma]\psi & \equiv \lnot[\beta][\gamma]\psi
\\ & \equiv \bigvee \big\{ {\textstyle{\bigwedge} X} \land
        \lnot\Box(\ol{\delta},[\gamma]\psi)
		\mid (X,\ol{\delta}) \in \Dset{\beta}\big\} & \text{by (IH $\beta$)}
\\ & \equiv \bigvee \big\{ {\textstyle{\bigwedge} X} \land
        \lnot\Box(\ol{\delta},[\gamma]\psi)
		\mid (X,\ol{\delta}) \in \Dset{\beta}, \ol{\delta} \neq \emptylist \big\} &
\\ & \hspace*{5mm} \lor \bigvee \big\{ {\textstyle{\bigwedge} X} \land
        \lnot\Box(\emptylist,[\gamma]\psi)
		\mid (X,\emptylist) \in \Dset{\beta} \big\}
\\ & \equiv \bigvee \big\{ {\textstyle{\bigwedge} X} \land
        \lnot\Box(\ol{\delta}\gamma,\psi)
		\mid (X,\ol{\delta}) \in \Dset{\beta}, \ol{\delta} \neq \emptylist \big\}
\\ & \hspace*{5mm} \lor \bigvee \big\{ {\textstyle{\bigwedge} X} \land
        \lnot[\gamma] \psi
		\mid (X,\emptylist) \in \Dset{\beta} \big\}
\\ & \equiv \bigvee \big\{ {\textstyle{\bigwedge} X} \land
        \lnot\Box(\ol{\delta}\gamma,\psi)
		\mid (X,\ol{\delta}) \in \Dset{\beta}, \ol{\delta} \neq \emptylist \big\}
\\ & \hspace*{5mm} \lor \bigvee \big\{
        {\textstyle{\bigwedge} X} \land {\textstyle{\bigwedge} Y} \land
		\lnot\Box(\ol{\lambda},\psi)
		\mid (X,\emptylist) \in \Dset{\beta}, (Y,\ol{\lambda}) \in \Dset{\gamma} \big\} & \text{by (IH $\gamma$)}
\\ & \equiv \bigvee \big\{ {\textstyle{\bigwedge} Z} \land
        \lnot\Box(\ol{\eta},\psi)
		\mid (Z,\ol{\eta}) \in \Dset{\alpha} \big\}.
\end{align*}

\item[Case $\alpha=\beta^\ast$]
Define
\[
\rho \isdef \bigvee \big\{
    {\textstyle{\bigwedge} X} \land \lnot \Box(\ol{\delta},\psi)
   \mid (X,\ol{\delta}) \in \Dset{\beta^{*}} \big\}.
\]
Then our goal will be to show that $\lnot[\beta^{\ast}]\psi \equiv \rho$.

We first show that $\rho \models \lnot[\beta^{\ast}]\psi$ as we will use it for the other direction.
For this it suffices to show that each of the disjuncts of $\rho$ implies $\lnot[\beta^{\ast}]\psi$, so take an arbitrary pair $(X,\ol{\delta}) \in \Dset{\beta^{\ast}}$.
In case $(X,\ol{\delta}) = (\emptyset,\emptylist)$ we find
$\bigwedge X \land \lnot \Box(\ol{\delta},\psi) \equiv \top \land \neg\psi$, and this formula obviously implies $\lnot[\beta^{\ast}]\psi$.
If $(X,\ol{\delta}) \neq (\emptyset,\emptylist)$, then $\ol{\delta}$ by definition of $\Dset{\beta^\ast}$ is of the form $\ol{\gamma}\beta^{\ast}$ with $\ol{\gamma} \neq \emptylist$ and $(X, \ol{\gamma}) \in \Dset{\beta}$.
In this case we find
\begin{align*}
{\textstyle{\bigwedge} X} \land \lnot \Box(\ol{\delta},\psi)
   & = {\textstyle{\bigwedge} X} \land \lnot \Box(\ol{\gamma},[\beta^{\ast}]\psi)
\\ & \models \neg [\beta][\beta^{\ast}]\psi & \text{by (IH $\beta$)}
\\ & \models \neg [\beta^{\ast}]\psi.
\end{align*}

In order to prove that $\lnot[\beta^{\ast}]\psi \models \rho$ we will apply \autoref{l:guardToStarDiamond}.
For this, let $x$ be a fresh proposition letter not occurring in $\beta$ and therefore, by \autoref{f:keepFreshDset}, also not occurring in $\Dset{\beta}$.
Define
\[\begin{array}{lll}
\phi_{0} &\isdef& \bigvee \big\{ {\textstyle{\bigwedge} X}
     \mid (X,\emptylist) \in \Dset{\beta} \big\}, \text{ and}
\\ \phi_{1} &\isdef& \bigvee \big\{
     {\textstyle{\bigwedge} X} \land \lnot\Box(\ol{\delta},\neg x)
     \mid (X,\ol{\delta}) \in \Dset{\beta}, \ol{\delta} \neq \emptylist \big\}.
\end{array}\]
It now suffices to show that these formulas meet the conditions of \autoref{l:guardToStarDiamond}.

First, note that $\phi_0$ and $\phi_1$ are based on disjoint subsets of $\Dset{\beta}$ for the empty and non-empty program sequences, respectively.
This is similar to $T_0$ and $T_1$ in the proof of \autoref{l:localBoxTruth}.
It follows from (IH $\beta$) and distributivity that $\lnot[\beta]\lnot x \equiv (x \land \phi_{0}) \lor \phi_{1}$.
Second, we have
$\lnot\psi
= \lnot\Box(\emptylist,\psi)
\equiv \bigwedge \emptyset \land \lnot\Box(\emptylist,\psi)
\models \rho$
because $(\emptyset,\emptylist) \in \Dset{\beta^{\ast}}$.
It thus remains to show the last condition of \autoref{l:guardToStarDiamond}, namely that $\subst{\substAt{x}{\rho}}{\phi_{1}} \models \rho$; this we prove as follows, where the first equality holds because $x$ does not occur in $\Dset{\beta}$.
\begin{align*}
\hspace{.8em}\subst{\substAt{x}{\rho}}{\phi_{1}} & =
\bigvee \big\{
     {\textstyle{\bigwedge} X} \land \lnot\Box(\ol{\delta}, \neg\rho)
     \mid (X,\ol{\delta}) \in \Dset{\beta}, \ol{\delta} \neq \emptylist \big\}
   &
\\ & \models
\bigvee \big\{
     {\textstyle{\bigwedge} X} \land \lnot\Box(\ol{\delta}, \neg\neg[\beta^{\ast}]\psi)
     \mid (X,\ol{\delta}) \in \Dset{\beta}, \ol{\delta} \neq \emptylist \big\}
   & \text{by $\rho \models \lnot[\beta^{\ast}]\psi$}
\\ & \models
\bigvee \big\{
     {\textstyle{\bigwedge} X} \land \lnot\Box(\ol{\delta},[\beta^{\ast}]\psi)
     \mid (X,\ol{\delta}) \in \Dset{\beta}, \ol{\delta} \neq \emptylist \big\}
\\ & =
\bigvee \big\{
     {\textstyle{\bigwedge} X} \land \lnot\Box(\ol{\delta}\beta^{\ast},\psi)
     \mid (X,\ol{\delta}) \in \Dset{\beta}, \ol{\delta} \neq \emptylist \big\}
\\ & \models
\bigvee \big\{
     {\textstyle{\bigwedge} X} \land \lnot\Box(\ol{\eta},\psi)
     \mid (X,\ol{\eta}) \in \Dset{\beta^{\ast}} \big\}
   & \text{by def.~of~$\Dset{\beta^\ast}$}
  \\ & = \rho. & & \qedhere
\end{align*}
\end{description}
\end{proof}

This finishes the main result of this section.
Only later (for \autoref{l:loadedDiamondPaths}) will we need the following lemma.
As it is a simple result about diamonds and our $\DsetName$ function, we include it here.
Note that we use $R_{\ol{\delta}}$ from \autoref{d:relateSeq}.
Informally, \autoref{l:existsDiamondDset} is the diamond analogue of \autoref{l:FP} for boxes.
It says that no matter how a Kripke model executes a program, some element of $\Dset{\gamma}$ will match the program execution in the model.

\begin{lem}[\checklean{existsDiamondDset}]\label{l:existsDiamondDset}
Let $\gamma$ be a program, let $\kmodel = (W,R,V)$ be a Kripke model,
and let $v$ and $w$ be states in $\kmodel$ such that $(v,w) \in R_{\gamma}$.
Then there is a pair $(X,\ol{\delta}) \in \Dset{\gamma}$ such that $\kmodel, v \Vdash X$
and $(v,w) \in R_{\ol{\delta}}$.
\end{lem}
\begin{proof}
By a straightforward induction on the shallow program structure of $\gamma$.
We cover only two cases of the induction step and refer to the Lean proof for full details.

\begin{description}
\item[Case $\gamma = \alpha; \beta$]
It follows from $(v, w) \in R_{\gamma}$ that there is some $u$ such that
$(v,u) \in R_{\alpha}$ and $(u,w) \in R_{\beta}$.
By the induction hypothesis on $\alpha$ we find a pair $(X,\ol{\delta}) \in \Dset{\alpha}$
such that $\kmodel, v \Vdash X$ and $(v,u) \in R_{\ol{\delta}}$.

Now distinguish cases.
If $\ol{\delta} \neq \emptylist$ then we find
$(X,\ol{\delta}\beta) \in \Dset{\gamma}$ and $(v,w) \in R_{\ol{\delta}\beta}$
as required.
On the other hand, if $\ol{\delta} = \emptylist$ we have $v = u$;
furthermore,
by the induction hypothesis it follows from $(u,w) \in R_{\beta}$ that there
is a pair $(Y,\ol{\eta}) \in \Dset{\beta}$ such that $\kmodel, v \Vdash Y$ and
$(v,w) \in R_{\ol{\eta}}$.
Combining all of this we find $(X\cup Y, \ol{\eta}) \in \Dset{\gamma}$,
$\kmodel, v \Vdash X \cup Y$ and $(v,w) \in R_{\ol{\eta}}$, as required.

\item[Case $\gamma = \beta^{\ast}$]
It follows from $(v,w) \in R_{\beta^{\ast}}$ that there are two possibilities.
If $v = w$ then we can take the pair $(\emptyset,\emptylist)\in
\Dset{\beta^{\ast}}$:
It follows that $\kmodel, v \Vdash \bigwedge\emptyset$ and
$(v,w) \in R_{\emptylist}$.

If, on the other hand, $v \neq w$ then there is some $v_{1} \neq v$ such that
$(v,v_{1}) \in R_{\beta}$ and $(v_{1},w) \in R_{\beta^{\ast}}$.
By the induction hypothesis there is some pair $(X,\ol{\delta}) \in \Dset{\beta}$
such that $\kmodel, v \Vdash X$ and $(v,v_{1}) \in R_{\ol{\delta}}$.
Note that we must have $\ol{\delta} \neq\emptylist$ since $v \neq v_{1}$.
This implies that the pair $(X,\ol{\delta}\beta^{\ast})$ belongs to
$\Dset{\beta^{\ast}}$.
Furthermore, we already saw that $\kmodel, v \Vdash X$ and we have
$(v,w) \in R_{\ol{\delta}}\seqc R_{\beta^{\ast}} = R_{\ol{\delta}\beta^{\ast}}$.
In other words, we are done.
\qedhere
\end{description}
\end{proof}

\subsection{Loaded Diamonds}\label{subsec:LoadedDiamonds}

The unfolding of boxes and diamonds that we defined in the previous two subsections is sufficient to deal with local repeats.
However, repeats can also be generated by non-local, that is, modal reasoning.
To recognize these situations in our tableau system, we will \emph{load} diamond formulas before applying the modal rule.

\begin{defi}[\lean{NegLoadFormula}]\label{d:NegLoadFormula}
  The set $\loaded{\progs}$ of \emph{loaded programs} and the set $\loaded{\pdlforms}$ of \emph{loaded formulas} are given by:
  \begin{align*}
    \loaded{\progs} \ni \beta &\Coloneqq \loaded{\alpha}, \\
    \loaded{\pdlforms} \ni \xi &\Coloneqq [\loaded{\alpha}]\phi \mid [\loaded{\alpha}]\xi
  \end{align*}
  where $\phi \in \pdlforms$ and $\alpha \in \progs$.

  A sequent is \emph{loaded} if it contains a loaded formula, otherwise it is \emph{free}.
  A node in a tableau is loaded (free) if it is labeled with a loaded (free) sequent.

  Given any (possibly) loaded formula $\xi$, we let $\xi^-$ denote its \emph{unloaded version} obtained by replacing each $\loaded{\alpha}$ with $\alpha$, without changing anything else and leaving the structure of the formula as it is.
  Given a set $\Gamma \subseteq \pdlforms \cup \loaded{\pdlforms}$, we will let $\Gamma^-$ denote the set of formulas we obtain by unloading every formula in $\Gamma$, that is, $\Gamma^- \isdef \{\xi^- \mid \xi \in \Gamma\}$.
\end{defi}

For example, $[\loaded{a}][\loaded{b \cup c^\ast}]p$ is a loaded formula and its unloaded version is $[a][b \cup c^\ast]p$.
All loaded formulas are of the form
$[\loaded{\alpha_1}]\dots[\loaded{\alpha_{n+1}}]\phi$ where
$\phi \in \pdlforms$, $\loaded{\alpha_i} \in \loaded{\progs}$ for $i \leq n + 1$,
for some $n \geq 0$.
For simplicity we sometimes also write $\loaded{[\alpha_{1}]\cdots[\alpha_{n}]}\phi$
to mean a formula $[\loaded{\alpha_{1}}][\loaded{\alpha_{2}}]\cdots[\loaded{\alpha_{n}}]\phi$.

The loading mechanism is purely syntactic.
That is, there is no semantic difference between $\alpha$ and $\loaded{\alpha}$, as the following definition and fact make explicit.

\begin{defi}
  We extend \autoref{d:evaluate} from $\pdlforms$ to $\loaded{\pdlforms}$ and from $\progs$ to $\loaded{\progs}$ by adding two cases.
  First, let $v \Vdash [\loaded{\alpha}] \phi$ iff $v \Vdash [\alpha] \phi$.
  Second, define $\reach{\loaded{\alpha}} \isdef \reach{\alpha}$.
\end{defi}

\begin{fact}\label{f:loadedSemantics}
  Let $\kmodel$ be a Kripke model and $w$ be a state of $\kmodel$.
  We then have $w \Vdash \loaded{[\alpha_{1}]\cdots[\alpha_{n}]}\phi$ if and only if $w \Vdash [\alpha_{1}]\cdots[\alpha_{n}]\phi$.
\end{fact}

\begin{defi}[\lean{loadMulti}]\label{d:loadMulti}
  Given a list $\ol{\alpha}$ of programs and a formula $\phi$,
  we define the loaded formula $\loaded{\Box}(\ol{\alpha},\phi)$ as follows:
  \[
    \begin{array}{lll}
      \loaded{\Box}(\emptylist,\phi)       & \isdef & \phi \\
      \loaded{\Box}(\alpha\ol{\beta},\phi) & \isdef & [\loaded{\alpha}]\loaded{\Box}(\ol{\beta},\phi)
    \end{array}
  \]
\end{defi}

For example, if $\ol{\alpha}$ is the list [$a$, $b^\ast$] then $\loaded{\Box}(\ol{\alpha}, [c]p) = [\loaded{a}][\loaded{b^\ast}][c]p$.

\begin{defi}[\lean{unfoldDiamondLoaded}]\label{d:unfoldDiamondLoaded}
For all programs $\alpha$ and all possibly loaded formulas $\xi$, we define:
\[
\loaded{\unfold}_{\diam}(\alpha,\xi) \isdef
\{ X \cup \{ \lnot\loaded{\Box}(\ol{\delta},\xi) \} \mid (X, \ol{\delta}) \in \Dset{\alpha} \}
\]
\end{defi}

The following lemma is analogous to \autoref{l:unfoldDiamondInFL}, and we will use it to prove termination of our system.

\begin{lem}\label{l:unfoldLoadedDiamondInFL}
  Let $\alpha$ be a program and let $\psi$ and $\phi$ be formulas such that $\phi \in \bigcup \loaded{\unfold}_{\diam}(\alpha, \psi)$.
  Then we have $\phi^- \in \FL(\lnot[\alpha]\psi)$ (as in \autoref{d:fischerLadner}).
\end{lem}
\begin{proof}
  By inspection of \autoref{d:unfoldDiamondLoaded} and \autoref{d:Dset}.
\end{proof}

\subsection{Local Rules}\label{subsec:LocalRules}

Recall from \autoref{d:isBasic} that a set of formulas is called \emph{basic} if it only contains formulas of the forms $\bot$, $\lnot\bot$, $p$, $\lnot p$, $[a]\phi$ and $\lnot [a]\phi$.

\begin{defi}[\lean{LocalRule}]\label{d:localTableau}
A \emph{local PDL-tableau} for $\Gamma$ is a tableau over the following
rules, called the local rules:

\begin{center}
  \AxiomC{$\Delta , \lnot \lnot \phi$}
  \LeftLabel{$(\lnot)$}
  \UnaryInfC{$\Delta,\phi$}
  \DisplayProof
  \hspace{1.5em}
  \AxiomC{$\Delta , \phi \land \psi$}
  \LeftLabel{($\land$)}
  \UnaryInfC{$\Delta,\phi,\psi$}
  \DisplayProof
  \hspace{1.5em}
  \AxiomC{$\Delta , \lnot (\phi \land \psi)$}
  \LeftLabel{$(\lnot \land)$}
  \UnaryInfC{$\Delta , \lnot \phi \splitCase \Delta , \lnot \psi$}
  \DisplayProof

  \smallskip

  \AxiomC{$\Delta , [\alpha]\phi$}
  \LeftLabel{$(\Box)$}
  \RightLabel{\ \ $\alpha$ non-atomic}
  \UnaryInfC{$\big\{\Delta, \Gamma
      \mid  \Gamma \in \unfold_{\square}(\alpha, \phi) \big\}$}
  \DisplayProof

  \smallskip

  \AxiomC{$\Delta , \lnot[\alpha]\phi$}
  \LeftLabel{$(\diam)$}
  \RightLabel{\ \ $\alpha$ non-atomic}
  \UnaryInfC{$\{\Delta, \Gamma \mid \Gamma \in \unfold_{\diam}(\alpha,\phi)\}$}
  \DisplayProof

  \smallskip

  \AxiomC{$\Delta , \lnot\loaded{[\alpha]}\xi$}
  \LeftLabel{$(\loaded{\diam})$}
  \RightLabel{\ \ $\alpha$ non-atomic}
  \UnaryInfC{$\big\{\Delta, \Gamma \mid \Gamma \in
      \loaded{\unfold}_{\diam}(\alpha, \xi)\big\}$}
  \DisplayProof

  \smallskip

  \AxiomC{$\Delta , \lnot[a][\alpha_1]\ldots[\alpha_n]\phi$}
  \LeftLabel{$(L+)$}
  \RightLabel{\ \ $\Delta$ free and basic and $n \geq 0$ maximal}
  \UnaryInfC{$\Delta , \lnot\loaded{[a][\alpha_1]\ldots[\alpha_n]}\phi$}
  \DisplayProof

  \smallskip

  \AxiomC{$\Delta , \lnot\loaded{[a][\alpha_1]\ldots[\alpha_n]}\phi$}
  \LeftLabel{$(L-)$}
  \RightLabel{\ \ $\Delta$ basic}
  \UnaryInfC{$\Delta , \lnot[a][\alpha_1]\ldots[\alpha_n]\phi$}
  \DisplayProof
\end{center}
where $\phi, \psi \in \pdlforms$ are formulas, $\alpha_1, \dots, \alpha_n \in \progs$ are programs, $a$ is an atomic program, $\xi \in \pdlforms \cup \loaded{\pdlforms}$, i.e., $\xi$ is a possibly loaded formula, and $\Delta$ is a possibly empty set of possibly loaded formulas.

Furthermore, if $t$ is a child of a node $s$ and $L(s) = (L+)$ then $L(t) \neq (L-)$.
That is, we are not allowed to apply $(L-)$ immediately after $(L+)$.
\end{defi}

Notice that the rules involving $\unfold$ are only applicable when the program $\alpha$ is non-atomic.
Furthermore, notice that the $(L+)$ rule loads \emph{all} boxes at the front of the principal formula.
For example, given $\lnot[a][b]p$ we will obtain $\lnot\loaded{[a][b]}p$ and not $\lnot\loaded{[a]}[b]p$.
That is, we always load maximally.

We will use the following fact implicitly throughout the rest of the article.

\begin{fact}[\checklean{basic_iff_noLocalRuleApp}]\label{f:basicIffNoLocalRule}
A set $\Gamma$ of formulas is basic if and only if no local rules other than $(L+)$ and $(L-)$ can be applied to $\Gamma$.
\end{fact}

\begin{lem}[\checklean{localRuleTruth}]\label{l:localRuleTruth}
All local rules are locally sound and locally invertible.
That is, for any local rule with parent $\Delta$ and children $\Gamma_1,\ldots,\Gamma_n$,
and any pointed model $(\kmodel, w)$, we have
$\kmodel, w \Vdash \Delta$ if and only if there is an $i$ with $\kmodel, w \Vdash \Gamma_i$.
\end{lem}
\begin{proof}
  For the propositional rules $(\lnot)$, $(\land)$ and $(\lnot\land)$ this is easy to check.
  The rules $(L+)$ and $(L-)$ do not affect the semantic meaning of the sequent, so for them the claim holds as well.
  There remain the three modal rules.
  For the $(\Box)$ and the $(\diam)$ rules, use, respectively, \autoref{l:localBoxTruth} and \autoref{l:localDiaTruth}.
  For the $(\loaded{\diam})$ rule, we also use \autoref{l:localDiaTruth} via \autoref{f:loadedSemantics}.
\end{proof}

Next, we introduce a measure on formulas and sequents that we use to show that all local PDL-tableaux are finite and that we will also use in the soundness proof.

\begin{defi}[\lean{lmOfFormula}]\label{d:lmOfFormula}
    Let $m \colon \pdlforms \cup \loaded{\pdlforms} \to \mathbb{N}$ be the following measure:
    \begin{align*}
    m(\bot) &\isdef 0
    \\
    m(\lnot\bot) &\isdef 0
    \\
    m(p) &\isdef 0
    \\
    m(\lnot p) &\isdef 0
    \\
    m(\lnot \lnot \phi) &\isdef 1 + m(\phi)
    \\
    m(\phi \land \psi) &\isdef 1 + m(\phi) + m(\psi)
    \\
    m(\lnot(\phi \land \psi)) &\isdef 1 + \max(m(\lnot \phi), m(\lnot \psi))
    \\
    m([a]\phi) &\isdef 0
    \\
    m(\lnot[a]\phi) &\isdef 0
    \\
    m([\alpha]\phi) &\isdef 1 + m(\phi) + \sum_{\tau \in \Test(\alpha)} m(\lnot \tau)
    \\
    m(\lnot[\alpha]\phi) &\isdef 1 + m(\neg\phi) + \sum_{\tau \in \Test(\alpha)} m(\tau)
    \\
    m(\lnot\loaded{[\alpha_1]\dots[\alpha_n]}\phi) &\isdef  m(\lnot[\alpha_1]\dots[\alpha_n]\phi)
    \end{align*}
where $\phi$ and $\psi$ are PDL formulas, $p$ is a proposition letter, $a$ is an atomic program, $\alpha$ is a non-atomic program, and $\alpha_1, \ldots, \alpha_n$ are arbitrary programs.
Recall that sequents are finite sets.
We thus extend the domain of~$m$ to sequents as follows:
    \[
    m(\Gamma) \isdef \sum_{\psi \in \Gamma} m(\psi).
    \]
\end{defi}

We stress that $m(\cdot)$ from \autoref{d:lmOfFormula} is \emph{not} meant to be a generally useful measure of formula complexity, as for example we have $m(\lnot\lnot p) = 1$ but $m([a]\lnot\lnot p) = 0$.
Intuitively, $m(\phi)$ counts how many steps we can make by applying local rules to $\phi$.
No local rule can be applied to $[a]\lnot\lnot p$, and end nodes of local PDL-tableaux will have measure zero.

\begin{lem}[\checklean{measureProp}]\label{l:propertiesOfm}
    The following properties hold:
    \begin{enumerate}[(a)]
        \item $m(\lnot\lnot\phi) > m(\phi)$,
        \item $m(\phi_1 \land \phi_2) > m(\phi_1) + m(\phi_2)$,
        \item $m(\lnot(\phi_1 \land \phi_2)) > m(\lnot \phi_i)$ for all $i \in \{1, 2\}$,
        \item if $\alpha$ is not atomic,
          then $m([\alpha]\phi) > m(\Gamma)$ for all $\Gamma \in \unfold_{\square}(\alpha, \phi)$,
        \item if $\alpha$ is not atomic,
          then $m(\lnot[\alpha]\phi) > m(\Gamma)$ for all $\Gamma \in \unfold_{\diam}(\alpha, \phi)$.
    \end{enumerate}
\end{lem}
\begin{proof}
The properties (a), (b), (c) are immediate by the definition of $m$.

We now show (d).
Let $\Gamma \in \unfold_{\square}(\alpha, \phi)$.
By \autoref{d:unfoldBox} of $\unfold_{\square}(\alpha, \phi)$, it follows that
$\Gamma = \Fset{\ell}{\alpha} \cup \big\{ \Box(\ol{\delta},\phi) \mid \ol{\delta} \in \Pset{\ell}{\alpha} \big\}$,
for some $\ell \in \TT(\alpha)$. 
By \autoref{l:boxHelperTermination}, for all $\ol{\delta} \in \Pset{\ell}{\alpha}$, either $\ol{\delta}=\emptylist$, in which case $m(\Box(\ol{\delta},\phi)) = m(\phi)$, or the first element in $\ol{\delta}$ is an atomic program, in which case $m(\Box(\ol{\delta},\phi)) = 0$. 
Hence $m(\big\{ \Box(\ol{\delta},\phi) \mid \ol{\delta} \in \Pset{\ell}{\alpha} \big\}) \leq m(\phi)$. 
Furthermore, note that $\Fset{\ell}{\alpha} \subseteq \{ \neg\tau \mid \tau \in \Test(\alpha)\}$.
Consequently, 
\[ m(\Gamma)
  \ \leq \ m(\phi) + \sum_{\tau \in \Test(\alpha)} m(\lnot\tau) 
  \ < \ 1 + m(\phi) + \sum_{\tau \in \Test(\alpha)} m(\lnot\tau) 
  \ = \ m([\alpha]\phi),
\]
which shows (d). Note that the last equality above holds because $\alpha$ is not atomic.

For (e), let $\Gamma \in \unfold_{\diam}(\alpha, \phi)$.
Then, by the definition of $\unfold_{\diam}$, there exists $(X, \ol{\delta}) \in \Dset{\alpha}$ such that $\Gamma = X \cup \{\lnot\Box(\ol{\delta},\phi)\}$.
Note that $X \subseteq \Test(\alpha)$ by the definition of $\Dset{\alpha}$. 
If $\ol{\delta} = \emptylist$ then $m(\lnot\Box(\ol{\delta},\phi)) = m(\neg\phi)$.
Otherwise, by \autoref{f:terminationHelperDiamond}, the leading program (i.e., the first element of $\ol{\delta}$ in $\lnot\Box(\ol{\delta},\phi)$) is atomic, and hence $m(\lnot\Box(\ol{\delta},\phi)) = 0 \leq m(\neg\phi)$.
In both cases it follows that:
\[ m(\Gamma)
  \ = \ m(\lnot\Box(\ol{\delta},\phi)) + m(X) 
  \ \leq \ m(\neg\phi) + m(X)
  \ < \ 1 + m(\neg\phi) + \hspace{-8pt} \sum_{\tau \in \Test(\alpha)} \hspace{-10pt}m(\tau)\hspace{3pt}
  \ = \ m(\lnot [\alpha]\phi).
\]
This shows (e). Note again that the last equality holds because $\alpha$ is not atomic.
\end{proof}

\begin{lem}[\checklean{localRuleApp.decreases_lm}]\label{l:localMeasureDown}
Let $\tab = (V, \edge, r, \Lambda, L)$ be a tableau over the local rules which has no applications of either $(L+)$ or $(L-)$.
If $s, t \in V$ are such that $s \edge t$, then $m(\Lambda(s)) > m(\Lambda(t))$.
\end{lem}
\begin{proof}
    The key idea behind the proof is that for each rule, \autoref{l:propertiesOfm} provides exactly what is required to show that $m(\Lambda(s)) > m(\Lambda(t))$.

We only consider the case $L(s) = (\square)$; the other cases are analogous.
For this case, we have that $\Lambda(s) = \Delta \cup \{[\alpha]\phi\}$ where $[\alpha]\phi \not \in \Delta$.
We also have that
$\Lambda(t) = \Delta \cup \Gamma$ where $\Gamma \in \unfold_{\square}(\alpha, \phi)$.
From this and \autoref{l:propertiesOfm}(d) it follows that:
\[ m(\Lambda(t)) \ \leq \ m(\Delta) + m(\Gamma) \ < \ m(\Delta) + m([\alpha]\phi) \ = \ m(\Lambda(s)) \qedhere \]
\end{proof}

\begin{lem}\label{l:localTabFinite}
  All local PDL-tableaux are finite.
\end{lem}
\begin{proof}
Recall that sequents are finite sets.
Towards a contradiction, assume there exists an infinite local PDL-tableau for some sequent.
From \autoref{l:unfoldBoxContent}, \autoref{l:unfoldDiamondInFL}, \autoref{l:unfoldLoadedDiamondInFL} and inspection of the rules it follows that all branching is finite, that is, for every node $s$ in a local PDL-tableau, there are finitely many children of $s$.
As all branching is finite, by K\H{o}nig's Lemma there exists an infinite branch ${(s_i)}_{i < \omega}$.

Assume for contradiction that $(L+)$ is applied at some node $s_i$.
This means that the sequent $\Lambda(s_{i+1})$ is basic and loaded, but then by \autoref{f:basicIffNoLocalRule} and the fact that $(L-)$ may not be applied directly after $(L+)$, no local rule would be applicable at node $s_{i+1}$.
Similarly, if the rule $(L-)$ were applied at some node $s_i$, then the sequent $\Lambda(s_{i+1})$ is basic and free, and so the only applicable local rule at node $s_{i+1}$ would be $(L+)$, but we just saw that this is impossible.

It follows that $L(s_i) \neq (L+), (L-)$ for all $i \in \mathbb{N}$.
Then, by \autoref{l:localMeasureDown}, we obtain an infinite descending chain $m(\Lambda(s_0)) > m(\Lambda(s_1)) > \dots$, which provides the desired contradiction.
\end{proof}

\section{PDL-Tableaux}\label{sec:PDLTableaux}

We obtain the system of PDL-tableaux by adding a modal rule to the rules of local tableaux, and adding some loop checking conditions that will ensure that all closed tableaux are finite.

\subsection{Modal Rules and PDL-Tableaux}

Intuitively, the application of a modal rule represents a move from one state to another state in a Kripke model.
A modal rule application is necessary to satisfy atomic diamond formulas, i.e., formulas of the form $\neg[a]\psi$, since such a formula requires the existence of an $a$-successor where $\neg\psi$ is true.
Due to the semantics of box formulas, for any formula of the form $[a]\phi$ in the context of $\neg[a]\psi$, the formula $\phi$ must also be satisfied at this $a$-successor. 
These formulas $\phi$ are obtained via the notion of projection.

\begin{defi}[\lean{projection}]\label{d:projection}
For a sequent $\Gamma$ and an atomic program $a$, we define the \emph{projection} of $\Gamma$ under $a$ as $\Gamma_a \isdef \{ \phi \mid [a]\phi \in \Gamma \}$.
\end{defi}

For example, if $\Gamma \isdef \{[a]p, [b]q, [a^*]r, \lnot [b] r, r\}$ then $\Gamma_a= \{p\}$ and  $\Gamma_b=\{q\}$.
In particular, note that $r \not\in \Gamma_b$ because the projection does not consider any negated formulas.

Since loading as defined in \autoref{d:NegLoadFormula} can extend beyond the leading atomic diamond, the modal rule takes the following form.
\pagebreak[2]

\begin{defi}\label{d:Modalrule}
The \emph{modal rule} is:
\begin{center}
  \AxiomC{$\Delta , \lnot\loaded{[a]}\xi$}
  \LeftLabel{$(M)$}
  \RightLabel{\ \ $\Delta$ basic}
  \UnaryInfC{$\Delta_a , \lnot \xi$}
  \DisplayProof
\end{center}
where $a$ is an atomic program, $\Delta_a$ is from \autoref{d:projection} and $\xi \in \pdlforms \cup \loaded{\pdlforms}$, i.e., $\xi$ is a possibly loaded formula.
\end{defi}

Note that the modal rule can only be applied to formulas starting with a loaded \emph{atomic} diamond and that the context of the rule must be basic.
This amounts to requiring that first all local rules are exhausted before applying a modal rule, which ensures that unsatisfiability due to propositional constraints is not missed. Without this requirement, the system would not be sound.

\begin{defi}[\lean{LoadedPathRepeat}]\label{d:repeat}
Let $\tab$ be a tableau over the local rules (\autoref{d:localTableau}) and the modal rule (\autoref{d:Modalrule}).
A \emph{repeat} is a node $t$ such that $\Lambda(t) = \Lambda(s)$ for some ancestor $s$ of $t$; in this case, the nearest such ancestor to $t$ is called its \emph{companion} and we denote it by $\compan{t}$.
A \emph{\lpr} is a repeat $t$ such that the path from $\compan{t}$ to $t$ (including end points) consists only of loaded nodes.
A \emph{free repeat} is a repeat $t$ where $t$ is a free node, i.e., $\Lambda(t)$ is not loaded.
\end{defi}

Note that local tableaux cannot contain repeats. This follows from observing that:
(i) the rules $(\Box)$, $(\diam)$, and $(\loaded{\diam})$ (by design) eliminate repeated formulas (see the discussion at the start of \autoref{sec:LocalTableaux}), (ii) the rules for the propositional connectives always remove the principal formula,
and (iii) we do not allow applying $(L-)$ immediately after $(L+)$.

\begin{fact}\label{f:lprAreCritical}
  Let $\tab$ be a tableau over the local rules and the modal rule.
  Let $t$ be a \lpr in $\tab$.
  Then on the path from $\compan{t}$ to $t$ the modal rule~$(M)$ is applied at least once.
\end{fact}

In our system, we enforce that free repeats and \lprs are leaves.

\begin{defi}[\lean{Tableau}]\label{d:Tableau}
  The \emph{PDL-rules} are the local rules from \autoref{d:localTableau} together with the modal rule from \autoref{d:Modalrule}.
A \emph{PDL-tableau} is a tableau over the PDL-rules which adheres to the following two stopping conditions:
\begin{enumerate}[label=(T\arabic*)]
  \item[(\lprCond)] every \lpr is a leaf;
  \item[(\freeCond)] every free repeat is a leaf.
\end{enumerate}
\end{defi}
We note that our system would still be sound without any special condition about free repeats, but to ensure termination of proof search (cf.~\autoref{l:pdlTabFinite}), we also require free repeats to be leaves per Condition (\freeCond) from \autoref{d:Tableau}.

One other type of repeat may be encountered: a node $t$ such that $\Lambda(t)$ is loaded and the path from $\compan{t}$ to $t$ goes through an unloaded sequent.
There is no stopping condition for such repeats so the construction of the PDL-tableau proceeds in such cases.

\begin{defi}[\lean{Sequent.closed}]\label{d:closedNode}
A sequent $\Gamma$ is \emph{closed} if we have $\bot \in \Gamma$, or $\{\phi,\lnot\phi\}\subseteq \Gamma$ for some~$\phi$.
A sequent that is not closed is called \emph{open}.
A node $t$ in a tableau is \emph{closed} if $\Lambda(t)$ is closed, and $t$ is \emph{open} if $\Lambda(t)$ is open.
\end{defi}

\begin{defi}[\lean{inconsistent}, \lean{consistent}]\label{def-tab-closed}
A PDL-tableau is \emph{closed} if each of its leaves is either closed or a \lpr.
A PDL-tableau is \emph{open} if it is not closed.
A sequent $\Gamma$ is \emph{inconsistent} if there is a closed PDL-tableau for $\Gamma$.
Otherwise $\Gamma$ is said to be \emph{consistent}.
\end{defi}

Now that we have defined open and closed tableaux, we can explain the role of the loading mechanism in our system.
A loaded sequent $\Delta = \neg\loaded{[\alpha_1] \cdots [\alpha_n]}\phi, \Gamma$ can be seen as a promise to show that the loaded formula $\neg[\alpha_1] \cdots [\alpha_n]\phi$, in the context $\Gamma$, is not satisfiable.
That is, we want to show that there is no execution of the loaded program sequence $\alpha \isdef \alpha_1 \seqc \cdots \seqc \alpha_n$ 
that leads to a state that satisfies $\phi$.
This promise can be fulfilled by showing the existence of a (sub)tableau for $\Delta$ in which all leaves are closed leaves or \lprs,
but it will remain unfulfilled if we can construct a (sub)tableau for $\Delta$ which has a free repeat.
In particular, a \lpr represents a looping execution of (part of) $\alpha$. 
In order to fulfil the promise, it is necessary that we exit the loop, and that all exits lead to a closed leaf or another \lpr. 
Taking the perspective that the aim is to refute the satisfiability of $\Delta$ by constructing a closed tableau, using terminology from the literature, we also think of \lprs as \emph{successful repeats} and free repeats as \emph{unsuccessful repeats}.
We provide two examples to illustrate the two kinds of repeats.

\begin{exa}\label{ex:freeRepeat}
  To illustrate why we treat free repeats as open leaves,
  consider the PDL formula $[a^\ast]\lnot[a]p \to p$ which is \emph{not} valid, i.e., its negation is satisfiable.
  We show an open tableau for $\{ [a^\ast]\lnot[a]p \+ \lnot p \}$ in the left side of \autoref{f:freeRepeat}.
  The free repeat is visualized by the dashed arrow labeled with $\infty$.

  \begin{figure}[H]
    \begin{center}
      \begin{tikzpicture}[node distance=4em, >=latex]
        \node (0) {$[a^\ast]\lnot[a]p \+ \lnot p$};
        \node [below of=0] (1) {$\lnot[a]p \+ [a][a^\ast]\lnot[a]p \+ \lnot p$}; \arr{0}{$(\Box)$}{1}
        \node [below of=1] (2) {$\lnot\loaded{[a]}p \+ [a][a^\ast]\lnot[a]p \+ \lnot p$}; \arr{1}{$(L+)$}{2}
        \node [below of=2] (3) {$[a^\ast]\lnot[a]p \+ \lnot p$}; \arr{2}{$(M)$}{3}
        \draw (3.east) edge [dashed, gray, ->, thick, bend right=60, looseness=1.5] node [right] {$\infty$} (0.east);
      \end{tikzpicture}
      \hspace{2cm}
      \begin{tikzpicture}[>=latex, baseline=(hidden.base)]
        \node [draw, circle, minimum size=1cm] (0) {$w$};
        \draw (0.south) edge [->, thick, out=180, in=100, bend right=100, looseness=3] node [right] {$a$} (0.east);
        \node [below of=0, node distance=3cm] (hidden) {};
      \end{tikzpicture}
    \end{center}
    \caption{An open tableau ending in a free repeat and a model satisfying its root.}\label{f:freeRepeat}
  \end{figure}

  In the right side of \autoref{f:freeRepeat}, we show a Kripke model that satisfies the root sequent of the tableau and thus falsifies the formula $[a^\ast]\lnot[a]p \to p$.
  This model can be obtained from the open tableau by making one state~$w$ where $\lnot p$ holds, and adding one $a$-edge from that state to itself --- we do this ad hoc here, a systematic approach will be given in the completeness proof below in \autoref{sec:Completeness}.
  We stress that the companion and the repeat are satisfied by the same state.
  Moreover, on the way from the companion to the repeat in the tableau we have executed \emph{all} loaded programs (in this case: just one).
  This means the self-loop suffices to satisfy the whole sequence, and no other state is needed.
\end{exa}

\begin{exa}\label{ex:lprEx}
  We now show an example of a tableau that has a \lpr as well as an open leaf.
  Consider $[a][a^\ast]p \to [a][a^\ast]q$ which is \emph{not} a validity.
  We show an open tableau for its negation $\{ [a][a^\ast]p, \lnot [a][a^\ast]q \}$ in \autoref{f:lprEx}.
  Note that this tableau is only open because the left branch ends in an open node (whereas a tableau for $\{ [a][a^\ast]p, \lnot [a][a^\ast]p \}$, which is not satisfiable, could have the same shape but would be closed).
  The \lpr is visualized by the dashed arrow labeled with $\comp$.

  \begin{figure}[ht]
    \begin{center}
      \begin{tikzpicture}[node distance=4em, >=latex]
        \node (0) {$[a][a^\ast]p \+ \lnot [a][a^\ast]q$};
        \node [below of=0] (1) {$[a][a^\ast]p \+ \lnot\loaded{[a][a^\ast]}q$}; \arr{0}{$(L+)$}{1}
        \node [below of=1] (2) {$[a^\ast]p \+ \lnot\loaded{[a^\ast]}q$}; \arr{1}{$(M)$}{2}
        \node [below of=2] (3) {$p, [a][a^\ast]p \+ \lnot\loaded{[a^\ast]}q$}; \arr{2}{$(\Box)$}{3}
        \node [below of=3] (4) {};
        \node [left  of=4, node distance=2cm] (4l) {$p, [a][a^\ast]p \+ \lnot q$}; \arr{3}{$(\loaded{\diam})$}{4l}
        \node [right of=4, node distance=2cm] (4r) {$p, [a][a^\ast]p \+ \lnot\loaded{[a][a^\ast]}q$}; \arr{3}{$(\loaded{\diam})$}{4r}
        \node [below of=4r] (5r) {$[a^\ast]p \+ \lnot\loaded{[a^\ast]}q$}; \arr{4r}{$(M)$}{5r}
        \draw (5r.east) edge [dashed, gray, ->, thick, bend right=60, looseness=1.5] node [right] {$\comp$} (2.east);
      \end{tikzpicture}
      \hspace{2cm}
      \begin{tikzpicture}[>=latex, baseline=(hidden.base)]
        \node [draw, circle, minimum size=1cm] (0) {$w$};
        \node [draw, circle, minimum size=1cm, below of=0, node distance=4cm] (1) {$v : p$}; \arr{0}{$a$}{1}
        \draw (1.south) edge [->, thick, out=180, in=100, bend right=100, looseness=3] node [right] {$a$} (1.east);
        \node [below of=1, node distance=3cm] (hidden) {};
      \end{tikzpicture}
    \end{center}
    \caption{An open tableau with a loaded-path repeat and a model satisfying its root.}\label{f:lprEx}
  \end{figure}

  Next to the tableau in \autoref{f:lprEx}, we again show a model that satisfies its root.
  It consists of two states, $w$ which satisfies the sequents of the first two nodes in the tableau and 
  $v$ which satisfies the sequents of all other nodes.
  In the figure, we write $v : p$ to indicate that $v$ satisfies $p$.
  Also here the companion and the repeat are satisfied by the same state.
  Note that the two applications of the $(M)$ rule correspond to the two edges in the model.
  A crucial difference to \autoref{ex:freeRepeat} is that here we need to make not only $p$ true but also $q$ false.
  This information is coming from the left branch, not from the loaded path.
\end{exa}

In spite of the presence of a \lpr in the tableau in \autoref{f:lprEx}, we can still construct a model, since it is possible to exit the looping execution represented by the \lpr and reach an open leaf, which witnesses satisfiability.

Of course these two examples are just that, examples.
But the difference between free repeats and loaded-path repeats will be crucial for
the soundness proof, where in \autoref{l:loadedDiamondPaths} we deal with \lprs,
and for the completeness proof, where we use free repeats to build a model.

The first main result on PDL-tableaux that we prove concerns their finiteness.

\begin{lem}\label{l:pdlTabFinite}
  Every PDL-tableau is finite.
\end{lem}

\begin{proof}
Suppose towards a contradiction that there is an infinite PDL-tableau~$\tab$ with root sequent $\Gamma$.
Since all PDL-rules have a finite number of children, the underlying rooted tree is finitely branching, hence by K\H{o}nig's Lemma it must have an infinite branch.
For each node~$s$ on the branch, $\Lambda(s)$ is a sequent.
Every element of $\Lambda(s)$ is either a formula from the Fischer-Ladner closure of $\Gamma$ or a loaded version of such a formula.
Since the Fischer-Ladner closure is finite,
it follows that the number of distinct labels that can be encountered is finite.
Since any repeat of a free sequent would yield a leaf in the PDL-tableau (\autoref{d:Tableau}), the infinite branch cannot repeat any free sequents.
This means that from some node onwards, every descendant on the branch is loaded, and hence there must be a \lpr.
The latter is a leaf (\autoref{d:Tableau}) so the branch is finite, reaching a contradiction.
\end{proof}

The following theorem states the soundness and completeness of our tableau system.
Soundness is the statement that satisfiability implies consistency (or, equivalently, that no sequent with a closed tableau is satisfiable); completeness is the converse statement that for every unsatisfiable sequent we may find a closed tableau.

\begin{thm}[\checklean{consIffSat}]\label{t:consIffSat}
Let $\Phi \subseteq \pdlforms$ be a sequent.
Then $\Phi$ is consistent if and only if it is satisfiable.
\end{thm}

\begin{proof}
The two directions of the statement will be established in the next two sections: soundness (the direction from right to left) as \autoref{t:soundness} and completeness (the direction from left to right) as \autoref{t:completeness}.
\end{proof}

\subsection{Properties of PDL-Tableaux}\label{subsec:PropertiesOfPDLTableuax}

We finish this section by making some useful observations concerning the shape of PDL-tableaux.

\begin{defi}\label{d:loadedPath}
Let $s_1 \edge \ldots \edge s_n$ be a path of a PDL-tableau.
If all $s_i$ are loaded, then we say that the path is \emph{loaded}, otherwise it is \emph{free}.
\end{defi}

We note that the modal rule $(M)$ is the only additional rule that can be used in PDL-tableaux but not in local tableaux.
Therefore each PDL-tableau $\tab$ can be partitioned into maximal local sub-tableaux that are glued together by applications of the modal rule:
to each leaf of a local sub-tableau that is not already a leaf of $\tab$ the rule $(M)$ is applied to obtain the root of another local sub-tableau.
We illustrate this with a simple example.

\begin{exa}\label{ex:validity}
  We prove the validity
  $([a;b](p \land q) \land [c]\bot) \to ([a;b]q \land [c]r)$
  in our system by constructing a closed tableau for the set
  $\{ [a;b](p \land q) \land [c]\bot, \lnot([a;b]q \land [c]r) \}$
  in \autoref{f:fourLocTab}.
It consists of four local tableaux, connected by three applications of the modal rule, marked with dotted lines.
\end{exa}

\begin{figure}[ht]
  \centering
  \footnotesize
  \begin{tikzpicture}[node distance=3.5em,->]
    \node (0) {
      $ [a;b](p \land q) \land [c]\bot
      , \lnot([a;b]q \land [c]r)$  };
    \node (1) [below of=0] {
      $ [a;b](p \land q)
      , [c]\bot
      , \lnot([a;b]q \land [c]r)$  };
    \draw (0) -- node [right] {$(\land)$} (1);
    \node (2) [below of=1] {
      $ [a][b](p \land q)
      , [c]\bot
      , \lnot([a;b]q \land [c]r)$  };
    \draw (1) -- node [right] {$(\Box)$} (2);
    \node (2-3) [below of=2, node distance = 1.75em] {$(\lnot\land)$};
    \node (3) [below of=2] {};
    \node (3a) [left of=3, node distance=2.3cm] {
      $ [a][b](p \land q)
      , [c]\bot
      , \lnot[a;b]q$ };
    \draw (2) -- (3a);
    \node (4a) [below of=3a] {
      $ [a][b](p \land q)
      , [c]\bot
      , \lnot[a][b]q$ };
    \draw (3a) -- node [right] {$(\diam)$} (4a);
    \node (5a) [below of=4a] {
      $ [a][b](p \land q)
      , [c]\bot
      , \lnot\loaded{[a][b]}q$ };
    \draw (4a) -- node [right] {$(L+)$} (5a);
    \node (6a) [below of=5a] {
      $ [b](p \land q)
      , \lnot\loaded{[b]}q $ };
    \draw [thick,dotted] (5a) -- node [right] {$(M)$} (6a);
    \node (7a) [below of=6a] {
      $ p \land q
      , \lnot q $ };
    \draw [thick,dotted] (6a) -- node [right] {$(M)$} (7a);
    \node (8a) [below of=7a] {
      $ p
      , q
      , \lnot q $ };
    \draw [thick,dotted] (7a) -- node [right] {$(\land)$} (8a);
    \node (8ax) [below of=8a, node distance=1.5em] { $\times$ };
    \node (3b) [right of=3, node distance=2.3cm] {
      $ [a][b](p \land q)
      , [c]\bot
      , \lnot[c]r$  };
    \draw (2) -- (3b);
    \node (4b) [below of=3b] {
      $ [a][b](p \land q)
      , \loaded{[c]}\bot
      , \lnot[c]r$  };
    \draw (3b) -- node [right] {$(L+)$} (4b);
    \node (5b) [below of=4b] {
      $ \bot
      , \lnot r$  };
    \draw [thick,dotted] (4b) -- node [right] {$(M)$} (5b);
    \node (5bx) [below of=5b, node distance=1.5em] { $\times$ };
    \begin{scope}[on background layer,
      highlight/.style = {rounded corners, fill = gray!10, inner ysep=0.5mm, inner xsep=2mm},
      ]
      \node [highlight, fit=(0) (3a) (4a) (4b) (3b)] {};
      \node [highlight, fit=(3a) (5a)] {};
      \node [highlight, fit=(6a) (6a)] {};
      \node [highlight, fit=(7a) (8ax)] {};
      \node [highlight, fit=(5b) (5bx)] {};
    \end{scope}
  \end{tikzpicture}
  \caption{A PDL-tableau consisting of four local tableaux.}\label{f:fourLocTab}
\end{figure}

Recall that we write $s \edge t$ when $s$ is a parent of $t$.
We now define another relation $\cEdge$, which extends $\edge$ by adding
\emph{back edges}, that is, connections \emph{from \lprs to their companions}.
Alternatively, we could say that we now view our tableaux as \emph{cyclic} proofs~\cite{Docherty2019:NonWellfoundedPDL,BGP2012:GenCyclicProver}.

\begin{defi}[\lean{companion}]\label{d:companion}\label{d:cEdge}
Given a PDL-tableau $\tab$, we let $\comp$ be the \emph{companion relation} such
that
\[
s \comp t\ \ \defiff \ \  \text{$s$ is a leaf that is a \lpr with the companion $t$},
\]

and we define the \emph{edge relation} $\cEdge$ by putting:
\[ s \cEdge t \ \ \defiff\ \ s \edge t \text{ or } s \comp t,
\]
that is, $\cEdge = \edge \cup \comp$.
A sequence $t_1 \cEdge \ldots \cEdge t_n$ is called a \emph{$\cEdge$-path};
it is a \emph{proper path} if $n>1$.
Finally, we let $\cEdgeT$ and $\cEdgeRT$ denote, respectively, the transitive
and the reflexive transitive closure of $\cEdge$.
\end{defi}

Clearly we have $s \cEdgeRT t$ iff there is a path from $s$ to $t$, and
accordingly we think of $\cEdgeRT$ as a reachability relation in the underlying
graph  of $\tab$ (that is, the version where the edge relation $\cEdge$
includes back edges).
On the basis of this we define two further relations.

\begin{defi}\label{d:cEquiv}
Let $\tab$ be a PDL-tableau.
We define the relations $\equiv_c$ and $\simpler{}{}$ by putting, for any pair of nodes $s, t \in \tab$:
\[ \begin{array}{lcl}
   s \cEquiv t    & \defiff & s \cEdgeRT t \text{ and } t \cEdgeRT s
\\ \simpler{t}{s} & \defiff & s \cEdgeRT t \text{ but not } t \cEdgeRT s.
\end{array}\]
\end{defi}

Informally, we have $s \cEquiv t$ if there are $\cEdge$-paths from $s$ to $t$
and from $t$ to $s$, and we have $\simpler{t}{s}$ if there is an
$\cEdge$-path from $s$ to $t$, but not vice versa.
The following lemma states some immediate observations.

\begin{lem}[\checklean{eProp}]\label{l:eProp}
Let $\tab$ be a closed PDL-tableau.
Then the following hold:

\begin{enumerate}
\item[(a)]
  the relation $\cEquiv$ is an equivalence relation, and its equivalence classes are node sets of maximal strongly connected components of $\cEdge$;
\item[(b)]
  the relation $\simpler{}{}$ is a strict partial order which is both well-founded and conversely well-founded.
\end{enumerate}
\end{lem}
\begin{proof}
  For part (a), note that by \autoref{d:cEquiv} the relation $\cEquiv$ is transitive, reflexive and symmetric.
  Moreover, for all nodes $s, t$, by the definition of $\cEquiv$, it is immediate that $s \cEquiv t$ iff $s$ and $t$ belong to the same strongly connected component for $\cEdge$.
  Hence the node set of every maximal strongly connected component is an equivalence class. Conversely, every equivalence class is the node set of a maximal strongly connected component due to equivalence classes being disjoint.

  For part (b), first note that by \autoref{d:cEquiv} the relation $\simpler{}{}$ is transitive and asymmetric.
  Using the transitivity of $\cEdgeRT$ we can also show that $\simpler{}{}$ is irreflexive.
  Second, by \autoref{l:pdlTabFinite} there are only finitely many nodes in $\tab$.
  Hence $\simpler{}{}$ is a relation over a finite set and to show that it is (conversely) well-founded it suffices to show that its transitive closure is irreflexive.
  But because  $\simpler{}{}$ is transitive, its transitive closure is $\simpler{}{}$ itself, and we already noted that  $\simpler{}{}$ is irreflexive, so we are done.
\end{proof}

Based on \autoref{l:eProp}(b), we may use the converse of $\simpler{}{}$
as a basis for well-founded induction.
That is, in order to prove that every node of $\tab$ satisfies some property $P$,
it suffices to show that $P$ holds of any node $t$,
given that $P$ already holds of all $s$ such that $\simpler{s}{t}$.

Based on \autoref{l:eProp}(a) we introduce some further terminology.

\begin{defi}\label{d:cluster}
Let $\tab$ be a closed PDL-tableau.
The equivalence classes of $\cEquiv$ are called \emph{clusters}.
Non-singleton clusters will be called \emph{proper}.
\end{defi}

\begin{rem}
To explain the terminology `proper' we first note that in every non-singleton
cluster any pair of nodes is connected by a proper $\cEdge$-path.
However, this is not the case for singleton clusters.
The key observation here is that in any tableau $\tab$, the relation $\cEdge$
is irreflexive --- we leave it for the reader to verify this.
As a consequence, we can only have $s \cEdge^{+} s$ if there is a
\emph{different} node $t$ such that $s \cEdge^{+} t$ and $t \cEdge^{+} s$.
This means that a singleton $\{ s \}$ can only be a cluster of $\cEquiv$ if
there is no proper path from $s$ to itself.
In summary, a cluster $C$ has more than one element iff any pair of nodes
in $C$ is connected by a proper $\cEdge$-path.
\end{rem}

\begin{exa}
We prove the validity $[a^\ast]q \to [a][{(a \cup p?)}^\ast]q$ in our system, by
constructing a closed tableau for the set
$\{ [a^\ast]q, \ \lnot[a][{(a \cup p?)}^\ast]q \}$, cf.~\autoref{f:ex1}.
The solid arrows here show the $\edge$ relation.
\pagebreak[1]

\begin{figure}[htb]
\begin{center}
  \begin{tikzpicture}[node distance=16.5mm, >=latex]
    \draw[black!30,dashed,fill=yellow!10] (-2.6,-2.9) -- (6.7,-2.9) -- (6.7, -8.75) -- (0.5,-8.75) -- (-2.6,-7) -- cycle;
    \node (-2) {$[a^\ast]q  \+ \lnot[a][{(a \cup p?)}^\ast]q$};
    \node [below of=-2] (-1) {$q \+ [a][a^\ast]q  \+ \lnot[a][{(a \cup p?)}^\ast]q$}; \arr{-2}{$(\Box)$}{-1}
    \node [below of=-1] (0) {$q \+  [a][a^\ast]q  \+ \lnot\loaded{[a][{(a \cup p?)}^\ast]}q$}; \arr{-1}{$(L+)$}{0}
    \node [below of=0] (1) {$[a^\ast]q  \+  \lnot\loaded{[{(a \cup p?)}^\ast]}q$}; \arr{0}{$(M)$}{1}
    \node [below of=1] (2) {$q \+ [a][a^\ast]q  \+  \lnot\loaded{[{(a \cup p?)}^\ast]}q$}; \arr{1}{$(\Box)$}{2}
    \node [below of=2, node distance=8mm] (2-3) {{\footnotesize $(\loaded{\diam})$}};
    \node [below of=2] (3) {};
    \node [left of=3, node distance=30mm] (3a) {$ q  \+  [a][a^\ast]q  \+  \lnot q$}; \arr{2}{}{3a}
    \node [right of=3, node distance=30mm] (3b) {$   q  \+  [a][a^\ast]q  \+  \lnot\loaded{[a][{(a \cup p?)}^\ast]}q$};  \arr{2}{}{3b}
    \draw (3b.east) edge [dashed, gray, ->, thick, bend right=60, looseness=1.5] node [right] {$\comp$} (0.east);
    \draw (3b.east) edge [dashed, gray, ->, thick, bend right=20, looseness=1.2] node [left] {$\cEdgeT$} (1.east);
  \end{tikzpicture}
\end{center}
\caption{A closed tableau with a highlighted cluster.}\label{f:ex1}
\end{figure}

This tableau has two leaves.
The left leaf labeled with $\{ q \+ [a][a^\ast]q \+ \lnot q \}$ is closed and free. Recall that all free leaves of a tableau must be closed for the tableau to be closed.
The right leaf labeled with $\{ q \+ [a][a^\ast]q \+ \lnot\loaded{[a][{(a \cup p?)}^\ast]}q \}$ is not closed, but it is a \lpr,
since all nodes on the $\edge$-path from the companion to the \lpr are loaded.
The right dashed arrow shows the $\comp$ relation from the \lpr to its companion, a node with the same label.
The other dashed arrow is an example of the $\cEdgeT$ relation (\autoref{d:cEdge}) given by first making a $\comp$-step and then a $\edge$-step.

The dashed zone with yellow background indicates the \emph{proper cluster} to which the \lpr and companion belong.
The other nodes of the tableau each form a \emph{singleton cluster}.
We stress that clusters are \emph{not} the same as local tableaux (as highlighted in \autoref{f:fourLocTab}).

Lastly, if we call the left leaf $s$ then for all nodes $t$ in the dashed cluster we have $\simpler{s}{t}$.
On the other hand, if $u$ is the root of the tableau then for all nodes $t$ in the dashed cluster, we have $\simpler{t}{u}$.
\end{exa}

In \autoref{l:ePropB} below we gather some observations on the relations that we just introduced.
In particular, these will be needed in the soundness proof we give in the next section.

\begin{lem}[\checklean{ePropB}]\label{l:ePropB}
  Let $\tab$ be a PDL-tableau.
  For all $s, u, t \in \tab$ we have:
  \begin{enumerate}[(a)]
  \item if $s \edge t$ then $\simpler{t}{s}$ or $t \cEquiv s$;
  \item if $s \comp t$ then $t \cEquiv s$;
  \item if $s$ is free and $s \edgeT t$, then $\simpler{t}{s}$;
  \item if $t$ is free and $s \edgeT t$, then $\simpler{t}{s}$;
  \item if $s$ is loaded, $t$ is free and $s \edge t$, then $\simpler{t}{s}$;
  \item if $\simpler{u}{s}$ and $\simpler{t}{u}$ then $\simpler{t}{s}$;
  \item if $t \cEdgeT s$ and $t \not\cEquiv s$ then $\simpler{s}{t}$;
  \item if $\simpler{t}{s}$ then $s \not\cEquiv t$;
  \item if $\simpler{u}{s}$ and $s \cEquiv t$ then $\simpler{u}{t}$.
  \end{enumerate}
\end{lem}
\begin{proof}
    We only show (d) since the proof of (c) is similar, and all other items are straightforward from the definitions.
    
        First we show the following Claim.

        \emph{Claim.}
        If $u$ is free, then for all \lprs $l$ such that $u < l$, we have that $u < c(l)$.\\
        \emph{Proof of the Claim:}
        First note that since $u < l$ and $u$ is free, we have either $c(l) \edgeT u$ or $u \edgeT c(l)$.
        Assume that $c(l) \edgeT u$. Then $c(l) \edgeT u \edgeT l$, which means that the unique path from $c(l)$ to $l$ contains a free node, which contradicts that $l$ is a \lpr.
        This shows the claim.

        Assume that $s \edgeT t$. This implies that $s \cEdgeRT t$. 
        We must therefore show that it is not the case that $t \cEdgeRT s$.
        Assume towards a contradiction that there exists a $\cEdge$-path of the form 
        $t = t_0 \cEdge \ldots \cEdge t_m = s$.
        Fix such a $\cEdge$-path with a minimal number $n$ of $\comp$-steps.
        If $n=0$ then $t \edgeT s$, which contradicts $s \edgeT t$.
        If $n \geq 1$, then there is a smallest $j$ such that $t_j \comp t_{j+1}$, i.e., $t_j$ is the first \lpr in the path and $t_{j+1}=\compan{t_j}$.
        By the Claim we get $t < t_{j+1}$, hence
        $t = t_0 \edgeT t_{j+1} \cEdge  \ldots \cEdge t_m = s$
        which means that there is a $\cEdge$-path from $t$ to $s$ with $n-1$ $\comp$-steps.
        This contradicts $n$ being minimal.
        \qedhere
\end{proof}

\section{Soundness}\label{sec:Soundness}

In this section, we will prove the soundness of the PDL-tableau proof system, stating that every satisfiable sequent is consistent.
In fact, contrapositively we will prove the following statement.

\begin{thm}[\checklean{soundness}]\label{t:soundness}
Let $\tab$ be a closed PDL-tableau.
Then for any node $t$ in $\tab$, the sequent $\Lambda(t)$ is unsatisfiable.
\end{thm}

We will prove \autoref{t:soundness} by induction on the well-founded strict partial order
$\simpler{}{}$ from \autoref{d:cEquiv} above.

The final preparations for the soundness proof are the following two technical lemmas.

\begin{lem}\label{l:loadedSucc}
Assume that $t$ is a node in some closed PDL-tableau $\tab$ and that $v$ is a state of a Kripke model $\kmodel$ such that $\kmodel, v \Vdash \Lambda(t)$.
Suppose $t$ is loaded.

Then there is a $\cEdge$-path $\pi$ of satisfiable nodes from $t$ to a node $s$ such that the loaded formula is principal at $s$ and $\kmodel, v \Vdash \Lambda(r)$, for any node $r$ on $\pi$.
\end{lem}
\begin{proof}
First, note that every loaded node has children with respect to $\cEdge$, obtained by going to the companion or by applying a rule either to the loaded formula or other formulas.

If the loaded formula is not the principal formula at $t$, then the rule applied at $t$ must be local and leave the loaded formula untouched.
Moreover, by \autoref{l:localRuleTruth} the label set of at least one of the children of the node is also true in $v$.
Now consider any $\cEdge$-path starting from $t$ in which no rule is ever applied to the loaded formula, and the state $v$ satisfies the label set of every node on the path.
This means that only local rules may be applied on this path.
By \autoref{l:localMeasureDown} the measure of the sequents along this path strictly decreases at each $\edge$-step.
But after each $\comp$-step in the path there must be a $\edge$-step, so the path must be finite.
Furthermore, every label set on such a path contains the original loaded formula.

Hence if we consider a \emph{maximal} path of this kind, the only rule that is applicable at its final node must have the loaded formula as its principal formula.
Thus any such maximal path meets the conditions of the lemma.
\end{proof}

\begin{lem}[\checklean{loadedDiamondPaths}]\label{l:loadedDiamondPaths}
Assume that $t$ is a node in some closed PDL-tableau $\tab$ and that $v$ is a state of a Kripke model $\kmodel$ such that $\kmodel, v \Vdash \Lambda(t)$.
Furthermore, let $\xi$ be a (possibly loaded) formula such that $\lnot\underline{[\alpha]}\xi \in \Lambda(t)$, and suppose that $w$ in $\kmodel$ is such that $v \rel{\alpha} w$ and  $\kmodel, w \Vdash \lnot\xi$.

Then there is a $\cEdge$-path $\pi$ of satisfiable nodes from $t$ to a node $s$ such that either $s \not\cEquiv t$ or every node on $\pi$ belongs to the same cluster as $t$, and $s$ satisfies all these three conditions:
\begin{enumerate}[(a)]
\item $\lnot\xi \in \Lambda(s)$;
\item $\kmodel, w \Vdash \Lambda(s)$;
\item $\Lambda(s) \setminus \{ \neg\xi \}$ is free.
\end{enumerate}
\end{lem}

Note that ``$\Lambda(s) \setminus \{ \neg\xi \}$ is free'' just means that there can be no loaded formula in $\Lambda(s)$ other than $\neg\xi$ (which may itself be loaded or not).

\begin{proof}
We will prove the Lemma by induction on the shallow program structure of $\alpha$.
Before distinguishing the different cases of $\alpha$, we note that by \autoref{l:loadedSucc} we may without loss of generality assume that either the modal rule or the loaded diamond rule is applied to $\lnot\underline{[\alpha]}\xi$ at node $t$.

To see this, assume that $\lnot\underline{[\alpha]}\xi$ is not principal at $t$.
Then by \autoref{l:loadedSucc} there is a $\cEdge$-path $\rho$ of satisfiable nodes from $t$ to a node $t'$ such that the loaded formula is principal at $t'$ and $\kmodel, v \Vdash \Lambda(r)$, for any node $r$ on $\rho$.
If some node $r$ on this path falls out of the cluster of $t$, we may put $s \isdef r$ and we are done.
Hence, we assume otherwise, i.e., we have $t \cEquiv t'$.
Now if the rule applied at $t'$ is the unloading rule $(L-)$, then we are done as well.
For, in this case $t'$ has a unique descendant, say, $s$, and this descendant is free.
By \autoref{l:ePropB}(e) this implies $\simpler{s}{t'}$, and so by \autoref{l:ePropB}(f) we find $\simpler{s}{t}$ which means in particular that $s \not\cEquiv t$.
Then this $s$ meets the requirements of the lemma.
Hence we may assume that at $t'$ some rule other than $(L-)$ is applied to the loaded formula $\lnot\loaded{[\alpha]}\xi$.
Concretely, this rule must either be $(M)$ (if $\alpha$ is atomic) or $(\loaded{\diam})$ (if it is not).

We now turn to the inductive proof of the Lemma.

\begin{description}
\item[Base case]
  In the base case of the proof we consider the following options.
  In both cases we assume that $\lnot\underline{[\alpha]}\xi$ is principal at $t$.
  \begin{description}
  \item[\it Case $\alpha=a$]
    In this case the rule applied at $t$ is the modal rule $(M)$; let $s$ be the unique child of $t$, then $s$ is labeled with the sequent
    ${\big(\Lambda(t)\big)}_{a}, \neg\xi$.
    Furthermore, we have $v \rel{a} w$.
    Since $\Lambda(t)$ holds at $v$, it is then obvious that $s$
    satisfies the required properties (a), (b), and (c).
  \item[\it Case $\alpha={\tau?}$]
    In this case the rule applied at $t$ is $(\loaded{\diam})$.
    Moreover, we have $w = v$ and $\kmodel, v \Vdash \tau,\neg\xi$, so we may take the
    child $s$ of $t$ that is labeled with the sequent $\Lambda(t), \tau, \neg\xi$.
  \end{description}

\item[Induction step]
In the induction step of the proof we prove the statement of the lemma for an arbitrary complex program $\alpha$.
Given that $\alpha$ is not atomic, the rule applied at $t$ cannot be $(M)$.
Moreover, as argued above, without loss of generality we may assume that the rule applied
at $t$ is $(\loaded{\diam})$, with $\lnot\underline{[\alpha]}\xi$ as the
principal formula.

As our induction hypothesis, we may suppose that the lemma holds for all programs
simpler than $\alpha$, but in fact, it will be convenient to use the following
extended version of this statement that applies to \emph{lists} of such programs.

\emph{Claim (Extended Induction Hypothesis)}:
Let $\ol{\delta}$ be a sequence of programs that are all simpler than $\alpha$.
Let $\psi$ be some (possibly loaded) formula, and assume that $r$ in $\tab$ and
$v$ in $\kmodel$ are such that
$\kmodel, v \Vdash \Lambda(r)$,
$\lnot\loaded{\Box}(\ol{\delta},\psi) \in \Lambda(r)$, $(v,w) \in R_{\ol{\delta}}$ and
$\kmodel, w \Vdash \neg\psi$.
\\
Then there is a $\cEdge$-path $\pi$ of satisfiable nodes from $r$ to a node $s$ such that either $s \not\cEquiv r$ or every node on $\pi$ belongs to the same cluster as $r$, and $s$ satisfies the following three conditions:
(a) $\lnot\psi \in \Lambda(s)$;
(b) $\kmodel, w \Vdash \Lambda(s)$;
(c) $\Lambda(s) \setminus \{ \neg\psi \}$ is free.

One may \emph{prove} this Claim by induction on the length of the sequence~$\ol{\delta}$, applying the original induction hypothesis for each of the programs simpler than $\alpha$.
We return to the proof of the induction step of the main proof.
Here we make a case distinction as to the nature of $\alpha$.

\begin{description}
\item[Case $\alpha=\beta\cup\gamma$ or $\alpha=\beta;\gamma$]
By \autoref{l:existsDiamondDset} there is a pair $(F,\ol{\delta}) \in \Dset{\alpha}$ such
that $\kmodel, v \Vdash F$ and $(v,w) \in R_{\ol{\delta}}$.
Since the rule applied at $t$ is $(\loaded{\diam})$, $t$ has a child $t'$ such
that $F \cup \{ \lnot\loaded{\Box}(\ol{\delta},\xi) \} \subseteq \Lambda(t')$,
and $\kmodel, v \Vdash  \Lambda(t')$.

Since all programs in $\ol{\delta}$ are proper subprograms of $\alpha$,
by the Extended Induction Hypothesis we obtain a $\cEdge$-path $\rho$ in
$\tab$ from $t'$ to some node $s$ such that
(a) $\lnot\xi \in \Lambda(s)$;
(b) $\kmodel, w \Vdash \Lambda(s)$; and
(c) $\Lambda(s) \setminus \{ \neg\xi \}$ is free.
From this and the fact that $t \cEdge t'$ the statement of the Lemma follows.

\item[Case $\alpha = \beta^{\ast}$]
Here we would like to use the same argument as in the previous cases, and it is
true that there is a pair $(F,\ol{\delta}) \in \Dset{\alpha}$ such that $\kmodel, v
\Vdash F$ and $(v,w) \in R_{\ol{\delta}}$.
The problem, however, is that here we may encounter $\beta^{*}$ itself as the
final program on the list $\ol{\delta}$, so that the Extended Induction Hypothesis
may not apply directly.

We therefore reason as follows.
Since $(v,w) \in R_{\beta^{*}}$, there is a path
\[
v = w_0 \rel{\beta} w_1 \rel{\beta} \ldots \rel{\beta} w_n = w.
\]
Here we may take $n$ to be minimal, which means that all $w_{i}$ are distinct.

We now first show the following claim.

\emph{Claim}:
For all $k$ such that $0 \leq k \leq n$ there is a satisfiable $\cEdge$-path
from $t$ to a node $s_k$ such that either $s_k \not\cEquiv t$ or we have
(a$^k$)~$\lnot\loaded{[\beta^\ast]}\xi \in \Lambda(s_k)$,
(b$^k$)~$\kmodel, w_k \Vdash \Lambda(s_k)$,
(c$^k$)~$\Lambda(s_k) \setminus \{\lnot\loaded{[\beta^\ast]}\xi\}$ is free,
and (d$^k$)~the formula $\lnot\loaded{[\beta^\ast]}\xi$ is principal at $s_k$.

We prove the claim by an inner induction on $k$.
For the base case $k = 0$ we take $s_0$ to be any node as given by \autoref{l:loadedSucc}.
Then clearly $s_0$ either satisfies $s_0 \cEquiv t$ or else it has the properties (a$^0$) to (d$^0$).

For the induction step to $k+1$, we assume as our inner inductive hypothesis
that the Claim holds for $k$.
If $s_k \not\cEquiv t$ then let $s_{k+1} \isdef s_k$ and check that we are done.
Otherwise, there is a satisfiable $\cEdge$-path $\pi^{k}$ from $t$ to $s_k$
satisfying the properties (a$^k$) to (d$^k$).

Furthermore, by \autoref{l:existsDiamondDset} there is a pair $(F,\ol{\delta}) \in
\Dset{\beta}$ such that $\kmodel, w_{k} \Vdash F$ and $(w_{k},w_{k+1}) \in
R_{\ol{\delta}}$.
Note that since $w_{k} \neq w_{k+1}$ it cannot be the case that $\ol{\delta} =
\varepsilon$; from this it follows that $(F,\ol{\delta}\beta^{*}) \in
\Dset{\beta^{*}}$.
Then, since $\lnot\loaded{[\beta^\ast]}\xi$ is principal at $s_k$ by (d$^k$),
$s_{k}$ must have a successor $s'_{k}$ labeled
\[
\Lambda(s'_{k}) =
\big(\Lambda(s_{k}) \setminus \{ \lnot\loaded{[\beta^\ast]}\xi \} \big)
\cup F \cup \{ \lnot\loaded{\Box}(\ol{\delta}\beta^\ast,\xi) \}.
\]
Now observe that, since all programs in $\ol{\delta}$ are subprograms of $\beta$
and hence proper subprograms of $\beta^{*}$, we may apply the outer inductive
hypothesis, or rather its strengthened version, to the node $s'_{k}$, the list
$\ol{\delta}$, and the formula $\lnot\loaded{\Box}(\ol{\delta}\beta^\ast,\xi) =
\lnot\loaded{\Box}(\ol{\delta},\loaded{[\beta^\ast]}\xi)$.
This gives a satisfiable path $\rho$ from $s'_{k}$ to some node $t_{k+1}$ such that either $t_{k+1} \cEquiv s_{k}$, or we have
(a)~$\lnot\loaded{[\beta^\ast]}\xi \in \Lambda(t_{k+1})$,
(b)~$\kmodel, w_{k+1} \Vdash \Lambda(t_{k+1})$, and
(c)~$\Lambda(t_{k+1}) \setminus \{\lnot\loaded{[\beta^\ast]}\xi\}$ is free.
Finally, if $t_{k+1} \cEquiv s_{k}$ we are done; otherwise, let $s_{k+1}$ be a node as given by \autoref{l:loadedSucc},
in particular, the formula $\lnot\loaded{[\beta^\ast]}\xi$ is principal at
$s_{k+1}$, while in addition $s_{k+1}$ satisfies the properties
(a$^{k+1}$) to (d$^{k+1}$).

It is then straightforward to verify that the concatenation of the path
$\pi^{k}$ (from $t$ to $s_{k}$), the one-edge path $s_{k} \cEdge s'_{k}$,
the path $\rho$ (from $s'_{k}$ to $t_{k+1}$) and the path from
$t_{k+1}$ to $s_{k+1}$ satisfies the required conditions.
\emph{This finishes the proof of the Claim.}
\medskip

Finally, returning to the main proof of the lemma, in order to obtain the
required node $s$ in the inductive case where  $\alpha = \beta^{*}$, we consider
the node $s_{n}$ as given by the Claim instantiated with $k \isdef n$.
Since by condition (d$^{n}$) the formula $\lnot\loaded{[\beta^\ast]}\xi$ is
principal at $s_n$ and the pair $(\emptyset,\varepsilon)$ belongs to
$\Dset{\alpha}$, the node $s_{n}$ has a $\cEdge$-successor $s$ such that
\[
\Lambda(s) =
\big(\Lambda(s_{n}) \setminus \{ \lnot\loaded{[\beta^\ast]}\xi \} \big)
\cup \{ \lnot\xi \}.
\]
It is then straightforward to verify that this node $s$ has the required
properties.
\qedhere
\end{description}
\end{description}
\end{proof}

We are now ready to prove the soundness of our system.

\begin{proof}[Proof of \autoref{t:soundness}]
We prove the theorem by induction on the relation $\simpler{}{}$.
Let $\tab$ be a closed PDL-tableau, take an arbitrary node $t$ in $\tab$ and
assume inductively that all nodes~$s$ with $\simpler{s}{t}$ are unsatisfiable.
It suffices to show that $t$ is unsatisfiable.
We make a case distinction as to the nature of $t$.

First, assume that $t$ is free.
If $t$ is a leaf, then since $\tab$ is assumed to be closed, $t$ must be closed
and hence, $\Lambda(t)$ is unsatisfiable.
If $t$ is not a leaf, then since it is free, by \autoref{l:ePropB}(c) each of
its children is $\simpler{}{}$-smaller than $t$.
It thus follows by the induction hypothesis that no child of $t$ carries a
satisfiable sequent.
But since the rule applied at node $t$ is sound (as all rules are), this means
that $\Lambda(t)$ cannot be satisfiable either.

The interesting case is where $t$ is loaded; say, we have
$\lnot\underline{[\alpha_{1}]\cdots[\alpha_{n}]}\xi \in \Lambda(t)$,
for some (unloaded) PDL formula $\xi$.

Let $C$ be the cluster to which $t$ belongs.
We claim that any $\cEdge$-path of satisfiable nodes starting at $t$ must
remain in $C$.
This easily follows from \autoref{l:ePropB}(g): if we leave the cluster we arrive at a node which is unsatisfiable by the induction hypothesis.

Now assume for contradiction that $\Lambda(t)$ is satisfiable.
Satisfiability of $\lnot[\alpha_{1}]\cdots[\alpha_{n}]\xi$
and successive applications of \autoref{l:loadedDiamondPaths}
(together with the observation that paths of satisfiable nodes must stay in $C$) imply
that there is a $\cEdge$-path of satisfiable nodes leading from $t$ to a
node $s$ such that $\lnot\xi \in \Lambda(s)$ and $\Lambda(s) \setminus \{ \neg\xi \}$
is free.
Since $\xi$ is not loaded this means that $s$ itself is a free node.
Since $s$ is free, it cannot belong to the same cluster as $t$.
Since there exists a $\cEdge$-path from $t$ to $s$ and $t \not\cEquiv s$, we get that $\simpler{s}{t}$.
Thus we have arrived at the desired contradiction:
We had that $\Lambda(s)$ was satisfiable, but by the induction hypothesis, $\Lambda(s)$ is unsatisfiable.
We thus conclude that $\Lambda(t)$ is unsatisfiable.
\end{proof}

\section{Completeness}\label{sec:Completeness}

The aim of this section is to prove the completeness of the PDL-tableau system, stated as follows.

\begin{thm}[\checklean{completeness}]\label{t:completeness}
Let $\Gamma$ be a sequent.
If $\Gamma$ is consistent, then it is satisfiable.
\end{thm}

In \autoref{subsec:modelgraphs} we will first define \emph{model graphs}, a class of models based on sets of formulas and similar to canonical models used in many other completeness proofs~\cite[Section~4.2]{BRV}.
Then in \autoref{subsec:TableauGame} we describe how proof search in our tableau system can be seen as a two-player game, and in \autoref{subsec:StrToMoG} we show how to convert a winning strategy for one of the players to a model graph.
Combining all of this we then prove \autoref{t:completeness} in \autoref{subsec:CompProof}.

\subsection{Saturated sets and model graphs}\label{subsec:modelgraphs}

As mentioned, the states of model graphs are given by certain sets of formulas, which we will call \emph{saturated}.

\begin{defi}[\lean{saturated}]\label{d:saturated}
A set of formulas $X$ is \emph{saturated} if it fulfils the following conditions for all formulas $\phi$, $\phi_1$, $\phi_2$ and all programs $\alpha$:
\[
  \begin{array}{lcl}
    \lnot \lnot \phi \in X             & \text{ implies } & \phi \in X   \\[0.3em]
    \phi_1 \land \phi_2 \in X          & \text{ implies } & \phi_1 \in X \text{ and } \phi_2 \in X \\[0.3em]
    \lnot (\phi_1 \land \phi_2) \in X  & \text{ implies } & \lnot \phi_1 \in X \text{ or } \lnot \phi_2 \in X \\[0.3em]
    [\alpha]\phi \in X      & \text{ implies } & \Delta \subseteq X, \text{ for some } \Delta \in \unfold_{\Box}(\alpha,\phi) \\[0.3em]
    \lnot[\alpha]\phi \in X & \text{ implies } & \Delta \subseteq X, \text{ for some } \Delta \in \unfold_{\diam}(\alpha,\phi) \\[0.3em]
  \end{array}
\]
A saturated set $X$ is \emph{locally consistent} if $\bot \notin X$ and $p \in X \text{ implies } \lnot p \notin X$, for all proposition letters $p$.
\end{defi}

As in the setting of completeness proofs that involve the canonical model, the key statement on model graphs is a \emph{truth lemma} stating that formulas that belong to a saturated set are also satisfied by that set, taken as a state.
In the definition of model graphs and in the proof of this truth lemma we use a variation $Q_\alpha$ of the accessibility relation $R_\alpha$ --- the main difference being in the case of tests, where we use, in anticipation of the truth lemma, the membership relation rather than the truth relation.

\begin{defi}[\lean{Modelgraphs.Q}]\label{d:Qrelation}
Let $W$ be some collection of saturated sets, and let, for every atomic program $a$, $R_{a}$ be a binary relation on $W$.
Then by program induction we define a relation $Q_{\alpha}$ on $W$ for every $\alpha \in \progs$:
\[
  \begin{array}{lll}
    Q_{a} &\isdef& R_{a} \\[0.3em]
    Q_{\tau?}  &\isdef& \{ (X,X) \mid \tau \in X \} \\[0.3em]
    Q_{\alpha\cup\beta} &\isdef&  Q_{\alpha} \cup Q_{\beta} \\[0.3em]
    Q_{\alpha; \beta} &\isdef& Q_{\alpha} \seqc Q_{\beta} \\[0.3em]
    Q_{\alpha^{\ast}} &\isdef& Q_{\alpha}^{*} \\[0.3em]
  \end{array}
\]
\end{defi}

\begin{defi}[\lean{ModelGraph}]\label{d:ModelGraph}
A \emph{model graph} is a Kripke model $\kmodel = (W, \{\reach{a} \mid a \in \progs_0\}, V)$ such that:
\begin{enumerate}[(a)]
\item Each element of $W$ is a locally consistent saturated set of formulas;
\item If $X \in W$, then $p \in X$ iff $X \in V(p)$;
\item If $X \rel{a} Y$ and $[a]\phi \in X$ for some atomic program $a$, then $\phi \in Y$;
\item If $X \in W$ and $\lnot [\alpha]\phi \in X$ for any program $\alpha$, then there exists a $Y \in W$ such that $(X,Y) \in  Q_{\alpha}$ and $\lnot \phi \in Y$.
\end{enumerate}
\end{defi}

\begin{lem}[\checklean{truthLemma}]\label{l:truthLemma}
Suppose $\kmodel$ is a model graph.
Then for all states $X$ and all formulas $\phi \in X$ it holds that $\kmodel, X
\Vdash \phi$.
\end{lem}

\begin{proof}
The proof of this Lemma is more or less standard.
By a simultaneous induction on formulas and programs we prove that the following
statements hold, for all formulas $\phi$ and all programs $\alpha$:

(C1) If $\phi \in X$, then $\kmodel, X \Vdash \phi$;

(C2) If $\lnot\phi \in X$, then $\kmodel, X \not\Vdash \phi$;

(C3) $Q_{\alpha} \subseteq R_{\alpha}$;

(C4) If $[\alpha]\phi \in X$ and $X \rel{\alpha} Y$, then $\phi \in Y$.

Note that we work with a fixed model graph $\kmodel$ in this proof, so for the sake of brevity we will only write $X \Vdash \phi$ for $\kmodel, X \Vdash \phi$.

\emph{Base cases for (C1) and (C2):}
If $\phi=\bot$ or $\phi$ is a proposition letter, then statements (C1) and (C2) hold by conditions (a) and (b) in \autoref{d:ModelGraph}.\\
\emph{Base case for (C3):}
For atomic programs $\alpha = a$, we have $Q_{a} = R_{a}$ by \autoref{d:Qrelation}.\\
\emph{Base case for (C4):}
For atomic programs, statement (C4) holds by condition (c) in \autoref{d:ModelGraph}.

\emph{Induction steps for (C1) and (C2):}
Let $\phi$ be a formula. By induction hypothesis, statements (C1) and (C2) hold for all subformulas of $\phi$, and statements (C3) and (C4) hold for all programs in $\phi$.

\begin{description}
\item[Case $\phi = \lnot \psi$]
(C1)
If $\neg\psi \in X$,
then by the induction hypothesis (C2) for the subformula $\psi$, we have $X \not\Vdash\psi$, hence $X  \Vdash \lnot \psi$.
(C2)
If $\lnot \phi = \lnot \lnot \psi \in X$, then because $X$ is saturated, we have
$\psi \in X$ and by induction hypothesis (C1) for $\psi$, we have $X \Vdash \psi$ and hence $X \not\Vdash \lnot \psi$.

\item[Case $\phi = \psi \land \chi$]
(C1)
If $\psi \land \chi \in X$, then by saturation, $\psi \in X$ and $\chi \in X$, which by induction hypothesis implies
$X \Vdash \psi$ and $X \Vdash \chi$, thus $X \Vdash \psi \land \chi = \phi$.
(C2)
If $\lnot(\psi \land \chi) \in X$, then $\lnot \psi \in X$ or $\lnot \chi \in X$, which
by induction hypothesis (C2) for $\psi$ and $\chi$ implies
$X \not\Vdash \psi \land \chi$.

\item[Case $\phi = {[\alpha]}\psi$]
(C1) If $[\alpha]\psi \in X$, then by induction hypothesis (C4) for $\alpha$ we have that
$\psi \in Y$ for all $Y$ such that $X \rel{\alpha} Y$.
By induction hypothesis (C1) for $\psi$ we also have that $\psi \in Y$ implies
$Y \Vdash \psi$, hence we have $X \Vdash [\alpha]\psi$.\\
(C2)
If $\lnot [\alpha]\psi \in X$, then by condition (d) there is a $Y$ such that $(X,Y)
\in Q_{\alpha}$ and $\lnot \psi \in Y$.
By induction hypothesis (C3) for $\alpha$, we have $Q_{\alpha} \subseteq
R_{\alpha}$ and by induction hypothesis (C2) for $\psi$, we have $Y \Vdash
\lnot \psi$.
Combining these observations we find $X \not\Vdash [\alpha]\psi$.
\end{description}

\emph{Induction step for (C3):}
If $\alpha$ is a test, say, $\alpha = {\tau?}$, then by \autoref{d:Qrelation} we have $Q_{\tau?} = \{ (X,X) \mid \tau \in X \}$.
By the induction hypothesis (C1) for subformula $\tau$ we find that $Q_{\tau?} \subseteq \{ (X,X) \mid \kmodel, X \Vdash \tau \} = R_{\tau?}$.
The other cases are straightforward by using the induction hypothesis for (C3).
For instance, we have $Q_{\alpha; \beta} = Q_{\alpha} \seqc Q_{\beta} \subseteq  R_{\alpha} \seqc R_{\beta} = R_{\alpha; \beta}$.

\emph{Induction step for (C4):}
Let $\alpha$ be a non-atomic program, and assume as induction hypothesis that statement (C4) holds for all subprograms of $\alpha$, and that statements (C1) and (C2) hold for all subformulas of $\alpha$, in particular for formulas $\tau \in \Test(\alpha)$.

\begin{description}
\item[Case $\alpha = {\tau?}$, $\alpha = \beta \cup \gamma$ or $\alpha = \beta;
\gamma$]
Assume that $[\alpha]\phi \in X$ and $X \rel{\alpha} Y$.
Then by saturation of $X$ there is some test profile $\ell \in \TT(\alpha)$ such
that $\Bset{\ell}{\alpha}{\phi} = \Fset{\ell}{\alpha} \cup
\{ \Box(\ol{\delta},\phi) \mid \ol{\delta} \in \Pset{\ell}{\alpha} \} \subseteq X$.
Recall that the formulas in $\Fset{\ell}{\alpha}$ are of the form $\neg\tau$ where ${\tau?}$ is a subprogram of $\alpha$, and hence $\tau$ is a subformula of $\alpha$.
By induction hypothesis (C2) on these subformulas $\tau$, we obtain that $\kmodel, X \Vdash \bigwedge \Fset{\ell}{\alpha}$.
By \autoref{l:FP} there is $(X,Y) \in R_{\ol{\delta}}$ for some
$\ol{\delta} = \delta_{1}\cdots\delta_{n} \in \Pset{\ell}{\alpha}$.
Now observe that by \autoref{l:boxHelperTermination}(c) each~$\delta_i$ belongs to $\Prog(\alpha)$, but is not equal to~$\alpha$.
Hence each~$\delta_i$ is a proper subprogram of~$\alpha$, so that we may apply the inductive hypothesis to it.
Now $n$ successive applications of the inductive hypothesis (C4) on the programs $\delta_{1},\ldots,\delta_{n}$ imply that $\phi\in Y$, as required.

\item[Case $\alpha = \beta^\ast$]
Assume that $[\beta^\ast]\phi \in X$ and  $X \rel{\beta^\ast} Y$.
Then there is a minimal $n \in \mathbb{N}$ such that
$(X,Y) \in {(R_{\beta})}^{n}$.
We now prove (C4) by an inner induction over $n$.

Base case $n=0$: In this case $X = Y$, so we must show that $\phi \in X$. By saturation of $X$,
there is some $\ell \in \TT(\beta^{\ast})$ such that $\Bset{\ell}{\beta^\ast}{\phi} \subseteq X$.
Since $\emptylist \in \Pset{\ell}{\beta^\ast}$, we obtain that
$\phi = \Box(\emptylist,\phi) \in \Bset{\ell}{\beta^\ast}{\phi}$, hence $\phi\in X$.

Induction step:
Let $n>0$ and assume the claim holds for $m \isdef n-1$.
It follows from $(X,Y) \in R_{\beta}^{n}$ and minimality of $n$ that there is a $Z\neq X$ such that $X \rel{\beta} Z$ and $Z \rel{\beta^{m}} Y$.

By saturation of $X$ there is some $\ell \in \TT(\beta^{\ast})$ such that
$\Bset{\ell}{\beta^{\ast}}{\phi} = \Fset{\ell}{\beta^{\ast}} \cup
\{ \Box(\ol{\delta},\phi) \mid \ol{\delta} \in \Pset{\ell}{\beta^{\ast}} \}
\subseteq X$.
It follows by the induction hypothesis (C2) on the subformulas $\tau$ such that $\neg\tau \in \Fset{\ell}{\beta^{\ast}} = \Fset{\ell}{\beta}$ that $\kmodel, X \Vdash \bigwedge \Fset{\ell}{\beta}$.
By \autoref{l:FP} we may conclude that $(X,Z) \in R_{\ol{\delta}}$, for some $\ol{\delta} \in \Pset{\ell}{\beta}$.
Note that $\ol{\delta} \neq \emptylist$ since $n$ was minimal and thus $X \neq Z$.
It follows that $\ol{\delta}\beta^{\ast} \in \Pset{\ell}{\beta^{\ast}}$, and because of this we have $\Box(\ol{\delta}\beta^{\ast},\phi) = \Box(\ol{\delta},[\beta^{\ast}]\phi) \in X$.
Now, as in the previous case, through successive applications of the inductive hypothesis for (C4) on the programs in $\ol{\delta}$ we find $[\beta^{\ast}]\phi
\in Z$.
But now we may apply the inner induction hypothesis on $Z$ and $Y$, for the
formula $[\beta^{\ast}]\phi$.
This gives $\phi \in Y$ as required.
\qedhere
\end{description}
\end{proof}

Note that condition (d) in the definition of a model graph requires that certain pairs of saturated sets belong to the relation $Q_\alpha$, for some program $\alpha$.
In the next subsection, where we will define a concrete model graph, we will prove that it satisfies this condition (d) by induction on the complexity of programs.
For this proof it will be convenient to extend the definition of the $Q$-relations as follows.

\begin{defi}[\lean{Qcombo}]\label{d:Qcombo}
Let $X$ be some set of formulas, and let $\ol{\delta}$ be a list of programs.
We define the relations $Q_{X}, Q_{\ol{\delta}}$ and $Q_{X,\ol{\delta}}$ on $W$ as
follows:
\[
  Q_{X} \isdef \mathsf{Id}_W \cap \bigcap_{\tau\in X} Q_{\tau?},
\]
while $Q_{\ol{\delta}}$ is defined by induction on the length of the list
$\ol{\delta}$:
\[
  \begin{array}{lll}
    Q_{\emptylist} &\isdef& \mathsf{Id}_W \\[0.3em]
    Q_{\alpha\ol{\gamma}} &\isdef& Q_{\alpha}\seqc Q_{\ol{\gamma}}
  \end{array}
\]
Finally we set
\[
  Q_{X,\ol{\delta}} \isdef Q_{X} \seqc Q_{\ol{\delta}}.
\]
\end{defi}

\begin{lem}[\checklean{cpHelpA}]\label{l:cp3a}
Let $\alpha$ be a PDL-program.
Then we have $Q_{X,\ol{\delta}} \subseteq Q_{\alpha}$ for every pair $(X,\ol{\delta}) \in \Dset{\alpha}$.
\end{lem}
\begin{proof}
  By induction on the shallow program structure of $\alpha$.
  The claim for the base cases with an atomic program $\alpha=a$ or with a test $\alpha={\tau?}$ follows directly from \autoref{d:Qcombo}.
  For the inductive step we have three cases.
  If $\alpha = \beta \cup \gamma$ then the claim follows from applying the induction hypothesis to~$\beta$ and~$\gamma$.
  For the cases $\alpha = \beta ; \gamma$ and $\alpha = \beta^\ast$ before applying the induction hypothesis to $\beta$ (and $\gamma$) we make an additional case distinction on whether the unfolding of $\beta$ yields only tests or (also) a non-empty list of programs.
  Notably, in the case $\alpha = \beta^\ast$ there is \emph{no} need for an additional inner induction.
  The proof overall is very similar to that of \autoref{l:FP}, so we omit further details here (note that this lemma has been verified in Lean).
\end{proof}

\subsection{Tableau game}\label{subsec:TableauGame}

The search space of tableaux for a sequent $\Gamma$ can be described as a game between two players.
Given a sequent $\Gamma$, the \emph{tableau game} $\Gm(\Gamma)$ is a two-player graph game with players Builder and Prover.
Intuitively, Builder aims to show that the sequent $\Gamma$ is satisfiable, whereas it is Prover's goal to construct a closed tableau for $\Gamma$.

The game proceeds by the two players moving a token from one position to another.
A \emph{position} of $\Gm(\Gamma)$ is either:
\begin{itemize}
\item a (possibly loaded) sequent $\Delta$, or
\item a triple of the form $(\Delta,\phi,R)$ such that $\Delta \cup \{ \phi \}$ is a (possibly loaded) sequent such that $\phi \not\in \Delta$, and $R$ is a PDL-rule applicable to $\Delta, \phi$ where $\phi$ is the principal formula, or
\item $v_\bot$, a sink node that has no outgoing edges.
\end{itemize}
Positions of the first type belong to Prover; at a position $\Delta$ he picks
a formula $\phi \in \Delta$ and a rule $R$ to be applied to $\phi$, thus moving
the token to the position $(\Delta \setminus \{\phi\}, \phi, R)$.
Positions of the second type belong to Builder; at a position $(\Delta,\phi,R)$
she picks one of the rule children of $R$, and this rule child is then the next position.
(In many cases, she has no real choice, since the rule has only one rule child.)
The sink node also belongs to Builder. It is needed for technical reasons to define the winning condition (see below).
A \emph{match} of $\Gm(\Gamma)$ consists of a sequence of positions, starting at $\Gamma$, induced by the players' moves.
We also call it a $\Gm(\Gamma)$-match.
Such a match is \emph{won by Prover} if it reaches a closed sequent or a \lpr.
On the other hand, a match is \emph{won by Builder} if it reaches a sequent to which no rule is applicable, or a free repeat.
Here, we define \lprs and free repeats essentially in the same way as in \autoref{d:repeat}, i.e.,
a repeat in the game $\Gm(\Gamma)$ is a position of the form $\Delta$ which has occurred previously in the match.
Such a repeat is free, if $\Delta$ is not loaded, and it is a \lpr if $\Delta$ is loaded, and all sequent positions (i.e., positions of the first kind) in the match between the two occurrences of $\Delta$ are also loaded.
More precisely, in terms of the game graph, positions labelled with a free repeat have no outgoing edges and belong to Prover. Positions that are labelled by a \lpr or a closed sequent have one outgoing edge to the sink node $v_\bot$ which belongs to Builder. Hence the winning condition reduces to the usual graph game condition that a player who is stuck (i.e., it is their turn but they cannot make a move) loses the match.

A \emph{strategy} for a player $P$ prescribes a move at every position in the game where it is $P$'s turn. 
Such a strategy is \emph{winning} if it leads to a win for $P$, no matter which moves the other player makes.
A strategy for a player $P$ can be represented as a subgraph of the game tree
obtained by deleting all except one outgoing edge (corresponding to the prescribed move) from the positions where it is $P$'s turn. If a position for $P$ has no outgoing edges, it stays this way. At positions for the other player, all outgoing edges remain. This results in a tree structure in which the branches are all the possible matches that can arise when $P$ plays according to the strategy.
We simplify this structure further by identifying the positions for $P$ that have out-degree 1 with their unique successor, resulting in a tree in which each vertex corresponds to a partial match ending in a position of the first kind (labelled with a sequent) and the local branching corresponds to choices of the other player.

Formally, we define a strategy tree for Prover as follows.

\begin{defi}\label{d:prover-tree}
  A \emph{strategy tree for Prover in $\Gm(\Gamma)$} is a structure $\Stree = (T,\tedge,r,\Lambda,f)$, where $(T,\tedge,r)$ is a tree with root $r$, $\Lambda$ is a labelling which assigns to each node $s \in T$ a (possibly loaded) sequent $\Lambda(s)$, and $f$ assigns to each node $s \in T$ a pair $(R,\phi)$ where $R$ is a PDL-rule and $\phi$ is a formula
such that
\begin{itemize}
\item $\Lambda(r) = \Gamma$;
\item For every $s \in T$ with children $t_1, \dots, t_n$ and such that $f(s) = (R,\phi)$
we have that
$\Lambda(s) ~/~ \Lambda(t_1) | \dots | \Lambda(t_n)$ is an instance of $R$ with principal formula $\phi$.
\end{itemize}
We say that $\Stree$ is \emph{winning} for Prover if $\Stree$ is finite and each of its branches represents a $\Gm(\Gamma)$-match that is won by Prover.
\end{defi}

It follows immediately from the definition of the game and \autoref{d:prover-tree} that a strategy for Prover corresponds to a tableau over the PDL rules, and the strategy is winning if and only if the corresponding tableau is a closed PDL-tableau.
Hence, a winning strategy tree for Prover yields a closed PDL-tableau.

\begin{lem}[\checklean{gameP}]\label{lem:gameP}
A strategy tree $\Stree$ for Prover corresponds to a tableau over the PDL rules, and 
if $\Stree$ is winning for Prover in $\Gm(\Gamma)$ then $\Gamma$ has a closed PDL-tableau.
\end{lem}
\begin{proof}
Let $\Stree = (T,\tedge,r,\Lambda,f)$ be a strategy tree for Prover in $\Gm(\Gamma)$. Define a labelling function $L$ by taking $L(s) \coloneqq R$ if $f(s) = (R,\phi)$. By the rules of the game, $\mathcal{T} = (T,\tedge,r,\Lambda,L)$ is a tableau over the PDL rules (cf.~\autoref{d:tableau-general}). Furthermore, by the winning condition of the game, $\Stree$ is winning for Prover if and only if $\mathcal{T}$ is a closed PDL-tableau.
\end{proof}

In a strategy tree for Builder, at a node $s$, for every possible rule application that can be chosen by Prover there is a child $t$ labelled with the rule child prescribed by the strategy in response to Prover's move, and all children of $s$ arise in this way.
More precisely, a strategy tree for Builder is defined as follows.

\begin{defi}[\lean{BuildTree}]\label{d:builder-tree}
  A \emph{strategy tree for Builder in $\Gm(\Gamma)$} is a structure $\Stree = (T,\tedge,r,\Lambda,g)$,
  where $(T,\tedge,r)$ is a tree with root $r$, $\Lambda$ is a labelling which assigns to each node $s$ in $T$ a (possibly loaded) sequent $\Lambda(s)$,
  and $g$ assigns to each node different from $r$ a pair $(R,\phi)$
such that
\begin{itemize}
\item $\Lambda(r) = \Gamma$;
\item For all $s \in T$: if rule $R$ can be applied to $\Lambda(s)$ with principal formula $\phi$ then there is a $t \in T$ such that $s \tedge t$ and $g(t) = (R,\phi)$;
\item For all $s, t \in T$: if $s \tedge t$ and $g(t) = (R,\phi)$ then rule $R$ can be applied to $\Lambda(s)$ with principal formula $\phi$ and $\Lambda(t)$ is one of the rule children of the resulting instance of $R$.
\end{itemize}

We say that $\Stree$ is \emph{winning} for Builder if it is finite and each of its branches represents a $\Gm(\Gamma)$-match that is won by Builder.
\end{defi}

In \autoref{d:builder-tree}, $g(t)$ is a possible move for Prover at $\Lambda(s)$
which would lead to position $(\Lambda(s)\setminus\{\phi\},\phi,R)$ in the game,
and $\Lambda(t)$ is the rule child chosen by Builder at this position 
according to the strategy represented by $\Stree$.
Winning strategy trees for Builder give rise to satisfying models, as we will explain in more detail in the next subsection.

\begin{lem}[\checklean{matchesFinite}]\label{l:matchesFinite}
All matches of $\Gm(\Gamma)$ are finite.
\end{lem}
\begin{proof}
  Recall from \autoref{l:localTabFinite} and \autoref{l:pdlTabFinite} that all (local) PDL-tableaux are finite.
  The proofs there used that all rule applications in a tableau stay inside the Fischer-Ladner closure of the sequent at the root.
  Similarly, here in the game $\Gm(\Gamma)$ all Prover positions are sequents that are subsets of the (loading-extended) Fischer-Ladner closure of $\Gamma$.
  We use a similar argument to \autoref{l:pdlTabFinite}, proceeding as follows.
  Suppose towards a contradiction that there is an infinite match.
  Consider the infinite list of sequents given by all Builder positions along this match.
  All sequents consist of (possibly loaded) formulas from the Fischer-Ladner closure of $\Gamma$.
  Because the Fischer-Ladner closure is finite there can only be finitely many different sequents occurring in the match.
  There cannot be a free repeat, because it would end the match with a win for Builder.
  This means that from some position onwards the sequents at all later positions are loaded.
  As the match from this point is infinite by assumption, some sequent must be repeated within this loaded part.
  This gives us a \lpr that ends the match (with a win for Prover), contradicting the assumption that the match is infinite.
\end{proof}

As a corollary of \autoref{l:matchesFinite}, we obtain \emph{determinacy} of the tableau game.
\begin{cor}[\checklean{gamedet}]\label{lem:gamedet}
$\Gm(\Gamma)$ is determined, that is, either Builder or Prover has a winning strategy in $\Gm(\Gamma)$.
\end{cor}

As a consequence of \autoref{lem:gameP} and \autoref{lem:gamedet}, in order to prove the completeness of the PDL-tableau system, it suffices to show that, for any sequent $\Gamma$, if Builder has a winning strategy in $\Gm(\Gamma)$ then $\Gamma$ is satisfiable.

\subsection{From winning strategies to model graphs}\label{subsec:StrToMoG}

Throughout this section, we assume a winning strategy for Builder in $\Gm(\Gamma)$
which is represented by the strategy tree $\Stree = (T,\tedge,r,\Lambda,g)$.

We first gather some basic observations and definitions concerning $\Stree$ --- some of which are analogous to the ones we used for tableaux.
We will write $\tedgeT$ for the transitive closure of $\tedge$ and $\tedgeRT$
for the reflexive transitive closure of $\tedge$.
A \emph{path} in $\Stree$ is a sequence $\pi = {(s_{i})}_{1\leq i\leq k}$ such
that $s_{i} \tedge s_{i+1}$, for all $i$; a \emph{branch} of $\Stree$
is a maximal path.
Observe that branches of $\Stree$ correspond to matches of $\Gm(\Gamma)$ in which
Builder plays her winning strategy; in particular (cf.~\autoref{d:builder-tree}), if $s \tedge t$ and $g(t) = (R,\phi)$ then 
$\Lambda(t)$ is a rule child of $\Lambda(s)$ of the application of $R$ with principal formula $\phi$.
It follows that a leaf of $\Stree$ is either labeled with a sequent to which no
rule is applicable, or it is a \emph{free repeat leaf}; that is, it has a
unique ancestor in $\Stree$ with the same, free, label.
Given such a free repeat leaf $u$, we will denote this ancestor as $c(u)$ and refer to it
as the companion of $u$.
By the same argumentation that we used to establish the finiteness of local PDL-tableaux, one can use \autoref{f:lprAreCritical} to show that $(M)$ must be used on the path from $c(u)$ to $u$.

\begin{defi}[\lean{PreState}]\label{d:prestate}
A node $t$ in $\Stree$ is called \emph{initial} if either $t = r$ or else $t$
is the child of a modal rule instance (i.e., $g(t) = ((M),\phi)$ for some $\phi$).
A \emph{pre-state} is a maximal sequence $\pi = {(s_{i})}_{1\leq i\leq k}$ of nodes in $\Stree$ such that $s_{1}$ is initial, there is no application on $\pi$ of the rules $(M)$, $(L+)$ or $(L-)$, i.e., for all $s_i$ on $\pi$, $g(s_i) = (R,\phi)$ implies that $R \notin \{(M),(L+),(L-)\}$,
and 
\begin{itemize}
\item either $\pi$ is a path on $\Stree$, 
\item or $\pi$ is a special concatenation $\pi = \pi'\pi''$ of two paths $\pi'$ and $\pi''$.
In this case we require that the last node $s$ of $\pi'$ is a free repeat leaf, and that $\pi''$ starts from its companion node $c(s)$, which means that the first node of $\pi''$ is an ancestor of the last node of $\pi'$.
The final requirement is then that the entire path $\pi''$ is an initial segment of the (unique) path from $c(s)$ to $s$.
\end{itemize}
Given any path or pre-state $\pi = {(s_{i})}_{1\leq i\leq k}$ on $\Stree$, we
define $\Lambda(\pi) \isdef \bigcup_{i} \Lambda(s_{i})$.
\end{defi}

Recall that, given a (possibly) loaded formula $\phi$, we let $\phi^-$ denote its \emph{unloaded version}; we will let $\Lambda^-(\pi)$ denote the set of formulas we obtain by unloading every formula in $\Lambda(\pi)$.
\autoref{l:cp1} below states that every set of the form $\Lambda^-(\pi)$ where $\pi$ is a pre-state is a locally consistent and saturated set.
In its proof, the following auxiliary result will be necessary, which provides information on how we can collect the formulas in $\Lambda^-(\pi)$ from the label sets on its path.

We say that $\phi$ is the principal formula at $s_{i}$ if $s_i \tedge s_{i+1}$ and $g(s_{i+1})=(R,\phi)$, cf.~\autoref{d:builder-tree}, i.e.,
$\phi \in \Lambda(s_{i})$, and $\Lambda(s_{i+1})$ is one of the rule children resulting from applying $R$ to the sequent $\Lambda(s_{i})$ with $\phi$ as its principal formula.

\begin{lem}[\checklean{PreState.formsCases}]\label{l:cp1a}
Let $s_{j}$ be some node on a pre-state $\pi = {(s_{i})}_{1\leq i \leq k}$ and let $\phi \in \Lambda(s_{j})$ be some formula. Then, either
\begin{enumerate}[(a)]
\item $\phi \in \Lambda(s_{k})$, and $\phi$ has one of the following forms: $p$, $\lnot p$, $[a]\psi$, $\lnot[a]\psi$ or
$\lnot\loaded{[a]}\psi$ (in the case $\phi = \lnot\loaded{[a]}\psi$, we also have $s_{j} \tedgeRT s_{k}$); or
\item
there is some $i$ with $j \leq i < k$ such that $\phi$ is the principal formula
at $s_{i}$ and, in case $\phi$ is loaded, $s_{j} \tedgeRT s_{i}$.
\end{enumerate}
\end{lem}
\begin{proof}
Suppose that no rule other than possibly $(M)$, $(L+)$ or $(L-)$ can make $\phi$ principal.
In that case, it must be that $\phi$ has one of the forms specified in case (a). 
Since no rule makes~$\phi$ principal, and since every non-principal formula passes unchanged to each child (i.e., the parent and every child have the same context), it follows that $\phi \in \Lambda(s_{k})$.
Now suppose that $\phi = \lnot\loaded{[a]}\psi$. This means that every $s_{j'}$ with $j\leq j'\leq k$ is loaded (since $(L-)$ is not applied). This means that we cannot encounter a free repeat leaf along this subsequence, hence this subsequence corresponds to a path in $\Stree$, i.e., $s_{j} \tedgeRT s_{k}$.

Now suppose that some rule other than $(M)$, $(L+)$ or $(L-)$ can make $\phi$ principal.
By the maximality condition in the definition of pre-state, and the `iff' in~\autoref{d:builder-tree} that ensures that every applicable rule is applied, it follows that $\phi$ is principal 
by some rule other than $(M)$, $(L+)$ or $(L-)$ at $s_i$ such that $j\leq i$; then one of the children of this rule appears in the pre-state (by maximality) as the label at $s_{i+1}$. Thus $i<k$.
Finally, if $\phi$ is loaded, then the path from $s_j$ to $s_i$ is loaded. As in the argument in the preceding paragraph, this means that we cannot encounter a free repeat leaf along this subsequence. Hence, this subsequence corresponds to a path in $\Stree$, i.e., $s_{j} \tedgeRT s_{i}$.
\end{proof}
Note that statement~(b) cannot be strengthened to $s_{j} \tedgeRT s_{k}$ since it is possible for loading to be removed even without~$(L-)$ being applied, due to the test operator.

\begin{lem}\label{l:cp3}
Let $\lnot[\alpha]\phi$ be a PDL formula where $\alpha$ is not atomic and let $\pi$ be a pre-state such that $\lnot[\alpha]\phi \in \Lambda(\pi)$.
Then there is a pair $(X,\ol{\delta}) \in \Dset{\alpha}$ such that $X \cup \{ \lnot\Box(\ol{\delta},\phi) \} \subseteq \Lambda(\pi)$.
Similarly, if $\lnot\loaded{[\alpha]}\phi \in \Lambda(\pi)$, then there is a pair $(X,\ol{\delta}) \in \Dset{\alpha}$ such that $X \cup \{ \lnot\loaded{\Box}(\ol{\delta},\phi) \} \subseteq \Lambda(\pi)$.
\end{lem}

\begin{proof}
This follows from the proof rules (\autoref{d:localTableau}), the definition of the $\unfold$ functions, and the previous lemma which guarantees that the formula will be made principal at some node in the pre-state (by a $(\diam)$ or $(\loaded{\diam})$ rule).
\end{proof}

\begin{lem}[\checklean{PreState.locConsSatBas}]\label{l:cp1}
Let $\pi$ be a pre-state.
Then $\Lambda^-(\pi)$ is saturated and locally consistent, and the sequent
labelling the last node of $\pi$ is basic.
\end{lem}

\begin{proof}
Saturation of $\Lambda^-(\pi)$ follows from \autoref{l:cp1a} and \autoref{l:cp3}.
Local consistency follows from the soundness of the local PDL rules (\autoref{l:localRuleTruth}) 
and the fact that $\Stree$ is winning for Builder, hence, in particular, no sequent label occurring in $\Stree$ is closed.
By maximality of $\pi$, the sequent $\Gamma$ labelling the last node $s$ on $\pi$ is either a sequent to which no rule can be applied, which implies that $\Gamma$ is basic, or it is a sequent such that all children of $s$ are obtained via an application of the modal rule (\autoref{d:Modalrule}). In this case, the side condition of $(M)$ entails that $\Gamma$ is basic.
\end{proof}

In the following key definition of the completeness proof, we see how to define a model graph from the winning strategy tree for Builder.
\begin{defi}[\lean{BuildTree.toModel}]\label{d:theModelGraph}
The structure $\kmodel_{\Stree} \isdef (W,\{\reach{a} \mid a \in \progs_0 \},V)$ is defined as follows.
Its domain $W$ consists of all sets of the form $\Lambda^-(\pi)$, where $\pi$ is a pre-state of $\Stree$.
For an atomic program $a$ we define the relation $R_{a} \subseteq W \times W$ by putting
\[
(X,Y) \in R_{a} \ \defiff \ X_{a} \cup \{ \lnot\phi \} \subseteq Y, \text{ for some } \lnot[a]\phi \in X.
\]
Finally, the valuation $V$ is given by
\[
V(p) \isdef \{ X \in W \mid p \in X \}.
\]
\end{defi}

Our goal will be to prove \autoref{thm:strmg} below, which states that $\kmodel_{\Stree}$ is indeed a model graph.
That is, we need to show that $\kmodel_{\Stree}$ satisfies the conditions (a), (b), (c) and (d) from \autoref{d:ModelGraph}.
Condition (a) will follow from the already stated \autoref{l:cp1}.
Conditions (b, c) will follow immediately from the definition of $\kmodel_{\Stree}$.
The main challenge is to show that $\kmodel_{\Stree}$ satisfies condition (d).
To show this, it suffices to prove that for every pre-state $\pi$, whenever $\lnot[\alpha]\phi \in \Lambda^-(\pi)$, there exists a pre-state $\pi'$ such that $(\Lambda^-(\pi),\Lambda^-(\pi')) \in Q_{\alpha}$ and $\lnot\phi \in \Lambda^-(\pi')$.

Observe that if $\lnot[\alpha]\phi \in \Lambda^-(\pi)$, then there are two possible cases:
either $\lnot[\alpha]\phi \in \Lambda(\pi)$, or
$\lnot\loaded{[\alpha]}\psi \in \Lambda(\pi)$ for some loaded formula $\psi$ such that $\phi = \psi^-$.
We will prove these two cases separately, starting with the loaded case; the
unloaded case will be shown in \autoref{l:cp5}.

\begin{lem}[\checklean{PreState.loadedExists}]\label{l:cp2new}
Let $t$ be a node of $\Stree$ on the pre-state $\pi$, and let
$\lnot\loaded{[\alpha]}\phi$ be a loaded formula in $\Lambda(t)$.
Then there is a node $u$ on $\Stree$ such that $t \tedgeT u$, $\lnot\phi \in
\Lambda(u)$ and $u$ lies on some pre-state $\rho$ such that
$(\Lambda^-(\pi), \Lambda^-(\rho)) \in Q_{\alpha}$
where $Q_\alpha$ is obtained by applying \autoref{d:Qrelation} to the family $\{\reach{a} \mid a \in \progs_0 \}$ in $\kmodel_{\Stree}$ from \autoref{d:theModelGraph}.
\end{lem}

\begin{proof}
The intuition underlying the proof is to look at matches that continue the
partial match represented by the node $t$ along $\pi$, and where Builder plays
her winning strategy.
If Prover keeps the same formula loaded, he can ensure that the
resulting path from $t$ stays loaded, until a node $u$ is reached such that
$\lnot\phi \in \Lambda(u)$.
In particular, no back edge of $\Stree$ is taken.

Formally, we prove this lemma by an outer induction on $\alpha$, and an inner
induction on the converse of the tree order $\tedgeT$.

In the base case of the outer induction we have $\alpha = a$ for some atomic program $a$.
Then by \autoref{l:cp1a} we find $\lnot[\loaded{a}]\phi \in \Lambda(s_{k})$
and $t \tedgeRT s_{k}$.
Now suppose that at this position $s_{k}$ Prover applies the modal rule with
principal formula $\lnot[\loaded{a}]\phi$.
Let $u$ be the induced $\tedge$-successor of $s_{k}$, then we find $t \tedgeT u$
and $\lnot\phi \in \Lambda(u)$.
Furthermore, let $\rho$ be any pre-state starting at $u$, then we have
$\lnot\phi \in \Lambda(\rho)$ and $(\Lambda^-(\pi),\Lambda^-(\rho)) \in Q_{a}$.
That is, all conditions of the Lemma are satisfied.

In the inductive step of the outer induction proof $\alpha$ is a non-atomic
program.
By \autoref{l:cp1a} there is a node $t_{1}$ on $\pi$ such that $t \tedgeT t_{1}$
and $\lnot\loaded{[\alpha]}\phi$ is the principal formula at $t_{1}$.
It~then follows by the loaded diamond rule (\loaded{$\diam$}) that for some pair $(X,\ol{\delta}) \in
\Dset{\alpha}$ we have $X \cup \{ \lnot\loaded{\Box}(\ol{\delta},\phi) \}
\subseteq \Lambda(t_{2})$, where $t_{2}$ is the immediate successor of $t_{1}$
on $\pi$.
It follows that $X \subseteq \Lambda(\pi)$, so that we find $(\Lambda^{-}(\pi),
\Lambda^{-}(\pi)) \in Q_{X}$.

Now make a case distinction.

\begin{description}
\item[Case $\ol{\delta} = \varepsilon$]
In this case it is clear that $u \isdef t_{2}$ and $\rho \isdef \pi$ satisfy the
conditions of the lemma.
For instance, $(\Lambda^{-}(\pi), \Lambda^{-}(\pi)) \in Q_{\alpha}$ is immediate
from \autoref{l:cp3a}.

\item[Case $\ol{\delta} \neq \varepsilon$, $\alpha$ is not of the form $\beta^{\ast}$]
In this case by \autoref{f:terminationHelperDiamond} $\ol{\delta}$ is of the form $\ol{\delta} =
a\delta_{1}\cdots \delta_{n}$, where $a$ is atomic and every $\delta_{i}$ is a
\emph{proper} subprogram of $\alpha$.
Hence we may apply the outer inductive hypothesis to $a, \delta_{1}, \ldots,
\delta_{n}$, successively.
It follows that there are nodes $u_{0}, \ldots, u_{n}$ in $\Stree$ such that
$t \tedgeT u_{0} \tedgeT \cdots \tedgeT u_{n}$, $\lnot\phi \in \Lambda(u_{n})$,
as well
as pre-states $\rho_{0},\ldots,\rho_{n}$ such that $u_{i}$ lies on $\rho_{i}$
(for each $i$), $(\Lambda^{-}(\pi),\Lambda^{-}(\rho_{0})) \in Q_{a}$ and
$(\Lambda^{-}(\rho_{i}),\Lambda^{-}(\rho_{i+1})) \in Q_{\delta_{i}}$
(for each $i<n$).
We then have $t \tedgeT u_{n}$ and from \autoref{l:cp3a} it easily follows that
$(\Lambda^{-}(\pi),\Lambda^{-}(\rho_{n})) \in Q_{X} \seqc Q_{\ol{\delta}}
\subseteq Q_{\alpha}$ as required.

\item[Case $\ol{\delta} \neq \varepsilon$, $\alpha$ is of the form $\alpha =
\beta^{\ast}$]
In this case $\ol{\delta}$ is of the form $\ol{\delta} =
a\delta_{1}\cdots \delta_{n}\beta^{*}$, where $a$ is atomic and
$(X,a\delta_{1}\cdots \delta_{n}) \in \Dset{\beta}$.
In particular, every $\delta_{i}$ is a subprogram of $\beta$, and hence a
\emph{proper} subprogram of $\alpha$.

Reasoning as in the previous case (but with $\lnot\loaded{[\beta^{*}]}\phi$ taking
the role of $\lnot\phi$) we find a node $u'$ on $\Stree$ and a pre-state $\rho'$
such that $u'$ lies on $\rho'$, $t \tedgeT u'$, $\lnot\loaded{[\beta^{*}]}\phi
\in \Lambda(u')$, and $(\Lambda^{-}(\pi),\Lambda^{-}(\rho')) \in Q_{\beta}$.
We now use the \emph{inner} induction hypothesis on $u'$, $\rho'$ and
the loaded formula $\lnot\loaded{[\beta^{*}]}\phi = \lnot\loaded{[\alpha]}\phi$.
This gives a node $u$ and a pre-state $\rho$ such that $u$ lies on $\rho$,
$\lnot\phi \in \Lambda(u)$, $u' \tedgeT u$ and
$(\Lambda^{-}(\rho'),\Lambda^{-}(\rho)) \in Q_{\beta^{*}}$.
From this we immediately obtain $t \tedgeT u$ and
$(\Lambda^{-}(\pi),\Lambda^{-}(\rho)) \in Q_{\beta} \seqc Q_{\beta^{*}}
\subseteq Q_{\beta^{*}} = Q_{\alpha}$, which means that we are done.
\qedhere
\end{description}
\end{proof}

For unloaded diamond formulas we have to work a bit harder.
For this reason we introduce the following auxiliary lemma that will be used to show \autoref{l:cp5}, the analogue of \autoref{l:cp2new} for unloaded diamonds.

\begin{lem}[\checklean{PreState.freeAtomicStep}]\label{l:cp4}
Let $\pi$ be a pre-state with $\lnot[a][\delta_{1}]\cdots[\delta_{n}]\phi \in \Lambda(\pi)$, where $n$ is maximal, i.e., $\phi$ itself is not a box formula.
Then there is a pre-state $\pi' = {(s_{i})}_{1\leq i\leq k}$ such that $(\Lambda^-(\pi),\Lambda^-(\pi')) \in Q_{a}$ and $\lnot\loaded{[\delta_{1}]\cdots[\delta_{n}]}\phi \in \Lambda(s_{1})$.
\end{lem}

\begin{proof}
Let $\pi = {(t_{i})}_{1\leq i\leq m}$, then we find
$\lnot[a][\delta_{1}]\cdots[\delta_{n}]\phi \in \Lambda(t_{m})$ by
\autoref{l:cp1a}.
We make a case distinction.
\begin{description}
\item[Case $\Lambda(t_{m})$ is free]
Prover may now simply apply the rule $(L+)$, loading the formula $\lnot[a][\delta_{1}]\cdots[\delta_{n}]\phi$ to
$\lnot\loaded{[a][\delta_{1}]\cdots[\delta_{n}]}\phi$
(recall that $(L+)$ loads maximally by definition),
and subsequently apply the modal rule.
The next position in the game will contain the formula
$\lnot\loaded{[\delta_{1}]\cdots[\delta_{n}]}\phi$, and any pre-state $\pi'$ starting at $s_{1}$ will meet the criteria of the Lemma.
\item[Case $\Lambda(t_{m})$ is loaded]
Prover can now only load the formula $\lnot[a][\delta_{1}]\cdots[\delta_{n}]\phi$
if he first \emph{liberates} the loaded formula of $\Lambda(t_{m})$,
by means of an application of the rule $(L-)$.
Recall that we can only apply the rule $(L-)$ if $\Lambda(t_{m})$ is basic.
This is established by applying \autoref{l:cp1a} to each formula in $\Lambda(t_{m})$.
Let $s$ be the child of $t_{m}$ in the strategy tree $\Stree$ such that
$\Lambda(s)$ is the unloaded version of $\Lambda(t_{m})$.
We make a further case distinction.
\begin{description}
\item[Case $s$ is not a free repeat]
We may now simply continue similarly to the previous case:
Prover may apply the rule $(L+)$ at position $s$, etc.

\item[Case $s$ is a free repeat]
By definition of a repeat this means that $s$ has a companion, that is, an \emph{ancestor} $s^{-}$ in $\Stree$ such that $\Lambda(s^{-}) = \Lambda(s)$.
Let $\rho$ be the pre-state such that $s^-$ is either the final position of $\rho$ or obtained from this final position by the unloading rule.
We may now continue as before, finding the required pre-state $\pi'$ in $\Stree$ by applying first the loading rule and then the modal rule at the node $s^{-}$.
(Note that by definition of the relation $R_{a}$, $\Lambda^-(\pi')$ is not only an $a$-successor of $\Lambda^-(\rho)$, but also of $\Lambda^-(\pi)$.)
\qedhere
\end{description}
\end{description}
\end{proof}

We can now prove the analogue of \autoref{l:cp2new} for non-loaded diamonds.

\begin{lem}[\checklean{freeDiamondExistence}]\label{l:cp5}
Let $\pi$ be a pre-state, and let $\lnot[\alpha]\phi$ be a formula in $\Lambda(\pi)$.
Then there is a pre-state $\pi'$ such that $(\Lambda^-(\pi),\Lambda^-(\pi')) \in Q_{\alpha}$ and $\lnot\phi \in \Lambda^-(\pi')$.
\end{lem}

\begin{proof}
We will prove the lemma by induction on $\alpha$, but for the induction to work smoothly we need a slightly stronger version of the statement.

\emph{Claim}:
Let $\pi$ be a pre-state, and let $\lnot[\alpha]\phi$ be a formula in $\Lambda(\pi)$.
Write $\phi = [\eta_1]\cdots[\eta_k]\psi$, where $\psi$ itself is not a box formula, and define $\loaded{\phi} \isdef \loaded{[\eta_1]\cdots[\eta_k]}\psi$.
Then there is a pre-state $\pi'$ such that $(\Lambda^-(\pi),\Lambda^-(\pi')) \in Q_{\alpha}$ and $\lnot\phi$ or $\lnot\loaded{\phi}$ belongs to $\Lambda(\pi')$.

\begin{description}
\item[Case $\alpha=a$]
Immediate by \autoref{l:cp4}.
\item[Case $\alpha={\tau?}$]
It is easy to see that if $\lnot[\tau?]\phi \in \Lambda(\pi)$, then we have $\{ \tau, \lnot\phi \} \subseteq \Lambda(\pi)$.
From $\tau \in \Lambda(\pi)$ it follows that $(\Lambda^-(\pi),\Lambda^-(\pi)) \in Q_{\tau?}$ so that we can take $\pi$ itself as our $\pi'$.

\item[Case $\alpha$ is not atomic]
By \autoref{l:cp3} and \autoref{l:cp3a} there is a pair $(X,\ol{\delta}) \in \Dset{\alpha}$ such that
$X \subseteq \Lambda(\pi)$, $\lnot\Box(\ol{\delta},\phi) \in \Lambda(\pi)$, and $Q_{X,\ol{\delta}} \subseteq Q_{\alpha}$.

First consider the case where $\ol{\delta} = \emptylist$.
This implies that $(\Lambda^-(\pi),\Lambda^-(\pi)) \in Q_{X,\ol{\delta}} \subseteq Q_{\alpha}$ and $\lnot\Box(\emptylist,\loaded{\phi}) = \lnot\loaded{\phi} \in \Lambda(\pi)$, so that we are done.

Now consider the case where $\ol{\delta} \neq \emptylist$.
In this case $\ol{\delta}$ must be of the form $\ol{\delta} = a\ol{\gamma}$, for some atomic program $a$ and some program list $\ol{\gamma} = \gamma_{1}\cdots \gamma_{n}$.
It follows by \autoref{l:cp4} that there is some pre-state $\rho$ such that
\[
(\Lambda^-(\pi),\Lambda^-(\rho)) \in Q_{a}
\]
and $\lnot\loaded{[\gamma_{1}]\cdots[\gamma_{n}]}\loaded{\phi} \in \Lambda(\rho)$.
But then by $n$ iterative applications of \autoref{l:cp2new} we find a pre-state
$\pi'$ such that $\lnot\loaded{\phi} \in \Lambda(\pi')$ and
\[
(\Lambda^-(\rho),\Lambda^-(\pi')) \in Q_{\gamma_{1}} \seqc \cdots \seqc Q_{\gamma_{n}}.
\]
We note that here we do \emph{not} need the induction hypothesis since we refer to the earlier lemma on loaded formulas.
But from $X \subseteq \Lambda(\pi)$ it follows that $(\Lambda^-(\pi),\Lambda^-(\pi)) \in Q_{X}$, so that with $(\Lambda^-(\pi),\Lambda^-(\rho)) \in Q_{a}$ we find that
\[
(\Lambda^-(\pi), \Lambda^-(\pi')) \in
Q_{X} \seqc Q_{a} \seqc Q_{\gamma_{1}} \seqc \cdots \seqc Q_{\gamma_{n}}
= Q_{X} \seqc Q_{\ol{\delta}} \subseteq Q_{\alpha}
\]
as required.
\end{description}
Finally, it is easy to see how the Lemma follows from the Claim.
\end{proof}

\begin{thm}[\checklean{strmg}]\label{thm:strmg}
Let $\Stree$ be the strategy tree representing a winning strategy for Builder in $\Gm(\Gamma)$.
Then $\kmodel_{\Stree}$ is a model graph with some state $X \in W$ such that $\Gamma \subseteq X$.
\end{thm}

\begin{proof}
  By \autoref{d:theModelGraph}, $\kmodel_{\Stree}$ includes a state for every pre-state of $\Stree$.
  The root of $\Stree$ is labelled with $\Gamma$, hence there exists a pre-state $\pi$ of $\Stree$ such that $\Gamma \subseteq \Lambda^-(\pi) \in W$.
  It remains to check the four conditions of \autoref{d:ModelGraph} to confirm that $\kmodel_{\Stree}$ is indeed a model graph.
\begin{enumerate}[align=left]
\item[(a)] For this condition we need to check that each element of $W$ is a locally consistent saturated set of formulas.
Let $X$ be an element of $W$.
From the definition of $\kmodel_{\Stree}$, it follows that there exists a pre-state $\pi$ on $\Stree$ such that $X = \Lambda^-(\pi)$.
Thus, we can apply \autoref{l:cp1} to obtain that $X$ is locally consistent and saturated.

\item[(b, c)] Immediate from the definition of $\kmodel_{\Stree}$ in \autoref{d:theModelGraph}.
\item[(d)]
Assume that $X \in W$ and $\lnot [\alpha]\phi \in X$.
Let $\pi$ be any pre-state of $\Stree$ such that $X = \Lambda^-(\pi)$.
Then for some node $t$ on the pre-state $\pi$ we either have $\lnot [\alpha]\phi \in \Lambda(t)$, or else $\lnot [\loaded{\alpha}]\psi \in \Lambda(t)$ for some (possibly loaded) formula $\psi$ such that $\phi = \psi^-$.
To show that there exists a $Y \in W$ such that $(X,Y) \in  Q_{\alpha}$ and
$\lnot \phi \in Y$, we use \autoref{l:cp5} in the first case and
\autoref{l:cp2new} in the second.
\qedhere
\end{enumerate}
\end{proof}

\subsection{Completeness proof}\label{subsec:CompProof}

We now have all the material that is needed for the completeness proof.

\begin{proof}[Proof of \autoref{t:completeness}]
Suppose that $\Gamma$ is a consistent sequent, i.e., there is no closed tableau for $\Gamma$.
Then by \autoref{lem:gameP} and \autoref{lem:gamedet} Builder has a winning strategy $\Stree$ in the tableau game $\Gm(\Gamma)$.
By \autoref{thm:strmg} we obtain the model graph $\kmodel_{\Stree}$ with a state $X$ such that $\Gamma \subseteq X$.
Using the Truth \autoref{l:truthLemma} we find that $\kmodel_{\Stree}, X \models \Gamma$, hence $\Gamma$ is satisfiable.
\end{proof}

\section{Interpolation: Preliminaries and Overview}\label{sec:Interpolation}

In this section, we provide an overview of the interpolation argument.
In the next section, we define the tableau system that is used; it is an extension of the PDL-tableau system we encountered in the previous sections such that each sequent is formally split into a left and right component. The interpolating formulas are defined in~\autoref{sec:DefinitionInterpolants}; in the case of proper clusters, this introduces the crucial quasi-tableau construction. Finally, their correctness is established in~\autoref{section:MainTheorem}.

\subsection{Preliminaries}

\begin{defi}[\lean{Formula.voc}]\label{d:vocab}
Let $\xi$ be a formula or a program.
The vocabulary~$\voc(\xi)$ of $\xi$ is defined as the set of atomic formulas and atomic programs occurring in $\xi$.
Moreover, for a sequent $\Gamma$, we define $\voc(\Gamma) \isdef  \bigcup \{ \voc(\phi) \mid \phi \in \Gamma \}$.
\end{defi}

\begin{defi}[\lean{Interpolant}]\label{d:Interpolant}
Let $\phi$ and $\psi$ be PDL formulas such that $\models \phi \to \psi$.
An \emph{interpolant} for $\phi \to \psi$ is a formula $\theta$ such that
$\voc(\theta) \subseteq \voc(\phi) \cap \voc(\psi)$,
$\models \phi \to \theta$, and
$\models \theta \to \psi$ hold.
\end{defi}

We generally use the Greek letter $\theta$ to denote interpolants.

With these definitions, we state our main theorem.

\begin{thm}[\checklean{interpolation}]\label{t:interpolation}
PDL has the Craig Interpolation Property.
That is, for all formulas $\phi$ and $\psi$ such that $\models \phi \to \psi$, there exists an interpolant for $\phi \to \psi$.
\end{thm}

A typical and important consequence of the Craig Interpolation Property is the Beth definability property which states that implicit definability implies explicit definability.
\begin{defi}[\lean{Formula.impDef}, \lean{Formula.expDef}]\label{d:ImpExpDef}
  A formula $\phi(p)$ is said to \emph{implicitly define} $p$ in terms of $\voc(\phi)\setminus \{p\}$, if
  \[ \phi(p_0), \phi(p_1) \models  p_0 \leftrightarrow p_1 \]
  where $p_0, p_1 \notin \voc(\phi)$.

  A formula $\psi$ is an \emph{explicit definition} of $p$ relative to $\phi$, if $\voc(\psi) \subseteq \voc(\phi) \setminus \{ p \}$ and
  \[ \phi(p) \models p \leftrightarrow \psi. \]
\end{defi}

The argument that the Craig Interpolation Property entails the Beth definability property for PDL is similar to that for many other propositional logics.
We refer to the work of Hoogland~\cite{hoog:defi01} for further details and examples of other logics.

\begin{cor}[\checklean{beth}]\label{cor:beth}
  PDL has the Beth definability property.
  That is, for any formula~$\phi(p)$ that implicitly defines $p$ in terms of $\voc(\phi) \setminus \{p\}$,
  there exists an explicit definition of~$p$ relative to~$\phi$.
\end{cor}

\begin{proof}
  Fix a formula $\phi(p)$.
  Suppose that $\phi(p_0), \phi(p_1) \models  p_0 \leftrightarrow p_1$.
By \autoref{t:deduction}
$\models \phi(p_0) \to (\phi(p_1) \to  (p_0 \leftrightarrow p_1)) $.
By propositional reasoning we have:
\[ \models (\phi(p_0) \land p_0) \to (\phi(p_1) \to p_1). \]
By \autoref{t:interpolation}, there exists an interpolant, i.e., a formula $\theta$ such that:
\begin{enumerate}
\item[(0)] $\voc(\theta) \subseteq \voc(\phi(p_0) \land p_0) \cap \voc(\phi(p_1) \to p_1)$
\item[(1)] $\models (\phi(p_0) \land p_0) \to \theta$
\item[(2)] $\models \theta \to (\phi(p_1) \to p_1)$
\end{enumerate}

By Property~(0), we have that $\{p, p_0, p_1\}\cap\voc(\theta)=\emptyset$.
Hence, by uniformly substituting $p_0$ with $p$ in Property~(1), we  obtain
$\models (\phi(p) \land p) \to \theta$
which by \autoref{t:deduction} implies
$\phi(p) \models p \to \theta$.
Similarly, from Property~(2) we obtain
$\models \theta \to (\phi(p) \to p)$
which by \autoref{t:deduction} implies
$\phi(p) \models \theta \to p$.
Combining these, we obtain
\[ \phi(p) \models p \leftrightarrow \theta \]
which means that $\theta$ is an explicit definition of $p$ relative to $\phi$.
\end{proof}

The proof-theoretic tool that we use in this paper is a tableau calculus. Definitions about the tableau calculus are typically stated in terms of (un)satisfiability rather than validity/provability, so let us restate interpolation in those terms.
\begin{defi}[\lean{PartInterpolant}]\label{d:PartInterpolant}
  Let $\Gamma_1$ and $\Gamma_2$ be two sequents such that $\Gamma_1 \cup \Gamma_2$ is unsatisfiable.
  An \emph{interpolant} for $\Gamma_1$ and $\Gamma_2$ is a formula $\theta$ such that:
  \begin{itemize}
  \item $\voc(\theta) \subseteq \voc(\Gamma_1) \cap \voc(\Gamma_2)$,
  \item $\Gamma_1 \models \theta$, and
  \item $\Gamma_2 \models \lnot \theta$.
  \end{itemize}
\end{defi}
The following fact establishes the connection between \autoref{d:Interpolant} and the above:
\begin{fact}[\checklean{tautImp_iff_SequentUnsat}]\label{fact:intToPartInt}
    Let $\phi, \psi$ be formulas.
    Then $\models \phi \to \psi$ iff $\{\phi\}\cup\{\lnot \psi\}$ is unsatisfiable.
    Furthermore, if $\models \phi \to \psi$, then a formula
    $\theta$ is an interpolant for $\phi \to \psi$
    if and only if
    $\theta$ is an interpolant for $\{\phi\}$ and $\{\lnot \psi\}$.
\end{fact}

\subsection{An overview of the construction}\label{subsec:ProofIdea}
Before going into the technical details of our construction of the interpolant, we provide a motivating overview. 
The reader may want to skip parts of this sketch for now, and come back to it when engaged in the technical details that are to follow.

Let $\Sigma_1\cup\Sigma_2$ be an unsatisfiable set of PDL formulas.
The task is to compute an interpolant (\autoref{d:PartInterpolant}) for $\Sigma_1$ and $\Sigma_2$.
By completeness (\autoref{t:completeness}) and \autoref{def-tab-closed}, there is a closed PDL-tableau rooted at $\Sigma_1 \cup \Sigma_2$.\footnote{Readers more familiar with proof systems may interpret this as a proof of $\lnot((\bigwedge \Sigma_1) \land (\bigwedge \Sigma_2))$.}
Recognizing the distinct roles of $\Sigma_1$ and $\Sigma_2$ in the interpolation statement, the starting point is an extended `split PDL-tableau' system (defined formally in the next section) such that the closed PDL-tableau is rooted at the split sequent $\Sigma_1;\Sigma_2$. 
Thus, an interpolant for $\Sigma_1;\Sigma_2$ is an interpolant for $\Sigma_1$ and $\Sigma_2$ as per \autoref{d:PartInterpolant}.

We observe that the semicolon notation is overloaded: previously we saw that it denotes sequential composition of programs --- here it is used to demarcate the formulas with ancestry to the first/left component $\Sigma_1$ and second/right component $\Sigma_2$.
Each use of the semicolon is still unambiguous because sequential composition is only applied to programs and never to sequents or formulas.

The tableau rules have the property that at most one component contains a loaded formula; such a component is called the loaded component.
The split sequents $\Sigma_1; \Sigma_2$ and $\Sigma_1'; \Sigma_2'$ are considered equal if $\Sigma_1 = \Sigma_1'$ and $\Sigma_2 = \Sigma_2'$.

The standard proof-theoretic route to interpolation is Maehara's method (see~\cite{Maehara1961} and~\cite[p.~33]{Takeuti1975:ProofTheory}).
The idea is to define interpolants for the leaves of the closed tableau and combine the interpolants that have been obtained at the children of a rule to define an interpolant for its parent.
For example, we add a conjunction or disjunction, or prefix the interpolant obtained in the previous step with a modality. Defining an interpolant for a closed leaf is straightforward (\autoref{subsec:InterpolationForSingletonClusters}) because its unsatisfiability is directly evident. E.g., from a leaf labeled $\Sigma_1', p; \Sigma_2', \lnot p$
we can read off the interpolant $p$, and from $\Sigma_1', p, \lnot p; \Sigma_2$ we read off the interpolant $\bot$.
However, a closed PDL-tableau may also contain leaves that are loaded-path repeats. Recall that a loaded-path repeat has the same label as a sequent closer to the root (its companion), and every sequent on the path from companion to loaded-path repeat is loaded. Maehara's method cannot be applied here because it is not possible to read off an interpolant from the loaded-path repeat: the unsatisfiability of its label is not evident.

\textbf{Adapting Maehara's method to cyclic tableaux.}
By the above discussion, the real challenge is to construct an interpolant (\autoref{d:PartInterpolant}) at the companion of a loaded-path repeat in a closed PDL-tableau. It should be noted that a general analysis of such leaves requires a richer graph structure than a branch, to account for the nesting of companion nodes from multiple loaded-path repeats.
For this purpose, we consider a \emph{cluster}~$C$ (\autoref{d:cluster}) rooted at $\Gamma_1;\Gamma_2$; this is a maximal subgraph of the closed tableau that is strongly connected if for each loaded-path repeat we were to add an edge to its companion.
In the following, \emph{we  assume that $\Gamma_1$ is unloaded and $\Gamma_2$ is loaded}, and construct a formula~$\theta$ such that $\Gamma_1\models\theta$, $\Gamma_2\models\lnot\theta$, and
$\voc(\theta) \subseteq \voc(\Gamma_1) \cap \voc(\Gamma_2)$.
In the case that it is $\Gamma_1$ that is loaded and $\Gamma_2$ is unloaded, we first swap the sides of the cluster so that it is rooted at $\Gamma_2;\Gamma_1$ and then run the argument that follows to compute an interpolant. Finally, observe that the negation of an interpolant for $\Gamma_2;\Gamma_1$ is an interpolant for $\Gamma_1;\Gamma_2$.\footnote{An alternative to swapping the sides of the cluster is observing that the argument constructs a formula~$\theta$ that is relativized to the loading: the unloaded component entails $\theta$ and the loaded component entails $\lnot\theta$.}

Let us first consider a \emph{specific situation} where, on the loaded path from $c(t)$ to $t$, rules are only applied to the loaded component.
With the aim of reaching a fixed-point equation, we start by assigning
the proposition letter $q_{c(t)}$ as the candidate (`pre-interpolant') for the interpolant
at $t$.
Note that an interpolant at $t$ is also an interpolant at $c(t)$ since the two
nodes have the same label.
The idea is to proceed from $t$ towards $c(t)$, extending the pre-interpolant as
described below.
We ultimately compute a pre-interpolant $\iota_{c(t)}$ at $c(t)$.
We summarize the key ideas for this extension of Maehara's method in the items
\textbf{(\ref{loadedrule})} and \textbf{(\ref{loadedmodal})} below, and \textbf{(\ref{unloadedregion})} which is introduced after \autoref{e:properclusterinterpolant}.

\begin{enumerate}[label=\textbf{E\arabic*}, ref=E\arabic*]
\item\label{loadedrule} At a local rule in the loaded component, the pre-interpolant for the parent is the \emph{conjunction}\footnote{A disjunction formula in the loaded component contributes a conjunction to the pre-interpolant because we ultimately want the loaded component~$\Gamma_2$ to entail the \emph{negation} of the interpolant. Meanwhile, a conjunction formula in the loaded component does not change the pre-interpolant since the set of formulas in the sequent is already read conjunctively.}
of the pre-interpolants that were computed for the children of that rule.

\item\label{loadedmodal} For a modal rule on atomic program~$a$ --- the rule is applied to the unique loaded formula $\lnot\loaded{[a]}\phi$ in the loaded component --- the pre-interpolant for the parent is $[a]\iitp_{y}$ where~$\iota_y$ is the pre-interpolant already computed for the child.
\end{enumerate}

By repeatedly applying the above rules, and using the logical equivalence
$[a](\xi\land \chi)\equiv [a]\xi\land [a]\chi$, we have that the pre-interpolant $\iota_{c(t)}$ at $c(t)$ has
the form $[\alpha]q_{c(t)}\land\psi$, where $q_{c(t)}$ does not occur in $\alpha$ nor in~$\psi$.
Here~$\alpha$ is obtained as the composition of atomic programs that result from multiple applications of~\textbf{(\ref{loadedmodal})}.
Since $t$ and $c(t)$ are identically labeled, they have the same interpolant
so we reach the equation $q_{c(t)}\equiv [\alpha]q_{c(t)}\land\psi$.
Crucially, an equivalence of this form has a solution $q_{c(t)}=[\alpha^*]\psi$
in PDL, so this is how the pre-interpolant will be defined at companion nodes
\textbf{(\ref{unloadedregion})}.

So we have resolved the above specific situation where rules on the loaded path were only applied in the loaded component. Of course, in general, rules may also have been applied in the unloaded component, and branching rules in the unloaded component require particular attention. If we were to follow Maehara's method for such a branching rule, we would define the pre-interpolant for the
parent node as the \emph{disjunction}\footnote{Observe that a disjunction formula in the unloaded component contributes a disjunction to the pre-interpolant because we ultimately want the unloaded component $\Gamma_1$ to entail the interpolant.} of the pre-interpolants of its children.
However, when arriving at a companion node, the presence of disjunctions means that the fixed-point equation would not have the desired shape
$q_{c(t)}\equiv [\alpha]q_{c(t)}\land\psi$, and in such a case, it is unclear whether the equation even has a solution (in the language of PDL).

To arrive at a solution to this new issue, consider the \emph{special case} where
a branching rule is applied to the unloaded component of a node $t$ that has a child $s$ that belongs to the cluster~$C$, and a child $v$ that lies
outside of~$C$.
In this case, we may inductively assume that we have already found a proper
interpolant $\theta_v$ for the node $v$ and a pre-interpolant $\iitp_{s}$ for $s$, and
we define $\iitp_{t} \isdef [\lnot\theta_{v}?]\iitp_s$.
This formula is logically equivalent to the disjunction $\theta_{v}\lor\iitp_{s}$, but the difference is that we
have now `packaged' the interpolant $\theta_{v}$ which contains no occurrence of any
$q$-variable inside the boxed test program ${\lnot\theta_{v}?}$.
Consequently, taking advantage of the test operator, the equation that we obtain at the companion node will have the desirable shape we mentioned above. Hence, a solution inside PDL is evident.

\begin{exa}\label{e:properclusterinterpolant}
Consider the closed PDL-tableau in \autoref{f:overviewExampleInterpolant}
which consists of a cluster~$C$ and nodes that are adjacent to the cluster.
Let $C^+$ be the set of all nodes.
We use the term \emph{exits} to refer to the elements of $C^+\setminus C$.

\begin{figure}[htb]
\begin{center}
\begin{tikzpicture}[node distance=16.5mm, >=latex]
    \draw[black!30,dashed,fill=yellow!10]
    (-3.5, 0.5) -- (7.2, 0.5) -- (7.2, -8.7) -- (0.5, -8.7) --
    (0.5, -3.6) -- (-3.5, -3.6) -- cycle;

  \node (2) at (0,0)
    {$p \+ [a][a^\ast](p \lor [a^\ast]p) ~;~ \lnot\loaded{[a][a^\ast]}p$};

  \node [below of=2] (4)
    {$[a^\ast](p \lor [a^\ast]p) ~;~ \lnot\loaded{[a^\ast]}p$};
  \arr{2}{$(M)$}{4}

  \node [below of=4] (6)
    {$p \lor [a^\ast]p \+ [a][a^\ast](p \lor [a^\ast]p) ~;~ \lnot\loaded{[a^\ast]}p$};
  \arr{4}{$(\Box)$}{6}

  \node [below of=6, node distance=2em] {\footnotesize $(\lnot \land)$};

  \node [below of=6] (9) {};

    \node [left of=9, node distance=40mm] (8)
      {\fbox{$\lnot \lnot [a^\ast]p$}
      $\lnot \lnot [a^\ast]p \+ [a][a^\ast](p \lor [a^\ast]p) ~;~ \lnot\loaded{[a^\ast]}p$};
    \arr{6}{}{8}

    \node [right of=9, node distance=35mm] (10)
      {$\lnot \lnot p \+ [a][a^\ast](p \lor [a^\ast]p) ~;~ \lnot\loaded{[a^\ast]}p$};
    \arr{6}{}{10}

  \node [below of=10] (12)
    {$p \+ [a][a^\ast](p \lor [a^\ast]p) ~;~ \lnot\loaded{[a^\ast]}p$};
  \arr{10}{$(\lnot)$}{12}

  \node [below of=12, node distance=2em] (12l) {};
  \node [left of=12l, node distance=2em] {\footnotesize $(\loaded{\diam})$};

  \node [below of=12] (15) {};

    \node [left of=15, node distance=14em] (14)
      {\fbox{$p$}
      $p \+ [a][a^\ast](p \lor [a^\ast]p) ~;~ \lnot p$};
    \arr{12}{}{14}

    \node [right of=15, node distance=1em] (16)
      {$p \+ [a][a^\ast](p \lor [a^\ast]p) ~;~ \lnot\loaded{[a][a^\ast]}p$};
    \arr{12}{}{16}

  \draw (16.east) edge [dashed, gray, ->, thick, bend right=28, looseness=1.5]
    node [right] {$\comp$} (2.east);
\end{tikzpicture}
\end{center}
\caption{A closed tableau for $p \+ [a][a^\ast](p \lor [a^\ast]p) ~;~ \lnot\loaded{[a][a^\ast]}p$ with the highlighted proper cluster $C$.
The interpolants at $C^+\setminus C$ (already computed) have been indicated within a solid box and placed to the left of the node.
}\label{f:overviewExampleInterpolant}
\end{figure}

\begin{itemize}[label=$-$]
\item Let $t$ be the bottom rightmost node in \autoref{f:overviewExampleInterpolant}. It is a leaf that is a loaded-path repeat in the cluster, so assign its pre-interpolant as the variable~$q_{c(t)}$.
Observe that $t$ is the right child of a $(\loaded{\diam})$ rule; by inspection, the left child of the $(\loaded{\diam})$ rule has interpolant $p$.
\item
The parent of the $(\loaded{\diam})$ rule is assigned the pre-interpolant~$p\land q_{c(t)}$ by~\textbf{(\ref{loadedrule})}.
\item
The parent of the $(\loaded{\diam})$ rule belongs to a $\lnot[a^\ast]p$-region, i.e., the four nodes whose second component is $\lnot[a^\ast]p$.
The significance of regions in the general construction will be explained later on.
There is a single node adjacent to this region that is outside the cluster: the left child of the rule $(\lnot \land)$ applied to the unloaded component which, by inspection, has interpolant $\lnot\lnot[a^\ast]p$.
The right child of the rule $(\lnot \land)$ has pre-interpolant $p\land q_{c(t)}$ (this is because the parent of the $(\lnot)$ rule is given the same pre-interpolant as its child).
Hence the parent of $(\lnot \land)$ has pre-interpolant $[(\lnot\lnot\lnot[a^\ast]p)?](p\land q_{c(t)})$.
The $(\Box)$ rule in the unloaded component is non-branching so it does not change the pre-interpolant either, hence the pre-interpolant at the root of the $\lnot[a^\ast]p$-region is $[(\lnot\lnot\lnot[a^\ast]p)?](p\land q_{c(t)})$.
\item
Moving from this $\lnot[a^\ast]p$-region to the parent of the $(M)$ rule on atomic program~$a$, by~\ref{loadedmodal} we get $[a][(\lnot\lnot\lnot[a^\ast]p)?](p\land q_{c(t)})$ as the pre-interpolant~$\iota_{c(t)}$ for the root.
\end{itemize}
Now solve $q_{c(t)}\equiv \iota_{c(t)}$ to obtain $q_{c(t)}=[{(a;(\lnot\lnot\lnot[a^\ast]p)?)}^\ast]p$ as discussed in the paragraph following \textbf{(\ref{loadedmodal})} where a solution is identified for equations of this form.
Call this formula~$\theta$.
By inspection: $p \+ [a][a^\ast](p \lor [a^\ast]p)\models \theta$, and $\lnot[a][a^\ast]p\models\lnot\theta$.
\end{exa}

The trick of replacing a disjunction with a test-boxed formula works if the pre-interpolant of only one of the disjuncts contains an auxiliary $q$-variable. However, we cannot rule out the possibility that both pre-interpolants contain auxiliary variables.
The general solution requires a further abstraction, where the pre-interpolants are computed not for individual nodes, but instead for a region, i.e., the set of nodes in the cluster with the same loaded component.

\begin{enumerate} [label=\textbf{E\arabic*}, ref=E\arabic*]
\setcounter{enumi}{2}
\item\label{unloadedregion}
Let $C_{\Delta}$ be the set of nodes in the cluster $C$ of which the loaded
component is the set $\Delta$.
By enforcing a \emph{uniformity condition} on the tableau calculus which ensures that the same non-modal rule was applied with the same principal formula to the loaded component of every node $t\in C_{\Delta}$, we may presume that a pre-interpolant~$\iitp_{y}$ corresponding to the region(s) that follow from the application of \emph{the} rule to (all nodes with) the loaded component~$\Delta$ has been inductively found.
Define $\theta_{\Delta}$ as the disjunction of the interpolants from those nodes outside of the cluster but directly adjacent to $C_\Delta$ --- these interpolants exist by the inductive hypothesis, and do not contain any $q$-variables.
Then we define the interpolant for $C_{\Delta}$ as $[\lnot\theta_\Delta?]\iitp_y$.
\end{enumerate}

\textbf{The key observation} in the interpolation proof for PDL is that, in order to
obtain an interpolant for the root of a cluster $C$ of a closed split tableau,
we do not need to base the definition of the pre-interpolants on the tree
structure of $C$ itself.
Rather we define for each cluster $C$ an auxiliary structure $\Q$ which we call a \emph{quasi-tableau}
since it resembles a tableau.
The idea is then to define pre-interpolants for the nodes of $\Q$, inductively
from leaves to root.
To make this work, we need to synchronize the order in which rules are applied
in $C$: our uniformity condition states that when a non-modal rule is applied
to the loaded component~$\Delta$, the same rule must be applied
uniformly throughout $C_\Delta$.
Crucially, following the above rules \textbf{(\ref{loadedrule})} to \textbf{(\ref{unloadedregion})},
we make sure that pre-interpolants have a shape that ensures
that we can solve, inside PDL, the equations that appear at the
companion nodes.
Finally, having defined pre-interpolants for every node of the quasi-tableau,
we define the interpolant for the root of the cluster $C$ to be the
pre-interpolant that we associate with the root of its quasi-tableau
$\Q$.

\textbf{Why the construction yields an interpolant.}
Leaving aside the variable condition for this overview, consider a cluster~$C$ whose root~$r_C$ is labelled $\Gamma_1;\Gamma_2$, and suppose that $\Gamma_1$ is unloaded and $\Gamma_2$ contains a loaded formula $\lnot\loaded{[\alpha]}\phi$.
Let $\iitp_{r}$ be the pre-interpolant computed for the root $r$ of $\Q$. Let us motivate $\Gamma_1 \models \iitp_{r}$ and $\Gamma_2 \models \lnot\iitp_{r}$.

To establish that $\Gamma_1\models\iitp_{r}$ we prove, by a leaf-to-root
induction on the quasi-tableau $\Q$, that every node $x$ in $\Q$ satisfies the
following statement, where $\Delta_{x}$ is the label of $x$, $t$ is any node
in the region $C_{\Delta_{x}}$, and $\sigma$ is a substitution that
replaces every internal variable $q_{c(z)}$ with a formula representing the
region $C_{\Delta_{c(z)}}$: the unloaded component of $t$ entails $\subst{\sigma}{\iitp_x}$.
Since the pre-interpolant $\iitp_{r}$ contains no internal variables
and $\Gamma_{i}$ is the unloaded component at the root of the cluster, the result follows immediately.

To establish that $\Gamma_2\models \lnot\iitp_{r}$ we reason towards a
contradiction by assuming that $\Gamma_2 \land \iitp_r$ is satisfiable, i.e.,
there is a Kripke model $\kmodel$ such that
$\kmodel,v\Vdash\Gamma_2,\iitp_{r}$ for some state $v$ in $\kmodel$. Since $\neg[\alpha]\phi \in \Gamma_2$,
it follows that there is a finite path of minimal length (`witness distance')
in $\kmodel$ to some state~$w$ such that $\kmodel,w\Vdash\neg\phi$.
This finite path in $\kmodel$ determines a walk through the cluster~$C$ --- it is a walk in the
sense that some nodes in the cluster may be visited multiple times
due to loaded-path repeats --- such that at each state~$v'$ on the path in $\kmodel$, 
there is a node in $C^+$ which corresponds to 
a $z \in \Q$ such that $\kmodel,v'\Vdash \Delta_z,\iitp_{z}$.
Crucially, the last state $w$ on this path in $\kmodel$ must correspond to a node $s$ in the tableau just outside $C$, i.e., $s \in C^{+} \setminus C$,
since this is only way to witness $\neg\phi$ i.e., $\kmodel,w\Vdash\neg\phi$.
That implies $\kmodel, w\Vdash\Delta_z,\iitp_z$ for $z \in \Q$ corresponding to $s$.
While $\iitp_{z}$ is generally just a pre-interpolant, it is an interpolant in this special case of quasi-tableau nodes corresponding to nodes in the tableaux that are just outside the cluster (as we will show in~\autoref{f:iota}(a)); intuitively, Maehara's method is applicable to nodes outside the cluster, so we already have interpolants for nodes just outside $C$.
Given that $\iitp_{z}$ is an interpolant we thus have $\Delta_z\models\lnot\iitp_{z}$.
Hence $\kmodel, w \Vdash \lnot\iitp_{z}, \iitp_{z}$. The latter is impossible so we have reached a contradiction.

\section{Split Tableaux}\label{c}

We now define the split PDL-tableau system. The development is analogous to standard PDL-tableau, but we include the formal details for preciseness.

\begin{defi}[\lean{Sequent}]\label{d:split-sequent}
A \emph{split sequent} is a pair of sequents (i.e., a pair of finite sets of formulas from $\pdlforms \cup \loaded{\pdlforms}$).
A split sequent $(\Gamma_1, \Gamma_2)$ will usually be denoted as $\Gamma_1; \Gamma_2$, and we refer to $\Gamma_1$
and $\Gamma_2$ as, respectively, its first or left component, and its second or right component.
If $\Gamma_i$ ($i\in\{1,2\}$) is loaded (i.e., contains a loaded formula), then we call $\Gamma_i$ loaded, otherwise we call it unloaded.
\end{defi}

We now define the notion of a split PDL-tableau; nodes of the tableau are now labeled with split sequents instead of sequents.
\begin{defi}
  Given a set $\mathsf{L}$ of inference rules, a \emph{split tableau $\tab$ for a split sequent $\Gamma_1; \Gamma_2$} is a tuple $(V, \edge, r, \Lambda, L)$ where
  $(V, \edge, r)$ is a rooted tree,
  $L$ is a function that assigns to every interior node $v \in V$ a rule $L(v) \in \mathsf{L}$,
  $\Lambda$ is a function that assigns to every node $v \in V$ a split sequent $\Lambda(v)$,
  and the split sequent $\Lambda(r)$ of the root is equal to $\Gamma_1; \Gamma_2$.
  Furthermore, if $u \in V$ is an interior node with children $v_1, \dots, v_n$ then $\Lambda(u) ~/~ \Lambda(v_1) | \dots | \Lambda(v_n)$ is a rule instance of the rule $L(u)$ (respecting any side conditions).
\end{defi}

For a node $v \in V$ with $\Lambda(v) = \Gamma_1;\Gamma_2$, sometimes we will write $\Lambda_1(v)$ instead of $\Gamma_1$ and $\Lambda_2(v)$ instead of $\Gamma_2$.
We will call $\Lambda_1(v)$ and $\Lambda_2(v)$, respectively,
the first or left component of $v$, and the second or right component of $v$.

\begin{defi}[\lean{LocalTableau}]\label{d:split-PDL-tableau}
A \emph{local split PDL-tableau} for $\Delta; \Gamma$ is a tableau over the following set of rules, called the local split rules:
    \begin{center}
      \AxiomC{$\Delta , \lnot \lnot \phi; \Gamma$}
      \LeftLabel{${(\lnot)}_1$}
      \UnaryInfC{$\Delta,\phi; \Gamma$}
      \DisplayProof%
      \hspace{1.5em}
      \AxiomC{$\Delta , \phi \land \psi; \Gamma$}
      \LeftLabel{${(\land)}_1$}
      \UnaryInfC{$\Delta,\phi,\psi; \Gamma$}
      \DisplayProof%
      \hspace{1.5em}
      \AxiomC{$\Delta , \lnot (\phi \land \psi); \Gamma$}
      \LeftLabel{${(\lnot \land)}_1$}
      \UnaryInfC{$\Delta , \lnot \phi; \Gamma \splitCase \Delta , \lnot \psi; \Gamma$}
      \DisplayProof%

      \smallskip

      \AxiomC{$\Delta , [\alpha]\phi; \Gamma$}
      \LeftLabel{${(\Box)}_1$}
      \RightLabel{\ \ $\alpha$ non-atomic}
      \UnaryInfC{$\big\{\Delta, \Sigma  \mid  \Sigma \in \unfold_{\square}(\alpha, \phi) \big\}; \Gamma$}
      \DisplayProof

      \smallskip

      \AxiomC{$\Delta , \lnot[\alpha]\phi; \Gamma$}
      \LeftLabel{${(\diam)}_1$}
      \RightLabel{\ \ $\alpha$ non-atomic}
      \UnaryInfC{$\{\Delta, \Sigma \mid \Sigma \in \unfold_{\diam}(\alpha, \phi)\}; \Gamma$}
      \DisplayProof

      \smallskip

      \AxiomC{$\Delta , \lnot\loaded{[\alpha]}\xi; \Gamma$}
      \LeftLabel{${(\loaded{\diam})}_1$}
      \RightLabel{\ \ $\alpha$ non-atomic}
      \UnaryInfC{$\big\{\Delta, \Sigma \mid \Sigma \in
          \loaded{\unfold}_{\diam}(\alpha, \xi)\big\}; \Gamma$}
      \DisplayProof

      \smallskip

      \AxiomC{$\Delta , \lnot[a][\alpha_1]\ldots[\alpha_n]\phi; \Gamma$}
      \LeftLabel{${(L+)}_1$}
      \RightLabel{\ \ $\Delta, \Gamma$ free, basic and $n \geq 0$ maximal}
      \UnaryInfC{$\Delta , \lnot\loaded{[a][\alpha_1]\ldots[\alpha_n]}\phi; \Gamma$}
      \DisplayProof

      \smallskip

      \AxiomC{$\Delta , \lnot\loaded{[a][\alpha_1]\ldots[\alpha_n]}\phi; \Gamma$}
      \LeftLabel{${(L-)}_1$}
      \RightLabel{\ \ $\Delta; \Gamma$ basic}
      \UnaryInfC{$\Delta , \lnot[\alpha_1]\ldots[\alpha_n]\phi; \Gamma$}
      \DisplayProof

      \smallskip

      \AxiomC{$\Gamma; \Delta , \lnot \lnot \phi$}
      \LeftLabel{${(\lnot)}_2$}
      \UnaryInfC{$\Gamma; \Delta,\phi$}
      \DisplayProof
      \hspace{1.5em}
      \AxiomC{$\Gamma; \Delta , \phi \land \psi$}
      \LeftLabel{${(\land)}_2$}
      \UnaryInfC{$\Gamma; \Delta,\phi,\psi$}
      \DisplayProof%
      \hspace{1.5em}
      \AxiomC{$\Gamma; \Delta , \lnot (\phi \land \psi)$}
      \LeftLabel{${(\lnot \land)}_2$}
      \UnaryInfC{$\Gamma; \Delta , \lnot \phi \splitCase \Gamma; \Delta , \lnot \psi$}
      \DisplayProof

      \smallskip

      \AxiomC{$\Gamma; \Delta , [\alpha]\phi$}
      \LeftLabel{${(\Box)}_2$}
      \RightLabel{\ \ $\alpha$ non-atomic}
      \UnaryInfC{$\Gamma; \big\{\Delta, \Sigma  \mid  \Sigma \in \unfold_{\square}(\alpha, \phi) \big\}$}
      \DisplayProof

      \smallskip

      \AxiomC{$\Gamma; \Delta , \lnot[\alpha]\phi$}
      \LeftLabel{${(\diam)}_2$}
      \RightLabel{\ \ $\alpha$ non-atomic}
      \UnaryInfC{$\Gamma; \{\Delta, \Sigma \mid \Sigma \in \unfold_{\diam}(\alpha, \phi)\}$}
      \DisplayProof

      \smallskip

      \AxiomC{$\Gamma; \Delta , \lnot\loaded{[\alpha]}\xi$}
      \LeftLabel{${(\loaded{\diam})}_2$}
      \RightLabel{\ \ $\alpha$ non-atomic}
      \UnaryInfC{$\Gamma; \big\{\Delta, \Sigma \mid \Sigma \in  \loaded{\unfold}_{\diam}(\alpha, \xi)\big\}$}
      \DisplayProof

      \smallskip

      \AxiomC{$\Gamma; \Delta , \lnot[a][\alpha_1]\ldots[\alpha_n]\phi$}
      \LeftLabel{${(L+)}_2$}
      \RightLabel{\ \ $\Delta, \Gamma$ free, basic and $n \geq 0$ maximal}
      \UnaryInfC{$\Gamma; \Delta , \lnot\loaded{[a][\alpha_1]\ldots[\alpha_n]}\phi$}
      \DisplayProof

      \smallskip

      \AxiomC{$\Gamma; \Delta , \lnot\loaded{[a][\alpha_1]\ldots[\alpha_n]}\phi$}
      \LeftLabel{${(L-)}_2$}
      \RightLabel{\ \ $\Gamma; \Delta$ basic}
      \UnaryInfC{$\Gamma; \Delta , \lnot[\alpha_1]\ldots[\alpha_n]\phi$}
      \DisplayProof
    \end{center}
    In the above, $\phi, \psi \in \pdlforms$, $\alpha$ is a non-atomic program, $\alpha_1, \dots, \alpha_n \in \progs$, $a$ is an atomic program, and $\xi \in \pdlforms \cup \loaded{\pdlforms}$, i.e., $\xi$ is a possibly loaded formula.

    Furthermore, if $i \in \{1, 2\}$ and $t$ is a child of node $s$ such that $L(s) = {(L+)}_i$ then $L(t) \neq {(L-)}_i$.
    That is, if a node was obtained using ${(L+)}_i$ then it is not allowed to apply ${(L-)}_i$ immediately to it. The rule ${(L+)}_i$ is called a \emph{loading rule}, and ${(L-)}_i$ is called an \emph{unloading} rule.
\end{defi}

By some minor modifications of the reasoning in Subsection~\ref{subsec:LocalRules} we obtain the following split version of Lemma~\ref{l:localRuleTruth}.

\begin{lem}[\checklean{localRuleTruth}]\label{l:splitlocalRuleTruth}
All split local rules are locally sound and locally invertible.
\end{lem}

\begin{defi}\label{d:split-modalrule}
The split modal rules are:
\begin{center}
  \AxiomC{$\Delta, \lnot\loaded{[a]}\xi; \Gamma$}
  \LeftLabel{${(M)}_1$}
  \UnaryInfC{$\Delta_a, \lnot \xi; \Gamma_a$}
  \DisplayProof
  \quad
  \quad
  \AxiomC{$\Gamma; \Delta , \lnot\loaded{[a]}\xi$}
  \LeftLabel{${(M)}_2$}
  \UnaryInfC{$\Gamma_a; \Delta_a , \lnot \xi$}
  \DisplayProof
\end{center}
where $\Delta, \Gamma$ are basic and $\xi \in \pdlforms \cup \loaded{\pdlforms}$, i.e., $\xi$ is a possibly loaded formula.

The set of split PDL-rules is then defined as the set of split local rules together with the set of split modal rules.
\end{defi}

The rules of the form ${(R)}_1$ will be called left rules and the rules of the form ${(R)}_2$ right rules.

\begin{defi}\label{d:split-repeat}
Let $\tab$ be a split tableau over the split PDL-rules.
A node~$t$ is \emph{closed} if $\Lambda_1(t)\cup \Lambda_2(t)$ is closed (see \autoref{d:closedNode}).
A \emph{repeat} is a node $t$ such that $\Lambda_1(t) = \Lambda_1(s)$ and $\Lambda_2(t) = \Lambda_2(s)$, for some ancestor $s$ of $t$;
in this case, the nearest such ancestor to $t$ is called its \emph{companion} and denoted $c(t)$.
A node $s$ is \emph{free} if $\Lambda_1(s)\cup \Lambda_2(s)$ is a free sequent, and \emph{loaded} otherwise.
A repeat $t$ is a \emph{\lpr} if the path from $c(t)$ to $t$ consists only of loaded nodes.
\end{defi}

A \emph{cluster} is defined analogously to \autoref{d:cluster}, as a maximal strongly connected subset of nodes in $\tab$ after edges are added from each leaf that is a loaded-path repeat to its companion.

\begin{defi}
A \emph{split} PDL-tableau is a tableau over the split PDL-rules, which adheres to the following conditions:
\begin{enumerate}[label=(T\arabic*S)]
\item[(\lprCondSplit)] every \lpr is a leaf;
\item[(\freeCondSplit)] every free repeat is a leaf.
\end{enumerate}
\end{defi}

\begin{defi}\label{d:split-closed}
A split PDL-tableau is \emph{closed} if each of its leaves is either closed or a \lpr.
A split PDL-tableau is \emph{open} if it is not closed.
\end{defi}

The notions from \autoref{subsec:PropertiesOfPDLTableuax} of loaded path, the companion relation \comp, the edge relation~$\edge$, its transitive closure~$\edgeT$, its reflexive transitive closure~$\edgeRT$, the relations~$\cEquiv$ and~$\simpler{}{}$, the notions of clusters, proper clusters and singleton clusters can be defined completely analogously for split PDL-tableaux.

\begin{lem}\label{l:splitPdlTabFinite}
  Every split PDL-tableau is finite.
\end{lem}
\begin{proof}
  Analogous to \autoref{l:pdlTabFinite}, because any split PDL-tableau will only contain finitely many different split sequents, determined by the Fischer-Ladner closure.
\end{proof}

The following lemma will be useful to show that the vocabulary condition holds.

\begin{lem}[\checklean{cReach_fvoc_subset}]\label{l:vocIsPreserved}
    Let $\tab = (V, \edge, r, \Lambda, L)$ be a tableau over the split PDL-rules and let $s,t \in V$ be such that $s \edgeRT t$.
    Then, $\voc(\Lambda_i(s)) \supseteq \voc(\Lambda_i(t))$ for all $i \in \{1, 2\}$.
\end{lem}
\begin{proof}
    By inspection of the rules.
\end{proof}

\subsection{Uniformity}

Next, we introduce two \emph{uniformity conditions} for split tableaux, which will be needed to define the pivotal construction in the interpolation proof, viz., that of a \emph{quasi-tableau} (in \autoref{d:Qtableau}).
Informally, the uniformity conditions together say that for loaded nodes, the unloaded component must first be reduced to a basic sequent, and after that, nodes with the same loaded component must synchronize on rules for propositional connectives and the unfold rules ${(\Box)}_i$ and ${(\Diamond)}_i$ that are applied in the loaded component. Note that the synchronization does not apply to the loading rules ${(L+)}_i$ due to the side condition that the $i$-th component is basic. Hence, different diamond formulas can be loaded, so that the modal rule can be applied to different diamond formulas.

\begin{defi}[\lean{Tableau.isUniform}]\label{d:uniform}
A split PDL-tableau $\tab$ is \emph{uniform} if it satisfies the following conditions:
\begin{itemize}[align=left]
\item[U1.]
If $s\in \tab$ is such that $\Lambda_{i}(s)$ is loaded then a local rule is applied to a formula in the unloaded component of $s$, unless this component is basic;
\item[U2.]
If $s,t \in \tab$ are such that  $\Lambda_{i}(s) = \Lambda_{i}(t)$ is loaded and not basic, and the unloaded components of both $s$ and $t$ are basic, then at $s$ and at $t$ the same rule is applied, to the same formula in the loaded component of the node.
\end{itemize}
\end{defi}

\begin{rem}
Although condition U1 is not explicitly used in the proofs, it simplifies the formulation of condition U2.
In particular, U1 ensures that all rules are applied fully to the unloaded component before synchronizing the application of a rule to the loaded components.
While an alternative notion of uniformity that omits U1 and adjusts U2 accordingly is likely possible, we find using U1 and U2 more intuitive.
\end{rem}

\autoref{t:interpolation}, the main theorem of the document, follows directly from \autoref{t:unsatIffSplitClosed} and the following theorem.

\begin{thm}[\checklean{tabToInt}]\label{t:tabToInt}
Let $\tab$ be a uniform, closed split PDL-tableau for $\Gamma_1; \Gamma_2$.
Then $\Gamma_1; \Gamma_2$ has an interpolant.
\end{thm}

We will prove the theorem by induction on the size (number of nodes) of $\tab$, making a case distinction as to whether the root of $\tab$ belongs to a singleton or a proper cluster.
In the case for proper clusters, we will need the following definition and observation.

\begin{defi}[\lean{isExitOf}]\label{d:plus_exits}
Let $C$ be some cluster of a split PDL-tableau $\tab$.
Then we define
\[
C^{+} \isdef C \cup \{ s \in T \mid t \edge s, \text{ for some } t \in C\}
\]
as the set of nodes that either belong to $C$ or are the child of a node in $C$.
We refer to nodes in $C^+ \setminus C$ as \emph{exit nodes}, or simply \emph{exits}.
\end{defi}

The following lemma collects some properties of clusters in split tableaux that we will use implicitly from now on.

\begin{lem}\label{l:splitClusterProp}
Let $C$ be a proper cluster of a closed split tableau $\tab$.
\begin{enumerate}[(a)]
\item $C$ is a subtree\footnote{The pair $(T', \edge')$ is called a \emph{subtree} of a tree $(T, \edge)$ if it is a tree such that $T' \subseteq T$ and $\edge' \subseteq \edge$.}
of $(T,\edge)$; in particular, $C$ has a root.
\item Either $\Lambda_{1}(t)$ is loaded for every $t\in C$, or
$\Lambda_{2}(t)$ is loaded for every $t\in C$.
\item If $\Lambda_{i}(t) = \emptyset$ for some $t \in C$, then
$\Lambda_{i}(t) = \emptyset$ for all $t \in C$.
\end{enumerate}
\end{lem}

\begin{proof}
We have (a) since $C$~is a (maximal) strongly connected set of nodes (once the edges have been added from a leaf to its companion), and since $T$~is a tree.
Part (c)~is immediate from an inspection of the rules.
To show (b), we first establish the following claim:

\emph{Claim.}
If $s \cEquiv t$, $s \cEdge t$, and $\Lambda_i(s)$ is loaded, then $\Lambda_i(t)$ is loaded.
\medskip

By definition, $s \cEdge t$ iff $s \edge t$ or $s \comp t$.
If $s \comp t$, then the sequents of $s$ and $t$ must be equal (recall this means that the split sequent tuples are equal, so in particular they have the same loaded component).
If $s \edge t$, it follows by inspection of the rules that the loading must either remain in the same component, or it must be removed.
If it is removed, then we have that $t$ is free, and thus, by \autoref{l:ePropB}(d), that $s \not\cEquiv t$.
This establishes the claim.

Since $C$ is a proper cluster, there exists an $s \in C$ that is loaded.
For all nodes $t$ in $C$, there is a $\cEdge$-path from $s$ to $t$, and hence by induction on path length
using the above Claim, it follows that all nodes in $C$ are loaded in the same component as $s$.
\end{proof}

\subsection{Soundness and completeness for uniform split tableaux}

The following theorem states a soundness and completeness result for (uniform) split PDL-tableaux.
It is on the basis of this result that we can indeed use split tableaux to establish  interpolation for PDL\@.

\begin{thm}[Soundness and completeness for split tableaux.]\label{t:unsatIffSplitClosed}
For all split sequents $\Gamma_1; \Gamma_2$, the following statements are equivalent:
\begin{enumerate}
\item[(1)] $\Gamma_{1} \cup \Gamma_{2}$ is not satisfiable;
\item[(2)] there is a closed split tableau for $\Gamma_1; \Gamma_2$;
\item[(3)] there is a uniform, closed split tableau for $\Gamma_1; \Gamma_2$.
\end{enumerate}
\end{thm}

\begin{proof}
We prove the implications (1) $\Rightarrow$ (2),
(1) $\Rightarrow$ (3), (3) $\Rightarrow$ (2) and (2) $\Rightarrow$ (1).
\medskip

\textbf{(1) $\Rightarrow$ (2)} Due to the
similarities to the completeness proof in \autoref{sec:Completeness} we
confine ourselves to a sketch.

The crucial idea is to modify the tableau game of \autoref{subsec:TableauGame}.
The \emph{split tableau game} $\Gms(\Gamma_{1};\Gamma_{2})$ is a two-player graph
game with players Builder and Prover.
A position of $\Gms(\Gamma_1;\Gamma_2)$ is either:
\begin{itemize}
\item a (possibly loaded) split sequent $\Delta_{1};\Delta_{2}$; or
\item a pair of the form $\big((\Gamma,\phi,R_{1}),\Delta_{2}\big)$, such that $R_{1}$ is a left rule applicable to the split sequent $\Gamma, \phi; \Delta_{2}$
with $\phi$ as the principal formula; or
\item a pair of the form $\big(\Delta_{1}, (\Gamma,\phi,R_{2})\big)$, such that
$R_{2}$ is a right rule applicable to the split sequent $\Delta_{1};\Gamma, \phi$
with $\phi$ as the principal formula.
\end{itemize}
Positions of the first type belong to Prover; at a position $\Delta_{1};
\Delta_{2}$ he picks either a formula $\phi \in \Delta_{1}$ and a left rule
$R_{1}$ to be applied to $\phi$, or a formula $\phi \in \Delta_{2}$ and a right
rule $R_{2}$ to be applied to $\phi$.
He thus moves the token to a position of the second or the third type.
Positions of the second type belong to Builder; at a position $\big((\Gamma,
\phi,R_{1}),\Delta_{2}\big)$ she picks one of the rule children of $R_{1}$, and
this rule child is then the next position, of the first type.
Positions of the third type also belong to Builder, and her admissible moves are
defined similarly.

A \emph{match} of $\Gms(\Gamma_{1};\Gamma_{2})$ consists of a sequence of positions, starting at $\Gamma_{1};\Gamma_{2}$, induced by the players' moves.
Such a match is \emph{winning for Prover} if it reaches a closed split sequent
or a \emph{\lpr}.
On the other hand, a match is \emph{winning for Builder} if it reaches a sequent
to which no rule is applicable, or a free repeat.
Recall (\autoref{d:split-repeat}) that the notion of a repeat is now defined in terms of split sequents.

Similarly to the simpler game $\Gm$ it is easy to show that full matches of this
game must be finite, so that $\Gms(\Gamma_{1};\Gamma_{2})$ is determined.
It is also straightforward to prove that a winning strategy for Prover in
$\Gms(\Gamma_{1};\Gamma_{2})$ induces a closed split tableau.
Thus the crucial result to prove is the following:
\begin{equation}\label{claim:splcpl1}
\Gamma_{1} \cup \Gamma_{2} \text{ is satisfiable if Builder has a
winning strategy in } \Gms(\Gamma_{1};\Gamma_{2}).
\end{equation}

To prove \claimRef{claim:splcpl1} we assume that Builder has a winning strategy
in $\Gms(\Gamma_{1};\Gamma_{2})$, and as in \autoref{subsec:TableauGame} we
represent this winning strategy by its strategy tree $\Stree = (T,\tedge,r,
\Lambda)$.
Now, $\Lambda$ is a labelling that assigns to each node $s$ in $T$ a
(possibly loaded) \emph{split} sequent $\Lambda(s) = \Lambda_{1}(s);
\Lambda_{2}(s)$.
As before we will construct a \emph{model graph} $\kmodel_{\Stree}$ out of $\Stree$.

For this purpose, we define \emph{pre-states} completely analogously to \autoref{d:prestate}.
That is, a pre-state is a maximal sequence $\pi = {(s_{i})}_{1\leq i\leq k}$ on $\Stree$ such that $s_{1}$ is initial; there is no application on $\pi$ of a modal rule, a loading rule or an unloading rule; and $\pi$ is either a path or the composition $\pi = \pi'\pi''$ of two paths such that the last node $s$ of $\pi'$ is a repeat leaf, the first node of $\pi''$ is its companion, and $\pi''$ is an initial segment of the path from~$c(s)$ to~$s$.
Given any path $\pi = {(s_{i})}_{1\leq i\leq k}$ on $\Stree$ and $h \in \{1,2\}$, we define $\Lambda_{h}(\pi) \isdef \bigcup_{i} \Lambda_{h}(s_{i})$ and we let $\Lambda^-(\pi)$ denote the set of formulas we obtain from unloading the set $\Lambda_{1}(\pi) \cup \Lambda_{2}(\pi)$.

The goal is now to prove that this structure $\kmodel_{\Stree}$ is a model graph
indeed, and for this purpose we prove the following \emph{split} versions of
\autoref{l:cp2new}, \autoref{l:cp4}, and \autoref{l:cp5}.

\emph{\autoref{l:cp2new}${\,}^{s}$}
Let $t$ be a node of $\Stree$ on the pre-state $\pi$, let $h \in \{1,2\}$, and
let $\lnot\loaded{[\alpha]}\phi$ be a loaded formula in $\Lambda_{h}(t)$.
Then there is a node $u$ on $\Stree$ such that $t \tedgeT u$, $\lnot\phi \in
\Lambda_{h}(u)$ and $u$ lies on some pre-state $\rho$ such that
$(\Lambda^-(\pi), \Lambda^-(\rho)) \in Q_{\alpha}$.
\medskip

\emph{\autoref{l:cp4}${\,}^{s}$}
Let $\pi$ be a pre-state and let $h \in \{1,2\}$ be such that $\lnot[a][\delta_{1}]\cdots[\delta_{n}]\phi \in \Lambda_{h}(\pi)$, where $n$ is maximal, i.e., $\phi$ itself is not a box formula.
Then there is a pre-state $\pi' = {(s_{i})}_{1\leq i\leq k}$ such that $(\Lambda^-(\pi),\Lambda^-(\pi')) \in Q_{a}$ and $\lnot\loaded{[\delta_{1}]\cdots[\delta_{n}]}\phi \in \Lambda_{h}(s_{1})$.
\medskip

\emph{\autoref{l:cp5}${\,}^{s}$}
Let $\pi$ be a pre-state, let $h \in \{1,2\}$ and let $\lnot[\alpha]\phi$ be a
formula in $\Lambda_{h}(\pi)$.
Then there is a pre-state $\pi'$ such that $(\Lambda^-(\pi),\Lambda^-(\pi')) \in Q_{\alpha}$ and $\lnot\phi \in \Lambda_{h}^-(\pi')$.
\medskip

In all three cases, the proof of the split lemma is \emph{almost verbatim} the same as the proof of the original statement: the only modification is to replace at some spots a set $\Lambda(t)$ or $\Lambda(\pi)$ with its left- or right-hand side version $\Lambda_{h}(t)$ or $\Lambda_{h}(\pi)$.

From the `split lemmas', \autoref{l:cp2new}$^{s}$,
\autoref{l:cp4}$^{s}$ and \autoref{l:cp5}$^{s}$, it easily follows (as in the proof of
\autoref{thm:strmg}) that the structure $\kmodel_{\Stree}$ is a model graph,
and from this the statement \claimRef{claim:splcpl1} follows in the same manner as
in \autoref{subsec:TableauGame}.
This means that we are done, since, as we saw, \claimRef{claim:splcpl1} is the crucial
statement needed to prove the implication (1) $\Rightarrow$ (2).
\medskip

\textbf{(1) $\Rightarrow$ (3)} This is a minor variation of
the proof of the implication (1) $\Rightarrow$ (2).
The idea is simply to tweak the rules of the split tableau game $\Gms(\Gamma_{1};
\Gamma_{2})$ in such a way that a winning strategy for Prover in the resulting
game $\Gmt(\Gamma_{1};\Gamma_{2})$ induces a closed split tableau for $\Gamma_1;
\Gamma_2$ that in addition is \emph{uniform}.
This is quite simple: we restrict Prover's admissible rules at a position
$\Delta_{1};\Delta_{2}$ in the following way:
\begin{itemize}
\item
if one of the components is loaded and the unloaded component is not basic, then Prover must pick a formula from the unloaded component;
\item
if one of the components is loaded but not basic, and the unloaded component is basic, then we require Prover to pick the smallest admissible and unloaded formula in the loaded component with respect to some arbitrary but fixed order on the formulas in $\FL(\Gamma_{1} \cup \Gamma_{2})$.
\end{itemize}
Clearly, the absence of a uniform, closed split tableau for $\Gamma_1;\Gamma_2$ implies the non-existence of a winning strategy for Prover in the tweaked game $\Gmt(\Gamma_{1};\Gamma_{2})$.
Analogous to \autoref{l:matchesFinite} and \autoref{lem:gamedet} for $\Gm$ we can also show that $\Gmt$ is determined, hence there must exist a winning strategy for Builder in the tweaked game. The construction of a model graph from this winning strategy proceeds from the split versions of
\autoref{l:cp2new}, \autoref{l:cp4}, and \autoref{l:cp5} as in (1) $\Rightarrow$ (2) above. Indeed, observe that the saturation in the definition of model graph (\autoref{d:ModelGraph}) is a closure condition so the order in which formulas are picked does not matter as long as they are eventually picked (\autoref{l:cp1a} remains true so the latter holds).
\medskip

\textbf{(3) $\Rightarrow$ (2)}
This follows trivially by ignoring the uniformity property.
\medskip

\textbf{(2) $\Rightarrow$ (1)}
This statement comprises the soundness of the split tableau system, and, similarly to the completeness case, its proof is a routine variation of the soundness proof of the unsplit system.
We only provide a brief sketch.

Our argumentation for soundness roughly consists of two parts: a check that the split versions of the rules remain locally sound, and a ``global'' proof that any sequent in a closed split tableau is unsatisfiable.

The first task is a minor variation of the corresponding analysis in the unsplit case. Note that all the intricacies there revolve around the $(\Box)$ and $(\diam)$ rules, and the key argument concerns the properties of the $\unfold$ operations.
Whether the context of an application of $(\Box)$ and $(\diam)$ consists of a split or an unsplit sequent is not relevant for the argument.

Similarly, the second task can be executed by almost verbatim following the argumentation in Section~\ref{sec:Soundness}. To provide a bit more detail, in the soundness proof we formulate and prove split versions of Lemmas~\ref{l:loadedSucc} and~\ref{l:loadedDiamondPaths}, and then word by word follow the reasoning in the ``Proof of Theorem~\ref{t:soundness}'' part.
For the proofs of said split versions of Lemmas~\ref{l:loadedSucc} and~\ref{l:loadedDiamondPaths}, the core argument concerns an analysis of the development of loaded formulas in a closed tableau. Again, whether this development takes place in the context of split or unsplit sequents is of no consequence.
\end{proof}

\section{Defining Interpolants}\label{sec:DefinitionInterpolants}

\subsection{Interpolants for singleton clusters}\label{subsec:InterpolationForSingletonClusters}

In this subsection, we consider the relatively easy case where the tableau root forms a singleton cluster.
That is, we consider a cluster with respect to $\cEquiv$ as defined in \autoref{d:cEquiv} that consists of a single node.

\begin{lem}\label{l:itpLeaves}
Let the tableau $\tab$ for $\Gamma_1; \Gamma_2$ consist of a single, closed node.
Then $\Gamma_1; \Gamma_2$ has an interpolant.
\end{lem}
\begin{proof}
If the split sequent $\Gamma_1;\Gamma_2$ is closed, we must have at least one of
the following cases:
 \begin{itemize}
  \item $\bot \in \Gamma_1$ or $\{\phi,\lnot \phi\} \subseteq \Gamma_1$ for some formula $\phi$:\newline
    then we have $\Gamma_1 \models \bot$ and $\Gamma_2 \models \lnot \bot$, hence let $\theta \isdef  \bot$.
  \item $\bot \in \Gamma_2$ or $\{\phi,\lnot \phi\} \subseteq \Gamma_2$ for some formula $\phi$:\newline
    then we have $\Gamma_1 \models \lnot \bot$ and $\Gamma_2 \models \bot$, hence let $\theta \isdef  \lnot \bot$.
  \item $\phi \in \Gamma_1$ and $\lnot \phi \in \Gamma_2$ for some formula $\phi$:\newline
    then we have $\Gamma_1 \models \phi$ and $\Gamma_2 \models \lnot \phi$, hence let $\theta \isdef  \phi$.
  \item $\lnot \phi \in \Gamma_1$ and $\phi \in \Gamma_2$ for some formula $\phi$:\newline
    then we have $\Gamma_1 \models \lnot \phi$ and $\Gamma_2 \models \phi$, hence let $\theta \isdef  \lnot \phi$.
  \end{itemize}
  Moreover, in all cases we obviously have $\voc(\theta) \subseteq \voc(\Gamma_1)
\cap \voc(\Gamma_2)$.
Consequently, in all cases we have defined an interpolant $\theta$ for $\Gamma_1;
\Gamma_2$.
\end{proof}

\begin{lem}[\checklean{localInterpolantStep}]\label{l:easyItp}
Let $\tab$ be a uniform closed tableau for the split sequent $\Gamma_1; \Gamma_2$,
and assume that the root $r$ of $\tab$ forms a singleton cluster.
Furthermore, assume that for every child $t$ of $r$ we already have obtained an
interpolant $\theta_{t}$ for the split sequent $\Lambda_{1}(t);\Lambda_{2}(t)$.
Then $\Gamma_1; \Gamma_2$ has an interpolant.
\end{lem}
\begin{proof}
We make a case distinction as to the rule that is applied at the root of $\tab$.

If the $n$ children of $r$ are obtained using a rule different from $(M)$, then we have
one of the following cases:
\begin{itemize}
\item
  A rule was applied in the second component.
  The children of $\Gamma_1;\Gamma_2$ are then of the form $\Gamma_1;\Pi_1,\ldots,\Gamma_1;\Pi_n$,
  and we have $\Gamma_1 \models \theta_1 \land \ldots \land \theta_n$.
  Moreover $\Gamma_2$ only holds in a state in which also one of the $\Pi_i$ holds due to the fact that all local rules are locally sound by \autoref{l:localRuleTruth}.
  But then from $\Pi_i \models \lnot \theta_i$ for each $i \leq n$, we get that $\Gamma_2 \models \lnot \theta_1 \lor \ldots \lor \lnot \theta_n$, and thus $\Gamma_2 \models \lnot (\theta_1 \land \ldots \land \theta_n)$.
  Let $\theta \isdef  \theta_1 \land \ldots \land \theta_n$.
  Because we have $\voc(\Pi_i) \subseteq \voc(\Gamma_2)$ for each $i$, it follows that $\voc(\theta) \subseteq \voc(\Gamma_1) \cap \voc(\Gamma_2)$.
  Hence $\theta$ is an interpolant for $\Gamma_1;\Gamma_2$.
\item
  A rule was applied in the first component.
  The children of $\Gamma_1;\Gamma_2$ are then of the form $\Pi_1;\Gamma_2,\ldots,\Pi_n;\Gamma_2$,
  and we have $\Gamma_2 \models \lnot \theta_1 \land \ldots \land \lnot \theta_n$,
  and thus $\Gamma_2 \models \lnot (\theta_1 \lor \ldots \lor \theta_n)$.
  Moreover, $\Gamma_1$ only holds in a state in which also one of the $\Pi_i$ holds.
  From $\Pi_i \models \theta_i$ for each $i \leq n$ we get that $\Gamma_1 \models \theta_1 \lor \ldots \lor \theta_n$.
  Hence let $\theta \isdef  \theta_1 \lor \ldots \lor \theta_n$.
  Because we have $\voc(\Pi_i) \subseteq \voc(\Gamma_1)$ for each $i$, it follows that $\voc(\theta) \subseteq \voc(\Gamma_1) \cap \voc(\Gamma_2)$.
  Hence $\theta$ is an interpolant for $\Gamma_1;\Gamma_2$.
\end{itemize}

  Before we consider the case for the $(M)$ rule, we note that $\Gamma_a \models \phi$ implies $\Gamma \models [a]\phi$, for any atomic program~$a$.
  To~see this, suppose $\kmodel, v \Vdash \Gamma$ for some pointed model $(\kmodel, v)$.
  Then we have $\kmodel, w \Vdash \Gamma_a$ and thus $\kmodel, w \Vdash \phi$ for all $w$ from $\kmodel$ with $v \rel{a} w$.
  But this means we have $\kmodel, v \Vdash [a]\phi$.

  Moreover $\Gamma_a , \lnot \psi \models \phi$ implies $\Gamma , \lnot[a]\psi \models \lnot[a]\lnot \phi$, because if we have $\kmodel, v \Vdash \Gamma, \lnot [a]\psi$ for some pointed model $(\kmodel, v)$, then there exists a $w \in \kmodel$ with $v \rel{a} w$ and $\kmodel, w \Vdash \Gamma_a , \lnot \psi$.
  By the assumption, it follows that $\kmodel, w \Vdash \phi$, and hence $\kmodel, v \Vdash \lnot[a]\lnot \phi$.

  If now $\Gamma_1;\Gamma_2$ has a child $\Pi_1;\Pi_2$ with an interpolant $\theta$ which is obtained using the modal rule $(M)$, then we have one of the following cases:
  \begin{itemize}
  \item $\Gamma_1$ contains a loaded formula of the form $\lnot\underline{[a]}\chi$.
    Then we have $\Pi_1; \Pi_2 = \lnot \chi,{(\Gamma_1)}_a; {(\Gamma_2)}_a$, and a $\theta$ such that
    $\lnot\chi, {(\Gamma_1)}_a \models \theta$ and ${(\Gamma_2)}_a \models \lnot\theta$.

    If $\Pi_2$ is non-empty, then $a \in \voc(\Gamma_1) \cap \voc(\Gamma_2)$.
    As shown above, $\lnot \chi , {(\Gamma_1)}_a \models \theta$ implies $\Gamma_1 \models \lnot[a]\lnot \theta$ (since $\lnot[a]\chi\in \Gamma_1$),
    and ${(\Gamma_2)}_a \models \lnot \theta$ implies $\Gamma_2 \models [a]\lnot \theta$.
    Hence, $\lnot[a]\lnot \theta$ is an interpolant for $\Gamma_1;\Gamma_2$.

    If $\Pi_2$ is empty, then $\Gamma_1$ is not satisfiable.
    To see this, simply notice that the split tableau obtained from $\tab$ by only changing the label of the root to $\Gamma_1;\emptyset$ is also a closed tableau for $\Gamma_1;\emptyset$,
    and thus by \autoref{t:unsatIffSplitClosed}, we get that $\Gamma_1$ is inconsistent.
    Then we have $\Gamma_1 \models \bot$ and $\Gamma_2 \models \lnot \bot$, and $\bot$ is an interpolant for $\Gamma_1;\Gamma_2$.

  \item $\Gamma_2$ contains a loaded formula of the form $\lnot\underline{[a]}\chi$.
    Then we have $\Pi_1; \Pi_2= {(\Gamma_1)}_a; \lnot \chi,{(\Gamma_2)}_a$, and a $\theta$ such that
    ${(\Gamma_1)}_a \models \theta$ and $\lnot \chi, {(\Gamma_2)}_a \models \lnot\theta$.

    If $\Pi_1$ is non-empty, then $a \in \voc(\Gamma_1) \cap \voc(\Gamma_2)$.
    As shown above, ${(\Gamma_1)}_a \models \theta$ implies $\Gamma_1 \models [a]\theta$,
    and $\lnot \chi , {(\Gamma_2)}_a \models \lnot\theta$ implies $\Gamma_2 \models \lnot[a]\theta$ (since $\lnot[a]\chi\in \Gamma_2$).
    Hence, $[a] \theta$ is an interpolant for $\Gamma_1;\Gamma_2$.

    Analogous to the previous case, if $\Pi_1$ is empty, then $\Gamma_2$ is not satisfiable.
    Then we have $\Gamma_1 \models \lnot \bot$ and $\Gamma_2 \models \bot$, and $\lnot \bot$ is an interpolant for $\Gamma_1;\Gamma_2$.
  \end{itemize}
  This shows the lemma.
\end{proof}

\subsection{Defining interpolants for proper clusters}\label{subsec:InterpolationForProperClusters}

Next, we take care of the more difficult case of defining interpolants for roots of proper clusters,
assuming that interpolants have been found for all exits of the cluster.
The formal result is stated in \autoref{l:clusterInterpolation}, which provides the inductive step for proper clusters in the proof of
\autoref{t:tabToInt}.

\begin{lem}[\checklean{clusterInterpolation}]\label{l:clusterInterpolation}
Let $\tab$ be a uniform closed split PDL-tableau for $\Gamma_1;
\Gamma_2$, and assume that the root $r$ of $\tab$ belongs to a proper
cluster $C$.
Furthermore, assume that for all exit nodes $t \in C^{+}\setminus C$,
we have an interpolant $\theta_{t}$ for the split sequent
$\Lambda_{1}(t);\Lambda_{2}(t)$.
Then $\Gamma_1; \Gamma_2$ has an interpolant.
\end{lem}

In the remainder of this section and in \autoref{section:MainTheorem}, let $\tab$ be a uniform closed split PDL-tableau for
$\Gamma_1;\Gamma_2$ and let $C$ be a proper cluster of $\tab$ such that the root $r$ of $\tab$ is the root of $C$.
We first consider the case where $\Gamma_2$ is loaded, and then in the proof of \autoref{l:clusterInterpolation} we will show how to define the interpolant for the case where $\Gamma_1$ is loaded.

\begin{center}
\fbox{Fix $\tab, C, r, \Gamma_{1}, \Gamma_{2}$ for the rest of \autoref{sec:DefinitionInterpolants} and for \autoref{subsec:CorrectnessOfTheInterpolants}.}\label{fixedTabAndCluster}
\end{center}
where $\Gamma_{2}$ is assumed to be loaded.
Here are some first observations on $C$.

\begin{lem}\label{l:splitpcl}
For all $t \in C$, the following hold:
\begin{enumerate}[(a)]
\item $\Lambda_{2}(t)$ is loaded (and thus, nonempty).
\item $\Lambda_{1}(t)$ is empty iff $\Gamma_{1}$ is empty.
\item All children of $t$ are in $C^{+}$. If $t$ is not a \lpr then at least one of them belongs to $C$,
and if $t$ is a \lpr then its companion $c(t)$ belongs to $C$.
\end{enumerate}
\end{lem}
\begin{proof}
  We fix any $t \in C$ and show the three claims.
  \begin{enumerate}[(a)]
  \item
    By the assumptions, the root of cluster $C$ is labelled $\Gamma_1;\Gamma_2$ and $\Gamma_2$ is loaded.
    Since every node in $C$ lies on a path from some companion to its repeat, it follows that every node in $C$ is loaded.
    Since loading can only change components by first unloading, the loading cannot change sides inside $C$.
  \item
    If $\Gamma_1$ is empty then the label $\Lambda_1(r)$ of the first component at the root is empty.
    No rule introduces formulas into an empty component, so also $\Lambda_1(t)$ must be empty.
  \item
    The first part of the claim is immediate from the definition of $C^{+}$.
    The second part holds since otherwise $t$ would not belong to $C$.
    \qedhere
  \end{enumerate}
\end{proof}

\begin{rem}\label{r:itp-emptycase}
In case $\Gamma_{1} = \emptyset$, we can simply define the interpolant for $r$ as $\theta_{r} \isdef \top$, see \autoref{d:itp}.
In this case, it is easy to see that $\theta_r$ is an interpolant for $\Gamma_1;\Gamma_2$.
In the remainder of this section, we will show how to construct an interpolant for $r$ when $\Gamma_1 \neq \emptyset$.
This interpolant will be constructed from information that is stored in the loaded, right component of nodes in $C$.   
\end{rem}

It will be convenient to introduce some further notation.
Note that given any map we use square brackets for its direct image map, so that $\Lambda_{2}[C^{+}] \isdef \{ \Lambda_{2}(t) \mid t \in C^{+}\}$ is the set of (loaded) sequents that occur as second component in $C^+$, and similarly for $\Lambda_{2}[C]$.

\begin{defi}
Given a sequent $\Delta \in \Lambda_{2}[C^{+}]$, we define
\[\begin{array}{lll}
   C_{\Delta} &\isdef& \{ t \in C \mid \Lambda_{2}(t) = \Delta \},
\\ C_{\Delta}^{+} &\isdef& \{ t \in C^{+} \mid \Lambda_{2}(t) = \Delta \},
\end{array}\]
and we let $C_{\Delta}^{L}$ ($C_{\Delta}^{R}$) denote the set of nodes in $C_{\Delta}$ where a left rule (a right rule) is applied.
For example, $t\in C_{\Delta}^{L}$ means that $t$ is a node in $C$ whose second component is $\Delta$ --- here $\Delta$ is a sequent that occurs as the second component of some node in $C^+$ --- with the property that a left rule was applied at $t$.
\end{defi}

Below we gather some further observations.

\begin{lem}\label{l:itpbas}
For all sequents $\Delta \in \Lambda_{2}[C^{+}]$, the following hold:
\begin{enumerate}[(a)]
\item $C_{\Delta} = C_{\Delta}^{L} \cup C_{\Delta}^{R}$ and $C_{\Delta}^{L} \cap C_{\Delta}^{R} = \emptyset$.
\item If $t \in C_{\Delta}^{L}$ then $\Lambda_{1}(t)$ is not basic.
\item If $t \in C_{\Delta}^{L}$ then all of its children belong to $C_{\Delta}^+$, and at least one to $C_{\Delta}$.
\item If $C_\Delta$ is not empty then $C^R_\Delta$ is not empty.
\item If $\Delta$ is basic and $t \in C_{\Delta}^{R}$ then the modal rule is applied at $t$ to the loaded formula in $\Lambda_{2}(t) = \Delta$, and the unique child of $t$ also belongs to $C$.
\item If $\Delta$ is not basic and $C^R_{\Delta} \neq \emptyset$ then there is a unique local rule $R_2$ and a possibly loaded formula $\xi \in \Delta$ such that for all $t \in C_{\Delta}^{R}$, the rule $R_2$ with principal formula $\xi$ is applied at $t$.
Consequently, there are sequents $\Pi_{1},\ldots,\Pi_{n}$ such that
$\bigwedge \Delta \equiv \bigvee_{i} \bigwedge \Pi_{i}$, and, for all $t \in C_{\Delta}^{R}$, the children of $t$ can be listed as $s_{1},\ldots, s_{n}$ such that $\Lambda(s_{i}) = \Lambda_{1}(t);\Pi_{i}$, for all $i$.
\end{enumerate}
\end{lem}

\begin{proof}
Since every rule is either a left or a right rule but not both, we have (a).
Item~(b) follows from the earlier observation that $\Lambda_2(t)$ is loaded, and thus, $\Lambda_1(t)$ is free; if a rule on the first component was applied at~$t$, then it cannot be basic.
Item (c) is an instantiation  of \autoref{l:splitpcl}(c), together with the observation that, obviously, the rule applied at any $t \in C_{\Delta}^{L}$ is a left rule, distinct from $(M)$, and hence does not touch $\Delta$.

We prove item (d):
Assume that $C_\Delta \neq \emptyset$.
Then there exists $t \in C_\Delta$ such that $t$ is minimal in the tree relation, that is, there is no $t' \in C_\Delta$ such that $t' \edgeT t$.
Moreover, since all nodes in $C$ are loaded in the second component by \autoref{l:splitpcl}(a), it follows that $\Delta$ is loaded.
If $t \in C^R_\Delta$ then we are done. So suppose $t \notin C^R_\Delta$.
We will show that there exists $u \in C_\Delta$ such that $t \edgeT u$, $\Lambda_1(u)$ is basic and~$u$ is not a \lpr.
That implies that a rule was applied at $u$ and that it was not a rule on the first (i.e., left) component. Since every rule is a left or right rule (Item~(a)), it follows that $u \in C^R_\Delta$.
This~$u$ is obtained as follows.
Since $t \notin C^R_\Delta$, Item~(a) implies that $t \in C^L_\Delta$, i.e., some left rule is applied at $t$. By Item~(c), there is a child~$v$ of~$t$ that belongs to $C_\Delta$. By repeating this argument until no more left rules can be applied, we obtain a path $v_0 v_1 \cdots v_n$ in $C_\Delta$ such that $t=v_0$ and $\Lambda_1(v_n)$ is basic, and for all $i < n$ we have $v_i \in C^L_\Delta$.
We now prove that for all $i \leq n$, the node~$v_i$ cannot be a \lpr. Suppose towards a contradiction that $v_i$ is a \lpr. Then by minimality of $t$ we must have $t \edgeRT c(v_i) \edgeT v_i$.
Hence the unique path from $c(v_i)$ to $v_i$ in $\tab$ must be a segment of the path $v_0 v_1 \cdots v_n$.
By \autoref{f:lprAreCritical}, the path from~$c(v_i)$ to~$v_i$ must include an application of the modal rule,
but since only left rules are applied along this path, and the loaded formula is in the right component, this is impossible.
Hence we have a contradiction, and conclude that no~$v_i$ is a \lpr, so we can take $u=v_n$.

The first observation of Item~(e) follows from the fact that the right component of $t$ is loaded; the second observation is another instance of \autoref{l:splitpcl}(c).

Finally, for Item~(f), suppose that~$\Delta$ is not basic and $C^R_{\Delta} \neq \emptyset$.
By uniformity, the same (right) rule is applied at every $t \in C^R_{\Delta}$, and since $\Delta = \Lambda_2(t)$ is not basic, this rule must be local, and therefore only affects the right component $\Lambda_2(t)$.
The rest of the statement in Item~(f) then follows from the local soundness and invertibility of the rule applied at~$t$ (\autoref{l:splitlocalRuleTruth}).
\end{proof}

The next definition introduces the pivotal structure of the interpolation proof.
The basic idea is to represent the cluster $C$ by a finite labeled tree $\Q$ of which the nodes are labeled with \emph{right-side sequents only}.
Every node $x$ in $\Q$  will represent a certain subset $R_{x} \subseteq C^+$ (to be defined in \autoref{d:regfma}).

\begin{defi}[\lean{QuasiTab.build}, \lean{LoadedCluster.Q}]\label{d:Qtableau}
Step by step we will construct a structure $\Q = (Q,\qedge,k,\Delta)$ that we will refer to as the \emph{quasi-tableau} associated with the cluster $C$.
Here, $(Q,\qedge)$ is a finite tree, $k \colon Q \to \{ 1, 2, 3 \}$ assigns a
\emph{type} $k(x)$ to every node $x \in Q$, and $\Delta$ is a map that labels every node
$x \in Q$ with a finite sequent $\Delta_{x} \in \Lambda_{2}[C^{+}]$.

We construct $\Q$ starting from the root, and maintain the following invariant for all nodes $x \in Q$:
If $x$ is not a leaf then $C_{\Delta_x} \neq \emptyset$. (Note that $C_{\Delta_x} \neq \emptyset$ iff $\Delta_x \in \Lambda_2[C]$.)

To start the construction, we add the root node $r_{\Q}$ to $\Q$ and let $r_{\Q}$ have type 1 and
label $\Gamma_{2}$, i.e., $\Delta_{r_{\Q}}\isdef\Gamma_{2} = \Lambda_2(r)$ is the loaded right sequent of the
root $r$ of $C$. Since $r \in C_{\Gamma_2}$, the invariant holds at the root of $\Q$.

Inductively, given a node $x \in Q$, the set of children of $x$ is defined
by means of the following case distinction.
In each case, we show that if $x$ satisfies the invariant, then so do all its children.
\begin{description}
\item[Case $k(x) = 1$]
If $\Delta_{x} \in \Lambda_{2}[C^{+}] \setminus \Lambda_{2}[C]$, or if
there exists $z \in Q$ such that $z$ is an ancestor of $x$ and $\Delta_x=\Delta_z$,
then $x$ is a leaf.
In the latter case, we call $x$ a \emph{repeat}.
In this case, the invariant holds trivially.
Otherwise, that is, if $x$ is not a repeat and $\Delta_{x} \in \Lambda_{2}[C]$,
then $x$ has a unique child $y$ with type 2 and label $\Delta_{x}$.
In this case the invariant holds at $y$ since $\Delta_y = \Delta_{x} \in \Lambda_{2}[C]$.
\item[Case $k(x) = 2$]
$x$ has a unique child $y$ with type 3 and label $\Delta_{x}$.
The invariant holds at $y$ since $\Delta_y = \Delta_x$ and the invariant holds at $x$.
\item[Case $k(x) = 3$]
If $\Delta_{x}$ is not basic, then by the invariant for $x$, we have $C_{\Delta_x} \neq \emptyset$.
Hence by \autoref{l:itpbas}(d), $C^R_{\Delta_x} \neq \emptyset$, and
by \autoref{l:itpbas}(f),
there is a unique local rule $R_2$ and a unique, possibly loaded formula $\xi \in \Delta_x$ such that $R_2$ with principal formula $\xi$ is applied uniformly at all $t \in C^R_{\Delta_x}$.
Let $\Delta_x\,/\,\Pi_{1} | \ldots | \Pi_{n}$ be the rule instance of the non-split version of $R_2$  with principal formula $\xi$.
We define the set of children of $x$ in $\Q$ as the set $\{ y_{1},\ldots,y_{n} \}$, where each $y_{i}$ has type 1 and label $\Pi_{i}$.
We show that the invariant holds for all $y_i$:
If $y_i$ is not a leaf in $\Q$, then since $k(y_i)=1$ it follows from the above item for $k(x)=1$ that $\Delta_{y_i} \in \Lambda_2[C]$.

If $\Delta_{x}$ is basic and the unique loaded formula in
$\Delta_{x}$ is of the form $\lnot\loaded{[a][\delta_{1}]\cdots[\delta_{n}]}\phi$,
then $x$ has a unique child $y$ with type 1 and label
${(\Delta_{x})}_{a} \cup \{ \loaded{ \neg[\delta_{1}]\cdots[\delta_{n}]}\phi \}$.
We show that the invariant holds at $y$: Since it holds at $x$,  i.e.,
$C_{\Delta_x} \neq \emptyset$,
by \autoref{l:itpbas}(d), there exists $t \in C^R_{\Delta_x}$,
and by \autoref{l:itpbas}(e), the unique child $v$ of $t$ belongs to $C$.
Since $v$ is obtained by applying the modal rule ${(M)}_2$ at $t$,
it follows that $\Delta_y = \Lambda_2(v)$.
Hence $\Delta_y \in \Lambda_2[C]$.
\end{description}
\end{defi}

\begin{rem}\label{r:Qtableau-leaves}
Note that termination of the construction at a node~$z \in Q$ takes place in one of two ways: either $z$ is a repeat in $\Q$ or $\Delta_{z}\in\Lambda_{2}[C^{+}] \setminus \Lambda_{2}[C]$ (observe that $\Delta_z$ must be unloaded in the latter case).
In both cases, $k(z)=1$.
\end{rem}

Before discussing the role of the quasi-tableau $\Q$ in the interpolation
proof, we first gather some basic definitions and notations related to it.

As mentioned, each $x \in Q$ represents a subset $R_{x}$ of $C^{+}$ with fixed loaded component.
If $x$ has type 1 or 2 then $R_x$ consists of all nodes in $C^+$ that have $\Delta_x$ as their loaded (right) component.
If $x$ has type 3 then $R_x$ consists of all nodes in $C$ that have $\Delta_x$ as their loaded (right) component, and moreover a rule is applied in the loaded (right) component.

\begin{defi}[\lean{LoadedCluster.region}]\label{d:regfma}
For $x \in Q$, we define the set (sometimes called a \emph{region}) $R_{x} \subseteq C^{+}$ as follows:
\[
R_{x} \isdef \left\{\begin{array}{ll}
   C_{\Delta_x}^{+} & \text{if } k(x) \in \{1,2\}
\\[1mm] C_{\Delta_x}^{R} & \text{if } k(x) = 3.
\end{array}\right.
\]
\end{defi}

Recall that a repeat in $\Q$ is defined in~\autoref{d:Qtableau}.

\begin{defi}\label{def-companions}
We let $<_{\Q}$ and $\leq_{\Q}$ denote, respectively, the
transitive and the reflexive-transitive closure of the relation $\qedge$.

Along any branch of $\Q$, by construction, the sequent label changes only when transitioning from a node of type~3 to a node of type~1. Furthermore, repeats are identified at the first opportunity, since they are defined in case $k(x)=1$ of \autoref{d:Qtableau}.
Consequently, for every repeat leaf $x \in Q$, there is a unique node $z \in Q$ such that $z <_{\Q} x$, $\Delta_{x} = \Delta_{z}$ and $k(x) = k(z) = 1$.
Such a node~$z$ is called the \emph{companion of~$x$}, and it is denoted $\compan{x}$.

  We let $K_{\Q} = \{ c(x) \mid x \in Q \text{ is a repeat leaf}\}$ denote the \emph{set of companions}.
Finally, $L_{\Q}$ and $r_{\Q}$ denote, respectively, the set of leaves and the root of $\Q$.
Given a node $x \in Q$, we set
\[
  \cycs{x} \isdef \{ z \in L_{\Q} \mid  z \text{ is a repeat leaf and } c(z) <_{\Q} x \leq_{\Q} z \}.
\]
\end{defi}

In words, $\cycs{x}$ denotes the set of repeat leaves that are descendants of $x$ and whose companion node is a proper ancestor of $x$.

The following lemma collects the properties of these notions that we shall need.

\begin{lem}\label{l:Q.Properties}
Let $\Q$ be a quasi-tableau and let $x, y \in Q$. Then:
  \begin{enumerate}[(a)]
      \item If $x$ is a repeat leaf or a companion node in $\Q$, then $k(x) = 1$.
      \item $\cycs{r_{\Q}} = \emptyset$, and $\cycs{z} = \{ z \}$ for every repeat leaf $z$.
      \item If $x$ is a companion, then it has a unique child $y$. Furthermore:
      \[
        \cycs{y} = \cycs{x} \cup \{ z \in L_{\Q} \mid z \text{ is a repeat leaf and } c(z) = x \},
      \]
      where the two sets on the right-hand side are disjoint.
      \item If $x$ is neither a repeat leaf nor a companion, then
      $\cycs{x} = \bigcup \{ \cycs{y} \mid x \qedge y \}$.
      In particular $\cycs{y} \subseteq \cycs{x}$ for every $y$ with $x \qedge y$, and $\cycs{x} = \cycs{y}$ in case $y$ is the unique child of $x$.
  \end{enumerate}
\end{lem}
\begin{proof}
    Item~(a) is immediate from the definitions.

    For Item~(b), a repeat leaf $z \in \cycs{r_{\Q}}$ would satisfy $c(z) <_{\Q} r_{\Q}$, which is impossible since $r_{\Q}$ is the root of $\Q$.
    For the second part, let $z$ be a repeat leaf.
    Then $c(z) <_{\Q} z \leq_{\Q} z$, so that $z \in \cycs{z}$; and every $w \in \cycs{z}$ satisfies $z \leq_{\Q} w$, whence $w = z$ because $z$ is a leaf.

    For Item~(c), let $x$ be a companion.
    By Item~(a) we have $k(x) = 1$, and $x$ is not a leaf, since by the definition of companions (\autoref{def-companions}) $x$ has a repeat leaf among its proper descendants.
    By \autoref{d:Qtableau}, $x$ has a unique child $y$.
    Let $z \in \cycs{y}$, so that $c(z) <_{\Q} y$ and $y \leq_{\Q} z$.
    From $x \qedge y$ we obtain $c(z) \leq_{\Q} x$; either $c(z) = x$, or else $c(z) <_{\Q} x$ and then $z \in \cycs{x}$, since $x <_{\Q} y \leq_{\Q} z$.
    Conversely, let $z$ be a repeat leaf with $c(z) <_{\Q} x \leq_{\Q} z$ or $c(z) = x$.
    In both cases $x <_{\Q} z$, because $x$ is not a leaf, and hence $y \leq_{\Q} z$ as $y$ is the unique child of $x$; and in both cases $c(z) <_{\Q} y$, since $c(z) \leq_{\Q} x \qedge y$.
    So $z \in \cycs{y}$.
    The two sets are disjoint because $z \in \cycs{x}$ gives $c(z) <_{\Q} x$ and hence $c(z) \neq x$.

    For Item~(d), let $x$ be neither a repeat leaf nor a companion.
    For the inclusion from right to left, let $x \qedge y$ and $z \in \cycs{y}$, so that $c(z) <_{\Q} y$ and $y \leq_{\Q} z$.
    From $x \qedge y$ we obtain $c(z) \leq_{\Q} x$, and $c(z) \neq x$ since $x$ is not a companion, so that $c(z) <_{\Q} x$; and $x \qedge y \leq_{\Q} z$ gives $x \leq_{\Q} z$.
    Hence $z \in \cycs{x}$.
    For the inclusion from left to right, let $z \in \cycs{x}$, so that $c(z) <_{\Q} x \leq_{\Q} z$.
    Then $x \neq z$, since $z$ is a repeat leaf and $x$ is not, so that $x <_{\Q} z$.
    Let $y$ be the child of $x$ with $y \leq_{\Q} z$.
    From $c(z) <_{\Q} x <_{\Q} y$ we conclude that $z \in \cycs{y}$.
    Both consequences are immediate, the second one because the union then has $\cycs{y}$ as its only member.
\end{proof}

We are now ready to define interpolants.
The key idea is to define for each node $x$ in the quasi-tableau $\Q$ a so-called \emph{pre-interpolant}
$\iitp_{x}$.
The definition of pre-interpolants proceeds by a leaf-to-root
induction on the tree $(Q,\qedge)$, and ensures that
the pre-interpolant $\iitp_{r_{\Q}}$ of the root $r_{\Q}$ of $\Q$ is
an interpolant for the root $r$ of $C$. Hence, to obtain the main result, we can take $\theta_r = \iitp_{r_{\Q}}$.

To start the definition of the pre-interpolants, in the base case we need to
consider the \emph{leaves} of $\Q$.
As noted in \autoref{r:Qtableau-leaves}, a leaf $x$ has type 1, and there are two possibilities:
$x$ is a repeat or $\Delta_{x}$ belongs to $\Lambda_{2}[C^{+}] \setminus \Lambda_{2}[C]$.
In the second case, it follows that $R_x \subseteq C_{\Delta_x}^+\setminus C$, i.e.,
$x$ represents a region of exit nodes.
Recall that by the assumption in \autoref{l:clusterInterpolation} we already have an
interpolant $\theta_{t}$ for every node $t \in C^{+}\setminus C$.
This guarantees the correctness of the following definition, and the subsequent observation.

\begin{defi}\label{d:iota}
Given a sequent $\Delta \in \Lambda_{2}[C^{+}]$, we define
\[
\theta_{\Delta} \isdef \bigvee \{ \theta_{t} \mid t \in C^{+}_{\Delta}\setminus C_{\Delta}\}.
\]
\end{defi}

The following fact states that the formula $\theta_{\Delta}$ behaves as a ``region'' interpolant for the set $C^{+}_{\Delta}\setminus C_{\Delta}$ of exit nodes that have $\Delta$ as the second component.
(Note that if $C^{+}_{\Delta}\setminus C_{\Delta}=\emptyset$ then $\theta_\Delta \equiv \bot$.)

\begin{lem}\label{f:iota}
  Let $\Delta$ be a sequent in $\Lambda_{2}[C^{+}]$.
  Then we have:
  \begin{enumerate}[(a)]
  \item $\Lambda_{1}(t) \models \theta_{\Delta}$, for all $t \in C^{+}_{\Delta}\setminus C_{\Delta}$.
  \item $\Delta \models \lnot\theta_{\Delta}$.
  \item $\voc(\theta_\Delta)\subseteq \left(\bigcup_{t \in C^{+}_{\Delta}\setminus C_{\Delta}} \voc(\Lambda_1(t))\right)\cap\voc(\Delta)$.
  \end{enumerate}
\end{lem}

\begin{proof}
All three statements in the lemma are straightforward consequences of the assumption that each formula $\theta_t$ is an interpolant.
For instance, this assumption means that $\Lambda_{1}(t) \models \theta_{t}$, for all $t \in C^{+}_{\Delta}\setminus C_{\Delta}$, and so $\theta_{\Delta}$, the disjunction of these $\theta_t$, is also a consequence of each $\Lambda_1(t)$.
\end{proof}

Continuing our discussion of the base case of the inductive definition of the pre-interpolants, we now consider the leaves that are repeats.
For these we introduce some new propositional variables that we shall call \emph{internal}, and that will also feature in the pre-interpolants of the internal nodes of $\Q$ (but not in the one of the root).
More precisely, we will need a propositional variable $q_{x}$ for every companion node $x$ in $Q$.
For the specific syntactic shape of the pre-interpolants, we need to introduce some terminology.

\begin{defi}[\lean{QFormula}]\label{d:iptfma}
The set $\itpfmas$ of formulas (from which we will take the pre-interpolants)
is given by the following grammar:
\BNF{\itpfmas\ni\iota}{\psi \mid q \mid \iota \land \iota \mid
\Box(\ol{\alpha},\iota)}
where $q \in \{ q_{x} \mid x \in K_{\Q} \}$ (called the set of internal variables), and $\psi$ and $\ol{\alpha}$ are such that their vocabulary contains no internal variables.
Formulas in $\itpfmas$ will be referred to as \emph{$\Q$-formulas}.

A formula in $\itpfmas$ is \emph{simple} if it is of the form $\psi$ or $\Box(\ol{\alpha},q_{x})$, and \emph{in normal form} if it is a conjunction of such formulas.
\end{defi}

We now define a normal form for every formula $\iota \in \itpfmas$.

\begin{defi}[\lean{QFormula.Spl}]
By structural induction on $\itpfmas$-formulas, we define a function
$\Simple$ mapping every $\iota \in \itpfmas$ to a finite set of simple formulas:
\[\begin{array}{lll}
\Simple(\psi) &\isdef& \{ \psi \}
\\ \Simple(q_{x}) &\isdef& \{ [\top?]q_{x} \}
\\ \Simple(\iota_{1} \land \iota_{2}) &\isdef&
   \Simple(\iota_{1}) \cup \Simple(\iota_{2})
\\ \Simple([\alpha]\iota) &\isdef& \{ [\alpha]\theta \mid \theta\in \Simple(\iota)\}.
\end{array}\]
Define the \emph{normal form} of $\iota \in \itpfmas$ as the formula
\[
\nf{\iota} \isdef \bigwedge \Simple(\iota).
\]
\end{defi}

The following fact can be proved by a straightforward induction on the complexity of $\Q$-formulas.

\begin{fact}[\checklean{QFormula.eval_nf}, \checklean{QFormula.voc_nf}]\label{f:iota-nf}
Every formula $\iota \in \itpfmas$ is equivalent to its normal form.
Moreover, $\voc(\iota) = \voc(\nf{\iota})$.
\end{fact}

With all these preparations out of the way, we can now define pre-interpolants, starting at the leaves of the quasi-tableau and enriching the formula as we proceed towards the root: Informally, at a repeat leaf~$x$, the pre-interpolant is set as an internal variable~$q_{c(x)}$, and at a leaf~$x$ such that $\Delta_{x} \in {\Lambda_{2}[C^{+}] \setminus \Lambda_{2}[C]}$, the pre-interpolant is set as $\theta_{\Delta_x}$; informally, this is the interpolant for the exit node.
These are the only two base cases, see \autoref{r:Qtableau-leaves}.
The pre-interpolant is successively enriched as we proceed towards the root: the test operator is used where disjunction would usually be used, and at a companion~$c(x)$, a solution to the fixpoint equation --- existence guaranteed by equivalence to the normal form --- is employed to eliminate the internal variable~$q_{c(x)}$.

Later in \autoref{subsec:ExampleInterpolant} we will also provide a concrete example to compute an interpolant using \autoref{d:iitp}.
The reader might already want to have a look at \autoref{f:exk2Q} to get a better understanding of how pre-interpolants are defined.

\begin{defi}[\lean{LoadedCluster.iitp}]\label{d:iitp}
By a leaf-to-root induction on the relation $\qedge$ we define, for each node
$x \in Q$, a formula $\iitp_{x}$ (called a \emph{pre-interpolant}) as follows:
\begin{description}
\item[Case $k(x) = 1$, $x$ is a repeat] $\iitp_{x} \isdef q_{c(x)}$.
\item[Case $k(x) = 1$,
 $\Delta_{x} \in {\Lambda_{2}[C^{+}] \setminus \Lambda_{2}[C]}$]
$\iitp_{x} \isdef \theta_{\Delta_x}$ (see \autoref{d:iota}).

\item[Case $k(x) = 1$, $x$ is a companion]
In this case $x$ has a unique child $y$, and inductively we may assume that the formula $\iitp_{y} \in \itpfmas$ has already been defined.
By \autoref{f:iota-nf}, $\iitp_{y}$ is equivalent to its normal form
\[
\iitp_{y} \equiv {\bigwedge_{i}}[\alpha_{i}]q_{x} \land {\bigwedge_{j}}[\beta_{j}]q_{z_{j}} \land \psi,
\]
where $q_{x}$ does not occur in any of the $\alpha_{i}$ or in the formula $\bigwedge_{j}[\beta_{j}]q_{z_{j}} \land \psi$.
Now define
\[
\iitp_{x} \isdef
[{({\bigcup_{i}}\alpha_{i})}^{*}]
\big({\bigwedge_{j}}[\beta_{j}]q_{z_{j}} \land \psi\big).
\]
Here the conjunction ${\bigwedge_{i}}[\alpha_{i}]q_{x}$ may be empty, namely when $q_{x}$ does not occur in $\iitp_{y}$.
In that case we read ${\bigcup_{i}}\alpha_{i}$ as $\bot?$, the program denoting the empty relation, so that $\iitp_{x} \equiv \iitp_{y}$.
\item[Case $k(x) = 1$, $x$ neither a leaf nor a companion]
In this case, $x$ also has a unique child $y$, for which the formula $\iitp_{y} \in \itpfmas$ has already been defined.
We now set $\iitp_{x}\isdef\iitp_{y}$.
\item[Case $k(x) = 2$]
In this case, $x$ has a unique child $y$, and inductively we may assume that the formula $\iitp_{y}$ has already been defined.
\[
\iitp_{x} \isdef [\lnot\theta_{\Delta_{x}}?]\iitp_{y}.
\]
\item[Case $k(x) = 3$, $\Delta_{x}$ is not basic]
Here we define
\[
\iitp_{x} \isdef \bigwedge \{ \iitp_{y} \mid x \qedge y \}.
\]
\item[Case $k(x) = 3$, $\Delta_{x}$ basic]
In this case $x$ has a unique child $y$, and we define
\[
\iitp_{x} \isdef [a] \iitp_{y},
\]
where $a$ is the leading atomic program of the loaded formula in $\Delta_{x}$.
\end{description}
\end{defi}

\begin{rem}\label{r:dfitp}
We briefly discuss the two most interesting cases of \autoref{d:iitp}, which may also shed some light on the reasons why we need to differentiate the nodes of different types.

First, in the case where $k(x) = 1$ and $x$ is a companion node, one may think of $\iitp_{x}$ as the greatest fixpoint of the formula $\iitp_{y}$ with respect to the proposition letter $q_x$.
The fixpoint part of this is the statement that
\[
\iitp_{x} \equiv \subst{\substAt{q_{x}}{\iitp_{x}}}{\iitp_{y}}.
\]
To see this, notice that if $\iitp_{y} \equiv {\bigwedge_{i}}[\alpha_{i}]q_{x} \land {\bigwedge_{j}}[\beta_{j}]q_{z_{j}} \land \psi$ then we have
\begin{align*}
    \iitp_{x}
    &=
    [{({\bigcup_{i}}\alpha_{i})}^{*}]
    \big({\bigwedge_{j}}[\beta_{j}]q_{z_{j}} \land \psi\big)
    \\
    &\equiv
    [({\bigcup_{i}}\alpha_{i})][{({\bigcup_{i}}\alpha_{i})}^{*}]
    \big({\bigwedge_{j}}[\beta_{j}]q_{z_{j}} \land \psi\big)
    \land
    {\bigwedge_{j}}[\beta_{j}]q_{z_{j}} \land \psi
    \\
    &=
    [({\bigcup_{i}}\alpha_{i})] \iitp_{x}
    \land
    {\bigwedge_{j}}[\beta_{j}]q_{z_{j}} \land \psi
    \\
    &\equiv
    {\bigwedge_{i}}[\alpha_{i}]\iota_{x} \land {\bigwedge_{j}}[\beta_{j}]q_{z_{j}} \land \psi
\end{align*}

Second, in the case $k(x) = 2$, for the interpolant to be correct we need to combine the formula $\theta_{\Delta_x}$, which is an interpolant of the nodes in $C_{\Delta_x}$ that belong to $C^+ \setminus C$, with $\iota_y$.
For this reason, we assign the pre-interpolant $[\lnot \theta_{\Delta_x}?]\iota_y \equiv \theta_{\Delta_x} \lor \iota_y$ to nodes of type 2.
Using a test here, instead of a disjunction, guarantees that the normalisation succeeds when later on we arrive at a type-1 node in $\Q$.
\end{rem}

\begin{defi}[\lean{LoadedCluster.itp}]\label{d:itp}
The interpolant $\theta_r$ of the root $r$ of the cluster $C$ is given as
\[
\theta_{r} \isdef \left\{ \begin{array}{ll}
     \top          & \text{if } \Lambda_{1}(r) = \emptyset,
  \\ \iitp_{r_{\Q}} & \text{if } \Lambda_{1}(r) \neq \emptyset.
\end{array}\right.
\]
\end{defi}

\section{Interpolation: Main Theorem}\label{section:MainTheorem}

\subsection{Correctness of the interpolants}\label{subsec:CorrectnessOfTheInterpolants}

In this section, we first show that the formula $\theta_r$ from \autoref{d:itp} is indeed an interpolant at the root $r$ of $C$, where $C$ is the proper cluster that we fixed on page~\pageref{fixedTabAndCluster} together with a uniform closed split PDL-tableau $\tab$ for the sequent $\Gamma_1;\Gamma_2$ that labels the root $r$.
In this section, we furthermore let $\Q$ denote the quasi-tableau for $C$ from \autoref{d:Qtableau}.

First, we consider the vocabulary condition.
Recall that $\voc(\Gamma)$ is the set of atomic programs and proposition letters occurring in formulas in the sequent~$\Gamma$,
and for $x \in Q$, we defined
$\cycs{x} = \{ z \in L_{\Q} \mid  z \text{ is a repeat and } c(z) <_{\Q} x \leq_{\Q} z \}$ (see \autoref{def-companions}).

\begin{lem}[\checklean{LoadedCluster.iitp_voc}, \checklean{LoadedCluster.iitp_vars}]\label{l:itp-voc}
The vocabulary condition holds for $\theta_r$, that is, 
\[\voc(\theta_{r}) \subseteq \voc(\Gamma_{1}) \cap \voc(\Gamma_{2}).\]
\end{lem}

\begin{proof}
If $\Gamma_{1} = \emptyset$ then $\theta_{r} = \top$, and hence the vocabulary condition holds trivially.
If $\Gamma_{1} \neq \emptyset$ then $\theta_r = \iitp_{r_{\Q}}$.
We now show that for all nodes $x \in Q$, we have
\begin{equation}
\label{claim:itp-voc}
\voc(\iitp_{x}) \subseteq \big( \voc(\Gamma_{1}) \cap \voc(\Gamma_{2}) \big)
\cup \{ q_{c(z)} \mid z \in \cycs{x} \}.
\end{equation}
By taking $x=r_{\Q}$, the lemma statement follows, since $\cycs{r_{\Q}}=\emptyset$ (by \autoref{l:Q.Properties}(b)).

First, observe that by \autoref{l:vocIsPreserved}, it follows that $\voc(\Lambda_1(t))\subseteq\voc(\Gamma_1)$ for all $t \in C^{+}_{\Delta_x}\setminus C_{\Delta_x}$ and that $\voc(\Delta_x)\subseteq\voc(\Gamma_2)$.

We prove \claimRef{claim:itp-voc} by leaf-to-root induction on the relation $\qedge$.
The base cases are covered by the first two cases below, and the induction step is covered by the remaining cases.
\begin{description}
\item[Case $k(x) = 1$, $x$ is a repeat]
We have $\iitp_{x} = q_{c(x)}$ so the claim holds by \autoref{l:Q.Properties}(b), which gives $\cycs{x} = \{ x \}$.

\item[Case $k(x) = 1$,
 $\Delta_{x} \in {\Lambda_{2}[C^{+}] \setminus \Lambda_{2}[C]}$]
Note $\iitp_{x} = \theta_{\Delta_x}$. The claim holds by \autoref{f:iota}(c).

\item[Case $k(x) = 1$, $x$ is a companion]
In this case, $x$ has a unique child $y$, and $
\iitp_{y} \equiv {\bigwedge_{i}}[\alpha_{i}]q_{x} \land {\bigwedge_{j}}[\beta_{j}]q_{z_{j}} \land \psi$.
By the induction hypothesis and \autoref{f:iota-nf}, the vocabulary $\voc(\iitp_{y})$ is a subset of
$\big( \voc(\Gamma_{1}) \cap \voc(\Gamma_{2}) \big) \cup \{ q_{c(z)} \mid z \in \cycs{y}\}$.
Hence by the definition of $\iitp_{x}$, it follows that
\[
\voc(\iitp_x) \subseteq
\big( \voc(\Gamma_{1}) \cap \voc(\Gamma_{2}) \big)
\cup
\big( \{ q_{c(z)} \mid z \in \cycs{y}\} \setminus \{ q_x \}\big).
\]
By \autoref{l:Q.Properties}(c), every $z \in \cycs{y}$ with $c(z) \neq x$ lies in $\cycs{x}$, so that $\{ q_{c(z)} \mid z \in \cycs{y}\} \setminus \{ q_x \} \subseteq \{ q_{c(z)} \mid z \in \cycs{x}\}$. The claim follows.

\item[Case $k(x) = 1$, $x$ is neither a leaf nor a companion]
Then $x$ has a unique child~$y$, and we have $\iitp_{x} = \iitp_{y}$.
By \autoref{l:Q.Properties}(d), we have $\cycs{x} = \cycs{y}$.
By the induction hypothesis for $y$, we obtain that:
\[\voc(\iitp_{y}) \subseteq \big( \voc(\Gamma_{1}) \cap \voc(\Gamma_{2}) \big)
\cup \{ q_{c(z)} \mid z \in \cycs{y} \}.\]
Since $\iitp_{x} = \iitp_{y}$ and $\cycs{x} = \cycs{y}$, the claim follows.

\item[Case $k(x) = 2$]
In this case $x$ has a unique child $y$, and $
\iitp_{x} \isdef [\lnot\theta_{\Delta_x}?]\iitp_{y}$.
By the induction hypothesis, $\voc(\iitp_{x})$ is a subset of $\voc(\theta_{\Delta_x})\cup \big( \voc(\Gamma_{1}) \cap \voc(\Gamma_{2}) \big)
\cup \{ q_{c(z)} \mid z \in \cycs{y}\}$.
Notice that $x$ is neither a repeat leaf nor a companion by \autoref{l:Q.Properties}(a), and thus by \autoref{l:Q.Properties}(d) we get that
$\cycs{y}=\cycs{x}$.
Also, by \autoref{f:iota}(c) we find $\voc(\theta_{\Delta_x})\subseteq \voc(\Gamma_1)\cap \voc(\Gamma_2)$. Hence the claim is established.

\item[Case $k(x) = 3$, $\Delta_{x}$ is not basic]
Here $\iitp_{x} \isdef \bigwedge \{ \iitp_{y} \mid x \qedge y \}$.
By the induction hypothesis, $\voc(\iitp_{x})$ is a subset of $\big( \voc(\Gamma_{1}) \cap \voc(\Gamma_{2}) \big) \cup \big(\bigcup_{\{y\mid x \qedge y\}}\{ q_{c(z)} \mid z \in \cycs{y}\}\big)$.
Since $x$ is neither a repeat leaf nor a companion by \autoref{l:Q.Properties}(a),
we have that by \autoref{l:Q.Properties}(d) $\cycs{y} \subseteq \cycs{x}$ for all $y$ such that $x \qedge y$. Hence the claim is established.

\item[Case $k(x) = 3$, $\Delta_{x}$ is basic]
In this case, $x$ has a unique child $y$, and $\iitp_{x} \isdef [a] \iitp_{y}$ where $a$ is the leading atomic program of the loaded formula in $\Delta_{x}$.
By the induction hypothesis, $\voc(\iitp_{x})$ is a subset of $\{a\}\cup \big( \voc(\Gamma_{1}) \cap \voc(\Gamma_{2}) \big)
\cup \{ q_{c(z)} \mid z \in \cycs{y}\}$.
Since $x$ is neither a repeat leaf nor a companion by \autoref{l:Q.Properties}(a),
we have that by \autoref{l:Q.Properties}(d), $\cycs{y} = \cycs{x}$.
It remains to place the atomic program $a$ in $\voc(\Gamma_{1}) \cap \voc(\Gamma_{2})$.
For the second component this is immediate, as $a \in \voc(\Delta_x) \subseteq \voc(\Gamma_2)$.
For the first component, suppose towards a contradiction that $a \notin \voc(\Gamma_{1})$; then $a \notin \voc(\Lambda_{1}(t))$ for every node $t$ of $\tab$, by \autoref{l:vocIsPreserved}.
Since $x$ is not a leaf, the invariant of \autoref{d:Qtableau} gives $C_{\Delta_x} \neq \emptyset$, so by \autoref{l:itpbas}(d) there is a node $t \in C^{R}_{\Delta_x}$.
As $\Delta_{x}$ is basic, \autoref{l:itpbas}(e) tells us that the modal rule is applied at $t$ to the loaded formula of $\Lambda_{2}(t) = \Delta_{x}$, the leading atomic program of which is $a$, and that the unique child $u$ of $t$ belongs to $C$.
That rule is ${(M)}_{2}$, so $\Lambda_{1}(u) = {(\Lambda_{1}(t))}_{a}$ by \autoref{d:split-modalrule}.
No formula of $\Lambda_{1}(t)$ is of the form $[a]\zeta$, since such a formula would put $a$ in $\voc(\Lambda_{1}(t))$, and therefore $\Lambda_{1}(u) = \emptyset$.
This contradicts the observation in \autoref{r:itp-emptycase} that $\Lambda_{1}$ is nonempty at every node of $C$ when $\Gamma_{1} \neq \emptyset$.
Hence $a \in \voc(\Gamma_{1})$, and the claim follows.
\qedhere
\end{description}
\end{proof}

Now that we have shown that the formula $\theta_r$ only uses proposition letters and atomic programs that occur in both $\Gamma_1$ and $\Gamma_2$, we will show that it also has the desired semantic properties.

The following formula $\rho_{x}$ will be useful in the proof of the semantic properties of $\theta_r$.

\begin{defi}\label{d:rho}
Given a quasi-tableau $\Q$, for each node $x \in Q$, let
\[
\rho_{x} \isdef \bigvee \{ \textstyle{\bigwedge \Lambda_{1}(t)} \mid t \in R_{x} \}.
\]
\end{defi}

\begin{lem}[\checklean{LoadedCluster.left_unsat_neg_itp}]\label{l:correctLeftIP}
$\Gamma_{1} \models \theta_{r}$.
\end{lem}

\begin{proof}
In case $\Gamma_{1} = \emptyset$, we have $\theta_{r} = \top$, and so the Lemma holds trivially.
In the remainder of the proof we therefore assume that $\Gamma_{1} \neq \emptyset$, so that $\theta_{r} = \iitp_{r_{\Q}}$.

Consider the substitution $\sigma$ given by setting $\sigma \colon q_{x} \mapsto \rho_{x}$ where~$x \in Q$ is a companion node and $\rho_{x}$ is its region formula (\autoref{d:rho}).
We will show that
\begin{equation}\label{claim:itp-l1}
\text{for all $x \in Q$: } \rho_{x} \models \subst{\sigma}{\iitp_{x}}.
\end{equation}

We will prove \claimRef{claim:itp-l1} by leaf-to-root induction on the relation $\qedge$, and doing so we will frequently make implicit use of the fact that
\[
\rho_{x} \models \subst{\sigma}{\iitp_{x}} \quad\text{ iff }\quad
\text{ for all } t \in R_{x}: \Lambda_{1}(t) \models  \subst{\sigma}{\iitp_{x}}.
\]
The base cases are covered by the first two cases below, and the induction step is covered by the remaining cases.
\begin{description}
\item[Case $k(x) = 1$, $x$ is a repeat]
In this case, we have $\iitp_{x} = q_{c(x)}$.
Hence $\rho_{x} \models \subst{\sigma}{\iitp_{x}}$ reduces to $\rho_{x} \models \rho_{c(x)}$. Since $\rho_x=\rho_{c(x)}$ (see \autoref{d:regfma}), this is trivially true.

\item[Case $k(x) = 1$,
   $\Delta_{x} \in {\Lambda_{2}[C^{+}] \setminus \Lambda_{2}[C]}$]
In this case, we have that $\subst{\sigma}{\iitp_{x}}  = \subst{\sigma}{\theta_{\Delta_{x}}} = \theta_{\Delta_{x}}$.
Now take an arbitrary node $t \in R_{x}$.
By the case assumption we have $t \in C^{+}\setminus C$, so that \autoref{f:iota}(a) yields $\Lambda_{1}(t) \models \theta_{\Delta_{x}}$.
This suffices to show that $\rho_{x} \models \subst{\sigma}{\iitp_{x}}$, since $t$ was chosen arbitrarily.
\item[Case $k(x) = 1$, $x$ is a companion]
In this case, $x$ has a unique child $y$, and we have $\Delta_{x} = \Delta_{y}$.
This implies that $R_{x} = R_{y}$ so that also $\rho_{x} = \rho_{y}$; we simply write $\rho$ for this formula.

Recall from \autoref{f:iota-nf} that $\iitp_y$ is equivalent to its normal form.
Hence for the interpolants we find, writing $\alpha \isdef \bigcup_i \alpha_i$:
\[\begin{array}{lll}
\iitp_{y} & \equiv &
   [\alpha]q_{x} \land \bigwedge_{j}[\beta_{j}]q_{z_{j}} \land \psi,
\\ \iitp_{x} &=&
   [\alpha^{*}]\big({\bigwedge_{j}}[\beta_{j}]q_{z_{j}} \land \psi \big).
\end{array}\]
It then follows by the induction hypothesis on $y$ that
\[
\rho \models
[\alpha]\rho \land \bigwedge_{j}[\beta_{j}]\sigma(q_{z_{j}})
\land \psi,
\]
and from this we get by \autoref{l:stepToStar}
that
\[
\rho \models
   [\alpha^{*}]\big({\bigwedge_{j}}[\beta_{j}]\sigma(q_{z_{j}})
      \land \psi \big),
\]
so that we are done.

\item[Case $k(x) = 1$, $x$ is neither a leaf nor a companion]
In this case, $x$ has a unique child $y$, and by definition $\iota_x = \iota_y$.
Furthermore, we again have $\Delta_{x} = \Delta_{y}$, implying that also $\rho_{x} = \rho_{y}$.
We now immediately obtain $\rho_{x} \models \iota_x$ from the induction hypothesis, which states $\rho_{y} \models \iota_y$.

\item[Case $k(x) = 2$]
In this case, $x$ has a unique child $y$, and by definition we have $k(y)=3$
and $\Delta_{y} = \Delta_{x}$; we will simply write $\Delta$ for this sequent.
Furthermore recall that $R_{x} = C_{\Delta}^{+}$, but $R_{y} = C_{\Delta}^{R}$,
and that $\iitp_{x} = [\lnot\theta_{\Delta}?]\iitp_{y}$, so that
$\subst{\sigma}{\iitp_{x}} \equiv \theta_{\Delta} \lor
\subst{\sigma}{\iitp_{y}}$.

By an inner induction on the relation $\edge$ of $\tab$ we will show
that all nodes $t$ in $\tab$ satisfy the following:
\begin{equation}
\text{if } t \in R_{x} \text{ then }\Lambda_{1}(t) \models
  \subst{\sigma}{\iitp_{x}}.
\end{equation}
To prove this we make a case distinction.
If $t \in R_{x} \setminus C$ then by \autoref{f:iota}(a), we have
$\Lambda_{1}(t) \models \theta_{\Delta}$ which suffices since
$\theta_{\Delta} \models \subst{\sigma}{\iitp_{x}}$.
If $t \in R_x \cap C$ then we make a further distinction.
In case $t \in C_{\Delta}^{R}$ then by definition of $R_{y}$ we have $t \in
R_{y}$, and hence by the induction hypothesis on $y$,
$\Lambda_{1}(t) \models \subst{\sigma}{\iitp_{y}}$; again this suffices since
$\iitp_{y} \models \iitp_{x}$.
Finally, if $t \in C_{\Delta}^{L}$ then let $u_{1}, \ldots, u_{n}$ be the children of $t$ in $\tab$; note that since we applied a left rule to $t$, we have $\Lambda_{2}(u_{i}) = \Lambda_{2}(t)$, which means that all $u_{i}$ belong to $C_{\Delta}^{+} \subseteq R_{x}$.
Hence by the inner induction hypothesis, we find for all $i$ that
$\Lambda_{1}(u_{i}) \models \subst{\sigma}{\iitp_{x}}$; from this it
is immediate that $\bigvee_{i} \bigwedge \Lambda_{1}(u_{i}) \models \subst{\sigma}{\iitp_{x}}$.
This means that we are done, since we have $\bigwedge \Lambda_{1}(t) \equiv
\bigvee_{i} \bigwedge \Lambda_{1}(u_{i})$ by the local invertibility of the rule applied
at $t$.

\item[Case $k(x) = 3$, $\Delta_{x}$ is not basic]
First of all observe that in this case we have $\Delta\isdef \Delta_{x}\in \Lambda_{2}[C]$.
Let $\Pi_1,\ldots,\Pi_n$ be sequents as in \autoref{l:itpbas} and in \autoref{d:Qtableau}.
Take an arbitrary $t \in C_{\Delta}^{R} = R_x$, then we may present the children of $t$ in $\tab$ as $t_{1},\ldots,t_{n}$, with $\Lambda(t_{i}) = \Lambda_{1}(t);\Pi_{i}$.
Furthermore, let $y_{1},\ldots,y_{n}$ be the corresponding enumeration of the children of $x$ in $\Q$, that is, we have $\Delta_{y_{i}} = \Pi_{i}$.
We now have $t_{i} \in R_{y_{i}}$, for all $i$, so by the induction hypothesis we find that $\Lambda_{1}(t_{i}) \models \subst{\sigma}{\iitp_{y_{i}}}$.
Since for all $i$, $\Lambda_{1}(t) = \Lambda_{1}(t_{i})$, it follows that $\Lambda_{1}(t) \models \bigwedge_{i} \subst{\sigma}{\iitp_{y_{i}}}$, and hence by definition of $\iitp_{x} = \bigwedge_{i}\iitp_{y_{i}}$,
we find that $\rho_{x} \models \subst{\sigma}{\iitp_{x}}$ indeed holds.

\item[Case $k(x) = 3$, $\Delta_{x}$ is basic]
In this case, $x$ has a unique child $y$, which has type 1
and label $\Delta_{y} = {(\Delta_{x})}_{a} \cup \{ \lnot \xi\}$, where
$\lnot\loaded{[a]}\xi$ is the loaded formula in $\Delta_{x}$.
Take an arbitrary~$t \in C_{\Delta_x}^{R} = R_x$.
Since~$\Delta_{x}$ is basic, by \autoref{l:itpbas}(e),
the modal rule is applied at $t$.
Let~$u$ be the unique child of $t$ in $\tab$, then
$\Lambda_{1}(u) = {(\Lambda_{1}(t))}_{a}$ and
$\Lambda_{2}(u) = {(\Lambda_{2}(t))}_{a} \cup \{ \neg\xi \} = \Delta_{y}$,
and by \autoref{l:itpbas}(e), we have $u \in C$.
This means that $u \in R_{y}$, so that we obtain $\Lambda_{1}(u) \models\iitp_{y}$
by the induction hypothesis.
Since $\Lambda_{1}(u) = {(\Lambda_{1}(t))}_{a}$, we have ${(\Lambda_{1}(t))}_{a} \models \iitp_y$
from which we easily derive that $\Lambda_{1}(t) \models [a]\iitp_{y}$.
So we are done since~$\iitp_{x} = [a]\iitp_{y}$.
\end{description}

Finally, to see why the lemma follows from \claimRef{claim:itp-l1}, take $x$ to be the root $r_{\Q}$ of $\Q$.
Then \claimRef{claim:itp-l1} implies that $\rho_{r_{\Q}} \models \subst{\sigma}{\iitp_{r_{\Q}}}$.
Now since $\voc(\iitp_{r_{\Q}}) \subseteq \voc(\Gamma_{1}) \cap \voc(\Gamma_{2})$ by \autoref{l:itp-voc}, we have $\subst{\sigma}{\iitp_{r_{\Q}}} = \iitp_{r_{\Q}}$; and because of our assumption, $\Gamma_{1} \neq \emptyset$, by definition of $\theta_r$ we have $\theta_{r} = \iitp_{r_{\Q}}$.
Thus we obtain $\rho_{r_{\Q}} \models \theta_{r}$.
Furthermore, since $\Lambda_{2}(r) = \Gamma_{2} = \Delta_{r_{\Q}}$ we find
$r \in R_{r_{\Q}}$, so that $\bigwedge\Gamma_{1} = \bigwedge\Lambda_{1}(r)$ is
one of the disjuncts of $\rho_{r_{\Q}}$.
But then $\rho_{r_{\Q}} \models \theta_{r}$ implies $\Gamma_{1} \models
\theta_{r}$, as required.
\end{proof}

To show that $\Gamma_2 \models \lnot \theta_r$, we will need the following notion of distance, relative to a given program.
Intuitively, the idea is that if we have $(v,w) \in R_\alpha$ in some Kripke model $\kmodel$, the $\alpha$-distance from $v$ to $w$ is the minimal number of atomic steps it takes to move from $v$ to $w$ along an $\alpha$-path.
Furthermore, the $\alpha$-distance from $v$ to $w$ is set to infinity if $(v,w) \not\in R_\alpha$.

\begin{defi}[\lean{distance}]\label{d:distance}
Let $v$ and $w$ be states in some model $\kmodel = (W, \{\reach{a} \mid a \in \progs_0\}, V)$, and let $\alpha$ be some program.
By induction on the shallow structure of $\alpha$, we define the \emph{$\alpha$-distance} $d_\alpha(v,w)$ from $v$ to $w$ as follows:
\[\begin{array}{lll}
d_{a}(v,w) &\isdef& \left\{
   \begin{array}{lll}
   1         & \text{if } (v,w) \in R_{a}
   \\ \infty & \text{otherwise}
   \end{array}\right.
\\ d_{\tau?}(v,w) &\isdef& \left\{
   \begin{array}{lll}
   0         & \text{if } \kmodel, v \Vdash \tau \text{ and } v=w
   \\ \infty & \text{otherwise}
   \end{array}\right.
\\ d_{\alpha\cup\beta}(v,w) &\isdef& \min\big\{ d_{\alpha}(v,w), d_{\beta}(v,w) \big\}
\\ d_{\alpha;\beta}(v,w) &\isdef&
   \min\big\{d_{\alpha}(v,u)+ d_{\beta}(u,w) \mid u \in W \big\}
\\ d_{\alpha^{\ast}}(v,w) &\isdef&
   \min\big\{\sum_{i=0}^{n-1} d_{\alpha}(v_{i},v_{i+1}) \mid
   v_{0},\ldots,v_{n} \in W, v_{0} = v, v_{n} = w \big\}.
\end{array}\]

We extend this definition to lists of programs in the obvious way:
\[\begin{array}{lll}
   d_{\emptylist}(v,w) &\isdef& \left\{
     \begin{array}{lll}
     0         & \text{if } v=w
     \\ \infty & \text{otherwise}
     \end{array}\right.
\\ d_{\alpha\ol{\gamma}}(v, w) &\isdef&
   \min\big\{d_{\alpha}(v,u)+ d_{\ol{\gamma}}(u,w) \mid u \in W \big\}.
\end{array}\]
\end{defi}

Note that if $d_{\alpha^{\ast}}(v,w) < \infty$ and $v \neq w$ then
there is a sequence $v=v_0, \ldots, v_n=w$ such that $d_{\alpha^{\ast}}(v,w) = \sum_{i=0}^{n-1} d_{\alpha}(v_{i},v_{i+1})$.
Moreover, we can assume $v \neq v_1$ because if $v=v_1$ then also $v_1, \ldots, v_n$ would already witness $d_{\alpha^{\ast}}(v,w)$. More generally, we can remove those initial elements of the sequence that are equal to $v$ until we obtain a non-empty sequence whose second element is distinct from~$v$ (the assumption $v\neq w$ implies that the sequence cannot be empty).

\begin{lem}[\checklean{distanceProps}]\label{l:distanceProps}
Let $v$ and $w$ be states in some model $\kmodel=(W, \{\reach{a} \mid a \in \progs_0\}, V)$.
Then the following hold, for any program $\alpha$ and lists of programs $\ol{\gamma}, \ol{\delta}$.
\begin{enumerate}[(a)]
\item $d_{\alpha}(v,w) < \infty$ iff $(v,w) \in R_{\alpha}$;
\item $d_{\alpha_1\dots\alpha_n}(v,w) = d_{\alpha_1; \dots; \alpha_n}(v, w)$;
\item if, for all $u$, $d_{\alpha}(v, u) \leq d_{\ol{\delta}}(v, u)$ then, for all $u$,
$d_{\alpha\ol{\gamma}}(v, u) \leq d_{\ol{\delta}\ol{\gamma}}(v, u)$;
\item if $(X,\ol{\delta}) \in  \Dset{\alpha}$ and $\kmodel, v \Vdash \bigwedge X$ then, for all $u$, $d_{\alpha}(v,u) \leq d_{\ol{\delta}}(v,u)$;
\item if $(X,\ol{\delta}) \in  \Dset{\alpha}$ and $\kmodel, v \Vdash \bigwedge X$ then, for all $u$,  $d_{\alpha\ol{\gamma}}(v,u) \leq d_{\ol{\delta}\ol{\gamma}}(v,u)$;
\item if $(v,w) \in R_{\alpha}$ then there is a pair $(X,\ol{\delta}) \in  \Dset{\alpha}$ such that $\kmodel, v \Vdash \bigwedge X$ and $d_{\ol{\delta}}(v,w) = d_{\alpha}(v,w)$;
\item if $(v,w) \in R_{\alpha\ol{\gamma}}$ then there is a pair $(X,\ol{\delta}) \in  \Dset{\alpha}$ such that $\kmodel, v \Vdash \bigwedge X$ and $d_{\ol{\delta}\ol{\gamma}}(v,w) = d_{\alpha\ol{\gamma}}(v,w)$;
\item if $\kmodel, v \Vdash \lnot\loaded{\Box}(\alpha\ol{\gamma}, \psi)$ where $\psi$ is an unloaded PDL formula, then there exists a pair $(X, \ol{\delta}) \in \Dset{\alpha}$ such that $\kmodel, v \Vdash \bigwedge X \land \lnot \loaded{\Box}(\ol{\delta}\ol{\gamma}, \psi)$ and
\[
\min \big\{ d_{\ol{\delta}\ol{\gamma}}(v,w) \mid \kmodel, w \Vdash \lnot\psi \big\} =
\min \big\{ d_{\alpha\ol{\gamma}}(v,w) \mid \kmodel, w \Vdash \lnot\psi \big\}.
\]
\end{enumerate}
\end{lem}

\begin{proof}
Parts~(a) and~(b) are straightforward.
\begin{enumerate}
    \item[(c)]
    Let $\beta = \delta_0; \dots; \delta_n$ where $\ol{\delta} = \delta_0 \dots \delta_n$.
    Making use of part~(b), it suffices to show for any $u$ that
    $d_{\alpha\ol{\gamma}}(v, u) \leq d_{\beta\ol{\gamma}}(v, u)$.

    Let $u'$ witness minimality in $d_{\beta\ol{\gamma}}(v, u)=\min\big\{d_{\beta}(v,w)+ d_{\ol{\gamma}}(w,u) \mid w \in W \big\}$ in the sense that $d_{\beta\ol{\gamma}}(v, u) = d_{\beta}(v,u')+ d_{\ol{\gamma}}(u',u)$.
    By definition of $d_{\alpha\ol{\gamma}}(v, u)$:
    $d_{\alpha\ol{\gamma}}(v, u) \leq d_{\alpha}(v,u')+ d_{\ol{\gamma}}(u',u)$. Furthermore, from the hypothesis and part~(b), we have
    $d_{\alpha}(v, u') \leq d_{\beta}(v, u')$. So:
    \begin{align*}
    d_{\alpha\ol{\gamma}}(v, u) &\leq d_{\alpha}(v,u') + d_{\ol{\gamma}}(u',u) \\
                           &\leq d_{\beta}(v, u') + d_{\ol{\gamma}}(u', u) \\
                           &= d_{\beta\ol{\gamma}}(v, u).
    \end{align*}

    \item[(d)]
    We proceed by induction on the shallow program structure of $\alpha$.
    Suppose
    $(X,\ol{\delta}) \in  \Dset{\alpha}$ and $\kmodel, v \Vdash \bigwedge X$, and let $u$ be an arbitrary state.

    \begin{description}
    \item[Case $\alpha = a$]
        Then $\Dset{a} = \{ (\emptyset, a) \}$, and thus $\ol{\delta} = a$, making $d_{\alpha}(v, u) = d_{\ol{\delta}}(v, u)$ hold trivially.

    \item[Case $\alpha = {\tau?}$]
        Then we have that $\Dset{\tau?} = \{ ( \{ \tau \}, \emptylist )\}$.
        If $v = u$, then since $\kmodel, v \Vdash \tau$ we have $d_{\tau?}(v, u) = 0 = d_{\emptylist}(v, u)$.
        If $v \neq u$, then we get that
        $d_{\tau?}(v, u) = \infty = d_{\emptylist}(v, u)$.

    \item[Case $\alpha = \beta \cup \beta'$]
        Then $\Dset{\beta\cup\beta'} = \Dset{\beta} \cup \Dset{\beta'}$.
        Assume without loss of generality that $(X, \ol{\delta}) \in \Dset{\beta}$.
        Then, by the inductive hypothesis, we have that $d_{\beta}(v, u) \leq d_{\ol{\delta}}(v, u)$.
        It then follows that $d_{\beta \cup \beta'}(v, u) = \min(d_{\beta}(v, u), d_{\beta'}(v,u)) \leq d_{\beta}(v, u) \leq d_{\ol{\delta}}(v, u)$.

    \item[Case $\alpha = \beta; \beta'$]
        From the definition we have
        \begin{align*}
        \Dset{\beta;\beta'} =~&
           \{(X,\ol{\delta}\beta' ) \mid (X,\ol{\delta}) \in \Dset{\beta}, \ol{\delta} \neq
              \emptylist \}
        ~\cup \\
           &\{(X \cup Y, \ol{\delta} ) \mid (X,\emptylist) \in \Dset{\beta},
              (Y, \ol{\delta}) \in \Dset{\beta'} \}.
        \end{align*}
        Consider first a pair $(X, \ol{\delta}\beta') \in \Dset{\beta;\beta'}$ with $(X, \ol{\delta}) \in \Dset{\beta}$ and $\ol{\delta} \neq \emptylist$.
        By inductive hypothesis, we have that $d_{\beta}(v, u') \leq d_{\ol{\delta}}(v, u')$ for all $u'$.
        By part~(c) and instantiating $u'=u$, we obtain that $d_{\beta\beta'}(v, u) \leq d_{\ol{\delta}\beta'}(v, u)$.

        Now consider a pair $(X \cup Y, \ol{\delta}) \in \Dset{\beta;\beta'}$ with $(X,\emptylist) \in \Dset{\beta}$ and $(Y, \ol{\delta}) \in \Dset{\beta'}$.
        By inductive hypothesis, for all $u'$, $d_{\beta}(v, u') \leq d_{\emptylist}(v, u')$, and then by part~(c), for all $u'$, $d_{\beta\beta'}(v, u') \leq d_{\beta'}(v, u')$.
        Applying the inductive hypothesis with $(Y, \ol{\delta}) \in \Dset{\beta'}$ we get that $d_{\beta'}(v, u') \leq d_{\ol{\delta}}(v, u')$ for all $u'$.
        Combining these two observations, and instantiating $u'=u$, we obtain $d_{\beta\beta'}(v, u) \leq d_{\ol{\delta}}(v, u)$.

    \item[Case $\alpha = \beta^\ast$]
        From the definition we have
        \[
        \Dset{\beta^{\ast}} = \{ (\emptyset,\emptylist) \} \cup
             \{ (X, \ol{\delta}\beta^{\ast})
        	   \mid (X,\ol{\delta}) \in \Dset{\beta}, \ol{\delta} \neq
        	   \emptylist \}.
        \]
        Consider first the pair $(\emptyset, \emptylist) \in \Dset{\beta^\ast}$.
        For this pair, we need to show that $d_{\beta^\ast}(v, u) \leq  d_{\emptylist}(v, u)$. If $v = u$, then $d_{\beta^\ast}(v, u) = 0 = d_{\emptylist}(v, u)$. If $v \neq u$, then $d_{\emptylist}(v, u) = \infty$, and thus $d_{\beta^\ast}(v, u) \leq d_{\emptylist}(v, u)$ holds trivially.

        Now consider a pair $(X, \ol{\delta}\beta^\ast) \in \Dset{\beta^\ast}$ with $(X, \ol{\delta}) \in \Dset{\beta}$ and $\ol{\delta} \neq \emptylist$. We need to show $d_{\beta^\ast}(v, u) \leq d_{\ol{\delta}\beta^\ast}(v, u)$.
        Let $v = v_0, \dots, v_n = u $ be a sequence such that
        $d_{\beta^\ast}(v, u) = \sum_{i = 0}^{n-1} d_{\beta}(v_i, v_{i+1})$,
        and let $v = u_0, \dots, u_k = u $ be a sequence such that
        $d_{\ol{\delta}\beta^\ast}(v, u) = d_{\ol{\delta}}(u_0, u_1) + \sum_{j = 1}^{k-1} d_{\beta}(u_j, u_{j+1})$, i.e., these sequences witness the minimality in the definitions of $d_{\beta^\ast}(v, u)$ and $d_{\ol{\delta}\beta^\ast}(v, u)$.
        We now have,\\
        $
        \begin{array}{rcll}
        d_{\beta^\ast}(v, u)
           & = &\sum_{i = 0}^{n-1} d_{\beta}(v_i, v_{i+1}) & \text{(by assumption on $v_0, \ldots, v_n$)}\\[.1em]
           &\leq &\sum_{j = 0}^{k-1} d_{\beta}(u_j, u_{j+1}) & \text{(by minimality of $v_0, \ldots, v_n$)}\\[.1em]
           &\leq & d_{\ol{\delta}}(v, u_1) + \sum_{j = 1}^{k-1} d_{\beta}(u_j, u_{j+1}) & \text{(by induction hypothesis)}\\[.1em]
           & = &d_{\ol{\delta}\beta^\ast}(v, u) & \text{(by assumption on $u_0, \ldots, u_k$).}
        \end{array}
        $
    \end{description}

    \item[(e)] This follows from part~(c) and part~(d).

    \item[(f)]
    By induction on the shallow program structure of $\alpha$.
    We only cover two cases of the induction step.

    \begin{description}
    \item[Case $\alpha = \gamma; \beta$]
        Since $(v, w) \in R_{\gamma; \beta}$, by part~(a) it follows that
        $d_{\gamma; \beta}(v, w) < \infty$.
        By definition of $d_{\gamma; \beta}(v, w)$,
        there exists a $u$ such that
        $d_{\gamma; \beta}(v, w) = d_{\gamma}(v, u) + d_{\beta}(u, w)$.

        By the induction hypothesis on $\gamma$ we find a pair $(X,\ol{\delta}) \in \Dset{\gamma}$
        such that $\kmodel, v \Vdash \bigwedge X$ and $d_{\gamma}(v,u) = d_{\ol{\delta}}(v, u)$.

        Consider the definition of $\Dset{\gamma; \beta}$.
        If $\ol{\delta} \neq \emptylist$ then
        $(X,\ol{\delta}\beta) \in \Dset{\gamma; \beta}$.
        From part~(e), we obtain
        $d_{\gamma; \beta}(v, w) \leq d_{\ol{\delta}\beta}(v, w)$.
        Due to the $\min$ in the definition, $d_{\ol{\delta}\beta}(v, w)\leq
        d_{\ol{\delta}}(v,u) + d_{\beta}(u, w)$.
        Thus
        \[
        d_{\gamma; \beta}(v, w)
        =d_{\gamma}(v, u) + d_{\beta}(u, w)
        =d_{\ol{\delta}}(v,u) + d_{\beta}(u, w) \geq d_{\ol{\delta}\beta}(v, w).
        \]
        We conclude that $d_{\gamma; \beta}(v, w) = d_{\ol{\delta}\beta}(v, w)$, as required.

        On the other hand, if $\ol{\delta} = \emptylist$, from $d_{\ol{\delta}}(v, u)=d_{\gamma}(v,u)\leq d_{\gamma; \beta}(v, w)<\infty$, and the definition of $d_{\emptylist}(v, u)$, we have $v = u$.
        Furthermore, since $d_{\beta}(u, w) \leq d_{\gamma; \beta}(v, w) < \infty$, by part~(a) we have $(u,w) \in R_\beta$.
        Hence $(v,w)=(u,w) \in R_{\beta}$.
        By the induction hypothesis it follows  that there
        is a pair $(Y,\ol{\eta}) \in \Dset{\beta}$ such that $\kmodel, v \Vdash \bigwedge Y$ and
        $d_{\ol{\eta}}(v,w) = d_{\beta}(v, w)$.
        Combining all of this, and utilizing the definition of $\Dset{\gamma; \beta}$, we have $(X\cup Y, \ol{\eta}) \in \Dset{\gamma; \beta}$, and
        $\kmodel, v \Vdash \bigwedge(X \cup Y)$, and $d_{\ol{\eta}}(v,w) = d_{\beta}(v, w) = d_{\emptylist}(v, v) + d_{\beta}(v, w) = d_{\gamma; \beta}(v, w)$, as required.

    \item[Case $\alpha = \beta^{\ast}$]
        Let $(v,w) \in R_{\beta^{\ast}}$. There are two possibilities.
        If $v = w$ then the claim is witnessed by $(\emptyset,\emptylist)\in
        \Dset{\beta^{\ast}}$:
        it follows from definition that $\kmodel, v \Vdash \bigwedge\emptyset$ and
        $d_{\emptylist}(v, v) = 0 = d_{\beta^{\ast}}(v,v)$ (in the definition of $d_{\beta^{\ast}}(v,v)$, we can take $n=0$ which yields the empty sum, which equals 0).

        If $v \neq w$ then, since $(v, w) \in R_{\beta^{\ast}}$, by part~(a) we get that $d_{{\beta^{\ast}}}(v, w) < \infty$.
        Therefore, by definition of $d_{{\beta^{\ast}}}(v, w)$, there exists some $v_{1} \neq v$ (see observation above this lemma)
        such that
        $d_{\beta^{\ast}}(v, w) = d_{\beta}(v,v_{1}) + d_{\beta^{\ast}}(v_{1},w) $.
        By the induction hypothesis there is a pair $(X,\ol{\delta}) \in \Dset{\beta}$
        such that $\kmodel, v \Vdash \bigwedge X$ and $d_{\ol{\delta}}(v,v_{1}) =  d_{\beta}(v, v_1)$.
        Note that we must have $\ol{\delta} \neq\emptylist$ since $v \neq v_{1}$.
        By definition of $\Dset{\beta^{\ast}}$, we then have that $(X,\ol{\delta}\beta^{\ast})\in \Dset{\beta^{\ast}}$.
        It remains to establish $d_{\beta^{\ast}}(v, w) = d_{\ol{\delta}\beta^\ast}(v, w)$. To see this, first notice that
        from  $(X,\ol{\delta}\beta^{\ast})\in \Dset{\beta^{\ast}}$ and $\kmodel, v \Vdash \bigwedge X$, by part~(e) we obtain
        $d_{\beta^\ast}(v, w) \leq d_{\ol{\delta}\beta^\ast}(v, w)$.
        On the other hand,
        $d_{\beta^\ast}(v, w)
        = d_{\beta}(v,v_{1}) + d_{\beta^{\ast}}(v_{1},w)
        = d_{\ol{\delta}}(v, v_1) + d_{\beta^\ast}(v_1, w)
        \geq d_{\ol{\delta}\beta^\ast}(v, w)$.
        We conclude that $d_{\beta^\ast}(v, w) = d_{\ol{\delta}\beta^\ast}(v, w)$, so the claim is proved.
    \end{description}

    \item[(g)]
    Assume $(v, w) \in R_{\alpha \ol\gamma}$.
    First note that for all lists $\ol\gamma = \gamma_1 \ldots \gamma_n$ we have $R_{\ol\gamma} = R_{\gamma_1 \seqc \ldots \seqc \gamma_n}$ and $d_{\ol\gamma} = d_{\gamma_1 \seqc \ldots \seqc \gamma_n}$.
    Since $\alpha \seqc \gamma_1 \seqc \ldots \seqc \gamma_n$ is a single program, we can apply part~(a) to obtain
    $d_{\alpha \ol\gamma}(v, w) < \infty$.
    Therefore, there exists a $u$ such that $(v, u) \in R_{\alpha}$, $(u, w) \in R_{\ol{\gamma}}$,
    and $d_{\alpha\ol{\gamma}}(v,w) = d_{\alpha}(v,u)+ d_{\ol{\gamma}}(u,w)$, i.e., $u$ witnesses the minimum in the definition of $d_{\alpha \ol\gamma}(v, w)$.

    From $(v, u) \in R_{\alpha}$ and part~(f), we obtain a pair $(X,\ol{\delta}) \in  \Dset{\alpha}$ such that $\kmodel, v \Vdash \bigwedge X$ and $d_{\ol{\delta}}(v,u) = d_{\alpha}(v,u)$.

    Note that $d_{\alpha\ol{\gamma}}(v, w) = d_{\alpha}(v,u)+ d_{\ol{\gamma}}(u,w) = d_{\ol{\delta}}(v,u) + d_{\ol{\gamma}}(u, w) \geq d_{\ol{\delta}\ol{\gamma}}(v, w)$.
    Since $(X, \ol{\delta}) \in \Dset{\alpha}$, by part~(e) we obtain $d_{\alpha\ol{\gamma}}(v, w) \leq d_{\ol{\delta}\ol{\gamma}}(v, w)$.
    Combining both inequalities, we conclude $d_{\alpha\ol{\gamma}}(v, w) = d_{\ol{\delta}\ol{\gamma}}(v, w)$, as required.

    \item[(h)]
    Let $w_0 \in W$ be such that $\kmodel, w_0 \Vdash \lnot\psi$ and
    \[d_{\alpha\ol{\gamma}}(v, w_0) = \min \big\{ d_{\alpha\ol{\gamma}}(v,w) \mid \kmodel, w \Vdash \lnot\psi \big\}\]
    Since $(v, w_0) \in R_{\alpha\ol{\gamma}}$, applying part~(g), we obtain a pair $(X,\ol{\delta}) \in  \Dset{\alpha}$ such that $\kmodel, v \Vdash \bigwedge X$ and $d_{\ol{\delta}\ol{\gamma}}(v,w_0) = d_{\alpha\ol{\gamma}}(v,w_0)$.
    Now
    \[\min \big\{ d_{\ol{\delta}\ol{\gamma}}(v,w) \mid \kmodel, w \Vdash \lnot\psi \big\} \leq d_{\ol{\delta}\ol{\gamma}}(v,w_0) = d_{\alpha\ol{\gamma}}(v,w_0) \]

    Since $(X, \ol{\delta}) \in \Dset{\alpha}$, by part~(e) we get that $d_{\alpha \ol{\gamma}}(v, w) \leq d_{\ol{\delta} \ol{\gamma}}(v, w)$ for all $w \in W$.
    From this it follows that:
    \[\min \big\{ d_{\alpha\ol{\gamma}}(v,w) \mid \kmodel, w \Vdash \lnot\psi \big\} \leq
    \min \big\{ d_{\ol{\delta}\ol{\gamma}}(v,w) \mid \kmodel, w \Vdash \lnot\psi \big\}\]

    Combining the above, we obtain the desired result, since
    \begin{align*}
    \min \big\{ d_{\alpha\ol{\gamma}}(v,w) \mid \kmodel, w \Vdash \lnot\psi \big\}
    &\leq \min \big\{ d_{\ol{\delta}\ol{\gamma}}(v,w) \mid \kmodel, w \Vdash \lnot\psi \big\} \\
    &\leq d_{\alpha\ol{\gamma}}(v, w_0)\\
    &= \min \big\{ d_{\alpha\ol{\gamma}}(v,w) \mid \kmodel, w \Vdash \lnot\psi \big\}
    \qedhere
    \end{align*}
\end{enumerate}
\end{proof}

In order to prove that $\Gamma_{2} \models \lnot\theta_{r}$, the key observation is \autoref{l:satDownQ} below.
In the proof we will need the following auxiliary result.

\begin{lem}[\checklean{LoadedCluster.repeat_basicBetween}]\label{l:clbas}
If $x$ is a repeat in $\Q$ then the path from $c(x)$ to $x$ passes through a node $y$ such that $\Delta_{y}$ is basic.
\end{lem}

\begin{proof}
Write $c(x) = w_{0} \qedge w_{1} \qedge \cdots \qedge w_{m} = x$ for the path from $c(x)$ to $x$, and suppose towards a contradiction that $\Delta_{w_{i}}$ is not basic for all $i \leq m$.
Both $c(x)$ and $x$ have type~1 by \autoref{l:Q.Properties}(a), so by \autoref{d:Qtableau} the path passes through a node of type~3.
By that same definition the label changes only from a node of type~3 to its children, and at a node of type~3 whose label is not basic the children are labelled with the children of an instance of a local rule applied to that label.
The sequents $\Delta_{w_{0}}, \ldots, \Delta_{w_{m}}$ therefore trace a path in a tableau over the local rules, along which at least one rule is applied.
Since $x$ is a repeat we have $\Delta_{w_{m}} = \Delta_{x} = \Delta_{c(x)} = \Delta_{w_{0}}$, so that this path contains a repeat.
This contradicts the observation, recorded after \autoref{d:repeat}, that local tableaux contain no repeats.
\end{proof}

For the formulation of \autoref{l:satDownQ}, recall that $\cycs{x}$ denotes the set of repeat leaves that are descendants of $x$ and whose companion node is a proper ancestor of $x$.

\begin{lem}[\checklean{LoadedCluster.satDown_root}]\label{l:satDownQ}
Let $x$ be a node in $\Q$.
If the sequent $\Delta_{x},\iitp_{x}$ is satisfiable,
then so is the sequent $\Delta_{z}, \iitp_{z}$, for some leaf $z \in \cycs{x}$.
\end{lem}

\begin{proof}
For the proof of this lemma we will look in more detail at \emph{how} $\Delta_{x},\iitp_{x}$ is satisfied in some model $\kmodel$ by considering the ``witness distance'' of the loaded formula at a given node in  $\kmodel$ that satisfies the sequent $\Delta_{x},\iitp_{x}$.
We need some further notation here: for $x \in Q$, let~$\ol{\delta_x}$ and $\psi_{x}$ be such that $\lnot \loaded{\Box}(\ol{\delta_{x}},\psi_{x})$ is the (unique) loaded formula in $\Delta_{x}$, while $\psi_{x}$ itself is unloaded.
Observe that $\ol{\delta_x}$ and $\psi_{x}$ are uniquely determined by $x$.
For a state~$v$ in $\kmodel$ and node~$x$ in $\Q$, define
\[
\witdis_{\kmodel}(v,x) \isdef \left\{ \begin{array}{ll}
   \infty & \text{if } \kmodel, v \not\Vdash \Delta_{x},\iitp_{x}
\\ \min \big\{ d_{\ol{\delta_{x}}}(v,w)
   \mid \kmodel, w \Vdash \lnot\psi_x \big\}
   & \text{otherwise}.
\end{array}\right.
\]
In words: if $\kmodel, v \Vdash \Delta_{x},\iitp_{x}$ then $\witdis_{\kmodel}(v,x)$ is the least $\ol{\delta_x}$-distance from $v$ to a state that witnesses $\lnot\psi_x$.

The main technical result is formulated in the following Claim.

\emph{Claim.}
Let $x$ be a node in $\Q$.
Then for any pointed Kripke model $(\kmodel, v)$ such that $\kmodel, v \Vdash \Delta_{x}, \iitp_{x}$ there is a leaf $z \in \cycs{x}$ and a state $u$ in $\kmodel$, such that
\\
(1) $\kmodel, u \Vdash \Delta_{z}, \iitp_{z}$ and
\\
(2) $\witdis_{\kmodel}(u,z) \leq \witdis_{\kmodel}(v,x)$, with
\\
(3) $\witdis_{\kmodel}(u,z) < \witdis_{\kmodel}(v,x)$ if there is a basic node between $x$ and $z$.
\medskip

The proof of the claim is by primary induction on the relation $\qedge$ from leaf to root, and secondary induction on $\witdis_{\kmodel}(v,x)$.
Assume that $\kmodel = (W,R,V)$, $v$ and $x$ are as in the formulation of the claim.
We make a case distinction as to the type of $x$.

\begin{description}
\item[Case $k(x) = 1$, $x$ is a repeat]
In this case, we have $x \in \cycs{x}$ by \autoref{l:Q.Properties}(b), and the claim holds by simply taking $z \isdef x$ and $u\isdef v$.

\item[Case $k(x) = 1$, $x$ is a leaf because $\Delta_{x} \in {\Lambda_{2}[C^+]}\setminus {\Lambda_{2}[C]}$]
Since $\iitp_x = \theta_{\Delta_x}$, it follows from \autoref{f:iota}(b) that the sequent $\Delta_{x}, \iitp_{x}$ is not satisfiable, so that the claim is vacuously true.

\item[Case $k(x) = 1$, $x$ is neither a leaf nor a companion]
By the definition of $\Q$, $x$ has a unique child $y$, and we have $\iitp_{x} = \iitp_{y}$ and $\Delta_x = \Delta_y$. Since $\kmodel, v \Vdash \Delta_x, \iota_x$, it follows that $\kmodel, v \Vdash \Delta_y, \iota_y$. Thus, we can apply the primary inductive hypothesis at $y$, which gives the existence of a leaf $z \in \cycs{y}$ and a state $v \in W$ such that
\[ \kmodel, v \Vdash \Delta_{z}, \iitp_{z}
\text { and  }
\witdis_{\kmodel}(v, z) \leq \witdis_{\kmodel}(v, y), \]
where the inequality is strict if there is a basic node between~$y$ and~$z$.

First, notice that since $\Delta_x = \Delta_y$, it follows that $\witdis_{\kmodel}(v, x) = \witdis_{\kmodel}(v, y)$, and that there is a basic node between $y$ and $z$ if and only if there is a basic node between~$x$ and~$z$.
Second, since $x$ is neither a repeat leaf nor a companion, \autoref{l:Q.Properties}(d) gives $\cycs{x} = \cycs{y}$.
From these observations, we conclude that there exists a leaf $z \in \cycs{x}$ and a state $v \in W$ such that
\[
\text{
$\kmodel, v \Vdash \Delta_{z}, \iitp_{z}$ and
$\witdis_{\kmodel}(v, z) \leq \witdis_{\kmodel}(v, x)$,
}
\]
where the inequality is strict if there is a basic node between $x$ and $z$.
This exactly matches the claim.

\item[Case $k(x) = 1$, $x$ is a companion]
Observe that by definition of $\Q$, $x$ has a unique child $y$ with $k(y) = 2$ and $\Delta_{x} = \Delta_{y}$ (which we simply write as $\Delta$).
Furthermore, as noted in \autoref{r:dfitp},
by construction of the quasi-tableau,
$\iitp_{y} \equiv {\bigwedge_{i}}[\alpha_{i}]q_{x} \land {\bigwedge_{j}}[\beta_{j}]q_{z_{j}} \land \psi$ and
$\iitp_{x}\equiv {\bigwedge_{i}}[\alpha_{i}]\iota_{x} \land {\bigwedge_{j}}[\beta_{j}]q_{z_{j}} \land \psi$, hence $\iitp_{x}\equiv \subst{\substAt{q_{x}}{\iitp_{x}}}{\iitp_{y}}$.

Let $\kmodel'$ be the Kripke model $\kmodel' \isdef (W,R,V')$, where $V'$ is the valuation  on $W$ that modifies $V$ by mapping $q_{x}$ to the set
$\{ u \in W \mid \kmodel, u \Vdash \iitp_{x} \}$.
The difference between $V$ and $V'$ may affect the truth of a
pre-interpolant that contains $q_{x}$, but it will not affect the truth of any formula in $\Delta_z$ for all $z \in Q$,
since such formulas are $q_{x}$-free.
Also note that $\kmodel, v\Vdash \Delta, \iitp_{x}$ iff
$\kmodel, v \Vdash \Delta, \subst{\substAt{q_{x}}{\iitp_{x}}}{\iitp_{y}}$ iff
$\kmodel', v \Vdash \Delta, \iitp_{y}$. The first `iff' is
due to the definitions of $\iitp_{x}$ and $\iitp_{y}$, and the second `iff' is due to \autoref{l:subst} about substitution.

Since $x$ is closer to the root of $\Q$ than $y$, from the primary induction hypothesis at $y$, we obtain a leaf $z \in \cycs{y}$ and a state $v' \in W$ such that
\begin{equation}\label{obs:PIH}
\kmodel', v' \Vdash \Delta_{z}, \iitp_{z} \text { and }
\witdis_{\kmodel'}(v',z) \leq \witdis_{\kmodel'}(v,y),
\end{equation}
where the inequality is strict if there is a basic node between $y$ and $z$.
Applying the Substitution Lemma to $\kmodel', v' \Vdash \Delta_{z}, \iitp_{z}$ gives us $\kmodel, v' \Vdash \Delta_{z}, \subst{\substAt{q_{x}}{\iitp_{x}}}{\iitp_{z}}$.
Now make a further case distinction.

\begin{description}
\item[\it Subcase $z \in \cycs{x}$]
By this assumption, $x$ is not the companion of $z$. Hence $\iitp_{z} = q_{c(z)} \neq q_x$,
so $\subst{\substAt{q_{x}}{\iitp_{x}}}{\iitp_{z}} = \iitp_{z}$.
Then $\kmodel, v' \Vdash \Delta_{z}, \subst{\substAt{q_{x}}{\iitp_{x}}}{\iitp_{z}}$ simplifies to
$\kmodel, v' \Vdash \Delta_{z},\iitp_{z}$.
From the latter and from $\kmodel', v'\Vdash\Delta_{z},\iitp_{z}$  obtained in \observationRef{obs:PIH},
we have that $\witdis_{\kmodel}(v',z) = \witdis_{\kmodel'}(v',z)$ as
$\kmodel$ and $\kmodel'$ have the same relational structure, and since in both cases $\witdis$ returns the witness distance of the same loaded formula.
For the same reason --- and since by assumption $\kmodel, v\Vdash \Delta, \iitp_{x}$ and hence
$\kmodel', v \Vdash \Delta, \iitp_{y}$ --- we have $\witdis_{\kmodel'}(v,y) = \witdis_{\kmodel}(v,x)$.
Chaining these observations with $\witdis_{\kmodel'}(v',z)\leq \witdis_{\kmodel'}(v,y)$ from \observationRef{obs:PIH}, we obtain $\witdis_{\kmodel}(v',z) \leq \witdis_{\kmodel}(v,x)$
which constitutes part (2) of the Claim.
Part (3) then follows by the observation that there is a basic node between $x$
and $z$ iff there is a basic node between $y$ and $z$.

\item[\it Subcase $z \not\in \cycs{x}$]
In this case, we have $z \in \cycs{y} \setminus \cycs{x}$ which by \autoref{l:Q.Properties}(c) means that the companion of~$z$ must be~$x$.
It follows that $\Delta_{z} = \Delta$ and $\iitp_{z} = q_{c(z)} = q_{x}$ and thus we find $\subst{\substAt{q_{x}}{\iitp_{x}}}{\iitp_{z}} = \iitp_{x}$.
Since $\kmodel', v'\Vdash\Delta_{z},\iitp_{z}$ from \observationRef{obs:PIH}, by the Substitution Lemma we also have $\kmodel, v'\Vdash\Delta_{z},\subst{\substAt{q_{x}}{\iitp_{x}}}{\iitp_{z}}$, and the latter is simply $\kmodel, v'\Vdash\Delta,\iitp_{x}$.
Thus $\witdis_{\kmodel}(v',x) = \witdis_{\kmodel'}(v',z)$, since $\kmodel$ and $\kmodel'$ have the same relational structure, and $\witdis$ returns the witness distance of the same loaded formula.
Also, by \autoref{l:clbas} there is a basic node between $x$ and $z$; this
node then also lies between $y$ and $z$, so that in fact $\witdis_{\kmodel'}(v',z)$
is strictly smaller than $\witdis_{\kmodel'}(v,y)$.

Combining the above we get that $\witdis_{\kmodel}(v',x)
= \witdis_{\kmodel'}(v',z) < \witdis_{\kmodel'}(v,y) = \witdis_{\kmodel}(v,x)$ where the last equality relies on $\kmodel',v'\Vdash\Delta,\iitp_{y}$ (already shown) and the fact that $\mathcal{M}$ and $\mathcal{M}'$ have the same relational structure.
(In passing we note that the strict inequality implies that the subcase $z \not\in \cycs{x}$ cannot occur when $\witdis_{\kmodel}(v,x) = 0$.)
Since $\witdis_{\kmodel}(v',x)<\witdis_{\kmodel}(v,x)$
and $\kmodel, v'\Vdash\Delta,\iitp_{x}$, we can apply the secondary induction hypothesis to obtain a node $z' \in \cycs{x}$ and a state $u \in W$ such that $\kmodel, u \Vdash \Delta_{z'},\iitp_{z'}$ and
$\witdis_{\kmodel}(u,z') \leq \witdis_{\kmodel}(v',x) < \witdis_{\kmodel}(v,x)$, which means that we are done.
\end{description}

\item[Case $k(x) = 2$]
Let $y$ be the unique child of $x$, then by definition of $\Q$, we have
$\Delta_{x} = \Delta_{y}$ (which we write simply as $\Delta$)
and $\iitp_{x} =
[\lnot\theta_{\Delta}?]\iitp_{y} \equiv \theta_{\Delta} \lor \iitp_{y}$.
By \autoref{f:iota} the sequent $\Delta,\theta_{\Delta}$ is not satisfiable, and so it must be the case that
$\kmodel, v \Vdash \Delta,\iitp_{y}$.
We may now apply the primary induction hypothesis to $y$ and obtain a leaf $z \in \cycs{y}$ and a state $u \in W$ such that $\kmodel, u \Vdash \Delta_{z}, \iitp_{z}$ and $\witdis_{\kmodel}(u,z) \leq \witdis_{\kmodel}(v,y)$, with the inequality being strict if there is a basic node between $y$ and $z$.
Since $k(x)\neq 1$, \autoref{l:Q.Properties}(a) implies that $x$ is neither a repeat leaf nor a companion, and thus $\cycs{y}=\cycs{x}$ by \autoref{l:Q.Properties}(d).
Hence the claim is established.

\item[Case $k(x) = 3$, $\Delta_{x}$ is basic]
In this case, the loaded formula in $\Delta_{x}$ must be of the form
$\lnot \loaded{\Box}(a\ol{\gamma},\psi_{x})$ for some atomic program $a$.
Then by definition we have $\iitp_{x} = [a]\iitp_{y}$,
and $\Delta_{y} = {(\Delta_{x})}_{a}, \lnot \loaded{\Box}(\ol{\gamma},\psi_{x})$,
where $y$ is the unique child of $x$.

In addition, since $\kmodel, v \Vdash \lnot\loaded{\Box}(a\ol{\gamma},\psi_{x})$,
there must be some $w$ at minimum witness distance (a finite value by~\autoref{l:distanceProps}(a))
such that $\kmodel, w \Vdash \lnot\psi_{x}$ with an $a$-successor $v'$ of $v$
such that $\kmodel, v' \Vdash \lnot\loaded{\Box}(\ol{\gamma},\psi_{x})$
and $d_{\ol{\gamma}}(v',w) = d_{a\ol{\gamma}}(v,w) - 1$.
It is then easy to see that $\kmodel, v' \Vdash \Delta_{y},\iitp_{y}$.
By the primary induction hypothesis applied to $y$, we obtain a leaf $z \in
\cycs{y}$ and a state $u$ in $\kmodel$ such that
($1'$) $\kmodel, u \Vdash \Delta_{z}, \iitp_{z}$ and
($2'$) $\witdis_{\kmodel}(u,z) \leq \witdis_{\kmodel}(v',y)$.
Since $k(x)\neq 1$, \autoref{l:Q.Properties}(a) implies that $x$ is neither a repeat leaf nor a companion, so that $\cycs{y} = \cycs{x}$ by \autoref{l:Q.Properties}(d), and thus we find $z \in \cycs{x}$ and
$\witdis_{\kmodel}(u,z) \leq \witdis_{\kmodel}(v',y) \leq d_{\ol{\gamma}}(v',w)
< d_{a\ol{\gamma}}(v,w) = \witdis_{\kmodel}(v,x)$.

\item[Case $k(x) = 3$, $\Delta_{x}$ is not basic]
Let $y_{1},\ldots,y_{n}$ be the children of $x$, labeled with, respectively,
the sequents $\Pi_{1},\ldots,\Pi_{n}$.
Notice that in this case, it follows from the definitions that $\iitp_{x}
= \bigwedge_{i}\iitp_{y_{i}}$.

We will first show that there exists a child $y_i$ such that $\witdis_{\kmodel}(v, y_i) =  \witdis_{\kmodel}(v, x)$ and that $\kmodel, v \Vdash \Pi_i, \iitp_{y_{i}}$.
To show this, we distinguish two subcases based on whether at the elements of $R_x = C^R_{\Delta_x}$ a rule was applied to an unloaded formula or to the loaded formula.

\begin{description}
\item[\it The rule was applied to an unloaded formula]
    If at the elements of $R_x = C^R_{\Delta_x}$ a rule was applied to an unloaded formula,
    then the loaded formula remains unchanged,
    and thus $\ol\delta_x = \ol\delta_{y_i}$ for all $i \leq n$.
    (Recall that $\ol\delta_x$ is the program list in the loaded formula of $\Delta_x$.)
    Therefore, $d_{\ol{\delta_x}}(v, w') = d_{\ol{\delta_{y_i}}}(v, w')$ for all $w' \in W$ and for all $i \leq n$.
    From this, it follows that $\witdis_{\kmodel}(v, y_i) = \witdis_{\kmodel}(v, x)$ for all $i \leq n$.
    By the local invertibility of the local rules we have $\bigwedge \Delta_x \equiv \bigvee_{i} \bigwedge \Pi_{i}$.
    Hence, it follows from $\kmodel, v \Vdash \Delta_{x}, \iitp_{x}$ that $\kmodel, v \Vdash \Pi_i, \iitp_{y_{i}}$, for some $i \leq n$.
    We have thus shown that there exists a child $y_i$ of $x$ such that $\witdis_{\kmodel}(v, y_i) =  \witdis_{\kmodel}(v, x)$ and $\kmodel, v \Vdash \Pi_i, \iitp_{y_{i}}$.
\item[\it The rule was applied to the loaded formula]
    If at the elements of $R_x = C^R_{\Delta_x}$ a rule $R_2$ was applied to the loaded formula, then $R_2$ must have been the rule ${(\loaded{\diam})}_2$, since $\Delta_x$ is not basic.
    Let $\alpha_x$ and $\ol{\gamma_x}$ be such that $\ol{\delta_x} = \alpha_x \ol{\gamma_x}$, and let
    \[ \Dset{\alpha_x} = \{(F_1, \ol{\beta_{1}}), \dots, (F_k, \ol{\beta_k})\}.\]
    Since the branching at $x$ in $\Q$ corresponds to the application of ${(\loaded{\diam})}_2$, it follows that $k=n$, and
    we can assume that for all $i \leq n$, the sequent $\Pi_i$ is obtained from $\Delta_x$ by replacing the loaded formula $\lnot \loaded{\Box}(\alpha_x \ol{\gamma_x}, \psi_x)$ with $F_{i} \cup \{\lnot\loaded{\Box}(\ol{\beta_{i}} \ol{\gamma_x}, \psi_x)\}$.
    Since $\kmodel, v \Vdash \lnot \loaded{\Box}(\alpha_x\ol{\gamma_x}, \psi_x)$, we have by \autoref{l:distanceProps}(h) that there exists a pair $(F_{i}, \ol{\beta_{i}}) \in \Dset{\alpha_x}$ such that $\kmodel, v \Vdash \bigwedge F_{i} \land \lnot\loaded{\Box}(\ol{\beta_{i}} \ol{\gamma_x}, \psi_x)$, and that \[\min \big\{ d_{\alpha_x\ol{\gamma_x}}(v,w) \mid \kmodel, w \Vdash \lnot\psi_x \big\} = \min \big\{ d_{\ol{\beta_{i}}\ol{\gamma_x}}(v,w) \mid \kmodel, w \Vdash \lnot\psi_x \big\}.\]

    Since $\Pi_i = (\Delta_x \setminus (\lnot \loaded{\Box}(\alpha_x\ol{\gamma_x}, \psi_x))) \cup F_i \cup \{ \lnot \loaded{\Box}(\ol{\beta_i}\ol{\gamma_x}, \psi_x)\}$ and $\kmodel, v \Vdash \Delta_x$, it follows that $\kmodel, v \Vdash \Pi_i$.
    Moreover, since $\kmodel, v \Vdash \iota_x$ and $\iota_x = \bigwedge \iota_{y_i}$, we infer that $\kmodel, v \Vdash \Pi_i, \iota_{y_i}$.

    Also notice that:
    \begin{align*}
        \witdis_{\kmodel}(v, x)
        &= \min \big\{ d_{\alpha_x\ol{\gamma_x}}(v,w) \mid \kmodel, w \Vdash \lnot\psi_x \big\} \\
        &= \min \big\{ d_{\ol{\beta_{i}}\ol{\gamma_x}}(v,w) \mid \kmodel, w \Vdash \lnot\psi_x \big\} \\
        &= \witdis_{\kmodel}(v, y_i).
    \end{align*}
    The first equality is justified by the fact that $\kmodel, v \Vdash \Delta_x, \iota_x$ and the third by having $\kmodel, v \Vdash \Pi_i, \iota_{y_i}$.

    This concludes the proof of this subcase.
\end{description}

We have thus shown that there exists a child $y_i$ of $x$ such that $\witdis_{\kmodel}(v, y_i) =  \witdis_{\kmodel}(v, x)$ and that $\kmodel, v \Vdash \Pi_i, \iitp_{y_{i}}$.
Let us fix such a child $y_i$.

Since $\kmodel, v \Vdash \Pi_i, \iitp_{y_i}$,
we can apply the inductive hypothesis at $y_i$, to obtain that
there exists a leaf $z \in \cycs{y_i}$ and a state $u$ in $\kmodel$, such that
$\kmodel, u \Vdash \Delta_{z}, \iitp_{z}$ and
$\witdis_{\kmodel}(u,z) \leq \witdis_{\kmodel}(v,y_i)$,
where the inequality is strict if there is a basic node between $y_i$ and $z$.

Notice that since $k(x)\neq 1$, \autoref{l:Q.Properties}(a) implies that $x$ is neither a repeat leaf nor a companion, and thus $\cycs{y_i} \subseteq \cycs{x}$ by \autoref{l:Q.Properties}(d).
From this it follows that $z \in \cycs{x}$.
Since $\witdis_{\kmodel}(v, y_i) = \witdis_{\kmodel}(v,x)$, and since $x$ is not basic, we conclude that
$\witdis_{\kmodel}(u,z) \leq \witdis_{\kmodel}(v,x)$,
where the inequality is strict if there is a basic node between $x$ and $z$.
These observations and the fact that $\kmodel, u \Vdash \Delta_z, \iota_z$ are exactly what we wanted to show for this case.\qedhere
\end{description}
\end{proof}

\begin{lem}[\checklean{LoadedCluster.right_unsat_itp}]\label{l:correctRightIP}
$\Gamma_{2} \models \lnot\theta_{r}$.
\end{lem}
\begin{proof}
  We distinguish the two cases for how $\theta_r$ is defined in \autoref{d:itp}.
  If $\Gamma_1 = \emptyset$ then $\theta_r = \top$.
  By our assumption on page~\pageref{fixedTabAndCluster} there is a uniform, closed split tableau for $\Gamma_1;\Gamma2$.
  By \autoref{t:unsatIffSplitClosed} then $\Gamma_2$ is unsatisfiable, which entails $\Gamma_2 \models \neg\top$, so we are done.

  If $\Gamma_1 \neq \emptyset$ then $\theta_r = \iitp_{r_{\Q}}$.
  Note that for the root $r_{\Q}$ of $\Q$ the set $\cycs{r_{\Q}}$ is empty by \autoref{l:Q.Properties}(b).
It then follows by \autoref{l:satDownQ} that the sequent $\Delta_{r_{\Q}},
\iota_{r_{\Q}}$ cannot be satisfied.
From this, it follows that the formula $\lnot\theta_{r} =
\lnot\iota_{r_{\Q}}$ is a consequence of the sequent $\Gamma_{2} =
\Delta_{r_{\Q}}$.
\end{proof}

We are now ready to combine various lemmas to prove that $\theta_r$ is indeed an interpolant at the root of the proper cluster.

\begin{proof}[Proof of \autoref{l:clusterInterpolation}]
Let $\tab = (V, \edge, r, \Lambda, L)$ be a uniform closed tableau for the split sequent $\Gamma_1; \Gamma_2$,
and assume that the root $r$ belongs to a proper cluster $C$.
By assumption the split sequent $\Gamma_1; \Gamma_2$ is loaded.

Assume that $\Gamma_2$ is loaded.
By \autoref{l:itp-voc}, \autoref{l:correctLeftIP} and \autoref{l:correctRightIP}, the formula $\theta_r$ as given in \autoref{d:itp}
is an interpolant for the root $r$ of $C$.

Now assume that $\Gamma_1$ is loaded.
Intuitively, we now ``flip'' the tableau, exchanging the left and right parts of all split sequents.
Formally, let $\tab' = (V, \edge, r, \Lambda', L')$ be given by:
\begin{align*}
\Lambda'_1(s) &\isdef \Lambda_2(s) \\
\Lambda'_2(s) &\isdef \Lambda_1(s) \\
L'(s) &\isdef \twopartdef{R_1}{\text{if } L(s) = R_2}{R_2}{\text{if } L(s) = R_1}
\end{align*}
It is easy to see that $\tab'$ is a uniform closed PDL-tableau for the split sequent $\Gamma_2; \Gamma_1$.

Furthermore, recall that we have assumed that for every node $t \in C^{+}\setminus C$ we already have an interpolant $\theta_{t}$ for the split sequent $\Lambda_{1}(t); \Lambda_{2}(t)$.
Then we also have an interpolant for every node $t \in C^{+}\setminus C$, for the split sequent $\Lambda_{2}(t); \Lambda_{1}(t)$, by simply taking $\lnot \theta_{t}$.

Now we can proceed as in the case where $\Gamma_2$ was loaded, using $\tab'$ instead of $\tab$.
By \autoref{l:itp-voc}, \autoref{l:correctLeftIP} and \autoref{l:correctRightIP} we get that the interpolant $\theta_r$ given by \autoref{d:itp} is a correct interpolant for the flipped split sequent $\Gamma_2; \Gamma_1$.
Therefore $\lnot \theta_r$ is a correct interpolant for our original split sequent $\Gamma_1; \Gamma_2$, as wanted.
\end{proof}

We can now complete the main induction argument for showing that split tableaux yield interpolants.

\begin{proof}[Proof of \autoref{t:tabToInt}]
Let $\tab$ be a uniform, closed tableau for the split sequent $\Gamma_1; \Gamma_2$.
We want to show that $\Gamma_1; \Gamma_2$ has an interpolant.

We proceed by induction on the size (number of nodes) of $\tab$.

For the base case, if $\tab$ is a single node, then its label must be a closed sequent and by \autoref{l:itpLeaves} we are done.

As the inductive hypothesis, suppose that for any uniform, closed, split tableau $\tab'$ with fewer nodes than $\tab$ there exists an interpolant for the label of the root of $\tab'$.
We make a case distinction.

If the root of $\tab$ forms a singleton cluster, then we find an interpolant by \autoref{l:easyItp}.

On the other hand, if the root of $\tab$ does not form part of a singleton cluster, then it is part of a proper cluster $C$, and in particular, loaded.
Let $\tab'$ be a sub-tableau of $\tab$ whose root is a node in $C^+ \setminus C$.
Observe that $\tab'$ is a tableau for the split PDL-rules, and uniform.
We will show that $\tab'$ is a PDL-tableau and that it is closed.
Observe that if $v$ is a free repeat or a \lpr in $\tab$, then it is a leaf in $\tab$, and thus it is also a leaf in~$\tab'$.
As all repeats in $\tab'$ are repeats in $\tab$, it follows that $\tab'$ is a split PDL-tableau.

We show that $\tab'$ is closed.
Let  $s$ be a leaf in $\tab'$. We must show that $s$ is closed or a \lpr.
So assume $s$ is not closed. Then $s$ must be a \lpr in $\tab$, as otherwise $\tab$ would be open.
Therefore, there exists a companion $c(s)$ in $\tab$ such that the path from $c(s)$ to $s$ is loaded.
It follows that $s$ and $c(s)$ are in the same cluster and thus, that $c(s)$ is in $\tab'$, which entails that $s$ is a \lpr in $\tab'$.
From this it follows that $\tab'$ is closed.

By inductive hypothesis, we can obtain interpolants for all nodes $t \in C^+ \setminus C$, and by applying \autoref{l:clusterInterpolation}, we obtain the claim.
\end{proof}

We now get our main result.

\begin{proof}[Proof of \autoref{t:interpolation}]
Suppose $\models \phi \to \psi$.
Then $\phi \land \lnot \psi$ is unsatisfiable and by \autoref{t:unsatIffSplitClosed} there is a closed tableau for the split sequent $\{\phi\} ; \{\lnot \psi\}$.
By \autoref{t:tabToInt} there is an interpolant $\theta$ for this split sequent and by \autoref{fact:intToPartInt} this $\theta$ is also an interpolant for $\phi \to \psi$.
\end{proof}

\subsection{Constructing an interpolant for a particular example}\label{subsec:ExampleInterpolant}

To demonstrate how to construct an interpolant, consider the PDL validity $(p \land [a][a^\ast](p \lor [a^\ast]p)) \to [a][a^\ast]p$.
To find an interpolant (in the sense of \autoref{d:Interpolant}) for this valid implication we now construct an interpolant (in the sense of \autoref{d:PartInterpolant})
for the split sequent $\Gamma_1; \Gamma_2$
where $\Gamma_1 = \{p \+ [a][a^\ast](p \lor [a^\ast]p)\}$
and $\Gamma_2 = \{\lnot[a][a^\ast]p\}$.
Notice that since $\Gamma_1 \cup \Gamma_2$ is unsatisfiable, we can use the algorithm to find an interpolant for $\Gamma_1; \Gamma_2$.

The first step is to find a closed tableau for $\Gamma_1; \Gamma_2$.
As there exists only a finite number of tableaux for $\Gamma_1; \Gamma_2$, it is possible to algorithmically find a closed split PDL-tableau for $\Gamma_1; \Gamma_2$.
As an example, the tableau in \autoref{f:exk2} is a closed split PDL-tableau for $\Gamma_1; \Gamma_2$.

Notice that this tableau contains only one proper cluster, the highlighted cluster $C$.
Next we will construct the quasi-tableau $\Q$ for the cluster $C$.
From this construction we will find the interpolant for the root of the cluster $C$.

\begin{figure}[htb]
\begin{center}
  \begin{tikzpicture}[node distance=17mm, >=latex]
    \draw[black!30,dashed,fill=yellow!10]
    (-3.5, -1.3)             -- (8, -1.3)  --
                              (8, -10.5) --
                   (1.25, -10.5)            --
                (-0.5, -6.1)                  --
    (-3.5, -6.1)                             --
                cycle;

    \node (0)
    {$p \+ [a][a^\ast](p \lor [a^\ast]p) ~;~ \lnot[a][a^\ast]p$};
    \node [below of=0] (2)
    {$p \+ [a][a^\ast](p \lor [a^\ast]p) ~;~ \lnot\loaded{[a][a^\ast]}p$}; \arr{0}{$(L+)_2$}{2}
    \node [below of=2] (4)
    {$[a^\ast](p \lor [a^\ast]p) ~;~ \lnot\loaded{[a^\ast]}p$}; \arr{2}{$(M)_2$}{4}
    \node [below of=4] (6)
    {$p \lor [a^\ast]p \+ [a][a^\ast](p \lor [a^\ast]p) ~;~ \lnot\loaded{[a^\ast]}p$}; \arr{4}{$(\Box)_1$}{6}
    \node [below of=6, node distance=8mm] {\footnotesize $(\lnot \land)_1$};
    \node [below of=6] (9)
    {};
        \node [left of=9, node distance=8em] (8)
        {$\lnot \lnot [a^\ast]p \+ [a][a^\ast](p \lor [a^\ast]p) ~;~ \lnot\loaded{[a^\ast]}p$}; \arr{6}{}{8}
        \node [right of=9, node distance=8em] (10)
        {$\lnot \lnot p \+ [a][a^\ast](p \lor [a^\ast]p) ~;~ \lnot\loaded{[a^\ast]}p$}; \arr{6}{}{10}
        \node [below of=10] (12)
        {$p \+ [a][a^\ast](p \lor [a^\ast]p) ~;~ \lnot\loaded{[a^\ast]}p$}; \arr{10}{$(\lnot)_1$}{12}
        \node [below of=12, node distance=8mm] {\footnotesize $(\loaded{\diam})_2$};
        \node [below of=12] (15)
        {};
            \node [left of=15, node distance=12em] (14)
            {$p \+ [a][a^\ast](p \lor [a^\ast]p) ~;~ \lnot p$}; \arr{12}{}{14}
            \node [right of=15, node distance=4.5em] (16)
            {$p \+ [a][a^\ast](p \lor [a^\ast]p) ~;~ \lnot\loaded{[a][a^\ast]}p$}; \arr{12}{}{16}
    \draw (16.east) edge [dashed, gray, ->, thick, bend right=26, looseness=1.2] node [right] {$\comp$} (2.east);
  \end{tikzpicture}
\end{center}
\caption{A closed tableau for $\Gamma_1; \Gamma_2$ with the highlighted proper cluster $C$.}\label{f:exk2}
\end{figure}

The interpolants for the two nodes in $C^+ \setminus C$ can be obtained following the method described in \autoref{subsec:InterpolationForSingletonClusters} as these are singleton clusters.
In particular, notice that for the leaf with split sequent
$p \+ [a][a^\ast](p \lor [a^\ast]p) ~;~ \lnot p$
we obtain the interpolant $p$,
and for the leaf with split sequent
$\lnot \lnot [a^\ast]p \+ [a][a^\ast](p \lor [a^\ast]p) ~;~ \lnot\loaded{[a^\ast]}p$
we obtain the interpolant
$\lnot \lnot [a^\ast]p$.
From this, and the fact that there is no leaf with the sequent $\{\lnot [a][a^\ast]p\}$ in the loaded component, it follows that:
\begin{align*}
\theta_{\{\lnot p\}} &= p, \\
\theta_{\{\lnot [a^\ast]p\}} &= \lnot \lnot [a^\ast]p, \\
\theta_{\{\lnot [a][a^\ast]p\}} &= \bot.
\end{align*}

Having made these computations, we can now construct the quasi-tableau $\Q$ for the cluster $C$ with the corresponding pre-interpolants according to \autoref{d:iitp}.
The result is shown in \autoref{f:exk2Q}.
\newcommand{\preInt}[1]{(\iota = #1)}

\begin{figure}[htb]
\begin{center}
  \begin{tikzpicture}[node distance=4em, >=latex]
    \node (a1)
    {$x_1: \lnot\loaded{[a][a^\ast]}p$};
    \node [below of=a1] (a2)
    {$x_2: \lnot\loaded{[a][a^\ast]}p$};
    \node [below of=a2] (a3)
    {$x_3: \lnot\loaded{[a][a^\ast]}p$};
    \node [below of=a3] (b1)
    {$y_1: \lnot\loaded{[a^\ast]}p$};
    \node [below of=b1] (b2)
    {$y_2: \lnot\loaded{[a^\ast]}p$};
    \node [below of=b2] (b3)
    {$y_3: \lnot\loaded{[a^\ast]}p$};
    \node [below of=b3] (cd)
    {};
        \node [left of=cd] (c1)
        {$l_1: \lnot p$};
        %
        \node [right of=cd] (d1)
        {$r_1: \lnot\loaded{[a][a^\ast]}p$};

    \arr{a1}{}{a2}
    \arr{a2}{}{a3}
    \arr{a3}{}{b1}
    \arr{b1}{}{b2}
    \arr{b2}{}{b3}
        \arr{b3}{}{c1}
        \arr{b3}{}{d1}

    \node [right of=a1, node distance=3.9em, anchor=west] (ia1)
    {$\preInt{[(\lnot\bot?; a; (\lnot\lnot\lnot[a^\ast]p)?)^\ast][\lnot\bot?; a; (\lnot\lnot\lnot[a^\ast]p)?]p}$};
    \node [right of=a2, node distance=3.9em, anchor=west] (ia2)
    {$\preInt{[\lnot\bot?][a][(\lnot\lnot\lnot[a^\ast]p)?](p \land q_{x_1})}$};
    \node [right of=a3, node distance=3.9em, anchor=west] (ia3)
    {$\preInt{[a][(\lnot\lnot\lnot[a^\ast]p)?](p \land q_{x_1})}$};
    \node [right of=b1, node distance=3.9em, anchor=west] (ib1)
    {$\preInt{[(\lnot\lnot\lnot[a^\ast]p)?](p \land q_{x_1})}$};
    \node [right of=b2, node distance=3.9em, anchor=west] (ib2)
    {$\preInt{[(\lnot\lnot\lnot[a^\ast]p)?](p \land q_{x_1})}$};
    \node [right of=b3, node distance=3.9em, anchor=west] (ib3)
    {$\preInt{p \land q_{x_1}}$};
    %
        \node [left of=c1, node distance=3em, anchor=east] (ic1)
        {$\preInt{p}$};
        %
        \node [right of=d1, node distance=3.9em, anchor=west] (id1)
        {$\preInt{q_{x_1}}$};
  \end{tikzpicture}
\end{center}
\caption{The quasi-tableau $\Q$ for the cluster $C$ annotated with pre-interpolants.}\label{f:exk2Q}
\end{figure}

Observe that the obtained pre-interpolant \[\iota_{x_1} = [{(\lnot\bot?; a; (\lnot\lnot\lnot[a^\ast]p)?)}^\ast][\lnot\bot?; a; (\lnot\lnot\lnot[a^\ast]p)?]p\] for the root $x_1$ of $\Q$ is indeed an interpolant for $p \+ [a][a^\ast](p \lor [a^\ast]p) ~;~ \lnot\loaded{[a][a^\ast]}p$.
To make this observation a little easier, note that $\iota_{x_1} \equiv [{(a; \lnot[a^\ast]p?)}^\ast][a; \lnot[a^\ast]p?]p$.
Finally, observe that the interpolant assigned to the root of the tableau is also $\iota_{x_1}$, and that this is indeed an interpolant for $\Gamma_1; \Gamma_2$.
This concludes the example.

\section{Conclusion}\label{sec:Conclusion}

We have shown that Propositional Dynamic Logic (PDL) has Craig Interpolation.
Our~proof is based on the key ideas of Borzechowski~\cite{MB1988}: we defined a tableau system for PDL, showed its soundness and completeness and then used an adaptation of Maehara's method to construct interpolants.
We end with some remarks about our ongoing project to implement the system and formally verify the proof, as well as remaining open questions.

\paragraph{Implementation.}
We implemented the tableau system in Haskell.\footnote{
  We permanently archived the code at
  \url{https://doi.org/10.5281/zenodo.22986497} (our proof system) and
  \url{https://doi.org/10.5281/zenodo.22986484} (the~\cite{MB1988} system).
  For the whole repository,
  please refer to \url{https://github.com/m4lvin/modal-tableau-interpolation}.}
The implementation comes in two variants:
one version stays close to the system developed by Borzechowski~\cite{MB1988} with $n$-formulas, the other version uses our proof system with $\unfold$ functions.
Both versions also compute interpolants.
We encourage the reader to try out the web interface available at \url{https://tools.malv.in/tapdleau}.

\paragraph{Verification.}
We formally verified the proof in the interactive theorem prover Lean~\cite{MU2021:Lean4} and marked verified claims with ``\checkmark'' throughout the article.\footnote{
  The formalization is permanently archived at
  \url{https://doi.org/10.5281/zenodo.23432787}
  and the full repository is \url{https://github.com/m4lvin/lean4-pdl}.
  }
The formalization project started in December 2023 and became sorry-free on 2 September 2026.
It covers all main results, including Soundness (\autoref{t:soundness}), Completeness (\autoref{t:completeness}) and Interpolation (\autoref{t:interpolation}).
Several more authors have contributed to our formalization, and parts of the proofs were written using the Aristotle AI tool~\cite{achim2025aristotle}.
For further details we refer to the documentation comments in the Lean code and the short article about a previous project~\cite{Gat2022:AiML:BMLean} with a proof of interpolation for the basic modal logic $\mathsf{K}$ in Lean 3.

\paragraph{Open questions.}
The somewhat unclear situation in the literature has arguably also blocked research building on top of the result.
Now that we have Craig interpolation for PDL, there are interesting follow-up questions that can be pursued:
\begin{itemize}
\item
As we also discuss in \autoref{apdx:history}, it was shown by Kracht~\cite{Kracht1999tools} that in order to show that PDL has interpolation it suffices to establish the property for its test-free fragment.
The other direction has \emph{not} been shown, and we do not know whether test-free PDL has interpolation.
Note that full PDL is more expressive than its test-free version, as was shown by Berman and Paterson~\cite[Main Theorem]{BerPat1981:testFreeWeaker}.

It is of course tempting, but not immediately clear how to adapt our method to test-free PDL\@.
Note that we use tests in the construction of interpolants for proper clusters --- in the case $k(x)=2$ of \autoref{d:iitp}, as discussed in \autoref{r:dfitp}.
These tests are often trivial, i.e., of the form ${\lnot\bot?}$, but not always.
For example, in \autoref{subsec:ExampleInterpolant} we started with a test-free split sequent but obtained an interpolant containing the non-trivial test ${(\lnot\lnot\lnot[a^\ast]p)?}$.
That is, even when given test-free validities our method will generally yield interpolants with tests.
Of course, this does not show that test-free PDL does not have interpolation and for the concrete example in \autoref{subsec:ExampleInterpolant} there exists a test-free interpolant, namely $[a][a^\ast]p$.
\item
In a similar direction, we leave it open whether the method would work for \emph{other fragments and variants} of PDL obtained by removing or adding additional operators such as intersection and converse.
\item
There is a well-known correspondence between interpolation for logics and amalgamation of structures and algebras.
  For an example, see the work of Hodges~\cite[Section~5.5]{Hodges1997:ShorterMT} where interpolation for first-order logic is shown in this way.
This correspondence was also the basis for the algebraic proof attempt by Kowalski in 2002~\cite{Kowalski2002}.
PDL contains \emph{Kleene Algebra with tests} (KAT) in the sense that a valid KAT-equation such as $\alpha = \beta$ corresponds to the PDL-validity $[\alpha]p \leftrightarrow [\beta]q$ where $p$ is fresh~\cite[Section~9]{Kozen1997:ComKAT}.
We leave open the question what the implications of the Craig Interpolation Property for PDL are for Kleene algebras with tests, and (possibly) other classes of algebras.
\end{itemize}

\paragraph{Acknowledgements.}
We thank additional former members of the PDL reading group that we started in 2022/23 in Amsterdam and Groningen: Johannes Marti, Jan Rooduijn and Guillermo Menéndez Turata.
We wholeheartedly thank the journal referees for taking the time to review our work, and we want to especially thank Bartosz Jan Bednarczyk for carefully providing many suggestions that improved our article substantially.

\clearpage

\appendix

\section{Notational Conventions}\label{apdx:notation}


\begin{tabular}{ll}
  $p$, $q$, $r$ & proposition letters \\
  $\phi$, $\psi$, $\rho$ & PDL formula \\
  $\tau$ & PDL formula used as a test \\
  $\theta$, $\iota$ & PDL formula used as interpolant \\
  $\chi$ & loaded formula, e.g., $[\loaded{a \cup b}]p$ \\
  $\xi$ & PDL formula or loaded formula \\
  $a, b, c$ & atomic program \\
  $\alpha$, $\beta$, $\gamma$ & program \\
  \\
  $X$, $Y$ & set of formulas \\
  $\ol{\delta}$, $\ol{\lambda}$, $\ol{\eta}$ & list of programs, e.g., [$a^\ast$, $b \cup c$, $a^\ast$] \\
  \\
  $\kmodel$ & Kripke model \\
  $W$ & set of states \\
  $u$, $v$, $w$ & state \\
  \\
  $\Dset{\alpha}$ & set of $(X,\ol{\delta})$ pairs used for $\unfold_{\diam}(\alpha,\psi)$, see Definitions~\ref{d:Dset}~and~\ref{d:unfoldDiamond} \\
  $\ell$ & test profile, see \autoref{d:TP} \\
  $\sigma^\ell$ & signature formula, see \autoref{d:TP} \\
  \\
  $\mathcal{T}$ & tableau \\
  $s$, $t$ & node in a tableau \\
  $\Q$ & quasi-tableau \\
  $x$, $y$, $z$ & node in a quasi-tableau \\
  $\Delta$, $\Gamma$, $\Upsilon$ & sequents (finite set of formulas) and set of formulas \\
  $\Lambda$ & function labelling tableau nodes with sequents, e.g., $\Lambda(s)=\Delta$ \\
  \\
  $s \edge t$ & $s$ is a \emph{parent} of $t$ \\
  $s \edgeT t$ & transitive closure of $\edge$, $s$ is an \emph{ancestor} of $t$ \\
  $s \edgeRT t$ & reflexive transitive closure of $\edge$ \\
  \\
  $t \comp s$ & $t$ has the companion $s$ (implies that $t$ is a \lpr) \\
  \\
  $s \cEdge t$ & $\edge$ plus repeat-companion pairs \\
  $s \cEdgeT t$ & transitive closure of $\cEdge$ \\
  $s \cEdgeRT t$ & reflexive transitive closure of $\cEdge$ \\
  $s \cEquiv t$ & equivalence relation induced by $\cEdgeRT$ and its converse \\
  $\simpler{t}{s}$ & partial (conversely) well-founded order induced by $\cEdge$ \\
  \\
  $\mathbb{T}$ & strategy tree \\
  $\tedge$ & successor in a strategy tree \\
\end{tabular}

\clearpage

\section{Rule Overview}\label{apdx:allRules}

Below are the local rules from \autoref{d:localTableau} and the modal rule from \autoref{d:Modalrule}.

\bigskip

\begin{center}
  \AxiomC{$\Delta , \lnot \lnot \phi$}
  \LeftLabel{$(\lnot)$}
  \UnaryInfC{$\Delta,\phi$}
  \DisplayProof
  \hspace{1.5em}
  \AxiomC{$\Delta , \phi \land \psi$}
  \LeftLabel{($\land$)}
  \UnaryInfC{$\Delta,\phi,\psi$}
  \DisplayProof
  \hspace{1.5em}
  \AxiomC{$\Delta , \lnot (\phi \land \psi)$}
  \LeftLabel{$(\lnot \land)$}
  \UnaryInfC{$\Delta , \lnot \phi \splitCase \Delta , \lnot \psi$}
  \DisplayProof

  \smallskip

  \AxiomC{$\Delta , [\alpha]\phi$}
  \LeftLabel{$(\Box)$}
  \RightLabel{\ \ $\alpha$ non-atomic}
  \UnaryInfC{$\big\{\Delta, \Gamma
      \mid  \Gamma \in \unfold_{\square}(\alpha, \phi) \big\}$}
  \DisplayProof

  \smallskip

  \AxiomC{$\Delta , \lnot[\alpha]\phi$}
  \LeftLabel{$(\diam)$}
  \RightLabel{\ \ $\alpha$ non-atomic}
  \UnaryInfC{$\{\Delta, \Gamma \mid \Gamma \in \unfold_{\diam}(\alpha,\phi)\}$}
  \DisplayProof

  \smallskip

  \AxiomC{$\Delta , \lnot\loaded{[\alpha]}\xi$}
  \LeftLabel{$(\loaded{\diam})$}
  \RightLabel{\ \ $\alpha$ non-atomic}
  \UnaryInfC{$\big\{\Delta, \Gamma \mid \Gamma \in
      \loaded{\unfold}_{\diam}(\alpha, \xi)\big\}$}
  \DisplayProof

  \smallskip

  \AxiomC{$\Delta , \lnot[a][\alpha_1]\ldots[\alpha_n]\phi$}
  \LeftLabel{$(L+)$}
  \RightLabel{\ \ $\Delta$ free and basic and $n \geq 0$ maximal}
  \UnaryInfC{$\Delta , \lnot\loaded{[a][\alpha_1]\ldots[\alpha_n]}\phi$}
  \DisplayProof

  \smallskip

  \AxiomC{$\Delta , \lnot\loaded{[a][\alpha_1]\ldots[\alpha_n]}\phi$}
  \LeftLabel{$(L-)$}
  \RightLabel{\ \ $\Delta$ basic}
  \UnaryInfC{$\Delta , \lnot[\alpha_1]\ldots[\alpha_n]\phi$}
  \DisplayProof

  \smallskip

  \AxiomC{$\Delta , \lnot\loaded{[a]}\xi$}
  \LeftLabel{$(M)$}
  \RightLabel{\ \ $\Delta$ basic}
  \UnaryInfC{$\Delta_a , \lnot \xi$}
  \DisplayProof

  \smallskip
\end{center}

where $\phi, \psi \in \pdlforms$, $\alpha_1, \dots, \alpha_n \in \progs$, $a$ atomic, and $\xi \in \pdlforms \cup \loaded{\pdlforms}$, i.e., $\xi$ is a possibly loaded formula.

Furthermore, if $t$ is a child of a node $s$ and $L(s) = (L+)$ then $L(t) \neq (L-)$.
That is, we are not allowed to apply $(L-)$ immediately after $(L+)$.

\section{History of the Problem}\label{apdx:history}

We are aware of three prior attempts in the literature to show that Propositional Dynamic Logic (PDL) has the Craig Interpolation Property.
Here we summarize the history of the problem and how our proof relates to these attempts.

Two years after the introduction of PDL by Fischer and Ladner~\cite{FL:PDL1979}, in 1981 Leivant defined a finitary sequent calculus for PDL to which he then applied Maehara's method in order to obtain interpolation~\cite{Leivant1981}.
However, that proof of interpolation is actually for CPDL, a constructive variant of PDL defined by Leivant, and not standard classical PDL\@.
To transfer interpolation from CPDL to PDL, Leivant claims that inserting double negations in front of every connective suffices to relate the two logics~\cite[Proposition~3.2.2]{Leivant1981}.
However, no proof details are given, and in particular it is not clear whether programs and test formulas also need to be modified in this translation.
It remains an open question whether and how this can actually be done.

Several years later, in 1988 Borzechowski (a co-author of the present work) wrote his Diplomarbeit (comparable to a master's thesis) on PDL and interpolation under the supervision of Rautenberg at FU Berlin~\cite{MB1988}.
The thesis was written in German and never published --- our suspicion is that not many people in the research community actually read it.
A translation to English was made available in 2020 by Gattinger~\cite{BorGat2020:AiML-talk}, another co-author of the present work.
Like with Leivant, the main technique used by Borzechowski is Maehara's method, this time applied to a tableau system for PDL itself instead of a sequent calculus for a constructive variant.
Borzechowski's proof~\cite{MB1988} provided the starting point for the present work.
We further discuss the changes and additions we made in the next subsection.

We continue the history of the problem another decade later.
In his 1999 book, Kracht also discusses ``The Unanswered Question''~\cite[Section~10.6]{Kracht1999tools}.
He mentions the proof attempts by Leivant and Borzechowski, and says that ``in neither case was it possible to verify the argument''~\cite[p.~493]{Kracht1999tools}.
Kracht does not further discuss the proof by Borzechowski, but says that Leivant ``makes use of the fact that if $\phi \vdash_{\textbf{PDL}} \psi$ then we can bound the size of a possible countermodel so that the iterat $\alpha^\ast$ only needs to search up to a depth $d$ which depends on $\phi$ and $\psi$''~\cite[p.~493]{Kracht1999tools}.
Indeed, this is what makes the finitary $\ast$-rule used in~\cite{Leivant1981} admissible.
Kracht then suggests that the argument by Leivant would require ``that interpolation is preserved under intersection''~\cite[p.~493]{Kracht1999tools}, and Kracht rightly notes that this is not the case in general.
However, Leivant's argument actually does not use an intersection of logics and does not need preservation of interpolation under intersection.
For more details, see~\cite{Gat2014:notesAboutLeivant}.
Overall, the criticism by Kracht does not apply to the proof by Leivant, but as mentioned above there are other gaps in it.

It is worth mentioning that Kracht shows that if test-free PDL has interpolation, then PDL has interpolation~\cite[Theorem 10.6.2]{Kracht1999tools}.
We do not use this result here but work with full PDL including tests.
To our knowledge it remains open whether test-free PDL has interpolation.

A third proof attempt was made in 2002, by Kowalski~\cite{Kowalski2002}.
This attempt was fundamentally different from the other two, as it used algebraic instead of proof-theoretic methods and claimed to show the superamalgamation property for free dynamic algebras.
The article was withdrawn in 2004~\cite{KowalskiRetraction2004}, after a mistake was found in the proof that could not be repaired.
Since then the question was considered open again, and for example highlighted by Van~Benthem in~\cite{vB2008:ManyFaces}.

\section{Differences from Borzechowski (1988)}\label{apdx:differences}

The proof we have presented here still uses the three key ideas from Borzechowski's \emph{Diplomarbeit}~\cite{MB1988}, but we made several modifications to the tableau system and the ways in which we prove its properties.
We summarize the four main modifications below.

\subsection*{Avoid n-formulas by using \texorpdfstring{$(\diam)$}{diamond} and \texorpdfstring{$(\Box)$}{box} rules}

Recall that our unfolding mechanism (see \autoref{d:unfoldDiamond} and \autoref{d:unfoldBox}) was specifically made to prevent \emph{local repeats}, i.e., to never generate the same formula within a local tableau.
Local repeats arise due to the interaction of iteration $^\ast$ and the test operator ${?}$, or due to the presence of nested iterations.
For example, formulas like $[{(a^\ast)}^\ast]p$, $\lnot[{(q?)}^\ast]p$, and $\lnot[{(a^\ast)}^\ast]p$ all produce local repeats.

The mechanism that deals with local repeats in Borzechowski's system~\cite{MB1988} is called $n$-formulas.
Next, we present this mechanism and compare it to the rules $(\diam)$ and $(\Box)$ that replace it in our approach.

An $n$-formula is defined as a string which is obtained by replacing in a PDL formula one or more programs of the form $\alpha^\ast$ by $\alpha^{(n)}$.
This replacement indicates that $\alpha^\ast$ has been unfolded, hence if a formula $[\alpha^{(n)}]\phi$ is regenerated in a node label, then it is a local repeat.
Semantically, the interpretation of an $n$-formula is the same as the original PDL formula from which the $n$-formula was obtained.
Thus, a formula $[\alpha^{(n)}]\phi$ is semantically equivalent to the PDL formula $[\alpha^\ast]\phi$.
It is worth noting that $(n)$ is simply a symbol and $n$ does not stand for a number.

Tableau nodes in this system are labeled with sets consisting of $n$-formulas and/or PDL formulas.
Furthermore, instead of having the rules $(\Box)$ and $(\diam)$ that deal with every possible program operator, Borzechowski's system has the following rules that follow the well-known Segerberg axioms:

\begin{figure}[H]
\begin{center}
    \AxiomC{$\Gamma , \lnot[\alpha \cup \beta]\phi$}
    \LeftLabel{$(\lnot\cup)$}
    \UnaryInfC{$\Gamma , \lnot[\alpha]\phi \splitCase \Gamma , \lnot [\beta]\phi$}
    \DisplayProof%
    \hspace{1em}
    \AxiomC{$\Gamma , \lnot[\tau?]\phi$}
    \LeftLabel{$(\lnot?)$}
    \UnaryInfC{$\Gamma , \tau , \lnot \phi$}
    \DisplayProof%
    \hspace{1em}
    \AxiomC{$\Gamma , \lnot [\alpha;\beta]\phi$}
    \LeftLabel{$(\lnot;)$}
    \UnaryInfC{$\Gamma , \lnot[\alpha][\beta]\phi$}
    \DisplayProof%

    \smallskip

    \AxiomC{$\Gamma , [\alpha \cup \beta]\phi$}
    \LeftLabel{$(\cup)$}
    \UnaryInfC{$\Gamma , [\alpha]\phi, [\beta]\phi$}
    \DisplayProof%
    \hspace{1.5em}
    \AxiomC{$\Gamma , [\tau?]\phi$}
    \LeftLabel{$(?)$}
    \UnaryInfC{$\Gamma , \lnot \tau \splitCase \Gamma , \phi$}
    \DisplayProof%
    \hspace{1.5em}
    \AxiomC{$\Gamma , [\alpha;\beta]\phi$}
    \LeftLabel{$(;)$}
    \UnaryInfC{$\Gamma , [\alpha][\beta]\phi$}
    \DisplayProof%

    \smallskip

    \AxiomC{$\Gamma , \lnot[\alpha^\ast]\phi$}
    \LeftLabel{$(\lnot n)$}
    \UnaryInfC{$\Gamma , \lnot \phi \splitCase \Gamma , \lnot[\alpha][\alpha^{(n)}]\phi$}
    \DisplayProof%
    \hspace{1.5em}
    \AxiomC{$\Gamma , [\alpha^\ast]\phi$}
    \LeftLabel{$(n)$}
    \UnaryInfC{$\Gamma , \phi , [\alpha][\alpha^{(n)}]\phi$}
    \DisplayProof%
\end{center}
\caption{Rules for boxes and negated boxes in Borzechowski's system~\cite{MB1988}.}\label{fig:MBRules}
\end{figure}

In addition to these rules, a Local Tableau must also satisfy these conditions:
\begin{enumerate}[(1)]
  \item A node label $\Gamma, \lnot[a]\phi$ or $\Gamma,[a]\phi$ where $\phi$ is an $n$-formula is replaced with $\Gamma , \lnot[a]\psi$ or $\Gamma , [a]\psi$, respectively, where $\psi$ is obtained by replacing each $(n)$ with $\ast$ in $\phi$.
  \item A node label $\Gamma , [\alpha^{(n)}]\phi$ is replaced with $\Gamma$.
  \item If it is possible, then a rule must be applied to an $n$-formula.
  \item No rule may be applied if the node label contains a formula of the form $\lnot [\alpha^{(n)}]\phi$.
\end{enumerate}

Finally, it is important to mention that the definition of a closed tableau~\cite[Definition~16]{MB1988} entails that leaves that contain a formula of the form $\lnot[\alpha^{(n)}]\phi$ are considered closed.

The $n$-formulas, the rules in \autoref{fig:MBRules}, the conditions (1)--(4), and the treatment of leaves that contain a formula of the form $\lnot[\alpha^{(n)}]\phi$ as closed, constitute the four elements that handle the local repeats in Borzechowski's system.

To see how this system handles particular examples of local repeats, consider the examples in \autoref{fig:MBExamples}.
These examples show how the $n$-formulas and conditions (1)--(4) deal with local repeats by deleting the formulas of the form $[\alpha^{(n)}]\phi$ or stopping the system from applying any rules to nodes that contain a formula of the form $\lnot[\alpha^{(n)}]\phi$.

In the left tableau, the crossed-out formula $\cancel{[{(a^\ast)}^{(n)}]p}$ is there to indicate that condition (2) has been applied to obtain that node, and thus the node only contains the formulas $p$ and $[a^\ast][{(a^\ast)}^{\ast}]p$.
Notice also that in the only leaf, condition (1) has been applied and thus the $n$-formula $[a][a^{(n)}][{(a^\ast)}^{(n)}]p$ has been replaced with its equivalent PDL formula $[a][a^{\ast}][{(a^\ast)}^{\ast}]p$.

In the right tableau, the cross below the right leaf indicates that condition (4) has been applied, and thus that no more rules can be applied to that node, and that this node is considered closed.

\begin{figure}[H]
\noindent
\centering
\begin{minipage}{0.4\textwidth}
  \begin{tikzpicture}[node distance=4em, >=latex]
    \node (0) {$[(a^\ast)^\ast]p$};
    \node [below of=0] (1) {$p, [a^\ast][(a^\ast)^{(n)}]p$}; \arr{0}{$(n)$}{1}
    \node [below of=1] (2) {$p, \cancel{[(a^\ast)^{(n)}]p}, [a][a^\ast][(a^\ast)^\ast]p$}; \arr{1}{$(n)$}{2}
  \end{tikzpicture}
\end{minipage}
\begin{minipage}{0.55\textwidth}
  \begin{tikzpicture}[node distance=4em, >=latex]
    \node (0) {$\lnot[(q?)^\ast]p$};
    \node [below of=0, node distance=2em] (1L) {$(\lnot n)$};
    \node [below of=0] (1) {};
    \node [left of=1, node distance=7em] (1a) {$\lnot p$}; \arr{0}{}{1a}
    \node [right of=1, node distance=7em] (1b) {$\lnot[q?][(q?)^{(n)}]p$};  \arr{0}{}{1b}
    \node [below of=1b] (2) {$q \+ \lnot[(q?)^{(n)}]p$}; \arr{1b}{$(\lnot ?)$}{2}
    \node [below of=2, node distance=2em] (3) {$\times$};
  \end{tikzpicture}
\end{minipage}
\caption{Local tableaux for $[{(a^*)}^*]p$ and $\lnot[{(q?)}^\ast]p$ in the tableau system with n-formulas.}\label{fig:MBExamples}
\end{figure}

Although the $n$-formulas suffice to deal with local repeats, Borzechowski also had to deal with $n$-formulas as a recurring special case in many lemmas and theorems.
Furthermore, the presence of $n$-formulas also added an additional step in the construction of the interpolant.

In contrast, our tableau system does not use $n$-formulas, but instead encapsulates all the reasoning about the $n$-formulas from~\cite{MB1988} including conditions (1)--(4) into the definition of the $\unfold$-functions that are used to define the rules
$(\Box)$ and $(\diam)$.
For example, compare the Local Tableaux of our system in \autoref{fig:MGExamples} to those of \autoref{fig:MBExamples}.

We explain the relationship between our $\unfold$-functions and the local tableaux of~\cite{MB1988} in more detail.
For clarity, we refer to these as local MB tableaux.
For boxes, consider the maximal local MB tableau for $[\alpha]\psi$ where $\psi$ and all $\tau \in \Test(\alpha)$ are treated as atomic.
Branching in this tableau arises from applications of $(?)$.
Each $\Gamma \in \unfold_{\Box}(\alpha,\psi)$ corresponds to some leaf of this local tableau.
Conversely, a leaf of this local tableau may be associated with multiple $\Gamma$. This can be seen by considering the maximal local MB tableau for $[p?;q?]r$. The leaf labeled with $\neg p$ corresponds to both sequents
$\{\neg p, \neg q\}$ and $\{\neg p\}$ in $\unfold_{\Box}(p? \seqc q?,r)$. Informally, if the test ${p?}$ fails then the outcome of ${q?}$ does not matter.

For diamonds, consider the maximal local MB tableau for $\neg[\alpha]\psi$ where $\psi$ and all $\tau \in \Test(\alpha)$ are treated as atomic.
Here, branching arises from applications of $(\neg \cup)$ and $(\neg n)$.
The sequents in $\unfold_{\diam}(\alpha,\psi)$ are the labels of the leaves in this local MB tableau that do not contain an $n$-formula $\lnot[\beta^{(n)}]\phi$.

Through the use of the $\unfold$ functions in rules $(\diam)$ and $(\Box)$, we have removed the need for $n$-formulas and the intermediate nodes that contain $n$-formulas.

\begin{figure}[H]
  \centering
  \begin{tikzpicture}[node distance=4em, >=latex]
    \node (1) {$[(a^\ast)^\ast]p$};
    \node [below of=1] (2) {$p, [a][a^\ast][(a^\ast)^\ast]p$}; \arr{1}{$(\Box)$}{2}
    \node [right of=1, node distance=14em] (1') {$\lnot[(q?)^\ast]p$};
    \node [below of=1'] (2') {$\lnot p$}; \arr{1'}{$(\diam)$}{2'}
  \end{tikzpicture}
  \caption{Local tableaux for $[{(a^*)}^*]p$ and $\lnot[{(q?)}^\ast]p$ using our system.}\label{fig:MGExamples}
\end{figure}

\subsection*{Completeness via games}

Both the completeness proof given by Borzechowski~\cite{MB1988} and the one presented here use the same standard approach:
given a consistent set $\Gamma$ of formulas, we know that there is no closed tableau for $\Gamma$.
We can then use the existence of a maximal open tableau to build a model graph, i.e., a Kripke model where states correspond to sets of formulas that are true at them.
In the work of Borzechowski~\cite{MB1988} this model is defined via a set $M_0$ of maximal tableaux that contains first a maximal local tableau for $\Gamma$ and then gets closed such that it ``also contains every maximal saturating tableau''~\cite[p.~36]{MB1988} for any sequent occurring in a tableau already in the set.
It is not clear whether $M_0$ is finite, nor whether it needs to be.
Then maximal local paths are used to define the states of the model graph.
The main work of the completeness proof is to show condition (d) of \autoref{d:ModelGraph}.
The argument for this is not entirely clear. In particular, the proof needs to ``jump'' between different tableaux in $M_0$ but the termination of this process is not shown.

Instead of $M_0$, in our completeness proof here we obtain the model graph via the \emph{tableau game} from \autoref{subsec:TableauGame}.
The role of maximal local paths is played by pre-states and the main work in our completeness proof is to show that a winning strategy for the player ``Builder'' suffices to define a model graph.
This use of games is by now a standard tool in the theory of (fixpoint) logics (see e.g.~\cite{ArnoldNiwinski:rudimentsGames,GraedelTW2001:LNCS2500,SEP:logic-games}), and they provide an intuitive, yet formal, setting for reasoning about proof search.

\subsection*{Interpolant correctness via semantics}

In Borzechowski's Diplomarbeit~\cite{MB1988}, the proof that the constructed interpolants are correct first extends the proof system with additional rules, argues for their admissibility and then builds tableaux in that new system to prove the correctness of the interpolants.
It should be noted this work does not include a complete proof of the admissibility of the additional rules, nor does it provide a convincing argument showing that the admissibility of these rules implies the completeness of the split system~\cite[Lemma~13]{MB1988}.
Furthermore, some definitions lack precision.
For instance, the definition of the structure obtained from a tableau by removing $n$-nodes~\cite[Definition~26]{MB1988}, and the definition of what we now call quasi-tableaux~\cite[Definition~31]{MB1988}, are both imprecise: the defined structures do not satisfy the conditions required to be considered tableaux, even though they are defined as such.
In addition, the proof of the correctness of the interpolant contains gaps.
The most significant one is in Lemma~25~\cite[p.~40]{MB1988} where some cases are handled incorrectly.

In this document, we establish the completeness of the split system directly instead of relying on the admissibility of additional rules. We also prove the correctness of interpolants directly via a semantic argument in \autoref{l:correctLeftIP} and \autoref{l:correctRightIP}, rather than by using a proof-theoretic approach.

\subsection*{Simplified loading}

In Borzechowski's system~\cite{MB1988}, the notation for loaded formulas looks like $[a][b]\phi^\phi$, but we write this as $\loaded{[a][b]}\phi$.
The purpose of this mechanism is the same and we still call it ``loading''.
However, we make two changes to the loading rule $(L+)$.
First, our $(L+)$ can only be applied to formulas that have an atomic program in their leading modality.
Second, our $(L+)$ always loads maximally,
meaning that, e.g., the formula $[a][b^*][c \cup d](p \lor [a^*]q)$ can only be loaded as
$\loaded{[a][b^*][c \cup d]}(p \lor [a^*]q)$.
These restrictions to $(L+)$ are not necessary, but they simplify the tableau system overall.

\bibliographystyle{alphaurl}
\bibliography{references}

\end{document}